\documentclass[11pt]{article}
\newcounter{ourcount}
\usepackage{comment}
\usepackage{amsmath,amsfonts,amssymb,graphicx,epsfig,pdflscape,multirow,gensymb,soul}
\usepackage{ntheorem}
\usepackage[matrix,arrow,curve]{xy}
\numberwithin{equation}{section}
\graphicspath{{figpdf/}}
\usepackage[nosort]{cite}
\usepackage[usenames]{color}
\usepackage{pstricks}
\usepackage[backref=false]{hyperref}
\usepackage[capitalise,noabbrev]{cleveref}

\hypersetup{
colorlinks=true,
citecolor=red,
linkcolor=darkblue,
urlcolor=darkblue
}
\definecolor{darkblue}{rgb}{0,0,.8}
\definecolor{red}{rgb}{1,0,0}

\theorembodyfont{\itshape} 
\theoremheaderfont{\scshape}
\theoremstyle{plain}
\newtheorem{Proposition}{Proposition}[section]
\newtheorem{Lemma}[Proposition]{Lemma}
\newtheorem{Theorem}{Theorem}
\newtheorem{Corollary}[Proposition]{Corollary}

\numberwithin{equation}{section}

\newcommand{\nc}{\newcommand}

\crefname{Proposition}{Proposition}{Propositions}
\crefname{Lemma}{Lemma}{Lemmas}

\nc{\ir}{\mathrm{i}}
\nc{\dd}{\mathrm{d}} 
\nc{\eE}{\mathsf{e}}
\renewcommand{\ge}{\geqslant}
\renewcommand{\le}{\leqslant}
\renewcommand{\geq}{\geqslant}
\renewcommand{\leq}{\leqslant}
\renewcommand{\mod}{\textrm{ mod }}
\nc{\mmod}{\,\mathrm{mod}\,}
\nc{\ba}{{\boldsymbol a}}
\nc{\pa}{\mathsf{a}}
\nc{\pb}{\mathsf{b}}
\nc{\pc}{\mathsf{c}}
\nc{\pd}{\mathsf{d}}
\nc{\ppi}{\mathsf{i}}
\nc{\ppj}{\mathsf{j}}
\nc{\g}{\mathfrak{g}}
\nc{\qq}{\mathfrak{q}}
\nc{\qqb}{\bar\mathfrak{q}}

\nc{\Cbb}{\mathbb{C}}
\nc{\Nbb}{\mathbb{N}}
\nc{\Rbb}{\mathbb{R}}
\nc{\Zbb}{\mathbb{Z}}
\nc{\wh}{\widehat}
\nc{\wt}{\widetilde}

\nc{\proof}{{\scshape Proof.\ }} 				
\nc{\eproof}{{\hfill \rule{0.5em}{0.5em}\medskip}}
\nc{\be}{\begin{equation}}
\nc{\ee}{\end{equation}}

\newcommand{\nn}{\nonumber}
\nc{\repB}{\mathsf{B}}
\nc{\repBa}{\mathsf{B}\mathrm{a}}
\nc{\repBd}{\mathsf{B}\mathrm{d}}
\nc{\repE}{\mathsf{E}}
\nc{\repI}{\mathsf{I}}
\nc{\repK}{\mathcal{K}}
\nc{\repL}{\mathsf{L}}
\nc{\repM}{\mathsf{M}}
\nc{\repN}{\mathsf{N}}
\nc{\repQ}{\mathsf{Q}}
\nc{\repR}{\mathsf{R}}
\nc{\repS}{\mathsf{S}}
\nc{\repV}{\mathsf{V}}
\nc{\repW}{\mathsf{W}}
\nc{\repX}{\mathsf{X}}
\nc{\Verma}{\mathcal{V}}
\nc{\tl}{\mathsf{TL}}
\nc{\eptl}{\mathsf{\mathcal EPTL}}
\nc{\bb}{\bar{b}}
\nc{\mb}{\boldsymbol{m}}
\nc{\Fb}{\overline{F}}
\nc{\Tb}{\boldsymbol{T}}
\nc{\Db}{\boldsymbol{D}}
\nc{\id}{\mathbf{1}}
\nc{\Op}{\mathcal{O}}
\nc{\timesf}{\times_{\!f}}
\nc{\tr}{\mathrm{tr}}
\nc{\eps}{\varepsilon}
\nc{\chib}{\raisebox{0.25ex}{$\bar\chi$}}
\nc{\chit}{\raisebox{0.25ex}{$\chi$}}
\nc{\chiK}[1]{\chit_{#1}} 
\nc{\chiKb}[1]{\bar\chit_{#1}} 
\nc{\cL}{\mbox{$\mathcal L$}}
\nc{\cLa}{\mbox{$\mathcal L\mathrm{a}$}}
\nc{\cLd}{\mbox{$\mathcal L\mathrm{d}$}}
\nc{\cLaa}{\mbox{$\Lambda\mathrm{a}$}}
\nc{\cLad}{\mbox{$\Lambda\mathrm{d}$}}
\nc{\cLk}[1]{\mbox{$\mathcal L^{\tinyx {#1}}$}}
\nc{\cLkd}[1]{\mbox{$\mathcal L\mathrm{d}^{\tinyx {#1}}$}}
\nc{\tinyL}{\textrm{\tiny$(\ell)$}}
\nc{\tinyx}[1]{\textrm{\tiny$(#1)$}}
\nc{\tinyz}[1]{\textrm{\tiny$[#1]$}}

\newcommand{\ket}[1]{| {#1} \rangle}

\newcommand{\kket}[1]{|\hspace{-0.03cm}| {#1} \rangle\!\rangle}
\newcommand{\aver}[1]{\langle {#1} \rangle}
\newcommand{\smallaver}[1]{\langle {#1} \rangle}

\newcommand{\aaver}[1]{\langle\!\langle {#1} \rangle\!\rangle}

\definecolor{lightblue}{rgb}{.7,.7,1}
\definecolor{lightestblue}{rgb}{.95,.95,1}
\definecolor{lightlightblue}{rgb}{.9,.9,1}
\definecolor{midblue}{rgb}{.7,.7,1}
\definecolor{babyblue}{rgb}{0.2, 0.75, 1}
\definecolor{purple}{rgb}{0.5,0,0.5}
\definecolor{peach}{rgb}{1, 0.854902, 0.72549}
\definecolor{pastelgreen}{rgb}{0.467, 0.867, 0.467}
\newrgbcolor{darkgreen}{0., 0.733333, 0.0621042}
\definecolor{darkpink}{rgb}{1,.2,1}
\definecolor{pastelpink}{rgb}{1, 0.82, 0.863}

\nc{\elegant}{1.5pt}
\nc{\moyen}{1.0pt}
\nc{\mince}{0.5pt}

\def\facegrid#1#2{
\psframe[fillstyle=solid,fillcolor=lightlightblue,linewidth=0pt]#1#2
\psgrid[gridlabels=0pt,subgriddiv=1]#1#2}

\def\loopa{
\psframe[linewidth=.25pt](0,0)(1,1)
\psarc[linewidth=1.5pt,linecolor=blue](1,0){.5}{90}{180}
\psarc[linewidth=1.5pt,linecolor=blue](0,1){.5}{-90}{0}
}
\def\loopb{
\psframe[linewidth=.25pt](0,0)(1,1)
\psarc[linewidth=1.5pt,linecolor=blue](0,0){.5}{0}{90}
\psarc[linewidth=1.5pt,linecolor=blue](1,1){.5}{180}{270}
}

\begin{document}

\topmargin -5mm
\oddsidemargin 5mm

\makeatletter 
\newcommand\Larger{\@setfontsize\semiHuge{19.00}{21.78}}
\makeatother 

\vspace*{-2cm}
\setcounter{page}{1}
\vspace{22mm}

\begin{center}
\Larger\bf Temperley--Lieb modules and local operators \\
for critical ADE models
\end{center}

\vspace{0.6cm}

\begin{center}
{\vspace{-5mm}\Large Yacine Ikhlef$^{\,\dagger}$ \qquad Alexi Morin-Duchesne}
\\[.4cm]
{\em $^{\dagger}$Sorbonne Universit\'{e}, CNRS, Laboratoire de Physique Th\'{e}orique}
\\
{\em et Hautes \'{E}nergies, LPTHE, F-75005 Paris, France}
\\[.5cm] 
{\tt ikhlef\,@\,lpthe.jussieu.fr \qquad alexi.morin.duchesne\,@\,gmail.com}
\end{center}

\vspace{10pt}

%
%
 
\begin{abstract}

We investigate critical restricted solid-on-solid models associated to Dynkin diagrams of type $A$, $D$ and $E$, with fixed, periodic and twisted periodic boundary conditions. These models are endowed with an action of the diagrams of the Temperley--Lieb category. For each model, we obtain the decomposition of the state space as a direct sum of irreducible modules over the Temperley--Lieb algebra $\tl_N(\beta)$ or its periodic incarnation $\eptl_N(\beta)$. This allows us to recover the known conformal partition functions for these models in the continuum scaling limit.\medskip

For each irreducible factor arising in the decompositions, we define an associated local operator on the lattice, which behaves like a connectivity operator. Using knowledge from the Temperley--Lieb representation theory at roots of unity, we show that these operators satisfy certain linear difference relations, which are lattice counterparts of the singular-vector relations in conformal field theory.
\end{abstract}
%
%
\tableofcontents
\clearpage

%
\section{Introduction}
%

Solid-on-solid (SOS) models are a class of exactly solvable models of statistical mechanics. They are defined by height variables living on the vertices of a square or rhombic lattice, with an interaction energy associated to the height configuration around each face. In the seminal papers \cite{ABF84,BF85}, a series of \emph{restricted} solid-on-solid (RSOS) models was introduced, where each height takes the values $\{1,2,\dots,n\}$, heights sitting on neighbouring vertices differ by $+1$ or $-1$, and the face interaction satisfies the star-triangle relation. The authors obtained exact formulas for local height probabilities in these models in terms of elliptic theta functions and generalized Rogers--Ramanujan $q$-series. Soon after, Huse realised \cite{Huse84} that the critical exponents associated to these local height probabilities coincide with the conformal dimensions of the unitary minimal models of Conformal Field Theory (CFT) \cite{BPZ84,FQS84}.\medskip

In the RSOS models of \cite{ABF84,BF85}, a coupling constant drives the system through a continuous phase transition. In \cite{PasquierADE87}, Pasquier focused on the unitary cases 
and showed that, at the critical point, these models admit a graph expansion in terms of a dense loop model, whose underlying symmetry is encoded in the Temperley--Lieb algebra \cite{TL71}, with the loop weight $\beta$ fixed to $2\cos(\pi/(n+1))$. He then constructed the \emph{ADE lattice models}, a series of critical RSOS models, whose heights take values on the nodes of a Dynkin diagram of type $A_n$, $D_n$ or $E_n$, and whose face interactions are designed to obey a similar graph expansion --- see also \cite{Pearce90} for an early review. The RSOS models of \cite{ABF84} simply correspond to the case of $A_n$. In \cite{PasquierOpContent87}, Pasquier defined a collection of local operators in the ADE lattice models, and determined their composition rules and their scaling dimensions. Based on these results, he argued that, in the continuum scaling limit, each ADE lattice model scales to the corresponding CFT in the ADE classification of modular invariant minimal models \cite{CIZ87}. This argument was also supported by the derivation of the torus partition functions \cite{PasquierZTorus87} through a mapping to the six-vertex model, assuming that the latter scales to a compact scalar field. This result in the continuum scaling limit motivated the intuition \cite{PasquierOpContent87,SZ90} that, for each ADE lattice model, the operator algebra translates into a complete decomposition of the state space in terms of indecomposable Temperley--Lieb modules, which can be deduced from the content of the corresponding minimal CFT.\medskip

Following this intense activity on solvable lattice models, mathematical physicists, algebraists and combinatorists started to unravel the representation theory of the Temperley--Lieb algebra $\tl_N(\beta)$ and its many variants. As nicely reviewed in \cite{RSA14}, these studies have a wide range of applications, including in statistical mechanics \cite{M91}, knot theory \cite{J83}, quantum spin chains \cite{GRS13}, and more. The representation theory turns out to be significantly harder in the cases where the loop weight is of the form $\beta=-q-q^{-1}$ with $q$ a root of unity. In this context, the approach introduced by Graham and Lehrer \cite{GL98}, based on the concept of cellular algebras \cite{GL96}, has the clear advantage of unifying a number of results concerning the algebra $\tl_N(\beta)$ and its periodic version, which we here denote as $\eptl_N(\beta)$. In this approach, one considers not only diagrams of non-intersecting curves on a system of $N$ inner and $N$ outer nodes, but its generalisation to the Temperley--Lieb \emph{category}, whereby the diagrams are the morphisms of the category, and have $N$ inner nodes and $N'$ outer nodes. Thus, in the framework of \cite{GL98}, the study of \emph{families of modules} \cite{IMD25} over the sequence of Temperley--Lieb algebras with varying numbers $N$ of nodes is natural. One of their most relevant results for the application to ADE lattice models is the structure of the {\it standard modules}, also called \emph{cell modules} in \cite{GL98}. These results cover all values of $q$, including when it is a root of unity.\medskip

Our present work is a systematic study of the relation between the critical ADE lattice models and the Temperley--Lieb algebras. We consider the transfer matrix of the ADE lattice model, acting on the space of states with fixed, periodic or twisted periodic boundary conditions. We prove from first principles that this space has the structure of a family of modules over the Temperley--Lieb category, and we determine its decomposition over the algebras' indecomposable modules. The models with fixed boundary conditions are described by the ordinary Temperley--Lieb category, whereas the models with periodic or twisted periodic boundary conditions involve the enlarged periodic (or {\it affine}) Temperley--Lieb category. As is expected for rational models, the only modules that arise in these decompositions are irreducible modules, in contrast with the case of logarithmic models where reducible yet indecomposable modules are common. From these decompositions, we obtain two results: (i) we check that the cylinder and torus partition functions scale to their known formulas in the minimal CFTs \cite{CIZ87,CardyBCFT89,BPZ98},
and (ii) we construct operators that generalise Pasquier's construction of local operators. For the periodic and twisted periodic cases, for each irreducible module $\repQ_{k,x}(N)$ arising in the decomposition of the state space, we construct a local operator living on a closed curve visiting $2k$ vertices of the lattice. The operators constructed by Pasquier correspond to the case $k=0$, and are defined directly in terms of the adjacency matrix of the Dynkin diagram. For the twisted case, our construction produces certain quasi-local operators whose insertion points are the endpoints of defect lines. For the boundary case, we instead construct a boundary operator associated to each module $\repQ_k(N)$ arising in the decomposition of the state space. In all the cases, we derive a set of discrete linear equations satisfied by the corresponding local operators. This is a major achievement of this paper, as these linear equations are the lattice analogs of the singular-vector relations satisfied by primary fields in minimal CFTs. More work will be needed to show that the newly derived lattice relations converge to the Virasoro singular-vector relations in the scaling limit.\medskip

The paper is organised as follows. \cref{sec:ADE} reviews the definition of the critical ADE lattice models, with some extra information given in \cref{app:Samu}. \cref{sec:EPTL} describes the diagram spaces associated to the Temperley--Lieb category, and the properties of its modules for $q$ a root of unity. \cref{sec:scaling.limit} describes the continuum scaling limit and the characters associated to these modules. In \cref{sec:boundaryADE,sec:periodicADE}, we investigate the action of the diagrams of the Temperley--Lieb category on the state space of the ADE models, for fixed boundary conditions and for periodic boundary conditions (both twisted and untwisted), respectively. In these sections, we obtain the modules decompositions and compute the corresponding modular covariant partition functions, with some details relegated to \cref{app:E678,app:dims,app:insertion,app:Z}. In \cref{sec:correlators}, we discuss Pasquier's local operators, construct its generalisations, and derive the relations satisfied by their correlation functions. Concluding remarks are presented in \cref{sec:conclusion}.

%
\section{Definition of the ADE lattice models}
\label{sec:ADE}
%

In this section, we review the definition of the critical ADE lattice models.

\subsection{Dynkin diagrams and automorphisms}\label{sec:Dynkin}

Let $\mathcal G$ be the Dynkin diagram with $n$ nodes associated to the Lie algebra $\g=A_n$, $D_n$, $E_6$, $E_7$ or $E_8$, and $A$ be the adjacency matrix of $\mathcal G$. These Dynkin diagrams are given in \cref{fig:Dynkin}. The eigenvalues of $A$ are of the form
\begin{equation}
\label{eq:beta.RSOS}
\beta_\mu = 2\cos \frac{\pi m_\mu}{p'} \,,
\end{equation}
where $p'$ is the {\it Coxeter number} associated to the algebra $\g$, $\mu$ is an {\it index} in $\{1,2,\dots,n\}$, and $m_\mu$ is an integer {\it exponent} satisfying $1 \leq m_\mu \leq p'-1$. The values of $p'$ and the possible exponents $m_\mu$ are listed in \cref{tab:Dynkin.data}.\medskip

The matrix $A$ is real symmetric and can thus be diagonalised in an orthonormal basis $S_{a\mu}$. The eigenvalue equation and the orthonormality conditions read
\begin{equation}
\sum_{b=1}^n A_{ab} S_{b\mu} = \beta_\mu \, S_{a\mu} \,,
\end{equation}
and
\begin{equation}
\sum_{a=1}^n S_{a\mu} \, S_{a\nu}^* = \delta_{\mu\nu} \,,
\qquad 
\sum_{\mu=1}^n S_{a\mu} \, S_{b\mu}^* = \delta_{ab}\,.
\end{equation}
For each algebra $\g$ in the ADE series, we give in \cref{app:Samu} the components $S_{a\mu}$ of an orthonormal basis of eigenstates.
\medskip

\begin{figure}[t] 
\centering
\begin{tabular}{ll}
\begin{tabular}{l}
$
A_n: \quad
\begin{pspicture}[shift=-0.9](0,-1)(6,1)
\multiput(0,0)(1,0){4}{\psline{-}(0,0)(1,0)}
\multiput(0,0)(1,0){5}{\pscircle[linewidth=1.5pt,linecolor=black,fillstyle=solid,fillcolor=white](0,0){.175}}
\rput(0,-0.375){\scriptsize$1$}\rput(1,-0.375){\scriptsize$2$}\rput(2,-0.375){\scriptsize$3$}\rput(3,-0.375){\scriptsize$\dots$}\rput(4,-0.375){\scriptsize$n$}
\end{pspicture}
$
\\
\\
$
D_n: \quad
\begin{pspicture}[shift=-1.9](0,-2)(5,1)
\multiput(0,0)(1,0){4}{\psline{-}(0,0)(1,0)}
\psline{-}(4,0)(4.7,0.7)\psline{-}(4,0)(4.7,-0.7)
\multiput(0,0)(1,0){5}{\pscircle[linewidth=1.5pt,linecolor=black,fillstyle=solid,fillcolor=white](0,0){.175}}
\rput(0,-0.375){\scriptsize$1$}\rput(1,-0.375){\scriptsize$2$}\rput(2,-0.375){\scriptsize$3$}\rput(3,-0.375){\scriptsize$\dots$}\rput(3.9,-0.375){\scriptsize$n\!-\!2$}
\pscircle[linewidth=1.5pt,linecolor=black,fillstyle=solid,fillcolor=white](4.7,0.7){.175}
\pscircle[linewidth=1.5pt,linecolor=black,fillstyle=solid,fillcolor=white](4.7,-0.7){.175}
\rput(4.7,1.07){\scriptsize$n$}
\rput(4.7,-1.10){\scriptsize$n\!-\!1$}
\end{pspicture}
$
\end{tabular}
\qquad
\begin{tabular}{l}
$
E_6: \quad
\begin{pspicture}[shift=-0.9](0,-1)(6,1.5)
\multiput(0,0)(1,0){4}{\psline{-}(0,0)(1,0)}\psline{-}(2,0)(2,1)
\multiput(0,0)(1,0){5}{\pscircle[linewidth=1.5pt,linecolor=black,fillstyle=solid,fillcolor=white](0,0){.175}}\pscircle[linewidth=1.5pt,linecolor=black,fillstyle=solid,fillcolor=white](2,1){.175}
\rput(0,-0.375){\scriptsize$1$}\rput(1,-0.375){\scriptsize$2$}\rput(2,-0.375){\scriptsize$3$}\rput(3,-0.375){\scriptsize$4$}\rput(4,-0.375){\scriptsize$5$}\rput(2,1.37){\scriptsize$6$}
\end{pspicture}
$
\\
$
E_7: \quad
\begin{pspicture}[shift=-0.9](0,-1)(6,1.5)
\multiput(0,0)(1,0){5}{\psline{-}(0,0)(1,0)}\psline{-}(2,0)(2,1)
\multiput(0,0)(1,0){6}{\pscircle[linewidth=1.5pt,linecolor=black,fillstyle=solid,fillcolor=white](0,0){.175}}\pscircle[linewidth=1.5pt,linecolor=black,fillstyle=solid,fillcolor=white](2,1){.175}
\rput(0,-0.375){\scriptsize$1$}\rput(1,-0.375){\scriptsize$2$}\rput(2,-0.375){\scriptsize$3$}\rput(3,-0.375){\scriptsize$4$}\rput(4,-0.375){\scriptsize$5$}\rput(5,-0.375){\scriptsize$6$}\rput(2,1.37){\scriptsize$7$}
\end{pspicture}
$
\\
$
E_8: \quad
\begin{pspicture}[shift=-0.9](0,-1)(6,1.5)
\multiput(0,0)(1,0){6}{\psline{-}(0,0)(1,0)}\psline{-}(2,0)(2,1)
\multiput(0,0)(1,0){7}{\pscircle[linewidth=1.5pt,linecolor=black,fillstyle=solid,fillcolor=white](0,0){.175}}\pscircle[linewidth=1.5pt,linecolor=black,fillstyle=solid,fillcolor=white](2,1){.175}
\rput(0,-0.375){\scriptsize$1$}\rput(1,-0.375){\scriptsize$2$}\rput(2,-0.375){\scriptsize$3$}\rput(3,-0.375){\scriptsize$4$}\rput(4,-0.375){\scriptsize$5$}\rput(5,-0.375){\scriptsize$6$}\rput(6,-0.375){\scriptsize$7$}\rput(2,1.37){\scriptsize$8$}
\end{pspicture}
$
\end{tabular}
\end{tabular}
\caption{The Dynkin diagrams for the ADE series.}
\label{fig:Dynkin}
\end{figure}
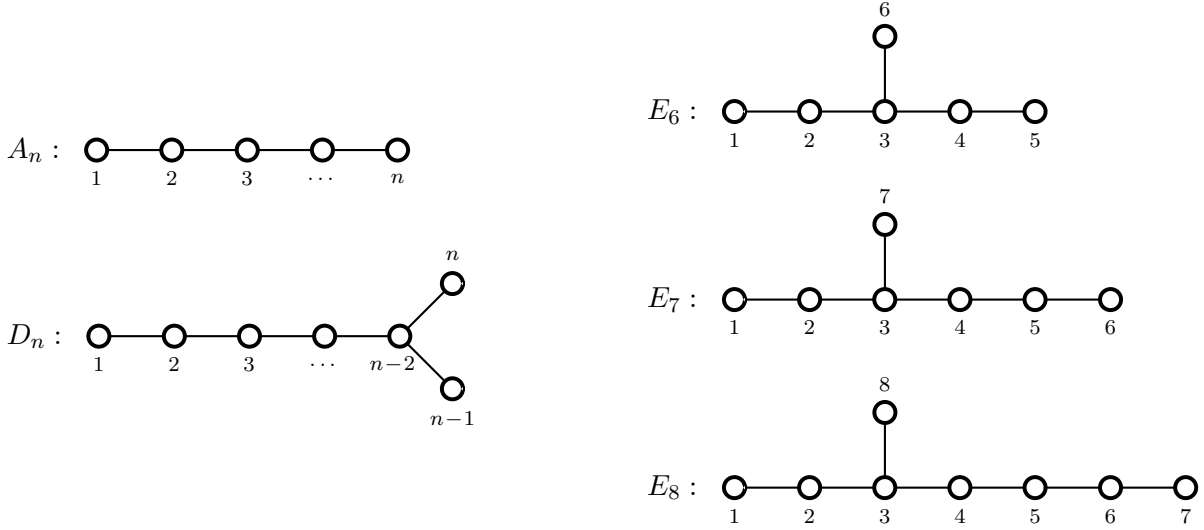

\begin{table}[t]
\begin{center}
\begin{tabular}{c|c|c}
Algebra $\mathfrak{\g}$ & Coxeter number $p'$ & Exponents $m_1,m_2,\dots,m_n$\\ 
\hline
$A_n$ & $n+1$ & $1,2,\dots,n$
\\
$D_n$ & $2n-2$ & $1,3,\dots,2n-3,n-1$
\\
$E_6$ & $12$ & $1,4,5,7,8,11$
\\
$E_7$ & $18$ & $1,5,7,9,11,13,17$ 
\\ 
$E_8$ & $30$ & $1,7,11,13,17,19,23,29$
\end{tabular}
\end{center}
\caption{Coxeter numbers and exponents for the ADE series.}
\label{tab:Dynkin.data}
\end{table}

\begin{table}[t]
\begin{center}
\begin{tabular}{c|c|c}
$(\g,K)$ & Number of fixed points $\wt{n}$ & Exponents $\wt m_1,\wt m_2,\dots, \wt m_{\wt n}$ \\ 
\hline
$(A_n,R)$ with $n$ odd & $1$ & $(n+1)/2$ \\
$(D_n,P_{(n-1,n)})$ & $n-2$ & $2,4,\dots,2n-4$ \\
$(D_4,P_{(134)})$ & $1$ & $3$ \\
$(E_6,P_{(15)(24)})$ & $2$ & $4,8$
\end{tabular}
\end{center}
\caption{The exponents for the subgraph of fixed points of the pairs $(\g,K)$.}
\label{tab:Dynkin.fixed.points}
\end{table}

If the graph $\mathcal G$ admits an automorphism $K$, namely a bijection $a \mapsto K(a)$ from $\mathcal G$ to itself such that each pair of adjacent vertices is mapped to a pair of adjacent vertices, then the matrix $K_{ab}=\delta_{K(a),b}$ commutes with $A$. Thus, the orthogonal basis can be chosen to diagonalise both $A$ and $K$: this is the case for the eigenstates of \cref{app:Samu}.
We denote by $\kappa_\mu$ the eigenvalue of $K$ associated to the eigenstate of $A$ with eigenvalue $\beta_\mu$. We have
\begin{equation}
S_{K(a)\mu} = \sum_{b=1}^n \delta_{K(a),b} \, S_{b\mu}
= (KS)_{a\mu} = \kappa_\mu \, S_{a\mu} \,.
\end{equation}
For the ADE series, the non-trivial homomorphisms are as follows:
\begin{itemize}
\item[1)] For all $n$, the Dynkin diagram of $A_n$ has a non-trivial reflection automorphism $K = R$ that maps $a \mapsto n+1 - a$. The eigenvalues of $R$ are $\kappa_\mu = (-1)^{\mu+1}$.
\item[2)] For all $n$, the Dynkin diagram of $D_n$ has the permutation automorphism $K = P_{(n-1,n)}$ that permutes the nodes $n-1$ and $n$ and leaves the other nodes unchanged. The eigenvalues of $P_{(n-1,n)}$ are $\kappa_\mu = 1$ for $\mu=1,2, \dots, n-1$, and $\kappa_n = -1$.
\item[3)] The Dynkin diagram of $D_4$ is invariant under the symmetric group $S_3$ of permutations of the nodes $1$, $3$ and $4$. This group has three conjugacy classes, with representing elements $\id$, $P_{(34)}$ and $P_{(134)}$. For $P_{(134)}$, the eigenvalues are $\kappa_\mu = 1,\omega,1,\omega^2$ for $\mu = 1,2,3,4$.
\item[4)] The Dynkin diagram of $E_6$ has the automorphism $K=P_{(15),(24)}$. Its eigenvalues are $\kappa_\mu = 1$ for $\mu=1,3,4,6$ and $\kappa_\mu=-1$ for $\mu = 2,5$.
\end{itemize}
The graphs $E_7$ and $E_8$ have no non-trivial automorphisms.
\medskip

The fixed points under $K$ of the Dynkin diagram $\mathcal G$ form a subgraph~$\wt{\mathcal G}$ of~$\mathcal G$. In each case, this subgraph is the Dynkin diagram of a Lie algebra $A_{\wt n}$, where $\wt n$ is the cardinality of $\wt{\cal G}$, namely the number of fixed points. We denote by $\wt A$ the adjacency matrix associated to $\wt{\cal G}$, and by $(\alpha_1,\alpha_2,\dots,\alpha_{\wt n})$ the nodes of $\wt{\cal G}$, so that $\alpha_j$ and $\alpha_{j+1}$ are adjacent for each $j$. The eigenvalues and eigenvector components of $\wt A$ read
\begin{equation} \label{eq:St.a.mu}
\wt{S}_{\alpha_j,\nu} = \sqrt{\frac{2}{\wt{n}+1}} \, \sin \left( \frac{\pi j\wt{m}_\nu}{p'}\right) \,,
\qquad \wt{\beta}_\nu =2 \cos \left( \frac{\pi \wt{m}_\nu}{p'}\right) \,,
\end{equation}
where $p'$ is the Coxeter number of the original graph $\cal G$, $\nu \in \{1,2,\dots, \wt n\}$, and $\wt{m}_\nu$ is an integer satisfying $1\leq \wt{m}_\nu\leq p'-1$. These eigenvectors form an orthonormal basis of $\Rbb^{\wt n}$. \cref{tab:Dynkin.fixed.points} lists this data for each pair $(\g,K)$ for which the subgraph $\wt{\cal G}$ is not empty.\medskip

\subsection{Interaction round-a-face models and transfer matrices}\label{sec:lattice.RSOS}

The critical ADE model $(\mathfrak{\g},\mu)$ is an interaction round-a-face model defined by a choice of an algebra $\mathfrak{\g}$ in the ADE series, and an index $\mu \in \{1, 2, \dots, n\}$, with the condition that the corresponding eigenvector of~$A$ has only non-vanishing entries $S_{a\mu}$. By the Perron--Froebenius theorem, the leading eigenvalue~$\beta_1$ has multiplicity one, and the corresponding eigenvector has real positive entries: $S_{a1}>0$. Thus the value $\mu=1$ is always a valid choice. The explicit formulas for the eigenvector components in \cref{app:Samu} reveal that $S_{a\mu} \neq 0$ for all $a$ if and only if $\mathrm{gcd}(m_\mu,p') = 1$.\medskip

We consider a square lattice $\Sigma$ embedded on a connected surface. This surface may have an arbitrary genus, as well as a boundary $B$ that may consist of multiple components 
$\ell$. Each vertex of $\Sigma$ is assigned a height $a \in \mathcal G$, with the constraint that the heights on neighbouring vertices are adjacent on $\mathcal G$. The Boltzmann weight of a face $f$ of $\Sigma$ is 
\begin{equation}
W
\left(
\psset{unit=0.5cm}
\begin{pspicture}[shift=-0.3cm](-0.4,-0.3)(1.4,1.3)
\facegrid{(0,0)}{(1,1)}
\psarc(0,0){0.25}{0}{90}
\rput(-0.2,1.2){$_a$}
\rput(1.2,1.2){$_b$}
\rput(1.2,-0.2){$_c$}
\rput(-0.2,-0.2){$_d$}
\end{pspicture}
\right)
= \, \delta_{bd} + \, \delta_{ac}\, \frac{S_{d\mu}}{S_{a\mu}} \,,
\end{equation}
where $a,b,c,d$ are the four heights attached to the corners of $f$.\medskip

Each boundary component $\ell$ consists in a closed curve that visits $M_\ell$ nodes of the lattice, for some non-negative even integer $M_\ell$. Let $u_\ba = \ket{a_0, a_1, \dots, a_{M_\ell}}$ be the state made of the heights on this boundary component in a given configuration, with $a_{M_\ell}=a_0$. We assign to each such configuration $u_\ba$ a boundary Boltzmann weight $\alpha_{\ba(\ell)}$. Some of these weights may be zero, thus restricting the possible boundary states. The resulting partition function for the critical ADE model $(\g,\mu)$ on the lattice $\Sigma$ is
\be
Z(\Sigma) = \sum_{\sigma} \prod_{\ell\, \in B}\alpha_{\pa(\ell)}\prod_{f \in \Sigma} 
W_f\left(
\psset{unit=0.5cm}
\begin{pspicture}[shift=-0.3cm](-0.4,-0.3)(1.4,1.3)
\facegrid{(0,0)}{(1,1)}
\psarc(0,0){0.25}{0}{90}
\rput(-0.2,1.2){$_a$}
\rput(1.2,1.2){$_b$}
\rput(1.2,-0.2){$_c$}
\rput(-0.2,-0.2){$_d$}
\end{pspicture}
\right).
\ee

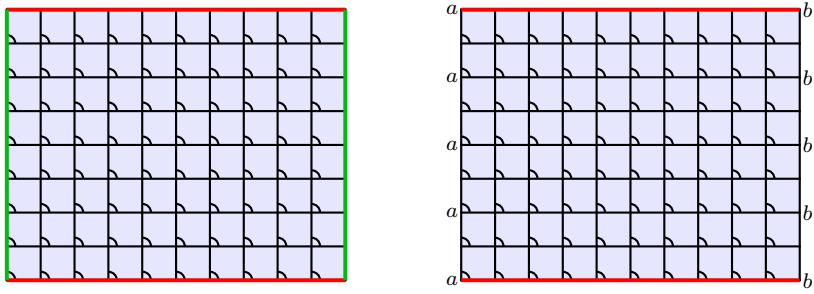
\begin{figure}[t]
\begin{equation*}
\psset{unit=0.45cm}
\begin{pspicture}(0,0)(10,8)
\facegrid{(0,0)}{(10,8)}
\multiput(0,0)(0,1){8}{\multiput(0,0)(1,0){10}{\psarc(0,0){0.25}{0}{90}}}
\psline[linewidth=1.5pt,linecolor=red](0,0)(10,0)
\psline[linewidth=1.5pt,linecolor=red](0,8)(10,8)
\psline[linewidth=1.5pt,linecolor=darkgreen](0,0)(0,8)
\psline[linewidth=1.5pt,linecolor=darkgreen](10,0)(10,8)
\end{pspicture}
\qquad
\qquad
\begin{pspicture}(0,0)(10,8)
\facegrid{(0,0)}{(10,8)}
\multiput(0,0)(0,1){8}{\multiput(0,0)(1,0){10}{\psarc(0,0){0.25}{0}{90}}}
\psline[linewidth=1.5pt,linecolor=red](0,0)(10,0)
\psline[linewidth=1.5pt,linecolor=red](0,8)(10,8)
\multiput(0,0)(0,2){5}{\rput(-0.25,0){$_a$}}
\multiput(10,0)(0,2){5}{\rput(0.25,0){$_b$}}
\end{pspicture}
\end{equation*}
\caption{The lattices corresponding to the $8\times 10$ torus and cylinder. The red and green segments indicate periodic boundary conditions.}
\label{fig:lattices}
\end{figure}

We now consider lattices consisting in rectangular array of $M\times N$ square faces, with toroidal and cylindrical boundary conditions, where the width $M$ and the height $N$ are non-negative even integers. These lattices are illustrated in \cref{fig:lattices}. The corresponding partition functions can be computed using transfer matrices. To compute the torus partition function, we use the single-row transfer matrix $T_N$, whose matrix elements are the Boltzmann weights associated to one row of the lattice
\be
\label{eq:Tab}
(T_N)_{\boldsymbol{a},\boldsymbol{b}} = W\left(
\psset{unit=0.5cm}
\begin{pspicture}[shift=-0.5cm](-0.4,-0.6)(8.4,1.6)
\facegrid{(0,0)}{(8,1)}
\multiput(0,0)(1,0){8}{\psarc(0,0){0.25}{0}{90}}
\rput(0,-0.3){$_{a_0}$}\rput(0,1.4){$_{b_0}$}
\rput(1,-0.3){$_{a_1}$}\rput(1,1.4){$_{b_1}$}
\rput(2,-0.3){$_{...}$}\rput(2,1.4){$_{...}$}
\rput(8,-0.3){$_{a_N}$}\rput(8,1.4){$_{b_N}$}
\end{pspicture}
\right)
=
\prod_{j=0}^{N-1} 
\,W
\left(
\psset{unit=0.5cm}
\begin{pspicture}[shift=-0.5cm](-0.7,-0.6)(1.9,1.6)
\facegrid{(0,0)}{(1,1)}
\psarc(0,0){0.25}{0}{90}
\rput(-0.2,1.4){$_{b_{j}}$}
\rput(1.3,1.4){$_{b_{j+1}}$}
\rput(1.3,-0.4){$_{a_{j+1}}$}
\rput(-0.2,-0.4){$_{a_{j}}$}
\end{pspicture}
\right)\,,
\ee
where $a_0 = a_N$ and $b_0 = b_N$. The torus partition function $Z_{\textrm{tor}}(M,N)$ is computed as
\be
Z_{\textrm{tor}}(M,N) = \textrm{tr}\, T_N^M\,.
\ee

For the cylinder, we choose the boundary condition on the left edge to be restricted to configuration $\ket{a,a_1,a,a_3,a,\dots,a,a_M}$, where $a$ takes a unique fixed value, whereas $a_1$, $a_3$, \dots, $a_M$ are free to take any values in $\mathcal G$ that are neighbors to $a$. Each of these configurations is assigned the same unit weight. Similarly, the boundary condition on the right boundary consists of states $\ket{b,b_1,b,b_3,b,\dots,b,b_M}$. The vertices with fixed heights $a$ and $b$ are then chosen to be aligned. The partition function of the model~$(g,\mu)$ on the cylinder with this boundary condition of type $(a,b)$ is computed using the double-row transfer matrix $D_N$, defined as
\be
\label{eq:Dab}
(D_N)_{\boldsymbol{a},\boldsymbol{b}} = 
W\left(
\psset{unit=0.5cm}
\begin{pspicture}[shift=-0.7cm](-0.4,-0.4)(8.4,2.4)
\facegrid{(0,0)}{(8,2)}
\multiput(0,0)(1,0){8}{\psarc(0,0){0.25}{0}{90}}
\multiput(0,1)(1,0){8}{\psarc(0,0){0.25}{0}{90}}
\rput(0,-0.3){$_{a_0}$}\rput(0,2.4){$_{b_0}$}
\rput(1,-0.3){$_{a_1}$}\rput(1,2.4){$_{b_1}$}
\rput(2,-0.3){$_{...}$}\rput(2,2.4){$_{...}$}
\rput(8,-0.3){$_{a_N}$}\rput(8,2.4){$_{b_N}$}
\end{pspicture}
\right)
=
\sum_{\boldsymbol{a'}}\prod_{j=0}^{N-1} 
\,
W\left(
\psset{unit=0.5cm}
\begin{pspicture}[shift=-0.5cm](-0.7,-0.6)(1.9,1.6)
\facegrid{(0,0)}{(1,1)}
\psarc(0,0){0.25}{0}{90}
\rput(-0.2,1.4){$_{a'_{j}}$}
\rput(1.3,1.4){$_{a'_{j+1}}$}
\rput(1.3,-0.4){$_{a_{j+1}}$}
\rput(-0.2,-0.4){$_{a_{j}}$}
\end{pspicture}
\right)
W
\left(
\psset{unit=0.5cm}
\begin{pspicture}[shift=-0.5cm](-0.7,-0.6)(1.9,1.6)
\facegrid{(0,0)}{(1,1)}
\psarc(0,0){0.25}{0}{90}
\rput(-0.2,1.4){$_{b_{j}}$}
\rput(1.3,1.4){$_{b_{j+1}}$}
\rput(1.3,-0.4){$_{a'_{j+1}}$}
\rput(-0.2,-0.4){$_{a'_{j}}$}
\end{pspicture}
\right)
\,,
\ee
where $a_0=b_0=a$ and $a_N=b_N=b$, and the sum over $\boldsymbol{a'}$ runs over all possible heights of the intermediate state $\ket{a'_0,a'_1,\dots, a'_{N}}$. The cylinder partition function $Z_{\textrm{cyl}}(M,N)$ is computed as
\be
Z_{\textrm{cyl}}(M,N) = \textrm{tr}\, D_{N}^{M/2}\,.
\ee

%
\section{Diagram spaces and families of modules}\label{sec:EPTL}
%

In this section, we review the definitions of the diagram spaces and Temperley--Lieb algebras, and describe their standard modules at roots of unity.

\subsection{Diagram spaces}\label{sec:LNN'}

The {\it diagram space} $\cL(N,N')$ is the vector space spanned by diagrams made of non-intersecting loop segments drawn on a ring that connect pairwise $N$ and $N'$ nodes on the outer and inner boundary, respectively. Diagrams may also contain non-contractible loops, namely loops winding around the inner boundary. A dashed segment is drawn connecting the two circles. Two diagrams that differ only by continuous deformations of the loop segments are considered equal. For example,
\begin{equation}
\label{eq:first diagrams}
\psset{unit=0.8cm}
\lambda_1 = \
\begin{pspicture}[shift=-1.7](-1.6,-1.6)(1.6,1.6)
\psarc[linecolor=black,linewidth=0.5pt,fillstyle=solid,fillcolor=lightlightblue]{-}(0,0){1.5}{0}{360}
\psarc[linecolor=black,linewidth=0.5pt,fillstyle=solid,fillcolor=white]{-}(0,0){0.7}{0}{360}
\rput{30}(0,0){\psbezier[linecolor=blue,linewidth=1.5pt]{-}(-0.388229, -1.44889)(-0.310583, -1.15911)(0.310583, -1.15911)(0.388229, -1.44889)
\psbezier[linecolor=blue,linewidth=1.5pt]{-}(-1.06066, -1.06066)(-0.848528, -0.948528)(0.848528, -0.948528)(1.06066, -1.06066)}
\rput{-90}(0,0){\psbezier[linecolor=blue,linewidth=1.5pt]{-}(-0.388229, -1.44889)(-0.310583, -1.15911)(0.310583, -1.15911)(0.388229, -1.44889)}
\rput{-150}(0,0){\psbezier[linecolor=blue,linewidth=1.5pt]{-}(-0.388229, -1.44889)(-0.310583, -1.15911)(0.310583, -1.15911)(0.388229, -1.44889)}
\rput{120}(0,0){\psbezier[linecolor=blue,linewidth=1.5pt]{-}(-0.388229, -1.44889)(-0.310583, -1.15911)(0.310583, -1.15911)(0.388229, -1.44889)}
\psline[linecolor=blue,linewidth=1.5pt]{-}(-1.06066, -1.06066)(-0.494975, -0.494975)
\psbezier[linecolor=blue,linewidth=1.5pt]{-}(0.494975, 0.494975)(0.989949, 0.989949)(-0.989949, 0.989949)(-0.494975, 0.494975)
\psbezier[linecolor=blue,linewidth=1.5pt]{-}(0.494975, -0.494975)(0.636396, -0.636396)(1.06252, -0.284701)(0.965926, 0.258819)
\psbezier[linecolor=blue,linewidth=1.5pt]{-}(0.968, 0.25)(0.777817, 0.777817)(0.310583, 1.15911)(0.388229, 1.44889)
\psline[linestyle=dashed, dash= 1.5pt 1.5pt,linewidth=0.5pt]{-}(0,-1.5)(0,-0.7)
\end{pspicture}
\ , \qquad \lambda_2 = \
\begin{pspicture}[shift=-1.5](-1.6,-1.6)(1.6,1.6)
\psarc[linecolor=black,linewidth=0.5pt,fillstyle=solid,fillcolor=lightlightblue]{-}(0,0){1.5}{0}{360}
\psarc[linecolor=black,linewidth=0.5pt,fillstyle=solid,fillcolor=white]{-}(0,0){0.7}{0}{360}
\rput{180}(0,0){\psbezier[linecolor=blue,linewidth=1.5pt]{-}(-1.06066, -1.06066)(-0.848528, -0.948528)(0.848528, -0.948528)(1.06066, -1.06066)}
\psbezier[linecolor=blue,linewidth=1.5pt]{-}(0.7, 0.)(1.1, 0.)(0.55, 0.952628)(0.35, 0.606218)
\psbezier[linecolor=blue,linewidth=1.5pt]{-}(-0.7, 0.)(-1.1, 0.)(-0.55, -0.952628)(-0.35, -0.606218)
\psbezier[linecolor=blue,linewidth=1.5pt]{-}(-0.35, 0.606218)(-0.7, 1.3)(-1.0,0.3)(-1.0, 0)
\psbezier[linecolor=blue,linewidth=1.5pt]{-}(-1.0, 0)(-1.0,-0.3)(-1, -1)(-1.06066, -1.06066)
\psbezier[linecolor=blue,linewidth=1.5pt]{-}(0.35, -0.606218)(0.55, -0.952628)(0.777817, -0.777817)(1.06066, -1.06066)
\psline[linestyle=dashed, dash= 1.5pt 1.5pt,linewidth=0.5pt]{-}(0,-1.5)(0,-0.7)
\end{pspicture}
\ , \qquad \lambda_3 = \
\begin{pspicture}[shift=-1.5](-1.6,-1.6)(1.6,1.6)
\psarc[linecolor=black,linewidth=0.5pt,fillstyle=solid,fillcolor=lightlightblue]{-}(0,0){1.5}{0}{360}
\psarc[linecolor=black,linewidth=0.5pt,fillstyle=solid,fillcolor=white]{-}(0,0){0.7}{0}{360}
\psbezier[linecolor=blue,linewidth=1.5pt]{-}(-0.75, 1.3)(-0.6, 1.04)(-1.2, 0.)(-1.5, 0.)
\psbezier[linecolor=blue,linewidth=1.5pt]{-}(-0.75, -1.3)(-0.6, -1.04)(0.6, -1.04)(0.75, -1.3)
\psbezier[linecolor=blue,linewidth=1.5pt]{-}(0.75, 1.3)(0.65, 1.13)(0.27, 1.01)(0., 1.)
\psbezier[linecolor=blue,linewidth=1.5pt]{-}(0., 1.)(-0.6, 1.04)(-1.2, 0.)(-0.87, -0.5)
\psbezier[linecolor=blue,linewidth=1.5pt]{-}(1.5, 0.)(1.3, 0.)(1.01, -0.27)(0.87, -0.5)
\psbezier[linecolor=blue,linewidth=1.5pt]{-}(0.87, -0.5)(0.6, -1.04)(-0.6, -1.04)(-0.87, -0.5)
\psline[linestyle=dashed, dash= 1.5pt 1.5pt,linewidth=0.5pt]{-}(0,-1.5)(0,-0.7)
\end{pspicture}\ ,
\ee
are connectivity diagrams in $\cL(12,4)$, $\cL(4,6)$ and $\cL(6,0)$, respectively. The product $\lambda_1 \lambda_2$ of two diagrams $\lambda_1 \in \cL(N,N')$ and $\lambda_2 \in \cL(N',N'')$ is obtained by drawing $\lambda_2$ inside $\lambda_1$. The resulting diagram in $\cL(N,N'')$ is read from the connectivities of the outer and inner boundaries. Each contractible loop is removed and replaced by a multiplicative factor of $\beta$. In contrast, non-contractible loops are not removed and remain in the diagram. For instance, we have
\be
\psset{unit=0.8cm}
\lambda_1 \lambda_2 = \
\psset{unit=1.2}
\begin{pspicture}[shift=-1.4](-1.5,-1.5)(1.5,1.5)
\psarc[linecolor=black,linewidth=0.5pt,fillstyle=solid,fillcolor=lightlightblue]{-}(0,0){1.5}{0}{360}
\psarc[linecolor=black,linewidth=0.5pt,fillstyle=solid,fillcolor=white]{-}(0,0){0.7}{0}{360}
\rput{30}(0,0){\psbezier[linecolor=blue,linewidth=1.5pt]{-}(-0.388229, -1.44889)(-0.310583, -1.15911)(0.310583, -1.15911)(0.388229, -1.44889)
\psbezier[linecolor=blue,linewidth=1.5pt]{-}(-1.06066, -1.06066)(-0.848528, -0.948528)(0.848528, -0.948528)(1.06066, -1.06066)}
\rput{-90}(0,0){\psbezier[linecolor=blue,linewidth=1.5pt]{-}(-0.388229, -1.44889)(-0.310583, -1.15911)(0.310583, -1.15911)(0.388229, -1.44889)}
\rput{-150}(0,0){\psbezier[linecolor=blue,linewidth=1.5pt]{-}(-0.388229, -1.44889)(-0.310583, -1.15911)(0.310583, -1.15911)(0.388229, -1.44889)}
\rput{120}(0,0){\psbezier[linecolor=blue,linewidth=1.5pt]{-}(-0.388229, -1.44889)(-0.310583, -1.15911)(0.310583, -1.15911)(0.388229, -1.44889)}
\psline[linecolor=blue,linewidth=1.5pt]{-}(-1.06066, -1.06066)(-0.494975, -0.494975)
\psbezier[linecolor=blue,linewidth=1.5pt]{-}(0.494975, 0.494975)(0.989949, 0.989949)(-0.989949, 0.989949)(-0.494975, 0.494975)
\psbezier[linecolor=blue,linewidth=1.5pt]{-}(0.494975, -0.494975)(0.636396, -0.636396)(1.06252, -0.284701)(0.965926, 0.258819)
\psbezier[linecolor=blue,linewidth=1.5pt]{-}(0.968, 0.25)(0.777817, 0.777817)(0.310583, 1.15911)(0.388229, 1.44889)
\psline[linestyle=dashed, dash= 1.5pt 1.5pt,linewidth=0.5pt]{-}(0,-1.5)(0,-0.7)
\psset{unit=0.4666}
\psarc[linecolor=black,linewidth=0.5pt,fillstyle=solid,fillcolor=lightlightblue]{-}(0,0){1.5}{0}{360}
\psarc[linecolor=black,linewidth=0.5pt,fillstyle=solid,fillcolor=white]{-}(0,0){0.7}{0}{360}
\rput{180}(0,0){\psbezier[linecolor=blue,linewidth=1.5pt]{-}(-1.06066, -1.06066)(-0.848528, -0.948528)(0.848528, -0.948528)(1.06066, -1.06066)}
\psbezier[linecolor=blue,linewidth=1.5pt]{-}(0.7, 0.)(1.1, 0.)(0.55, 0.952628)(0.35, 0.606218)
\psbezier[linecolor=blue,linewidth=1.5pt]{-}(-0.7, 0.)(-1.1, 0.)(-0.55, -0.952628)(-0.35, -0.606218)
\psbezier[linecolor=blue,linewidth=1.5pt]{-}(-0.35, 0.606218)(-0.7, 1.3)(-1.0,0.3)(-1.0, 0)
\psbezier[linecolor=blue,linewidth=1.5pt]{-}(-1.0, 0)(-1.0,-0.3)(-1, -1)(-1.06066, -1.06066)
\psbezier[linecolor=blue,linewidth=1.5pt]{-}(0.35, -0.606218)(0.55, -0.952628)(0.777817, -0.777817)(1.06066, -1.06066)
\psline[linestyle=dashed, dash= 1.5pt 1.5pt,linewidth=0.5pt]{-}(0,-1.5)(0,-0.7)
\end{pspicture}
\ = \beta \
\begin{pspicture}[shift=-1.4](-1.6,-1.5)(1.6,1.5)
\psarc[linecolor=black,linewidth=0.5pt,fillstyle=solid,fillcolor=lightlightblue]{-}(0,0){1.5}{0}{360}
\psarc[linecolor=black,linewidth=0.5pt,fillstyle=solid,fillcolor=white]{-}(0,0){0.7}{0}{360}
\rput{30}(0,0){\psbezier[linecolor=blue,linewidth=1.5pt]{-}(-0.388229, -1.44889)(-0.310583, -1.15911)(0.310583, -1.15911)(0.388229, -1.44889)
\psbezier[linecolor=blue,linewidth=1.5pt]{-}(-1.06066, -1.06066)(-0.848528, -0.948528)(0.848528, -0.948528)(1.06066, -1.06066)}
\rput{-90}(0,0){\psbezier[linecolor=blue,linewidth=1.5pt]{-}(-0.388229, -1.44889)(-0.310583, -1.15911)(0.310583, -1.15911)(0.388229, -1.44889)}
\rput{-150}(0,0){\psbezier[linecolor=blue,linewidth=1.5pt]{-}(-0.388229, -1.44889)(-0.310583, -1.15911)(0.310583, -1.15911)(0.388229, -1.44889)}
\rput{120}(0,0){\psbezier[linecolor=blue,linewidth=1.5pt]{-}(-0.388229, -1.44889)(-0.310583, -1.15911)(0.310583, -1.15911)(0.388229, -1.44889)}
\psbezier[linecolor=blue,linewidth=1.5pt]{-}(0.7, 0.)(1.1, 0.)(0.55, 0.952628)(0.35, 0.606218)
\psbezier[linecolor=blue,linewidth=1.5pt]{-}(-0.7, 0.)(-1.1, 0.)(-0.55, -0.952628)(-0.35, -0.606218)
\psbezier[linecolor=blue,linewidth=1.5pt]{-}(0.35, -0.606218)(0.5, -0.866025)(1.06252, -0.284701)(0.965926, 0.258819)
\psbezier[linecolor=blue,linewidth=1.5pt]{-}(0.968, 0.25)(0.777817, 0.777817)(0.310583, 1.15911)(0.388229, 1.44889)
\psbezier[linecolor=blue,linewidth=1.5pt]{-}(-0.35, 0.606218)(-0.5, 0.866025)(-1.1, 0.284701)(-0.965926,-0.258819)
\psbezier[linecolor=blue,linewidth=1.5pt]{-}(-0.968,-0.25)(-0.76,-0.9)(-0.919239, -0.919239)(-1.06066, -1.06066)
\psline[linestyle=dashed, dash= 1.5pt 1.5pt,linewidth=0.5pt]{-}(0,-1.5)(0,-0.7)
\end{pspicture}\ ,
\quad
\lambda_2\lambda_3 = \
\begin{pspicture}[shift=-1.4](-1.5,-1.5)(1.5,1.5)
\psarc[linecolor=black,linewidth=0.5pt,fillstyle=solid,fillcolor=lightlightblue]{-}(0,0){1.5}{0}{360}
\psarc[linecolor=black,linewidth=0.5pt,fillstyle=solid,fillcolor=white]{-}(0,0){0.7}{0}{360}
\psarc[linecolor=black,linewidth=0.5pt,fillstyle=solid,fillcolor=lightlightblue]{-}(0,0){1.5}{0}{360}
\psarc[linecolor=black,linewidth=0.5pt,fillstyle=solid,fillcolor=white]{-}(0,0){0.7}{0}{360}
\rput{180}(0,0){\psbezier[linecolor=blue,linewidth=1.5pt]{-}(-1.06066, -1.06066)(-0.848528, -0.948528)(0.848528, -0.948528)(1.06066, -1.06066)}
\psbezier[linecolor=blue,linewidth=1.5pt]{-}(0.7, 0.)(1.1, 0.)(0.55, 0.952628)(0.35, 0.606218)
\psbezier[linecolor=blue,linewidth=1.5pt]{-}(-0.7, 0.)(-1.1, 0.)(-0.55, -0.952628)(-0.35, -0.606218)
\psbezier[linecolor=blue,linewidth=1.5pt]{-}(-0.35, 0.606218)(-0.7, 1.3)(-1.0,0.3)(-1.0, 0)
\psbezier[linecolor=blue,linewidth=1.5pt]{-}(-1.0, 0)(-1.0,-0.3)(-1, -1)(-1.06066, -1.06066)
\psbezier[linecolor=blue,linewidth=1.5pt]{-}(0.35, -0.606218)(0.55, -0.952628)(0.777817, -0.777817)(1.06066, -1.06066)
\psline[linestyle=dashed, dash= 1.5pt 1.5pt,linewidth=0.5pt]{-}(0,-1.5)(0,-0.7)
\psset{unit=0.4666}
\psarc[linecolor=black,linewidth=0.5pt,fillstyle=solid,fillcolor=lightlightblue]{-}(0,0){1.5}{0}{360}
\psarc[linecolor=black,linewidth=0.5pt,fillstyle=solid,fillcolor=white]{-}(0,0){0.7}{0}{360}
\psbezier[linecolor=blue,linewidth=1.5pt]{-}(-0.75, 1.3)(-0.6, 1.04)(-1.2, 0.)(-1.5, 0.)
\psbezier[linecolor=blue,linewidth=1.5pt]{-}(-0.75, -1.3)(-0.6, -1.04)(0.6, -1.04)(0.75, -1.3)
\psbezier[linecolor=blue,linewidth=1.5pt]{-}(0.75, 1.3)(0.65, 1.13)(0.27, 1.01)(0., 1.)
\psbezier[linecolor=blue,linewidth=1.5pt]{-}(0., 1.)(-0.6, 1.04)(-1.2, 0.)(-0.87, -0.5)
\psbezier[linecolor=blue,linewidth=1.5pt]{-}(1.5, 0.)(1.3, 0.)(1.01, -0.27)(0.87, -0.5)
\psbezier[linecolor=blue,linewidth=1.5pt]{-}(0.87, -0.5)(0.6, -1.04)(-0.6, -1.04)(-0.87, -0.5)
\psline[linestyle=dashed, dash= 1.5pt 1.5pt,linewidth=0.5pt]{-}(0,-1.5)(0,-0.7)
\end{pspicture}
\ = \
\begin{pspicture}[shift=-1.4](-1.6,-1.5)(1.6,1.5)
\psarc[linecolor=black,linewidth=0.5pt,fillstyle=solid,fillcolor=lightlightblue]{-}(0,0){1.5}{0}{360}
\psarc[linecolor=black,linewidth=0.5pt,fillstyle=solid,fillcolor=white]{-}(0,0){0.5}{0}{360}
\psbezier[linecolor=blue,linewidth=1.5pt]{-}(-1.06066, 1.06066)(-0.848528, 0.948528)(0.848528, 0.948528)(1.06066, 1.06066)
\psbezier[linecolor=blue,linewidth=1.5pt]{-}(-1.06066, -1.06066)(-0.848528, -0.948528)(0.848528, -0.948528)(1.06066, -1.06066)
\psarc[linecolor=blue,linewidth=1.5pt]{-}(0,0){0.7}{0}{360}
\psline[linestyle=dashed, dash= 1.5pt 1.5pt,linewidth=0.5pt]{-}(0,-1.5)(0,-0.5)
\end{pspicture}\ .
\ee
A loop segment is called an {\it arch} if it connects two nodes of the same boundary, and a {\it bridge} if it connects the two boundaries. The elements 
\be
\label{eq:def.cj}
c_{N,j} = \ 
\psset{unit=0.8cm}
\begin{pspicture}[shift=-1.25](-1.2,-1.4)(1.5,1.2)
\psarc[linecolor=black,linewidth=0.5pt,fillstyle=solid,fillcolor=lightlightblue]{-}(0,0){1.2}{0}{360}
\psarc[linecolor=black,linewidth=0.5pt,fillstyle=solid,fillcolor=white]{-}(0,0){0.7}{0}{360}
\psline[linecolor=blue,linewidth=1.5pt]{-}(0.676148, 0.181173)(1.15911, 0.310583)
\psbezier[linecolor=blue,linewidth=1.5pt]{-}(0.494975, 0.494975)(0.671751, 0.671751)(0.245878, 0.91763)(0.181173, 0.676148)
\psline[linecolor=blue,linewidth=1.5pt]{-}(-0.181173, 0.676148)(-0.310583, 1.15911)
\psline[linecolor=blue,linewidth=1.5pt]{-}(-0.494975, 0.494975)(-0.848528, 0.848528)
\psline[linecolor=blue,linewidth=1.5pt]{-}(-0.676148, 0.181173)(-1.15911, 0.310583)
\psline[linecolor=blue,linewidth=1.5pt]{-}(-0.676148, -0.181173)(-1.15911, -0.310583)
\psline[linecolor=blue,linewidth=1.5pt]{-}(-0.494975, -0.494975)(-0.848528, -0.848528)
\psline[linecolor=blue,linewidth=1.5pt]{-}(-0.181173, -0.676148)(-0.310583, -1.15911)
\psline[linecolor=blue,linewidth=1.5pt]{-}(0.181173, -0.676148)(0.310583, -1.15911)
\psline[linecolor=blue,linewidth=1.5pt]{-}(0.494975, -0.494975)(0.848528, -0.848528)
\psline[linecolor=blue,linewidth=1.5pt]{-}(0.676148, -0.181173)(1.15911, -0.310583)
\psline[linestyle=dashed, dash= 1.5pt 1.5pt,linewidth=0.5pt]{-}(0,-1.2)(0,-0.7)
\rput(0.362347, -1.3523){$_1$}
\rput(0.989949, -0.989949){$_2$}
\rput(1.3523, -0.362347){$_{...}$}
\rput(1.57, 0.362347){$_{j-1}$}
\rput(-0.362347, 1.4023){$_{j}$}
\rput(-1.13066, -1.13066){$_{N-3}$}
\rput(-0.362347, -1.4423){$_{N-2}$}
\end{pspicture}\ , 
\qquad 
c_{N,j}^\dagger =\
\psset{unit=0.8cm}
\begin{pspicture}[shift=-1.1](-1.2,-1.4)(1.4,1.2)
\psarc[linecolor=black,linewidth=0.5pt,fillstyle=solid,fillcolor=lightlightblue]{-}(0,0){1.2}{0}{360}
\psarc[linecolor=black,linewidth=0.5pt,fillstyle=solid,fillcolor=white]{-}(0,0){0.7}{0}{360}
\psline[linecolor=blue,linewidth=1.5pt]{-}(0.676148, 0.181173)(1.15911, 0.310583)
\psbezier[linecolor=blue,linewidth=1.5pt]{-}(0.848528, 0.848528)(0.671751, 0.671751)(0.245878, 0.91763)(0.310583, 1.15911)
\psline[linecolor=blue,linewidth=1.5pt]{-}(-0.181173, 0.676148)(-0.310583, 1.15911)
\psline[linecolor=blue,linewidth=1.5pt]{-}(-0.494975, 0.494975)(-0.848528, 0.848528)
\psline[linecolor=blue,linewidth=1.5pt]{-}(-0.676148, 0.181173)(-1.15911, 0.310583)
\psline[linecolor=blue,linewidth=1.5pt]{-}(-0.676148, -0.181173)(-1.15911, -0.310583)
\psline[linecolor=blue,linewidth=1.5pt]{-}(-0.494975, -0.494975)(-0.848528, -0.848528)
\psline[linecolor=blue,linewidth=1.5pt]{-}(-0.181173, -0.676148)(-0.310583, -1.15911)
\psline[linecolor=blue,linewidth=1.5pt]{-}(0.181173, -0.676148)(0.310583, -1.15911)
\psline[linecolor=blue,linewidth=1.5pt]{-}(0.494975, -0.494975)(0.848528, -0.848528)
\psline[linecolor=blue,linewidth=1.5pt]{-}(0.676148, -0.181173)(1.15911, -0.310583)
\psline[linestyle=dashed, dash= 1.5pt 1.5pt,linewidth=0.5pt]{-}(0,-1.2)(0,-0.7)
\rput(0.362347, -1.3523){$_1$}
\rput(0.989949, -0.989949){$_2$}
\rput(1.3523, -0.362347){$_3$}
\rput(1.3523, 0.362347){$_{...}$}
\rput(0.989949, 0.989949){$_j$}
\rput(0.362347, 1.3923){$_{j+1}$}
\rput(-1.13066, -1.13066){$_{N-1}$}
\rput(-0.362347, -1.3523){$_N$}
\end{pspicture}\ ,
\qquad
j = 0, 1,\dots, N-1,
\ee
in $\cL(N-2,N)$ and $\cL(N,N-2)$, respectively, are generators for the diagram spaces. We usually omit the label $N$ and write them as $c_j$ and $c^\dag_j$. These operators satisfy a number of relations, given in \cite[Eq.~(2.6)]{IMD25}. The diagram spaces are then equivalently defined as words $w$ in the generators $c_j$ and $c^\dag_j$, with the numbers $n(w)$ and $n^\dag(w)$ counting the occurences of $c_j$ and $c^\dag_j$ operators in $w$ satisfying $n^\dag(w)-n(w) = \frac12(N-N')$. In particular, for all $N$, the space $\cL(N,N)$ includes the identity diagram~$\id_N$, which we usually write as $\id$.
\medskip

We define four subspaces of $\cL(N,N')$: 
\begin{itemize}
\item[1)] The subspace $\cLk{k}(N,N')$ of $\cL(N,N')$ is spanned by the diagrams with exactly $2k$ bridges, where $k \in \frac12 \Zbb_{\ge 0}$ and $2k$ has the same parity as $N$ and $N'$. 
\item[2)] For $k>0$, the subspace $\cLk{k,0}(N,N')$ of $\cLk{k}(N,N')$ is spanned by the diagrams without bridges crossing the dashed line. For $k=0$, the subspace $\cLk{0,0}(N,N')$ of $\cLk{0}(N,N')$ is spanned by the diagrams without non-contractible loops. 
\item[3)] The subspace $\cL_0(N,N')$ of $\cL(N,N')$ is spanned by the diagrams whose loop segments do not cross the dashed segment. For instance, $\lambda_2$ is in $\cL_0(4,6)$. It is common to draw diagrams in $\cL_0(N,N')$ inside a rectangular box with $N$ and $N'$ nodes on the bottom and top segments. For example, we have
\be
\lambda_2 = \
\begin{pspicture}[shift=-0.6](0,-0.7)(2.4,0.7)
\pspolygon[fillstyle=solid,fillcolor=lightlightblue,linewidth=0pt](0,-0.5)(2.4,-0.5)(2.4,0.5)(0,0.5)
\psarc[linecolor=blue,linewidth=1.5pt]{-}(0.8,0.5){0.2}{180}{0}
\psarc[linecolor=blue,linewidth=1.5pt]{-}(2.0,0.5){0.2}{180}{0}
\psarc[linecolor=blue,linewidth=1.5pt]{-}(1.2,-0.5){0.2}{0}{180}
\psbezier[linecolor=blue,linewidth=1.5pt]{-}(0.6,-0.5)(0.6,0)(0.2,0)(0.2,0.5)
\psbezier[linecolor=blue,linewidth=1.5pt]{-}(1.8,-0.5)(1.8,0)(1.4,0)(1.4,0.5)
\end{pspicture}\ .
\ee
The product $\lambda_1\lambda_2$ of diagrams $\lambda_1 \in \cL_0(N,N')$ and $\lambda_1 \in \cL_0(N',N'')$ is computed by stacking $\lambda_2$ above $\lambda_1$, and produces a diagram in $\cL_0(N,N'')$ times a power of $\beta$ for each loop removed. The diagram space $\cL_0(N,N')$ is equivalently defined as the vector space spanned by words $w$ in the generators $c_j$ and $c^\dag_j$ with $j \ge 1$, and satisfying $n^\dag(w)-n(w) = \frac12(N-N')$.
\item[4)] The subspace $\cL^{\tinyx k}_0(N,N')$ is spanned by the diagrams in $\cL_0(N,N')$ with exactly $2k$ bridges.
\end{itemize}

We parameterise the loop weight $\beta$ as
\be
\beta = -q-q^{-1}\,, \qquad q \in \mathbb C^\times.
\ee
For the applications to the ADE lattice models, we are particularly interested in the case where $q$ is a root of unity, which we parameterise as 
\begin{subequations}
\label{eq:q.beta.standard}
\be
\label{eq:q.standard}
q=\eE^{-\ir\pi p/p'} \qquad \textrm{with} 
\qquad 
p,p' \in \mathbb Z_{>0}\,, 
\qquad 
1\le p<p'\,,
\qquad 
\mathrm{gcd}(p,p') = 1\,,
\ee
so that
\be
\label{eq:beta.standard}
\beta = 2 \cos\bigg(\frac{\pi (p'-p)}{p'}\bigg)\,.
\ee
\end{subequations}
Values of $q$ which are not roots of unity are referred to as {\it generic}.

\subsection{Temperley--Lieb algebras and families of modules}\label{sec:TL}

\paragraph{The algebra $\boldsymbol{\eptl_N(\beta)}$.}
The enlarged periodic Temperley--Lieb algebra is spanned by the diagrams in $\cL(N,N)$. Its generators are
\begin{subequations}
\begin{alignat}{4}
&\Omega = \ 
\psset{unit=0.8cm}
\begin{pspicture}[shift=-1.3](-1.2,-1.4)(1.4,1.2)
\psarc[linecolor=black,linewidth=0.5pt,fillstyle=solid,fillcolor=lightlightblue]{-}(0,0){1.2}{0}{360}
\psarc[linecolor=black,linewidth=0.5pt,fillstyle=solid,fillcolor=white]{-}(0,0){0.7}{0}{360}
\psbezier[linecolor=blue,linewidth=1.5pt]{-}(0.676148, 0.181173)(0.91763, 0.245878)(0.91763, -0.245878)(1.15911, -0.310583)
\psbezier[linecolor=blue,linewidth=1.5pt]{-}(0.494975, 0.494975)(0.671751, 0.671751)(0.91763, 0.245878)(1.15911, 0.310583)
\psbezier[linecolor=blue,linewidth=1.5pt]{-}(0.181173, 0.676148)(0.245878, 0.91763)(0.671751, 0.671751)(0.848528, 0.848528)
\psbezier[linecolor=blue,linewidth=1.5pt]{-}(-0.181173, 0.676148)(-0.245878, 0.91763)(0.245878, 0.91763)(0.310583, 1.15911)
\psbezier[linecolor=blue,linewidth=1.5pt]{-}(-0.494975, 0.494975)(-0.671751, 0.671751)(-0.245878, 0.91763)(-0.310583, 1.15911)
\psbezier[linecolor=blue,linewidth=1.5pt]{-}(-0.676148, 0.181173)(-0.91763, 0.245878)(-0.671751, 0.671751)(-0.848528, 0.848528)
\psbezier[linecolor=blue,linewidth=1.5pt]{-}(-0.676148, -0.181173)(-0.91763, -0.245878)(-0.91763, 0.245878)(-1.15911, 0.310583)
\psbezier[linecolor=blue,linewidth=1.5pt]{-}(-0.494975, -0.494975)(-0.671751, -0.671751)(-0.91763, -0.245878)(-1.15911, -0.310583)
\psbezier[linecolor=blue,linewidth=1.5pt]{-}(-0.181173, -0.676148)(-0.245878, -0.91763)(-0.671751, -0.671751)(-0.848528, -0.848528)
\psbezier[linecolor=blue,linewidth=1.5pt]{-}(0.181173, -0.676148)(0.245878, -0.91763)(-0.245878, -0.91763)(-0.310583, -1.15911)
\psbezier[linecolor=blue,linewidth=1.5pt]{-}(0.494975, -0.494975)(0.671751, -0.671751)(0.245878, -0.91763)(0.310583, -1.15911)
\psbezier[linecolor=blue,linewidth=1.5pt]{-}(0.676148, -0.181173)(0.91763, -0.245878)(0.671751, -0.671751)(0.848528, -0.848528)
\psline[linestyle=dashed, dash= 1.5pt 1.5pt,linewidth=0.5pt]{-}(0,-1.2)(0,-0.7)
\rput(0.362347, -1.3523){$_1$}
\rput(0.989949, -0.989949){$_2$}
\rput(1.3523, -0.362347){$_3$}
\rput(1.3523, 0.362347){$_{...}$}
\rput(-1.13066, -1.13066){$_{N-1}$}
\rput(-0.362347, -1.3523){$_N$}
\end{pspicture} \ \ ,
\qquad
&&\Omega^{-1}\,= \
\psset{unit=0.8cm}
\begin{pspicture}[shift=-1.3](-1.2,-1.4)(1.4,1.2)
\psarc[linecolor=black,linewidth=0.5pt,fillstyle=solid,fillcolor=lightlightblue]{-}(0,0){1.2}{0}{360}
\psarc[linecolor=black,linewidth=0.5pt,fillstyle=solid,fillcolor=white]{-}(0,0){0.7}{0}{360}
\psbezier[linecolor=blue,linewidth=1.5pt]{-}(0.676148, 0.181173)(0.91763, 0.245878)(0.671751, 0.671751)(0.848528, 0.848528)
\psbezier[linecolor=blue,linewidth=1.5pt]{-}(0.494975, 0.494975)(0.671751, 0.671751)(0.245878, 0.91763)(0.310583, 1.15911)
\psbezier[linecolor=blue,linewidth=1.5pt]{-}(0.181173, 0.676148)(0.245878, 0.91763)(-0.245878, 0.91763)(-0.310583, 1.15911)
\psbezier[linecolor=blue,linewidth=1.5pt]{-}(-0.181173, 0.676148)(-0.245878, 0.91763)(-0.671751, 0.671751)(-0.848528, 0.848528)
\psbezier[linecolor=blue,linewidth=1.5pt]{-}(-0.494975, 0.494975)(-0.671751, 0.671751)(-0.91763, 0.245878)(-1.15911, 0.310583)
\psbezier[linecolor=blue,linewidth=1.5pt]{-}(-0.676148, 0.181173)(-0.91763, 0.245878)(-0.91763, -0.245878)(-1.15911, -0.310583)
\psbezier[linecolor=blue,linewidth=1.5pt]{-}(-0.676148, -0.181173)(-0.91763, -0.245878)(-0.671751, -0.671751)(-0.848528, -0.848528)
\psbezier[linecolor=blue,linewidth=1.5pt]{-}(-0.494975, -0.494975)(-0.671751, -0.671751)(-0.245878, -0.91763)(-0.310583, -1.15911)
\psbezier[linecolor=blue,linewidth=1.5pt]{-}(-0.181173, -0.676148)(-0.245878, -0.91763)(0.245878, -0.91763)(0.310583, -1.15911)
\psbezier[linecolor=blue,linewidth=1.5pt]{-}(0.181173, -0.676148)(0.245878, -0.91763)(0.671751, -0.671751)(0.848528, -0.848528)
\psbezier[linecolor=blue,linewidth=1.5pt]{-}(0.494975, -0.494975)(0.671751, -0.671751)(0.91763, -0.245878)(1.15911, -0.310583)
\psbezier[linecolor=blue,linewidth=1.5pt]{-}(0.676148, -0.181173)(0.91763, -0.245878)(0.91763, 0.245878)(1.15911, 0.310583)
\psline[linestyle=dashed, dash= 1.5pt 1.5pt,linewidth=0.5pt]{-}(0,-1.2)(0,-0.7)
\rput(0.362347, -1.3523){$_1$}
\rput(0.989949, -0.989949){$_2$}
\rput(1.3523, -0.362347){$_3$}
\rput(1.3523, 0.362347){$_{...}$}
\rput(-1.13066, -1.13066){$_{N-1}$}
\rput(-0.362347, -1.3523){$_N$}
\end{pspicture}\ \ ,
\qquad 
\id = \ 
\psset{unit=0.8cm}
\begin{pspicture}[shift=-1.3](-1.2,-1.4)(1.4,1.2)
\psarc[linecolor=black,linewidth=0.5pt,fillstyle=solid,fillcolor=lightlightblue]{-}(0,0){1.2}{0}{360}
\psarc[linecolor=black,linewidth=0.5pt,fillstyle=solid,fillcolor=white]{-}(0,0){0.7}{0}{360}
\psline[linecolor=blue,linewidth=1.5pt]{-}(0.676148, 0.181173)(1.15911, 0.310583)
\psline[linecolor=blue,linewidth=1.5pt]{-}(0.494975, 0.494975)(0.848528, 0.848528)
\psline[linecolor=blue,linewidth=1.5pt]{-}(0.181173, 0.676148)(0.310583, 1.15911)
\psline[linecolor=blue,linewidth=1.5pt]{-}(-0.181173, 0.676148)(-0.310583, 1.15911)
\psline[linecolor=blue,linewidth=1.5pt]{-}(-0.494975, 0.494975)(-0.848528, 0.848528)
\psline[linecolor=blue,linewidth=1.5pt]{-}(-0.676148, 0.181173)(-1.15911, 0.310583)
\psline[linecolor=blue,linewidth=1.5pt]{-}(-0.676148, -0.181173)(-1.15911, -0.310583)
\psline[linecolor=blue,linewidth=1.5pt]{-}(-0.494975, -0.494975)(-0.848528, -0.848528)
\psline[linecolor=blue,linewidth=1.5pt]{-}(-0.181173, -0.676148)(-0.310583, -1.15911)
\psline[linecolor=blue,linewidth=1.5pt]{-}(0.181173, -0.676148)(0.310583, -1.15911)
\psline[linecolor=blue,linewidth=1.5pt]{-}(0.494975, -0.494975)(0.848528, -0.848528)
\psline[linecolor=blue,linewidth=1.5pt]{-}(0.676148, -0.181173)(1.15911, -0.310583)
\psline[linestyle=dashed, dash= 1.5pt 1.5pt,linewidth=0.5pt]{-}(0,-1.2)(0,-0.7)
\rput(0.362347, -1.3523){$_1$}
\rput(0.989949, -0.989949){$_2$}
\rput(1.3523, -0.362347){$_3$}
\rput(1.3523, 0.362347){$_{...}$}
\rput(-1.13066, -1.13066){$_{N-1}$}
\rput(-0.362347, -1.3523){$_N$}
\end{pspicture}\ \ ,
\\[0.5cm] 
& e_0\,= \
\psset{unit=0.8cm}
\begin{pspicture}[shift=-1.3](-1.2,-1.4)(1.4,1.2)
\psarc[linecolor=black,linewidth=0.5pt,fillstyle=solid,fillcolor=lightlightblue]{-}(0,0){1.2}{0}{360}
\psarc[linecolor=black,linewidth=0.5pt,fillstyle=solid,fillcolor=white]{-}(0,0){0.7}{0}{360}
\psline[linecolor=blue,linewidth=1.5pt]{-}(0.676148, 0.181173)(1.15911, 0.310583)
\psline[linecolor=blue,linewidth=1.5pt]{-}(0.494975, 0.494975)(0.848528, 0.848528)
\psline[linecolor=blue,linewidth=1.5pt]{-}(0.181173, 0.676148)(0.310583, 1.15911)
\psline[linecolor=blue,linewidth=1.5pt]{-}(-0.181173, 0.676148)(-0.310583, 1.15911)
\psline[linecolor=blue,linewidth=1.5pt]{-}(-0.494975, 0.494975)(-0.848528, 0.848528)
\psline[linecolor=blue,linewidth=1.5pt]{-}(-0.676148, 0.181173)(-1.15911, 0.310583)
\psline[linecolor=blue,linewidth=1.5pt]{-}(-0.676148, -0.181173)(-1.15911, -0.310583)
\psline[linecolor=blue,linewidth=1.5pt]{-}(-0.494975, -0.494975)(-0.848528, -0.848528)
\psbezier[linecolor=blue,linewidth=1.5pt]{-}(-0.181173, -0.676148)(-0.245878, -0.91763)(0.245878, -0.91763)(0.181173, -0.676148)
\psbezier[linecolor=blue,linewidth=1.5pt]{-}(-0.310583, -1.15911)(-0.245878, -0.91763)(0.245878, -0.91763)(0.310583, -1.15911)
\psline[linecolor=blue,linewidth=1.5pt]{-}(0.494975, -0.494975)(0.848528, -0.848528)
\psline[linecolor=blue,linewidth=1.5pt]{-}(0.676148, -0.181173)(1.15911, -0.310583)
\psline[linestyle=dashed, dash= 1.5pt 1.5pt,linewidth=0.5pt]{-}(0,-1.2)(0,-0.7)
\rput(0.362347, -1.3523){$_1$}
\rput(0.989949, -0.989949){$_2$}
\rput(1.3523, -0.362347){$_3$}
\rput(1.3523, 0.362347){$_{...}$}
\rput(-1.13066, -1.13066){$_{N-1}$}
\rput(-0.362347, -1.3523){$_N$}
\end{pspicture}\ \ , \qquad
&&e_j =\
\psset{unit=0.8cm}
\begin{pspicture}[shift=-1.3](-1.2,-1.4)(1.4,1.2)
\psarc[linecolor=black,linewidth=0.5pt,fillstyle=solid,fillcolor=lightlightblue]{-}(0,0){1.2}{0}{360}
\psarc[linecolor=black,linewidth=0.5pt,fillstyle=solid,fillcolor=white]{-}(0,0){0.7}{0}{360}
\psline[linecolor=blue,linewidth=1.5pt]{-}(0.676148, 0.181173)(1.15911, 0.310583)
\psbezier[linecolor=blue,linewidth=1.5pt]{-}(0.494975, 0.494975)(0.671751, 0.671751)(0.245878, 0.91763)(0.181173, 0.676148)
\psbezier[linecolor=blue,linewidth=1.5pt]{-}(0.848528, 0.848528)(0.671751, 0.671751)(0.245878, 0.91763)(0.310583, 1.15911)
\psline[linecolor=blue,linewidth=1.5pt]{-}(-0.181173, 0.676148)(-0.310583, 1.15911)
\psline[linecolor=blue,linewidth=1.5pt]{-}(-0.494975, 0.494975)(-0.848528, 0.848528)
\psline[linecolor=blue,linewidth=1.5pt]{-}(-0.676148, 0.181173)(-1.15911, 0.310583)
\psline[linecolor=blue,linewidth=1.5pt]{-}(-0.676148, -0.181173)(-1.15911, -0.310583)
\psline[linecolor=blue,linewidth=1.5pt]{-}(-0.494975, -0.494975)(-0.848528, -0.848528)
\psline[linecolor=blue,linewidth=1.5pt]{-}(-0.181173, -0.676148)(-0.310583, -1.15911)
\psline[linecolor=blue,linewidth=1.5pt]{-}(0.181173, -0.676148)(0.310583, -1.15911)
\psline[linecolor=blue,linewidth=1.5pt]{-}(0.494975, -0.494975)(0.848528, -0.848528)
\psline[linecolor=blue,linewidth=1.5pt]{-}(0.676148, -0.181173)(1.15911, -0.310583)
\psline[linestyle=dashed, dash= 1.5pt 1.5pt,linewidth=0.5pt]{-}(0,-1.2)(0,-0.7)
\rput(0.362347, -1.3523){$_1$}
\rput(0.989949, -0.989949){$_2$}
\rput(1.3523, -0.362347){$_3$}
\rput(1.3523, 0.362347){$_{...}$}
\rput(0.989949, 0.989949){$_j$}
\rput(0.362347, 1.3523){$_{j+1}$}
\rput(-1.13066, -1.13066){$_{N-1}$}
\rput(-0.362347, -1.3523){$_N$}
\end{pspicture}
\qquad 1\le j \le N-1\,,
\end{alignat}
\end{subequations}
and satisfy the relations
\begin{subequations}
\label{eq:def.EPTL}
\begin{alignat}{4}
\label{eq:def.TL}
& e_j^2 = \beta \, e_j \,, \qquad &&e_j \,e_{j \pm 1 \, } e_j = e_j \,, \qquad &&e_i \,e_j = e_j \,e_i \qquad \text{for } |i-j|>1 \,, \\[0.1cm]
& \Omega \, e_j \, \Omega^{-1} = e_{j-1} \,, \qquad &&\Omega \, \Omega^{-1} = \Omega^{-1} \, \Omega = \id \,, \qquad &&\Omega^2 e_1 = e_{N-1}e_{N-2} \cdots e_2e_1\,,
\end{alignat}
\end{subequations}
for $N>2$. For $N=2$, the algebra is generated by $\Omega$, $\Omega^{-1}$, $e_0$ and $e_1$ satisfying
\be
e_j^2 = \beta \, e_j\,, 
\qquad 
\Omega \, e_j \, \Omega^{-1} = e_{j-1}\,, 
\qquad
\Omega \, \Omega^{-1} = \Omega^{-1} \, \Omega = \id\,, 
\qquad 
\Omega^2 e_j = e_j\,.
\ee
For $N=1$, the algebra is generated by $\Omega$ and $\Omega^{-1}$ satisfying $\Omega \, \Omega^{-1} = \Omega^{-1} \, \Omega = \id$. For $N=0$, the algebra is 
generated by the empty diagram denoted $\id$ and by the diagram $f$ depicted as
\be
\label{eq:def.f}
f = \ 
\psset{unit=0.8cm}
\begin{pspicture}[shift=-1.1](-1.2,-1.2)(1.4,1.2)
\psarc[linecolor=black,linewidth=0.5pt,fillstyle=solid,fillcolor=lightlightblue]{-}(0,0){1.2}{0}{360}
\psarc[linecolor=black,linewidth=0.5pt,fillstyle=solid,fillcolor=white]{-}(0,0){0.7}{0}{360}
\psarc[linecolor=blue,linewidth=1.5pt]{-}(0,0){0.95}{0}{360}
\psline[linestyle=dashed, dash= 1.5pt 1.5pt,linewidth=0.5pt]{-}(0,-1.2)(0,-0.7)
\end{pspicture}\ ,
\ee
with no extra relations. We note that we have the relations
\be
e_j = c^\dag_j c_j\,,
\qquad
\Omega = c_1c^\dag_0\,,
\qquad
\Omega^{-1} = c_0c^\dag_1\,,
\qquad
f=c_0c^\dag_1=c_1c^\dag_0\,.
\ee

\paragraph{The algebra $\boldsymbol{\tl_N(\beta)}$.}

The Temperley--Lieb algebra $\tl_N(\beta)$ is spanned by the diagrams in $\cL_0(N,N)$. It is generated by the diagrams $e_j$ with $j=1,2,\dots, N-1$, satisfying the relations \eqref{eq:def.TL} for $N \ge 2$. For $N=1$, the algebra $\tl_1(\beta)$ is spanned by the identity $\id$ on one strand. For $N=0$, the algebra $\tl_0(\beta)$ is spanned by the empty diagram, which we also denote by $\id$.

\paragraph{The uncoiled algebras $\boldsymbol{\eptl_N(\beta,\gamma)}$, $\boldsymbol{\eptl^{\tinyx 1}_N(\beta,\alpha)}$ and $\boldsymbol{\eptl^{\tinyx 2}_N(\beta,\gamma)}$.}

The uncoiled affine algebras were defined in \cite{LRMD22} as finite-dimensional quotients of the algebras $\eptl_N(\beta)$, whose generators satisfy extra quotient relations. There is one uncoiled affine algebra for $N$ odd, denoted here as $\eptl_N(\beta,\gamma)$, and two such algebras for $N$ even, denoted as $\eptl^{\tinyx 1}_N(\beta,\alpha)$ and $\eptl^{\tinyx 2}_N(\beta,\gamma)$, which differ in the choice of extra relations. These relations are
\begin{subequations}
\begin{alignat}{3}
&\eptl_N(\beta,\gamma): \quad &&\Omega^N = \gamma \id\,,
\\[0.1cm]
&\eptl^{\tinyx 1}_N(\beta,\alpha): \quad &&\Omega^N = \id\,, \quad && E\,\Omega\,E = \alpha E\,,
\\[0.1cm]
&\eptl^{\tinyx 2}_N(\beta,\gamma): \quad &&\Omega^N = \gamma \id\,, \quad && E = 0\,,
\end{alignat}
\end{subequations}
where $\gamma \in \mathbb C^\times$, $\alpha \in \mathbb C$, and
\be
\label{eq:def.E}
E = e_0 e_2 \cdots e_{N-2}\,.
\ee

\paragraph{Transfer matrices.} 
The single-row and double-row transfer matrices $\Tb$ and $\Db$ are respectively the elements of $\eptl_N(\beta)$ and $\tl_N(\beta)$ defined as
\be \label{eq:def.tm}
\Tb = \ 
\begin{pspicture}[shift=-1.1](-1.2,-1.3)(1.4,1.2)
\psarc[linecolor=black,linewidth=0.75pt,fillstyle=solid,fillcolor=lightlightblue]{-}(0,0){1.2}{0}{360}
\psarc[linecolor=black,linewidth=0.75pt,fillstyle=solid,fillcolor=white]{-}(0,0){0.7}{0}{360}
\psline[linewidth=0.75pt]{-}(0.7,0)(1.2,0)
\psline[linewidth=0.75pt]{-}(0.606218, 0.35)(1.03923, 0.6)
\psline[linewidth=0.75pt]{-}(0.35, 0.606218)(0.6, 1.03923)
\psline[linewidth=0.75pt]{-}(0,0.7)(0,1.2)
\psline[linewidth=0.75pt]{-}(-0.35, 0.606218)(-0.6, 1.03923)
\psline[linewidth=0.75pt]{-}(-0.606218, 0.35)(-1.03923, 0.6)
\psline[linewidth=0.75pt]{-}(-0.7,0)(-1.2,0)
\psline[linewidth=0.75pt]{-}(-0.606218, -0.35)(-1.03923, -0.6)
\psline[linewidth=0.75pt]{-}(-0.35, -0.606218)(-0.6, -1.03923)
\psline[linewidth=0.75pt]{-}(0., -0.7)(0., -1.2)
\psline[linewidth=0.75pt]{-}(0.35, -0.606218)(0.6, -1.03923)
\psline[linewidth=0.75pt]{-}(0.606218, -0.35)(1.03923, -0.6)
\rput(0.362347, -1.3523){$_1$}
\rput(0.989949, -0.989949){$_2$}
\rput(1.3523, -0.362347){$_{...}$}
\rput(-1.06066, -1.06066){$_{N-1}$}
\rput(-0.362347, -1.3523){$_N$}
\end{pspicture}\ \ , \qquad
\Db = 
\psset{unit=.7cm}
\begin{pspicture}[shift=-1.15](-0.6,-0.3)(5.6,2.2)
\facegrid{(0,0)}{(5,2)}
\psarc[linewidth=1.5pt,linecolor=blue]{-}(0,1){0.5}{90}{270}
\psarc[linewidth=1.5pt,linecolor=blue]{-}(5,1){0.5}{270}{90}
\rput(0.5,-0.25){$_1$}
\rput(1.5,-0.25){$_2$}
\rput(2.5,-0.25){$_{...}$}
\rput(4.5,-0.25){$_N$}
\end{pspicture} \ ,
\ee
where
\be
\psset{unit=.7cm}
\begin{pspicture}[shift=-.4](1,1)
\facegrid{(0,0)}{(1,1)}
\end{pspicture}
 \ = \
\begin{pspicture}[shift=-.4](1,1)
\facegrid{(0,0)}{(1,1)}
\rput(0,0){\loopa}
\end{pspicture}
 \ + \
\begin{pspicture}[shift=-.4](1,1)
\facegrid{(0,0)}{(1,1)}
\rput(0,0){\loopb}
\end{pspicture}\ \ .
\ee
By expanding each face as the sum of the two tiles, we obtain $\Tb$ and $\Db$ as sums of $2^N$ and $2^{2N}$ diagrams of $\eptl_N(\beta)$ and $\tl_N(\beta)$, respectively.

\paragraph{Families of modules.}
We recall from \cite{IMD25} that a \emph{family of modules} $\repM$ is a set of $\eptl_N(\beta)$-modules
\be
\label{eq:set.of.modules}
\repM = \{\repM(N)\,|\, N = N_0, N_0+2, N_0+4, \dots\}\,,
\qquad N_0 \in \{0,1\}\,,
\ee
endowed with an action $\cL(N,N'): \repM(N') \to \repM(N)$ for all $N,N'$ in $\{N_0, N_0+2, N_0+4, \dots\}$. We refer to $\{N_0, N_0+2, N_0+4, \dots\}$ as the set of {\it admissible} integers for $\repM$. A family is equivalently defined as a set of modules over $\eptl_N(\beta)$ endowed with an action of the generators $c_j$ and $c^\dag_j$ satisfying all the relations in \cite[Eq.~(2.6)]{IMD25}.\footnote{These definitions can be translated in the language of Graham and Lehrer \cite{GL98}. The objects in the even and odd subcategories of the enlarged periodic Temperley--Lieb category are respectively the even and odd non-negative integers. The spaces of morphisms are the diagram spaces $\cL(N,N')$. A family of modules with $N_0=0$ or $N_0=1$ can then be viewed as a functor from the even or odd subcategory to the category of complex vector spaces.}
\medskip 

Similary, a family $\repM$ of $\tl_N(\beta)$-modules is a set of modules, as in \eqref{eq:set.of.modules} but where each $\repM(N)$ is a module over $\tl_N(\beta)$, endowed with an action $\cL_0(N,N'): \repM(N') \to \repM(N)$. Equivalently, a family is a set of modules endowed with an action of the generators $c_j$ and $c^\dag_j$ with $j\ge 1$.\medskip

Let $\repM$ and $\repM'$ be two families of modules over $\eptl_N(\beta)$ or $\tl_N(\beta)$, with the same parity $N_0$. A \emph{family of homomorphisms} $\phi: \repM \to \repM'$ is a set of homomorphisms 
\be
\phi = \{\phi_N:\repM(N) \to \repM'(N)\,|\,N=N_0, N_0+2, N_0+4, \dots\}
\ee
satisfying
\begin{equation} \label{eq:def.family.morphism}
\phi_{N'}(\lambda \cdot u) = \lambda \cdot \phi_N(u)\,,
\end{equation}
for all $\lambda \in \cL(N',N)$ and $\cL_0(N',N)$ for $\eptl_N(\beta)$ and $\tl_N(\beta)$, respectively, and for all $u \in \repM(N)$ and all admissible integers $N$ and $N'$.

\subsection{Jones--Wenzl projectors}

The Jones--Wenzl projectors $P_n$ with $n = 1,2, \dots, N$ for the algebra $\tl_N(\beta)$ are defined recursively as \cite{J83,W88}
\be
\label{eq:WJ.def}
P_1 = \id\,, \qquad
P_{n+1} = P_{n}+\frac{[n]}{[n+1]} \, P_n \, e_n \, P_n\,,
\ee 
where the $q$-numbers are defined as
\be
[n] = \frac{q^n - q^{-n}}{q-q^{-1}}\,.
\ee
We depict the projector $P_n$ as the pink triangle
\be
P_n = \ 
\begin{pspicture}[shift=-0.05](0,0)(1.5,0.3)
\pspolygon[fillstyle=solid,fillcolor=pink](0,0)(1.5,0)(1.5,0.3)(0,0.3)(0,0)\rput(0.75,0.15){$_{n}$}
\end{pspicture} \ \, .
\ee
These projectors satisfy the relations
\be
\label{eq:WJ.relations}
c_n P_n c^\dag_n = -\frac{[n+1]}{[n]} P_{n-1}\,, \qquad
c_j P_n = P_n c^\dag_j = 0 
\qquad
\textrm{for}
\qquad
j=1, 2,\dots, n-1.
\ee
The algebra $\tl_N(\beta)$ has the one-dimensional module $\repV_{N/2}(N)$ described in \cref{sec:Vk}, and $P_{N}$ projects on this module. 
\medskip

For the uncoiled algebras, one can define the analogous Jones--Wenzl projectors $\wh P_{N,x}$ as\cite{LRMD22}
\begin{subequations}
\label{eq:hatP}
\begin{alignat}{2}
\eptl_N(\beta,\gamma)&: \quad
\wh P_{N,x} = \sum_{s=0}^{\frac{N-1}2} \sum_{\ell = 0}^{N-2s-1} \Gamma_{s,\ell} \, Z_{s,\ell}\,,
\\[0.15cm]
\label{eq:hatP1}
\eptl^{\tinyx 1}_N(\beta,\alpha)&: \quad
\wh P_{N,x} = \bigg[\sum_{s=0}^{\frac{N-2}2} \sum_{\ell = 0}^{N-2s-1} \Gamma_{s,\ell} \, Z_{s,\ell}\bigg]+\Gamma_{N/2,0} Z_{N/2,0}\,,
\\[0.15cm]
\eptl^{\tinyx 2}_N(\beta,\gamma)&: \quad
\wh P_{N,x} = \sum_{s=0}^{\frac{N-2}2} \sum_{\ell = 0}^{N-2s-1} \Gamma_{s,\ell} \, Z_{s,\ell}\,,
\end{alignat}
\end{subequations}
where
\be
Z_{s,\ell} = P_N\,(c_0^\dag)^s\, \Omega^\ell\, (c_0)^s\,P_N\,,
\ee
and the constants $\Gamma_{s,\ell}$ depend on $x$ and are given in \cref{app:GammasL}. These projectors satisfy the relations 
\be
\Omega\, \wh P_{N,x} = x\, \wh P_{N,x}\,,
\qquad
c_j \wh P_{N,x} = \wh P_{N,x} c^\dag_j = 0, \qquad j = 0, 1, \dots, N-1\,.
\ee
The algebra $\eptl_N(\beta)$ has a one-parameter family of one-dimensional modules $\repW_{N/2,x}(N)$, described in \cref{sec:W.and.Q}. The uncoiled algebras with $N \ge 1$ then have $N$ such one-dimensional modules $\repW_{N/2,x}(N)$, with 
\be
\label{eq:x.values}
x = \eE^{2 \pi \ir r / N} \times
\left\{\begin{array}{cl}
\gamma^{1/N} & \eptl_N(\beta,\gamma)\,,\, \eptl^{\tinyx 2}_N(\beta,\alpha)\,,\\[0.15cm]
1 & \eptl^{\tinyx 1}_N(\beta,\alpha)\,,
\end{array}\right.
\quad r \in\{0,1,\dots, N-1\}\,,
\ee 
and $\wh P_{N,x}$ projects on the module $\repW_{N/2,x}(N)$.

\subsection[Link modules over $\tl_N(\beta)$]{Link modules over $\boldsymbol{\tl_N(\beta)}$}
\label{sec:Vk}

The standard modules $\repV_k(N)$ over $\tl_N(\beta)$ are defined for $k \in \frac12\mathbb Z_{\ge 0}$ and $N \in \mathbb Z_{\ge 0}$ satisfying $\frac N2-k \in \mathbb Z$. For $N<2k$, $\repV_k(N)$ is the trivial zero module. For $N\ge 2k$, $\repV_k(N)$ is spanned by link states drawn over a horizontal segment with $N$ nodes and $2k$ defects. The non-intersecting loop segments connect nodes pairwise or are attached to defects. For instance, the bases for $\repV_0(4)$, $\repV_1(4)$ and $\repV_2(4)$ are
\begin{subequations}
\begin{alignat}{2}
&\repV_{0}(4): \quad \Big\{\
\begin{pspicture}[shift=-0.05](0,0)(1.6,0.6)
\psline[linewidth=0.5pt](0,0)(1.6,0)
\psarc[linecolor=blue,linewidth=1.5pt]{-}(0.4,0){0.2}{0}{180}
\psarc[linecolor=blue,linewidth=1.5pt]{-}(1.2,0){0.2}{0}{180}
\end{pspicture}\ ,
\quad
\begin{pspicture}[shift=-0.05](0,0)(1.6,0.6)
\psline[linewidth=0.5pt](0,0)(1.6,0)
\psarc[linecolor=blue,linewidth=1.5pt]{-}(0.8,0){0.2}{0}{180}
\psbezier[linecolor=blue,linewidth=1.5pt]{-}(0.2,0)(0.2,0.5)(1.4,0.5)(1.4,0)
\end{pspicture}\ \Big\}\ ,
\\[0.15cm]& 
\repV_{1}(4): \quad \Big\{\
\begin{pspicture}[shift=-0.05](0,0)(1.6,0.6)
\psline[linewidth=0.5pt](0,0)(1.6,0)
\psline[linecolor=blue,linewidth=1.5pt]{-}(0.2,0)(0.2,0.5)
\psline[linecolor=blue,linewidth=1.5pt]{-}(0.6,0)(0.6,0.5)
\psarc[linecolor=blue,linewidth=1.5pt]{-}(1.2,0){0.2}{0}{180}
\end{pspicture}\ , 
\quad
\begin{pspicture}[shift=-0.05](0,0)(1.6,0.6)
\psline[linewidth=0.5pt](0,0)(1.6,0)
\psline[linecolor=blue,linewidth=1.5pt]{-}(0.2,0)(0.2,0.5)
\psline[linecolor=blue,linewidth=1.5pt]{-}(1.4,0)(1.4,0.5)
\psarc[linecolor=blue,linewidth=1.5pt]{-}(0.8,0){0.2}{0}{180}
\end{pspicture}\ , 
\quad
\begin{pspicture}[shift=-0.05](0,0)(1.6,0.6)
\psline[linewidth=0.5pt](0,0)(1.6,0)
\psline[linecolor=blue,linewidth=1.5pt]{-}(1.0,0)(1.0,0.5)
\psline[linecolor=blue,linewidth=1.5pt]{-}(1.4,0)(1.4,0.5)
\psarc[linecolor=blue,linewidth=1.5pt]{-}(0.4,0){0.2}{0}{180}
\end{pspicture}\ \Big\} \ ,
\\[0.15cm]&
\repV_{2}(4): \quad \Big\{\
\begin{pspicture}[shift=-0.05](0,0)(1.6,0.6)
\psline[linewidth=0.5pt](0,0)(1.6,0)
\psline[linecolor=blue,linewidth=1.5pt]{-}(0.2,0)(0.2,0.5)
\psline[linecolor=blue,linewidth=1.5pt]{-}(0.6,0)(0.6,0.5)
\psline[linecolor=blue,linewidth=1.5pt]{-}(1.0,0)(1.0,0.5)
\psline[linecolor=blue,linewidth=1.5pt]{-}(1.4,0)(1.4,0.5)
\end{pspicture}\ \Big\} \ .
\end{alignat}
\end{subequations}
The dimension of these modules is
\be
\label{eq:dimV}
\dim \repV_{k}(N) = d_k(N)-d_{k+1}(N)\,,
\qquad d_k(N) = 
\begin{pmatrix}N\\\frac N2-k\end{pmatrix}, 
\ee
where we use the convention that $\left(\begin{smallmatrix}n\\m\end{smallmatrix}\right) = 0$ for $m<0$ or $m>n$. The action of a diagram $\lambda \in \cL_0(N,N')$ on a link state $u \in \repV_k(N')$ is obtained by drawing $u$ above $\lambda$ and reading the new link state from the bottom nodes. Each closed loop is erased and replaced by a factor $\beta$. For $k>0$, the result is set to zero if the resulting state has less than $2k$ defects. For example,
\be
e_2 \cdot 
\begin{pspicture}[shift=-0.05](0,0)(1.6,0.6)
\psline[linewidth=0.5pt](0,0)(1.6,0)
\psarc[linecolor=blue,linewidth=1.5pt]{-}(0.8,0){0.2}{0}{180}
\psbezier[linecolor=blue,linewidth=1.5pt]{-}(0.2,0)(0.2,0.5)(1.4,0.5)(1.4,0)
\end{pspicture}
\ = \beta \
\begin{pspicture}[shift=-0.05](0,0)(1.6,0.6)
\psline[linewidth=0.5pt](0,0)(1.6,0)
\psarc[linecolor=blue,linewidth=1.5pt]{-}(0.8,0){0.2}{0}{180}
\psbezier[linecolor=blue,linewidth=1.5pt]{-}(0.2,0)(0.2,0.5)(1.4,0.5)(1.4,0)
\end{pspicture}\ ,
\qquad
e_3 \cdot 
\begin{pspicture}[shift=-0.05](0,0)(1.6,0.6)
\psline[linewidth=0.5pt](0,0)(1.6,0)
\psline[linecolor=blue,linewidth=1.5pt]{-}(1.0,0)(1.0,0.5)
\psline[linecolor=blue,linewidth=1.5pt]{-}(1.4,0)(1.4,0.5)
\psarc[linecolor=blue,linewidth=1.5pt]{-}(0.4,0){0.2}{0}{180}
\end{pspicture} 
\ = 0\,.
\ee

Graham and Lehrer \cite{GL98} proved that the set
\be
\repV_{k} = \big\{\repV_{k}(N)\,|\, N=N_0, N_0+2, N_0+4, \dots\big\}\,,
\qquad 
N_0 = \left\{
\begin{array}{ll}
0 & k \in \mathbb Z\,,\\[0.1cm]
1 & k \in \mathbb Z + \frac12\,
\end{array}
\right.
\ee 
is a family of $\tl_N(\beta)$-modules, and that its modules are indecomposable for all $q\in\Cbb^\times$ and $k \in \frac12 \mathbb Z_{\ge 0}$. Moreover, we have
\begin{equation}
\repV_k(N) = \big\{
\lambda \cdot u_k \,|\, \lambda\in\cL_0(N,2k) 
\big\} \,,
\end{equation}
for all $q \in \mathbb C$, where $u_k$ is the unique link state of $\repV_k(2k)$. We write this as
\be
\repV_k(N) = \cL_0(N,2k) \cdot u_k\,.
\ee

\begin{Proposition}
\label{prop:seed.Vk}
Let $q \in \mathbb C^\times$ and $k \in \frac12 \mathbb Z_{\ge 0}$. Let also $\repM$ be a family of modules over $\tl_N(\beta)$, and $\xi$~be a nonzero element in $\repM(2k)$ satisfying
\begin{equation} \label{eq:insertion.state.Vk}
c_j\cdot \xi=0 \qquad \textrm{for}\qquad j=1,2, \dots, 2k-1.
\end{equation}
The linear map $\phi_\xi: \repV_k(N) \to \repM(N)$ defined as
\begin{equation}
\phi_\xi: \lambda\cdot u_k \mapsto \lambda \cdot \xi \,,
\qquad \lambda \in\cL_0(N,2k)
\end{equation}
is well-defined for all admissible $N$, and defines a family of homomorphisms from $\repV_k$ to $\repM$.
\end{Proposition}
\proof
Let $\lambda$ be a diagram in $\cL_0(N,2k)$. The product $\lambda \cdot u_k$ vanishes if and only if $\lambda$ is of the form $\lambda=\lambda'c_j$, for some $\lambda'\in\cL_0(N,2k-2)$ and some index $j\in\{1,2,\dots,2k-1\}$. Let $\lambda_1,\lambda_2\in\cL_0(N,2k)$ be such that $\lambda_1\cdot u_k=\lambda_2\cdot u_k$. Then the difference $(\lambda_1-\lambda_2)$ can be written as $\lambda_1-\lambda_2=\lambda'c_j$, for some $\lambda'\in\cL_0(N,2k-2)$ and $j\in\{1,2,\dots,2k-1\}$. Therefore, we have
\begin{equation}
\phi_\xi(\lambda_1\cdot u_k)-\phi_\xi(\lambda_2\cdot u_k)
= \phi_\xi\big((\lambda_1-\lambda_2)\cdot u_k \big)
= \phi_\xi(\lambda'c_j\cdot u_k) = 0 \,.
\end{equation}
This shows that the map $\phi_\xi$ is well-defined on $\repV_k(N)$. Moreover, let $u=\lambda\cdot u_k$ with $\lambda\in\cL_0(N,2k)$. For all $\lambda'\in\cL_0(N',N)$, we have
\begin{equation}
\phi_\xi(\lambda'\cdot u) = \phi_\xi(\lambda'\lambda\cdot u_k) = \lambda'\lambda\cdot \xi
= \lambda'\cdot \phi_\xi(u) \,,
\end{equation}
thus proving that $\phi_\xi$ defines a family of homomorphisms from $\repV_k$ to $\repM$.
\eproof

Let $\repM$ be a family of modules over $\tl_N(\beta)$, and $k \in \frac12 \mathbb Z_{\ge 0}$.
We say that a state $\xi\in\repM(2k)$ satisfying the property \eqref{eq:insertion.state.Vk} is an \emph{insertion state with $2k$ defects for $\repM$}. The module $\repM(2k)$ may in fact have more than one insertion state with $2k$ defects. The vector space spanned by these insertion states is a submodule of $\repM(2k)$, which we call the \emph{insertion space with $2k$ defects of $\repM$}.
\medskip

In the rest of this section, we consider the values of $q$ relevant for the critical ADE models, namely the roots of unity $q=\eE^{-\ir\pi p/p'}$ with $(p,p')$ as in \eqref{eq:q.standard}.
The structure of the families $\repV_k$ is as follows \cite{GL98,RSA14}:
\begin{itemize}
\item[(i)] For $2k\equiv p'-1 \mod p'$, $\repV_k$ is irreducible. We write this as $\repV_k\simeq\repQ_k$.
\item[(ii)] For $2k\not\equiv p'-1 \mod p'$, we write $r_k=\lfloor2k/p'\rfloor+1$, so that $(r_k-1)p'\leq 2k\leq r_kp'-2$. In this case, $\repV_k$ is indecomposable with two factors: a subfamily $\repR_k$ and a quotient family $\repQ_k = \repV_k/\repR_k$. Moreover, the subfamily $\repR_k$ is isomorphic to $\repQ_{r_kp'-1-k}$. The Loewy diagram for $\repV_k$ reads
\begin{equation} \label{eq:Loewy.Vk}
\repV_k\simeq [\repQ_k\to \repQ_{r_k p'-1-k}] \,.
\end{equation}
\end{itemize}

To complement this description, we give below some useful results about the modules $\repQ_k$ and~$\repR_k$. We first discuss the Jones--Wenzl projectors at roots of unity. For $q = \eE^{- \ir \pi p/p'}$ with $p$ and $p'$ as in~\eqref{eq:q.beta.standard}, the projectors $P_n$ with $1 \le n \le p'-1$ are well-defined. In contrast, the projector $P_{p'}$ is singular for $q \to \eE^{- \ir \pi p/p'}$, since the recursion relation \eqref{eq:WJ.def} for $n=p'-1$ contains a singularity. It is thus natural to wonder whether the projectors $P_n$ for $n>p'$ are also singular in the limit $q \to \eE^{- \ir \pi p/p'}$.

\begin{Proposition} \label{prop:Pnp'-1}
Let $p$ and $p'$ be as in \eqref{eq:q.standard}, and $n \in \mathbb Z_{\ge 1}$. Then the limit $q \to \eE^{-\ir \pi p/p'}$ of the projector $P_{np'-1}$ in $\tl_{np'-1}(\beta)$ exists.
\end{Proposition}
\proof We only sketch the proof, which was discussed in \cite{GST14} and uses the known representation theory of $\tl_N(\beta)$ at roots of unity. At $q = \eE^{-\ir \pi p/p'}$, the module $\repV_{(np'-1)/2}(N)$ is irreducible, as the determinant of its Gram matrix is nonzero. Moreover, the only non-trivial homomorphisms from $\repV_{(np'-1)/2}(N)$ to other indecomposable modules $\repM(N)$ are isomorphisms. Let us then suppose that $P_{np'-1}$ has a Laurent expansion around $q = \eE^{-\ir \pi p/p'}$ with a lowest order term $C(q-\eE^{-\ir \pi p/p'})^\ell$ for some $\ell<0$. The coefficient~$C$ is an element of $\tl_{np'-1}(\beta)$ with the following properties: (i) $c_j C = C c^\dag_j = 0$ for $j=1,2, \dots, np'-2$, and (ii) $C^2 = 0$, implying that $C$ does not have a component along the identity $\id$. We can then use $C$ to construct non-invertible homomorphisms from $\repV_{(np'-1)/2}(N)$ to other modules $\repM(N)$ that are not isomorphic to $\repV_{(np'-1)/2}(N)$. By contradiction, we deduce that $\ell$ cannot be smaller than zero.
\eproof

\begin{Proposition} \label{prop:Qk.Rk}
Let $q=\eE^{-\ir\pi p/p'}$ with $(p,p')$ as in \eqref{eq:q.standard}, $k\in\frac12\Zbb_{\geq 0}$ be such that $2k\not\equiv p'-1 \mod p'$, $r_k=\lfloor2k/p'\rfloor+1$, and $N \in \mathbb Z_{\ge 0}$. Then:
\begin{enumerate}
\item[(i)] The dimension of $\repQ_k(N)$ is given by
\be \label{eq:QTL.dim}
\dim \repQ_k(N) = D_k(N) - D_{k+1}(N) - D_{r_k p'-1-k}(N) + D_{r_kp'-k}(N)\,,
\ee
where
\be \label{eq:DkN}
D_k(N) = \sum_{j \ge 0} d_{k+p'j}(N)\,,
\ee
and $d_k(N)$ is defined in \eqref{eq:dimV}.
\item[(ii)] Let $k'=r_k p'-1-k$, and $v_k\in \repV_k(2k')$ be the state defined as
\be
\label{eq:vk}
v_k = (P_{k'+k} \otimes \id_{k'-k}) c^\dag_{k'+k}c^\dag_{k'+k-1} \cdots c^\dag_{2k+1} \cdot u_k = \
\begin{pspicture}[shift=-0.7](0,-0.7)(4.4,1.4)
\psline[linewidth=0.5pt](0,0)(4.4,0)
\psline[linecolor=blue,linewidth=1.5pt]{-}(0.2,0)(0.2,0.9)
\psline[linecolor=blue,linewidth=1.5pt]{-}(0.6,0)(0.6,0.9)
\psline[linecolor=blue,linewidth=1.5pt]{-}(1.0,0)(1.0,0.9)
\psarc[linecolor=blue,linewidth=1.5pt]{-}(2.8,0){0.2}{0}{180}
\psbezier[linecolor=blue,linewidth=1.5pt]{-}(2.2,0)(2.2,0.6)(3.4,0.6)(3.4,0)
\psbezier[linecolor=blue,linewidth=1.5pt]{-}(1.8,0)(1.8,0.9)(3.8,0.9)(3.8,0)
\psbezier[linecolor=blue,linewidth=1.5pt]{-}(1.4,0)(1.4,1.2)(4.2,1.2)(4.2,0)
\pspolygon[fillstyle=solid,fillcolor=pink](0,0)(2.8,0)(2.8,-0.3)(0,-0.3)(0,0)\rput(1.4,-0.15){$_{k'+k}$}
\rput(3.6,-0.4){$\underbrace{\hspace{1.35cm}}_{k'-k}$}
\rput(0.6,1.25){$\overbrace{\hspace{1.0cm}}^{2k}$}
\end{pspicture}\ \ .
\ee
The submodule $\repR_k(N)$ is given by
\be
\repR_k(N) = \cL_0(N,2k') \cdot v_k \,,
\ee
and the subfamily $\repQ_k$ is equivalently defined as
\be
\repQ_k = \repV_k \big/ \{v_k \equiv 0 \}.
\ee
\end{enumerate}
\end{Proposition}
\proof
We prove \eqref{eq:QTL.dim} by induction over decreasing values of $r_k$. We have
\begin{equation}
(r_k-1)p' \leq 2k \leq r_kp'-2 \,, \qquad r_kp' \leq 2k' \leq (r_k+1)p'-2 \,,
\end{equation}
and thus $r_{k'}=r_k+1$.
For $r_k>1+\frac N{p'}$, both $k$ and $k'$ are strictly greater than $\frac N2$, and hence \eqref{eq:QTL.dim} is trivially satisfied with $\dim \repQ_k(N)= 0$. Let $k$ be such that $1\leq r_k\leq 1+\frac N{p'}$. The inductive assumption is that \eqref{eq:QTL.dim} is satisfied for $\dim \repQ_{k'}(N)$ for each $k'$ with $r_{k'}=r_k+1$, namely
\be 
\dim \repQ_{k'}(N) = D_{k'}(N) - D_{k'+1}(N) - D_{k''}(N) + D_{k''+1}(N)\,,
\ee
where $k''=(r_k+1)p'-1-k'=p'+k$. We then obtain
\begin{alignat}{1}
\dim \repQ_k(N) &=\dim \repV_k(N) - \dim \repQ_{k'}(N) \nn\\[0.1cm]
&= \big(d_k(N)+D_{p'+k}(N)\big)-\big(d_{k+1}(N)+D_{p'+k+1}(N)\big)- D_{k'}(N) + D_{k'+1}(N) \,, \nn\\[0.1cm]
&= D_k(N)-D_{k+1}(N)- D_{k'}(N) + D_{k'+1}(N) \,,
\end{alignat}
which proves \eqref{eq:QTL.dim} for $\dim \repQ_k(N)$.\medskip

For the second property, we note that $k'+k=r_kp'-1$. Using \cref{prop:Pnp'-1}, we deduce that $P_{k'+k}$ is well-defined for $q\to \eE^{-\ir \pi p/p'}$. Using \eqref{eq:WJ.relations}, we find that $v_k$ is an insertion state with $2k'$ defects. Let us denote by $\varphi_k$ the corresponding family of homomorphisms from $\repV_{k'}$ to $\repV_k$. By construction, we have $\varphi_k\big(\repV_{k'}(N)\big)=\cL_0(N,2k')\cdot v_k$. Moreover, $\varphi_k(\repV_{k'})\simeq \repV_{k'}\big/\mathrm{Ker}\,\varphi_k$, and hence $\varphi_k(\repV_{k'})$ is isomorphic to a nonzero quotient of $\repV_{k'}$. From \eqref{eq:Loewy.Vk}, we conclude that $\varphi_k(\repV_{k'})$ is the unique subfamily of $\repV_k$ isomorphic to $\repQ_{k'}$, which proves that $\varphi_k(\repV_{k'})=\repR_k$.
\eproof

\subsection[Link modules over $\eptl_N(\beta)$]{Link modules over $\boldsymbol{\eptl_N(\beta)}$}
\label{sec:W.and.Q}

We review the definition and properties of the standard modules $\repW_{k,x}(N)$ over $\eptl_N(\beta)$. Let $x \in \mathbb C^\times$, $k \in \frac12\mathbb Z_{\ge 0}$ and $N \in \mathbb N_{\ge 0}$ satisfying $\frac N2-k \in \mathbb Z$. For $N<2k$, $\repW_{k,x}(N)$ is the trivial zero module. For $N\geq 2k$, $\repW_{k,x}(N)$ is spanned by link states with $2k$ defect. In this context, a link state is a diagram drawn on a disc with $N$~nodes on its boundary. A dashed segment is drawn connecting the marked point and the boundary of the disc. The link state has non-intersecting loop segments, which either tie the nodes pairwise, or are {\it defects}, namely loop segments connecting a node to the marked point without intersecting the dashed segment. To illustrate, here are the bases for $\repW_{0,x}(4)$, $\repW_{1,x}(4)$ and $\repW_{2,x}(4)$:
\begin{subequations}
\begin{alignat}{2}
&\repW_{0,x}(4): \quad
\begin{pspicture}[shift=-0.6](-0.7,-0.7)(0.7,0.7)
\psarc[linecolor=black,linewidth=0.5pt,fillstyle=solid,fillcolor=lightlightblue]{-}(0,0){0.7}{0}{360}
\psline[linestyle=dashed, dash= 1.5pt 1.5pt, linewidth=0.5pt]{-}(0,0)(0,-0.7)
\psarc[linecolor=black,linewidth=0.5pt,fillstyle=solid,fillcolor=darkgreen]{-}(0,0){0.07}{0}{360}
\psbezier[linecolor=blue,linewidth=1.5pt]{-}(0.494975, 0.494975)(0.20,0.20)(0.20,-0.20)(0.494975, -0.494975)
\psbezier[linecolor=blue,linewidth=1.5pt]{-}(-0.494975, 0.494975)(-0.20,0.20)(-0.20,-0.20)(-0.494975, -0.494975)
\end{pspicture}\ ,
\quad
\begin{pspicture}[shift=-0.6](-0.7,-0.7)(0.7,0.7)
\psarc[linecolor=black,linewidth=0.5pt,fillstyle=solid,fillcolor=lightlightblue]{-}(0,0){0.7}{0}{360}
\psline[linestyle=dashed, dash=1.5pt 1.5pt, linewidth=0.5pt]{-}(-0.5,0)(0,-0.7)
\psarc[linecolor=black,linewidth=0.5pt,fillstyle=solid,fillcolor=darkgreen]{-}(-0.5,0){0.07}{0}{360}
\psbezier[linecolor=blue,linewidth=1.5pt]{-}(0.494975, 0.494975)(0.20,0.20)(0.20,-0.20)(0.494975, -0.494975)
\psbezier[linecolor=blue,linewidth=1.5pt]{-}(-0.494975, 0.494975)(-0.20,0.20)(-0.20,-0.20)(-0.494975, -0.494975)
\end{pspicture}\ ,
\quad
\begin{pspicture}[shift=-0.6](-0.7,-0.7)(0.7,0.7)
\psarc[linecolor=black,linewidth=0.5pt,fillstyle=solid,fillcolor=lightlightblue]{-}(0,0){0.7}{0}{360}
\psline[linestyle=dashed, dash=1.5pt 1.5pt, linewidth=0.5pt]{-}(0.5,0)(0,-0.7)
\psarc[linecolor=black,linewidth=0.5pt,fillstyle=solid,fillcolor=darkgreen]{-}(0.5,0){0.07}{0}{360}
\psbezier[linecolor=blue,linewidth=1.5pt]{-}(0.494975, 0.494975)(0.20,0.20)(0.20,-0.20)(0.494975, -0.494975)
\psbezier[linecolor=blue,linewidth=1.5pt]{-}(-0.494975, 0.494975)(-0.20,0.20)(-0.20,-0.20)(-0.494975, -0.494975)
\end{pspicture}\ ,
\quad
\begin{pspicture}[shift=-0.6](-0.7,-0.7)(0.7,0.7)
\psarc[linecolor=black,linewidth=0.5pt,fillstyle=solid,fillcolor=lightlightblue]{-}(0,0){0.7}{0}{360}
\psline[linestyle=dashed, dash= 1.5pt 1.5pt, linewidth=0.5pt]{-}(0,0)(0,-0.7)
\psarc[linecolor=black,linewidth=0.5pt,fillstyle=solid,fillcolor=darkgreen]{-}(0,0){0.07}{0}{360}
\psbezier[linecolor=blue,linewidth=1.5pt]{-}(0.494975, 0.494975)(0.20,0.20)(-0.20,0.20)(-0.494975, 0.494975)
\psbezier[linecolor=blue,linewidth=1.5pt]{-}(0.494975, -0.494975)(0.20,-0.20)(-0.20,-0.20)(-0.494975, -0.494975)
\end{pspicture}\ ,
\quad
\begin{pspicture}[shift=-0.6](-0.7,-0.7)(0.7,0.7)
\psarc[linecolor=black,linewidth=0.5pt,fillstyle=solid,fillcolor=lightlightblue]{-}(0,0){0.7}{0}{360}
\psline[linestyle=dashed, dash=1.5pt 1.5pt, linewidth=0.5pt]{-}(0,-0.45)(0,-0.7)
\psarc[linecolor=black,linewidth=0.5pt,fillstyle=solid,fillcolor=darkgreen]{-}(0,-0.5){0.07}{0}{360}
\psbezier[linecolor=blue,linewidth=1.5pt]{-}(0.494975, 0.494975)(0.20,0.20)(-0.20,0.20)(-0.494975, 0.494975)
\psbezier[linecolor=blue,linewidth=1.5pt]{-}(0.494975, -0.494975)(0.20,-0.20)(-0.20,-0.20)(-0.494975, -0.494975)
\end{pspicture}\ ,
\quad
\begin{pspicture}[shift=-0.6](-0.7,-0.7)(0.7,0.7)
\psarc[linecolor=black,linewidth=0.5pt,fillstyle=solid,fillcolor=lightlightblue]{-}(0,0){0.7}{0}{360}
\psline[linestyle=dashed, dash=1.5pt 1.5pt, linewidth=0.5pt]{-}(0,0.5)(0,-0.7)
\psarc[linecolor=black,linewidth=0.5pt,fillstyle=solid,fillcolor=darkgreen]{-}(0,0.5){0.07}{0}{360}
\psbezier[linecolor=blue,linewidth=1.5pt]{-}(0.494975, 0.494975)(0.20,0.20)(-0.20,0.20)(-0.494975, 0.494975)
\psbezier[linecolor=blue,linewidth=1.5pt]{-}(0.494975, -0.494975)(0.20,-0.20)(-0.20,-0.20)(-0.494975, -0.494975)
\end{pspicture}\ ,
\\[0.2cm]
\label{eq:W14.basis}
&\repW_{1,x}(4): \quad
\begin{pspicture}[shift=-0.6](-0.7,-0.7)(0.7,0.7)
\psarc[linecolor=black,linewidth=0.5pt,fillstyle=solid,fillcolor=lightlightblue]{-}(0,0){0.7}{0}{360}
\psline[linestyle=dashed, dash= 1.5pt 1.5pt, linewidth=0.5pt]{-}(0,0)(0,-0.7)
\psbezier[linecolor=blue,linewidth=1.5pt]{-}(0.494975, 0.494975)(0.20,0.20)(0.20,-0.20)(0.494975, -0.494975)
\psline[linecolor=blue,linewidth=1.5pt]{-}(-0.494975, 0.494975)(0,0)
\psline[linecolor=blue,linewidth=1.5pt]{-}(-0.494975, -0.494975)(0,0)
\psarc[linecolor=black,linewidth=0.5pt,fillstyle=solid,fillcolor=darkgreen]{-}(0,0){0.09}{0}{360}
\end{pspicture}\ ,
\quad
\begin{pspicture}[shift=-0.6](-0.7,-0.7)(0.7,0.7)
\psarc[linecolor=black,linewidth=0.5pt,fillstyle=solid,fillcolor=lightlightblue]{-}(0,0){0.7}{0}{360}
\psline[linestyle=dashed, dash= 1.5pt 1.5pt, linewidth=0.5pt]{-}(0,0)(0,-0.7)
\psbezier[linecolor=blue,linewidth=1.5pt]{-}(0.494975, 0.494975)(0.20,0.20)(-0.20,0.20)(-0.494975, 0.494975)
\psline[linecolor=blue,linewidth=1.5pt]{-}(0.494975, -0.494975)(0,0)
\psline[linecolor=blue,linewidth=1.5pt]{-}(-0.494975, -0.494975)(0,0)
\psarc[linecolor=black,linewidth=0.5pt,fillstyle=solid,fillcolor=darkgreen]{-}(0,0){0.09}{0}{360}
\end{pspicture}\ ,
\quad
\begin{pspicture}[shift=-0.6](-0.7,-0.7)(0.7,0.7)
\psarc[linecolor=black,linewidth=0.5pt,fillstyle=solid,fillcolor=lightlightblue]{-}(0,0){0.7}{0}{360}
\psline[linestyle=dashed, dash= 1.5pt 1.5pt, linewidth=0.5pt]{-}(0,0)(0,-0.7)
\psbezier[linecolor=blue,linewidth=1.5pt]{-}(-0.494975, -0.494975)(-0.20,-0.20)(-0.20,0.20)(-0.494975, 0.494975)
\psline[linecolor=blue,linewidth=1.5pt]{-}(0.494975, -0.494975)(0,0)
\psline[linecolor=blue,linewidth=1.5pt]{-}(0.494975, 0.494975)(0,0)
\psarc[linecolor=black,linewidth=0.5pt,fillstyle=solid,fillcolor=darkgreen]{-}(0,0){0.09}{0}{360}
\end{pspicture}\ ,
\quad
\begin{pspicture}[shift=-0.6](-0.7,-0.7)(0.7,0.7)
\psarc[linecolor=black,linewidth=0.5pt,fillstyle=solid,fillcolor=lightlightblue]{-}(0,0){0.7}{0}{360}
\psline[linestyle=dashed, dash= 1.5pt 1.5pt, linewidth=0.5pt]{-}(0,0)(0,-0.7)
\psbezier[linecolor=blue,linewidth=1.5pt]{-}(-0.494975, -0.494975)(-0.20,-0.20)(0.20,-0.20)(0.494975, -0.494975)
\psline[linecolor=blue,linewidth=1.5pt]{-}(-0.494975, 0.494975)(0,0)
\psline[linecolor=blue,linewidth=1.5pt]{-}(0.494975, 0.494975)(0,0)
\psarc[linecolor=black,linewidth=0.5pt,fillstyle=solid,fillcolor=darkgreen]{-}(0,0){0.09}{0}{360}
\end{pspicture}\ ,
\qquad 
\repW_{2,x}(4): \quad
\begin{pspicture}[shift=-0.6](-0.7,-0.7)(0.7,0.7)
\psarc[linecolor=black,linewidth=0.5pt,fillstyle=solid,fillcolor=lightlightblue]{-}(0,0){0.7}{0}{360}
\psline[linestyle=dashed, dash= 1.5pt 1.5pt, linewidth=0.5pt]{-}(0,0)(0,-0.7)
\psline[linecolor=blue,linewidth=1.5pt]{-}(-0.494975, 0.494975)(0,0)
\psline[linecolor=blue,linewidth=1.5pt]{-}(-0.494975, -0.494975)(0,0)
\psline[linecolor=blue,linewidth=1.5pt]{-}(0.494975, -0.494975)(0,0)
\psline[linecolor=blue,linewidth=1.5pt]{-}(0.494975, 0.494975)(0,0)
\psarc[linecolor=black,linewidth=0.5pt,fillstyle=solid,fillcolor=darkgreen]{-}(0,0){0.09}{0}{360}
\end{pspicture}\ .
\end{alignat}
\end{subequations}
The dimension of these modules is
\be
\label{eq:dimW}
\dim \repW_{k,x}(N) = d_k(N)\,,
\ee
where we recall that $d_k(N)$ is defined in \eqref{eq:dimV}.\medskip

The action $\lambda \cdot u$ of $\lambda\in\cL(N,N')$ on a state $u \in \repW_{k,x}(N')$ is obtained by drawing $u$ inside $\lambda$, and reading the new link state from the outer boundary. Loops are removed and defects crossing the dashed segment are unwound following the rules
\begin{subequations}
\begin{alignat}{2}
\label{eq:Wkx.rules}
&\psset{unit=0.7cm}
\begin{pspicture}[shift=-0.6](-0.7,-0.7)(0.7,0.7)
\psarc[linecolor=black,linewidth=0.5pt,fillstyle=solid,fillcolor=lightlightblue]{-}(0,0){0.7}{0}{360}
\psline[linestyle=dashed, dash= 1.5pt 1.5pt, linewidth=0.5pt]{-}(0,0)(0,-0.7)
\psarc[linecolor=black,linewidth=0.5pt,fillstyle=solid,fillcolor=darkgreen]{-}(0,0){0.07}{0}{360}
\psarc[linecolor=blue,linewidth=1.5pt]{-}(0,0){0.27}{0}{360}
\end{pspicture}
\ = \alpha \ 
\begin{pspicture}[shift=-0.6](-0.7,-0.7)(0.7,0.7)
\psarc[linecolor=black,linewidth=0.5pt,fillstyle=solid,fillcolor=lightlightblue]{-}(0,0){0.7}{0}{360}
\psline[linestyle=dashed, dash= 1.5pt 1.5pt, linewidth=0.5pt]{-}(0,0)(0,-0.7)
\psarc[linecolor=black,linewidth=0.5pt,fillstyle=solid,fillcolor=darkgreen]{-}(0,0){0.07}{0}{360}
\end{pspicture}
\ ,
\qquad&&\psset{unit=0.7cm}
\begin{pspicture}[shift=-0.6](-0.7,-0.7)(0.7,0.7)
\psarc[linecolor=black,linewidth=0.5pt,fillstyle=solid,fillcolor=lightlightblue]{-}(0,0){0.7}{0}{360}
\psline[linestyle=dashed, dash= 1.5pt 1.5pt, linewidth=0.5pt]{-}(0,0)(0,-0.7)
\psbezier[linecolor=blue,linewidth=1.5pt]{-}(0,0)(0.4,-0.4)(-0.4,-0.4)(-0.495,-0.495)
\psarc[linecolor=black,linewidth=0.5pt,fillstyle=solid,fillcolor=darkgreen]{-}(0,0){0.09}{0}{360}
\end{pspicture}
\ = x \
\begin{pspicture}[shift=-0.6](-0.7,-0.7)(0.7,0.7)
\psarc[linecolor=black,linewidth=0.5pt,fillstyle=solid,fillcolor=lightlightblue]{-}(0,0){0.7}{0}{360}
\psline[linestyle=dashed, dash= 1.5pt 1.5pt, linewidth=0.5pt]{-}(0,0)(0,-0.7)
\psline[linecolor=blue,linewidth=1.5pt]{-}(0,0)(-0.495,-0.495)
\psarc[linecolor=black,linewidth=0.5pt,fillstyle=solid,fillcolor=darkgreen]{-}(0,0){0.09}{0}{360}
\end{pspicture}\ ,
\qquad
\begin{pspicture}[shift=-0.6](-0.7,-0.7)(0.7,0.7)
\psarc[linecolor=black,linewidth=0.5pt,fillstyle=solid,fillcolor=lightlightblue]{-}(0,0){0.7}{0}{360}
\psline[linestyle=dashed, dash= 1.5pt 1.5pt, linewidth=0.5pt]{-}(0,0)(0,-0.7)
\psline[linecolor=blue,linewidth=1.5pt]{-}(0,0)(0.2,0.2)
\psline[linecolor=blue,linewidth=1.5pt]{-}(0,0)(-0.2,0.2)
\psarc[linecolor=black,linewidth=0.5pt,fillstyle=solid,fillcolor=darkgreen]{-}(0,0){0.09}{0}{360}
\psbezier[linecolor=blue,linewidth=1.5pt]{-}(0.2,0.2)(0.5,0.5)(-0.5,0.5)(-0.2,0.2)
\end{pspicture}
\ = 0\,,
\\[0.2cm]
&\psset{unit=0.7cm}
\begin{pspicture}[shift=-0.6](-0.7,-0.7)(0.7,0.7)
\psarc[linecolor=black,linewidth=0.5pt,fillstyle=solid,fillcolor=lightlightblue]{-}(0,0){0.7}{0}{360}
\psline[linestyle=dashed, dash= 1.5pt 1.5pt, linewidth=0.5pt]{-}(0,0)(0,-0.7)
\psarc[linecolor=black,linewidth=0.5pt,fillstyle=solid,fillcolor=darkgreen]{-}(0,0){0.07}{0}{360}
\psarc[linecolor=blue,linewidth=1.5pt]{-}(0.3, 0.3){0.17}{0}{360}
\end{pspicture}
\ = \beta \ 
\begin{pspicture}[shift=-0.6](-0.7,-0.7)(0.7,0.7)
\psarc[linecolor=black,linewidth=0.5pt,fillstyle=solid,fillcolor=lightlightblue]{-}(0,0){0.7}{0}{360}
\psline[linestyle=dashed, dash= 1.5pt 1.5pt, linewidth=0.5pt]{-}(0,0)(0,-0.7)
\psarc[linecolor=black,linewidth=0.5pt,fillstyle=solid,fillcolor=darkgreen]{-}(0,0){0.07}{0}{360}
\end{pspicture}\ ,
\qquad&&\psset{unit=0.7cm}
\begin{pspicture}[shift=-0.6](-0.7,-0.7)(0.7,0.7)
\psarc[linecolor=black,linewidth=0.5pt,fillstyle=solid,fillcolor=lightlightblue]{-}(0,0){0.7}{0}{360}
\psline[linestyle=dashed, dash= 1.5pt 1.5pt, linewidth=0.5pt]{-}(0,0)(0,-0.7)
\psbezier[linecolor=blue,linewidth=1.5pt]{-}(0,0)(-0.4,-0.4)(0.4,-0.4)(0.495,-0.495)
\psarc[linecolor=black,linewidth=0.5pt,fillstyle=solid,fillcolor=darkgreen]{-}(0,0){0.09}{0}{360}
\end{pspicture}
\ = \frac1x \
\begin{pspicture}[shift=-0.6](-0.7,-0.7)(0.7,0.7)
\psarc[linecolor=black,linewidth=0.5pt,fillstyle=solid,fillcolor=lightlightblue]{-}(0,0){0.7}{0}{360}
\psline[linestyle=dashed, dash= 1.5pt 1.5pt, linewidth=0.5pt]{-}(0,0)(0,-0.7)
\psline[linecolor=blue,linewidth=1.5pt]{-}(0,0)(0.495,-0.495)
\psarc[linecolor=black,linewidth=0.5pt,fillstyle=solid,fillcolor=darkgreen]{-}(0,0){0.09}{0}{360}
\end{pspicture}\ .
\end{alignat}
\end{subequations}
For instance, we have
\be
e_3 \cdot \
\begin{pspicture}[shift=-0.6](-0.7,-0.7)(0.7,0.7)
\psarc[linecolor=black,linewidth=0.5pt,fillstyle=solid,fillcolor=lightlightblue]{-}(0,0){0.7}{0}{360}
\psline[linestyle=dashed, dash=1.5pt 1.5pt, linewidth=0.5pt]{-}(-0.5,0)(0,-0.7)
\psarc[linecolor=black,linewidth=0.5pt,fillstyle=solid,fillcolor=darkgreen]{-}(-0.5,0){0.07}{0}{360}
\psbezier[linecolor=blue,linewidth=1.5pt]{-}(0.494975, 0.494975)(0.20,0.20)(0.20,-0.20)(0.494975, -0.494975)
\psbezier[linecolor=blue,linewidth=1.5pt]{-}(-0.494975, 0.494975)(-0.20,0.20)(-0.20,-0.20)(-0.494975, -0.494975)
\end{pspicture}
\ = \alpha \
\begin{pspicture}[shift=-0.6](-0.7,-0.7)(0.7,0.7)
\psarc[linecolor=black,linewidth=0.5pt,fillstyle=solid,fillcolor=lightlightblue]{-}(0,0){0.7}{0}{360}
\psline[linestyle=dashed, dash=1.5pt 1.5pt, linewidth=0.5pt]{-}(-0.5,0)(0,-0.7)
\psarc[linecolor=black,linewidth=0.5pt,fillstyle=solid,fillcolor=darkgreen]{-}(-0.5,0){0.07}{0}{360}
\psbezier[linecolor=blue,linewidth=1.5pt]{-}(0.494975, 0.494975)(0.20,0.20)(0.20,-0.20)(0.494975, -0.494975)
\psbezier[linecolor=blue,linewidth=1.5pt]{-}(-0.494975, 0.494975)(-0.20,0.20)(-0.20,-0.20)(-0.494975, -0.494975)
\end{pspicture}\ ,
\qquad 
e_0 \cdot \
\begin{pspicture}[shift=-0.6](-0.7,-0.7)(0.7,0.7)
\psarc[linecolor=black,linewidth=0.5pt,fillstyle=solid,fillcolor=lightlightblue]{-}(0,0){0.7}{0}{360}
\psline[linestyle=dashed, dash= 1.5pt 1.5pt, linewidth=0.5pt]{-}(0,0)(0,-0.7)
\psbezier[linecolor=blue,linewidth=1.5pt]{-}(-0.494975, -0.494975)(-0.20,-0.20)(-0.20,0.20)(-0.494975, 0.494975)
\psline[linecolor=blue,linewidth=1.5pt]{-}(0.494975, -0.494975)(0,0)
\psline[linecolor=blue,linewidth=1.5pt]{-}(0.494975, 0.494975)(0,0)
\psarc[linecolor=black,linewidth=0.5pt,fillstyle=solid,fillcolor=darkgreen]{-}(0,0){0.09}{0}{360}
\end{pspicture}
\ = x \
\begin{pspicture}[shift=-0.6](-0.7,-0.7)(0.7,0.7)
\psarc[linecolor=black,linewidth=0.5pt,fillstyle=solid,fillcolor=lightlightblue]{-}(0,0){0.7}{0}{360}
\psline[linestyle=dashed, dash= 1.5pt 1.5pt, linewidth=0.5pt]{-}(0,0)(0,-0.7)
\psbezier[linecolor=blue,linewidth=1.5pt]{-}(-0.494975, -0.494975)(-0.20,-0.20)(0.20,-0.20)(0.494975, -0.494975)
\psline[linecolor=blue,linewidth=1.5pt]{-}(-0.494975, 0.494975)(0,0)
\psline[linecolor=blue,linewidth=1.5pt]{-}(0.494975, 0.494975)(0,0)
\psarc[linecolor=black,linewidth=0.5pt,fillstyle=solid,fillcolor=darkgreen]{-}(0,0){0.09}{0}{360}
\end{pspicture}\ ,
\qquad 
e_3 \cdot \ 
\begin{pspicture}[shift=-0.6](-0.7,-0.7)(0.7,0.7)
\psarc[linecolor=black,linewidth=0.5pt,fillstyle=solid,fillcolor=lightlightblue]{-}(0,0){0.7}{0}{360}
\psline[linestyle=dashed, dash= 1.5pt 1.5pt, linewidth=0.5pt]{-}(0,0)(0,-0.7)
\psbezier[linecolor=blue,linewidth=1.5pt]{-}(0.494975, 0.494975)(0.20,0.20)(0.20,-0.20)(0.494975, -0.494975)
\psline[linecolor=blue,linewidth=1.5pt]{-}(-0.494975, 0.494975)(0,0)
\psline[linecolor=blue,linewidth=1.5pt]{-}(-0.494975, -0.494975)(0,0)
\psarc[linecolor=black,linewidth=0.5pt,fillstyle=solid,fillcolor=darkgreen]{-}(0,0){0.09}{0}{360}
\end{pspicture}
\ = 0\,.
\ee
It was proven in \cite{GL98} that the set of modules 
\be
\label{eq:W.family}
\repW_{k,x} = \{\repW_{k,x}(N)\,|\, N = N_0, N_0+2, N_0+4, \dots\}\,, \qquad
N_0 = 
\left\{\begin{array}{ll}
0 & k \in \mathbb Z\,,\\[0.1cm] 
1 & k \in \mathbb Z + \frac12\,,
\end{array}\right.
\ee
is a family of modules over $\eptl_N(\beta)$, for all $q \in \mathbb C^\times$, $k \in \frac 12 \mathbb Z_{\ge 0}$ and $x \in \mathbb C^\times$. In addition, we have
\be
\repW_{k,x}(N) = \cL(N,2k) \cdot u_k\,,
\ee
for all $q \in \mathbb C^{\times}$, where $u_k$ is the unique link state of $\repW_{k,x}(2k)$.

\begin{Proposition} \label{prop:seed.W}
Let $q \in \mathbb C^\times$, $k \in \frac 12 \mathbb Z_{\ge 0}$ and $x \in \mathbb C^\times$. Let also $\repM$ be a family of modules over $\eptl_N(\beta)$, and $\xi$ be a nonzero state in $\repM(2k)$ satisfying 
\begin{subequations} \label{eq:def.insertion.kx}
\begin{alignat}{2}
&f\cdot\xi = (x+x^{-1})\,\xi \,, && k=0 \,,
\\
&c_0\cdot \xi=0 \quad\text{and}\quad \Omega\cdot\xi = x\,\xi \,, &\qquad& k>0 \,.
\label{eq:def.insertion.kx,k>0}
\end{alignat}
\end{subequations}
Then the linear map
\begin{equation} \label{eq:insertion.map.W}
\phi_\xi: \lambda\cdot u_k \mapsto \lambda \cdot \xi \,,
\qquad \lambda \in\cL(N,2k)
\end{equation}
is well-defined on $\repW_{k,x}(N)$ for all admissible $N$, and defines a family homomorphism from $\repW_{k,x}$ to~$\repM$.
\end{Proposition}
This can be proven with arguments similar to those used for \cref{prop:seed.Vk}. A state $\xi$ satisfying \eqref{eq:def.insertion.kx} is called an \emph{insertion state with parameters $(k,x)$ for $\repM$}. Moreover, the subspace of $\repM(2k)$ generated by the insertion states is called the \emph{insertion space with parameters $(k,x)$ of $\repM$}, and is a submodule of $\repM(2k)$.\medskip

The structure of the modules $\repW_{k,x}(N)$ at roots of unity depends on the values of $k$ and $x$, and is described in \cite{BSA18,PSA23}, based on the fundamental results of \cite{GL98}. Let $\repR_{k,x}$ be the maximal proper subfamily of $\repW_{k,x}$, and $\repQ_{k,x}=\repW_{k,x}\big/\repR_{k,x}$ be the corresponding quotient family.
For $x$ of the form $x = \eps q^s$ with $s-k\in \Zbb$, the structure of $\repW_{k,x}$ is described by one of two possible Loewy diagrams:
\begin{subequations}
\be
\label{eq:W.single.file}
\repW_{k,x} \simeq \big[\begin{pspicture}[shift=-0.43](-1.5,-0.5)(12,0.5)
\rput(-1,0){$\repQ_{k,x}$}
\rput(1.5,0){$\repQ_{k_1,x_1}$}
\rput(4,0){$\repQ_{k_2,x_2}$}
\rput(6.5,0){$\repQ_{k_3,x_3}$}
\rput(9,0){$\repQ_{k_4,x_4}$}
\rput(11.5,0){$\dots$}
\psline{->}(-0.3,0)(0.7,0)
\psline{->}(2.3,0)(3.2,0)
\psline{->}(4.8,0)(5.7,0)
\psline{->}(7.3,0)(8.2,0)
\psline{->}(9.8,0)(10.7,0)
\end{pspicture}
\big]
\ee
or
\be
\label{eq:W.double.file}
\repW_{k,x} \simeq \left[\begin{pspicture}[shift=-1.43](-0.5,-1.5)(11.7,1.5)
\rput(0,0){$\repQ_{k,x}$}
\rput(2,1){$\repQ_{k_1,x_1}$}
\rput(2,-1){$\repQ_{k_2,x_2}$}
\rput(5,1){$\repQ_{k_3,x_3}$}
\rput(5,-1){$\repQ_{k_4, x_4}$}
\rput(8,1){$\repQ_{k_5,x_5}$}
\rput(8,-1){$\repQ_{k_6, x_6}$}
\rput(11,1){$\dots$}
\rput(11,-1){$\dots$}
\psline{->}(0.6,0.2)(1.6,0.7)
\psline{->}(0.6,-0.2)(1.6,-0.7)
\psline{->}(3.0,1)(4.0,1)
\psline{->}(3.0,-1)(4.0,-1)
\psline{->}(6.0,1)(7.0,1)
\psline{->}(6.0,-1)(7.0,-1)
\psline{->}(9.0,1)(10.0,1)
\psline{->}(9.0,-1)(10.0,-1)
\psline{->}(2.2,-0.7)(4.8,0.7)
\psline{->}(2.2,0.7)(4.8,-0.7)
\psline{->}(5.2,-0.7)(7.8,0.7)
\psline{->}(5.2,0.7)(7.8,-0.7)
\psline{->}(8.2,-0.7)(10.8,0.7)
\psline{->}(8.2,0.7)(10.8,-0.7)
\end{pspicture}
\right] \,.
\ee
\end{subequations}
For other values of $x$ not of this form, $\repW_{k,x}$ is irreducible.
In the above diagrams, each $x_i$ is of the form $\eps_i q^{s_i}$ for some $\eps_i \in \{-1,1\}$ and $s_i \in \frac12 \mathbb Z$, satisfying
\be
k_i - k \in \mathbb Z\,, \qquad 
s_i - k \in \mathbb Z\,, \qquad 
k<k_1<k_2< \cdots\,, 
\ee
and
\be
x q^{k} + x^{-1} q^{-k} = x_i q^{k_i} + x_i^{-1} q^{-k_i}\,, 
\qquad
x q^{-k} + x^{-1} q^{k} = x_i q^{-k_i} + x_i^{-1} q^{k_i}\,.
\ee
The precise values of $k_i$ and $x_i$ are detailed in \cite{BSA18,PSA23}. In the above Loewy diagrams, each node represents an irreducible module $\repQ_{k,x}$ or $\repQ_{k_i,x_i}$, and the arrows indicate the action of $\eptl_N(\beta)$. A submodule is a subset $\mathcal S$ of the composition factors where all the arrows originating from the elements of~$\mathcal S$ point to other elements in~$\mathcal S$. The family $\repW_{k,\eps q^s}$ thus has a non-trivial pattern of nested subfamilies. For finite $N$, the module $\repW_{k,\eps q^s}(N)$ has finitely many composition factors, because $\repQ_{k,x}(N)=0$ for $2k>N$. Moreover, two irreducible modules $\repQ_{k,x}(N)$ and $\repQ_{\ell,y}(N)$ are inequivalent unless (i) $k=\ell=0$ and $x = y^{\pm 1}$, or (i) $k=\ell>0$ and $x = y$.\medskip

For our investigations of the ADE lattice models, we are particularly interested in the values $q=\eE^{-\ir\pi p/p'}$ with $(p,p')$ as in \eqref{eq:q.standard}, and the modules $\repW_{k,x}(N)$ with $0 \le k<\frac{p'}2$ and with the twist parameter $x$ taking the values
\be
\label{eq:xks}
x = \eps q^{s}\, \qquad \textrm{with} \qquad 
k < s < p'-k\,,\qquad \eps \in \{+1,-1\}\,.
\ee
In these cases, the Loewy diagram of $\repW_{k,\eps q^{s}}$ is \cite{BSA18,PSA23}
\be
\label{eq:W.Loewy}
\repW_{k,\eps q^{s}} \simeq
\left[\begin{pspicture}[shift=-1.4](-0.5,-1.5)(14.3,1.5)
\rput(0,0){$\repQ_{k,\eps q^{s}}$}
\rput(2,1){$\repQ_{s,\eps q^{k}}$}
\rput(2,-1){$\repQ_{p'-s,\eps q^{p'-k}}$}
\rput(5,1){$\repQ_{p'+k,\eps q^{p'+s}}$}
\rput(5,-1){$\repQ_{p'-k, \eps q^{p'-s}}$}
\rput(8,1){$\repQ_{p'+s,\eps q^{p'+k}}$}
\rput(8,-1){$\repQ_{2p'-s, \eps q^{2p'-k}}$}
\rput(11,1){$\repQ_{2p'+k,\eps q^{2p'+s}}$}
\rput(11,-1){$\repQ_{2p'-k, \eps q^{2p'-s}}$}
\rput(14,1){$\dots$}
\rput(14,-1){$\dots$}
\psline{->}(0.6,0.2)(1.6,0.7)
\psline{->}(0.6,-0.2)(1.6,-0.7)
\psline{->}(2.7,1)(4.0,1)
\psline{->}(3.1,-1)(3.9,-1)
\psline{->}(6.0,1)(7.0,1)
\psline{->}(6.0,-1)(6.85,-1)
\psline{->}(9.0,1)(9.85,1)
\psline{->}(9.15,-1)(9.85,-1)
\psline{->}(12.2,1)(13.3,1)
\psline{->}(12.2,-1)(13.3,-1)
\psline{->}(2.2,-0.7)(4.8,0.7)
\psline{->}(2.2,0.7)(4.8,-0.7)
\psline{->}(5.2,-0.7)(7.8,0.7)
\psline{->}(5.2,0.7)(7.8,-0.7)
\psline{->}(8.2,-0.7)(10.8,0.7)
\psline{->}(8.2,0.7)(10.8,-0.7)
\psline{->}(11.2,-0.7)(13.8,0.7)
\psline{->}(11.2,0.7)(13.8,-0.7)
\end{pspicture}\right]
\ee
The module $\repW_{k,\eps q^s}(N)$ has submodules isomorphic to $\repW_{s,\eps q^k}(N)$ and $\repW_{p'-s,\eps q^{p'-k}}(N)$. The Loewy diagrams for these two modules are read from \eqref{eq:W.Loewy} by considering all the factors produced from $\repQ_{s,\eps q^k}(N)$ and $\repQ_{p'-s,\eps q^{p'-k}}(N)$, respectively. Their intersection is nonzero, as both these submodules themselves have two submodules, isomorphic to $\repW_{p'+k,\eps q^{p'+s}}(N)$ and $\repW_{p'-k,\eps q^{p'-s}}(N)$.\medskip

The family $\repW_{k,\eps q^{-s}}$ with $s$ as in \eqref{eq:xks} also has a double-ladder structure. Its Loewy diagram is obtained from \eqref{eq:W.Loewy} by changing each factor $\repQ_{k_i,x_i}$ by $\repQ_{k_i,1/x_i}$. Using $q^{p'}=(-1)^p$, we can write $\repW_{k,\eps q^s}=\repW_{k,\eps'q^{-s'}}$ where $s'=p'-s$ and $\eps'=(-1)^p\eps$. Thus, to investigate all the families $\repW_{k,x}$ with $k$ and $x$ in \eqref{eq:xks}, it is sufficient to consider those with $k<s \le \frac{p'}2$ and $\eps \in \{-1,+1\}$.\medskip

We complete this description by explicitly describing the families of homomorphisms between the families~$\repW_{k,x}$. In \cite{IMD22}, we focused on generic values of $q$ and introduced the states in $\repW_{k,\eps q^{\sigma s}}(2s)$
\be
\label{eq:vks}
v^{\tinyx{\sigma,\eps}}_{k,s} = 
P_{2s} (c^\dag_0)^{s-k} \cdot u_k = \ \
\begin{pspicture}[shift=-1.6](-2.4,-1.7)(2.5,1.7)
\psarc[linecolor=black,linewidth=0.5pt,fillstyle=solid,fillcolor=pink]{-}(0,0){1.55}{-82.5}{97.5}
\psarc[linecolor=black,linewidth=0.5pt,fillstyle=solid,fillcolor=pink]{-}(0,0){1.55}{82.5}{262.5}
\psline[linecolor=black,linewidth=0.5pt](-0.182737, -1.38802)(-0.202316, -1.53674)
\psline[linecolor=black,linewidth=0.5pt](0.182737, -1.38802)(0.202316, -1.53674)
\psarc[linecolor=black,linewidth=0.5pt,fillstyle=solid,fillcolor=lightlightblue]{-}(0,0){1.4}{0}{360}
\psline[linestyle=dashed, dash= 1.5pt 1.5pt, linewidth=0.5pt]{-}(0,0.1)(0,-1.4)
\psline[linecolor=blue,linewidth=1.5pt]{-}(0,0.1)(1.3523, 0.362347)
\psline[linecolor=blue,linewidth=1.5pt]{-}(0,0.1)(0.989949, 0.989949)
\psline[linecolor=blue,linewidth=1.5pt]{-}(0,0.1)(0.362347, 1.3523)
\psline[linecolor=blue,linewidth=1.5pt]{-}(0,0.1)(-0.362347, 1.3523)
\psline[linecolor=blue,linewidth=1.5pt]{-}(0,0.1)(-0.989949, 0.989949)
\psline[linecolor=blue,linewidth=1.5pt]{-}(0,0.1)(-1.3523, 0.362347)
\psarc[linecolor=black,linewidth=0.5pt,fillstyle=solid,fillcolor=darkgreen]{-}(0,0.1){0.09}{0}{360}
\psbezier[linecolor=blue,linewidth=1.5pt]{-}(0.362347, -1.3523)(0.258819, -0.965926)(-0.258819, -0.965926)(-0.362347, -1.3523)
\psbezier[linecolor=blue,linewidth=1.5pt]{-}(0.989949, -0.989949)(0.565685, -0.565685)(-0.565685, -0.565685)(-0.989949, -0.989949)
\psbezier[linecolor=blue,linewidth=1.5pt]{-}(1.3523, -0.362347)(0.772741, -0.207055)(-0.772741, -0.207055)(-1.3523, -0.362347)
\rput(0.452933, -1.69037){$_1$}
\rput(1.23744, -1.23744){$_{\dots}$}
\rput(1.88867, -0.445874){$_{s-k}$}
\rput(2.09867, 0.465874){$_{s-k+1}$}
\rput(1.23744, 1.23744){$_{\dots}$}
\rput(-1.23744, 1.23744){$_{\dots}$}
\rput(-1.91867, 0.465874){$_{s+k}$}
\rput(-2.07867, -0.465874){$_{s+k+1}$}
\rput(-1.23744, -1.23744){$_{\dots}$}
\rput(-0.452933, -1.69037){$_{2s}$}
\end{pspicture}\,,
\ee
where $\sigma,\eps \in \{+1,-1\}$ and $s>k$. This states satisfy the properties
\be
\label{eq:ejv=0}
c_0 \cdot v^{\tinyx{\sigma,\eps}}_{k,s} = 0\,, \qquad
\Omega \cdot v^{\tinyx{\sigma,\eps}}_{k,s} = \eps q^{\sigma k} v^{\tinyx{\sigma,\eps}}_{k,s}\,,
\ee
and is therefore an insertion state with parameters $(s,\eps q^{\sigma k})$ for $\repW_{k,\eps q^{\sigma s}}$. \medskip

Let us now consider the state $v^{\tinyx{\sigma,\eps}}_{k,s}$ in the limit $q \to \eE^{-\ir \pi p/p'}$, with $s < \frac {p'}2$. In this case, the projector $P_{2s}$ is non-singular, which implies that $v^{\tinyx{\sigma,\eps}}_{k,s}$ is also non-singular. Here are examples of these states, which already appeared in \cite{IMD24}:
\begin{subequations}
\label{eq:vks.examples}
\begin{alignat}{2}
v^{\tinyx{\sigma,\eps}}_{0,1} &= \ \
\begin{pspicture}[shift=-0.6](-0.6,-0.7)(0.6,0.7)
\psarc[linecolor=black,linewidth=0.5pt,fillstyle=solid,fillcolor=lightlightblue]{-}(0,0){0.6}{0}{360}
\psline[linestyle=dashed, dash= 1.5pt 1.5pt, linewidth=0.5pt]{-}(0,0)(0,-0.6)
\psarc[linecolor=black,linewidth=0.5pt,fillstyle=solid,fillcolor=darkgreen]{-}(0,0){0.07}{0}{360}
\psbezier[linecolor=blue,linewidth=1.5pt]{-}(-0.6,0)(-0.3,0)(-0.3,-0.25)(0,-0.25)
\psbezier[linecolor=blue,linewidth=1.5pt]{-}(0.6,0)(0.3,0)(0.3,-0.25)(0,-0.25)
\end{pspicture}
\ + \eps \ \
\begin{pspicture}[shift=-0.6](-0.6,-0.7)(0.6,0.7)
\psarc[linecolor=black,linewidth=0.5pt,fillstyle=solid,fillcolor=lightlightblue]{-}(0,0){0.6}{0}{360}
\psline[linestyle=dashed, dash= 1.5pt 1.5pt, linewidth=0.5pt]{-}(0,0)(0,-0.6)
\psarc[linecolor=black,linewidth=0.5pt,fillstyle=solid,fillcolor=darkgreen]{-}(0,0){0.07}{0}{360}
\psbezier[linecolor=blue,linewidth=1.5pt]{-}(-0.6,0)(-0.3,0)(-0.3,0.25)(0,0.25)
\psbezier[linecolor=blue,linewidth=1.5pt]{-}(0.6,0)(0.3,0)(0.3,0.25)(0,0.25)
\end{pspicture}
\ \ ,
\\[0.3cm]
v^{\tinyx{\sigma,\eps}}_{0,2} &= \ \ 
\begin{pspicture}[shift=-0.6](-0.6,-0.7)(0.6,0.7)
\psarc[linecolor=black,linewidth=0.5pt,fillstyle=solid,fillcolor=lightlightblue]{-}(0,0){0.7}{0}{360}
\psline[linestyle=dashed, dash=1.5pt 1.5pt, linewidth=0.5pt]{-}(0,0.5)(0,-0.7)
\psarc[linecolor=black,linewidth=0.5pt,fillstyle=solid,fillcolor=darkgreen]{-}(0,0.5){0.07}{0}{360}
\psbezier[linecolor=blue,linewidth=1.5pt]{-}(0.494975, 0.494975)(0.20,0.20)(-0.20,0.20)(-0.494975, 0.494975)
\psbezier[linecolor=blue,linewidth=1.5pt]{-}(0.494975, -0.494975)(0.20,-0.20)(-0.20,-0.20)(-0.494975, -0.494975)
\end{pspicture} 
\ \ +\eps \ \
\begin{pspicture}[shift=-0.6](-0.6,-0.7)(0.6,0.7)
\psarc[linecolor=black,linewidth=0.5pt,fillstyle=solid,fillcolor=lightlightblue]{-}(0,0){0.7}{0}{360}
\psline[linestyle=dashed, dash=1.5pt 1.5pt, linewidth=0.5pt]{-}(-0.5,0)(0,-0.7)
\psarc[linecolor=black,linewidth=0.5pt,fillstyle=solid,fillcolor=darkgreen]{-}(-0.5,0){0.07}{0}{360}
\psbezier[linecolor=blue,linewidth=1.5pt]{-}(0.494975, 0.494975)(0.20,0.20)(0.20,-0.20)(0.494975, -0.494975)
\psbezier[linecolor=blue,linewidth=1.5pt]{-}(-0.494975, 0.494975)(-0.20,0.20)(-0.20,-0.20)(-0.494975, -0.494975)
\end{pspicture}
\ \ -\eps\, \beta \ \
\begin{pspicture}[shift=-0.6](-0.6,-0.7)(0.6,0.7)
\psarc[linecolor=black,linewidth=0.5pt,fillstyle=solid,fillcolor=lightlightblue]{-}(0,0){0.7}{0}{360}
\psline[linestyle=dashed, dash= 1.5pt 1.5pt, linewidth=0.5pt]{-}(0,0)(0,-0.7)
\psarc[linecolor=black,linewidth=0.5pt,fillstyle=solid,fillcolor=darkgreen]{-}(0,0){0.07}{0}{360}
\psbezier[linecolor=blue,linewidth=1.5pt]{-}(0.494975, 0.494975)(0.20,0.20)(-0.20,0.20)(-0.494975, 0.494975)
\psbezier[linecolor=blue,linewidth=1.5pt]{-}(0.494975, -0.494975)(0.20,-0.20)(-0.20,-0.20)(-0.494975, -0.494975)
\end{pspicture}
\ \ +\eps \ \
\begin{pspicture}[shift=-0.6](-0.6,-0.7)(0.6,0.7)
\psarc[linecolor=black,linewidth=0.5pt,fillstyle=solid,fillcolor=lightlightblue]{-}(0,0){0.7}{0}{360}
\psline[linestyle=dashed, dash=1.5pt 1.5pt, linewidth=0.5pt]{-}(0.5,0)(0,-0.7)
\psarc[linecolor=black,linewidth=0.5pt,fillstyle=solid,fillcolor=darkgreen]{-}(0.5,0){0.07}{0}{360}
\psbezier[linecolor=blue,linewidth=1.5pt]{-}(0.494975, 0.494975)(0.20,0.20)(0.20,-0.20)(0.494975, -0.494975)
\psbezier[linecolor=blue,linewidth=1.5pt]{-}(-0.494975, 0.494975)(-0.20,0.20)(-0.20,-0.20)(-0.494975, -0.494975)
\end{pspicture}
\ \ - \beta \ \ 
\begin{pspicture}[shift=-0.6](-0.6,-0.7)(0.6,0.7)
\psarc[linecolor=black,linewidth=0.5pt,fillstyle=solid,fillcolor=lightlightblue]{-}(0,0){0.7}{0}{360}
\psline[linestyle=dashed, dash= 1.5pt 1.5pt, linewidth=0.5pt]{-}(0,0)(0,-0.7)
\psarc[linecolor=black,linewidth=0.5pt,fillstyle=solid,fillcolor=darkgreen]{-}(0,0){0.07}{0}{360}
\psbezier[linecolor=blue,linewidth=1.5pt]{-}(0.494975, 0.494975)(0.20,0.20)(0.20,-0.20)(0.494975, -0.494975)
\psbezier[linecolor=blue,linewidth=1.5pt]{-}(-0.494975, 0.494975)(-0.20,0.20)(-0.20,-0.20)(-0.494975, -0.494975)
\end{pspicture}
\ \ + \ \
\begin{pspicture}[shift=-0.6](-0.6,-0.7)(0.6,0.7)
\psarc[linecolor=black,linewidth=0.5pt,fillstyle=solid,fillcolor=lightlightblue]{-}(0,0){0.7}{0}{360}
\psline[linestyle=dashed, dash=1.5pt 1.5pt, linewidth=0.5pt]{-}(0,-0.45)(0,-0.7)
\psarc[linecolor=black,linewidth=0.5pt,fillstyle=solid,fillcolor=darkgreen]{-}(0,-0.5){0.07}{0}{360}
\psbezier[linecolor=blue,linewidth=1.5pt]{-}(0.494975, 0.494975)(0.20,0.20)(-0.20,0.20)(-0.494975, 0.494975)
\psbezier[linecolor=blue,linewidth=1.5pt]{-}(0.494975, -0.494975)(0.20,-0.20)(-0.20,-0.20)(-0.494975, -0.494975)
\end{pspicture} \ \ ,
\\[0.3cm]
v^{\tinyx{\sigma,\eps}}_{1,2} &= \ \
\begin{pspicture}[shift=-0.6](-0.6,-0.7)(0.6,0.7)
\psarc[linecolor=black,linewidth=0.5pt,fillstyle=solid,fillcolor=lightlightblue]{-}(0,0){0.7}{0}{360}
\psline[linestyle=dashed, dash= 1.5pt 1.5pt, linewidth=0.5pt]{-}(0,0)(0,-0.7)
\psbezier[linecolor=blue,linewidth=1.5pt]{-}(-0.494975, -0.494975)(-0.20,-0.20)(0.20,-0.20)(0.494975, -0.494975)
\psline[linecolor=blue,linewidth=1.5pt]{-}(-0.494975, 0.494975)(0,0)
\psline[linecolor=blue,linewidth=1.5pt]{-}(0.494975, 0.494975)(0,0)
\psarc[linecolor=black,linewidth=0.5pt,fillstyle=solid,fillcolor=darkgreen]{-}(0,0){0.09}{0}{360}
\end{pspicture}
\ \ +\eps\, q^\sigma \ \
\begin{pspicture}[shift=-0.6](-0.6,-0.7)(0.6,0.7)
\psarc[linecolor=black,linewidth=0.5pt,fillstyle=solid,fillcolor=lightlightblue]{-}(0,0){0.7}{0}{360}
\psline[linestyle=dashed, dash= 1.5pt 1.5pt, linewidth=0.5pt]{-}(0,0)(0,-0.7)
\psbezier[linecolor=blue,linewidth=1.5pt]{-}(0.494975, 0.494975)(0.20,0.20)(0.20,-0.20)(0.494975, -0.494975)
\psline[linecolor=blue,linewidth=1.5pt]{-}(-0.494975, 0.494975)(0,0)
\psline[linecolor=blue,linewidth=1.5pt]{-}(-0.494975, -0.494975)(0,0)
\psarc[linecolor=black,linewidth=0.5pt,fillstyle=solid,fillcolor=darkgreen]{-}(0,0){0.09}{0}{360}
\end{pspicture}
\ \ +\eps \ \
\begin{pspicture}[shift=-0.6](-0.6,-0.7)(0.6,0.7)
\psarc[linecolor=black,linewidth=0.5pt,fillstyle=solid,fillcolor=lightlightblue]{-}(0,0){0.7}{0}{360}
\psline[linestyle=dashed, dash= 1.5pt 1.5pt, linewidth=0.5pt]{-}(0,0)(0,-0.7)
\psbezier[linecolor=blue,linewidth=1.5pt]{-}(0.494975, 0.494975)(0.20,0.20)(-0.20,0.20)(-0.494975, 0.494975)
\psline[linecolor=blue,linewidth=1.5pt]{-}(0.494975, -0.494975)(0,0)
\psline[linecolor=blue,linewidth=1.5pt]{-}(-0.494975, -0.494975)(0,0)
\psarc[linecolor=black,linewidth=0.5pt,fillstyle=solid,fillcolor=darkgreen]{-}(0,0){0.09}{0}{360}
\end{pspicture}
\ \ +\eps\, q^{-\sigma}\ \
\begin{pspicture}[shift=-0.6](-0.6,-0.7)(0.6,0.7)
\psarc[linecolor=black,linewidth=0.5pt,fillstyle=solid,fillcolor=lightlightblue]{-}(0,0){0.7}{0}{360}
\psline[linestyle=dashed, dash= 1.5pt 1.5pt, linewidth=0.5pt]{-}(0,0)(0,-0.7)
\psbezier[linecolor=blue,linewidth=1.5pt]{-}(-0.494975, -0.494975)(-0.20,-0.20)(-0.20,0.20)(-0.494975, 0.494975)
\psline[linecolor=blue,linewidth=1.5pt]{-}(0.494975, -0.494975)(0,0)
\psline[linecolor=blue,linewidth=1.5pt]{-}(0.494975, 0.494975)(0,0)
\psarc[linecolor=black,linewidth=0.5pt,fillstyle=solid,fillcolor=darkgreen]{-}(0,0){0.09}{0}{360}
\end{pspicture}\ \ .
\end{alignat}
\end{subequations}

We now consider the state $v^{\tinyx{\sigma',\eps'}}_{k,s'} \in \repW_{k,\eps' q^{\sigma' s'}}(2s')$, where $\sigma',\eps' \in \{+1,-1\}$ and $s' \in \frac12 \mathbb Z$ satisfying $s'-k \in \mathbb Z_{>0}$. This state satisfies \eqref{eq:ejv=0} with $\sigma \to \sigma'$, $\eps \to \eps'$ and $s \to s'$. For generic $q$, it is thus an insertion state with parameters $(s',\eps' q^{\sigma'k})$ in $\repW_{k,\eps' q^{\sigma' s'}}$. At the root of unity $q = \eE^{-\ir \pi p/p'}$, we have $\repW_{k,\eps q^{\sigma s}} = \repW_{k,\eps' q^{\sigma' s'}}$ for $\sigma' = -\sigma$, $\eps' = (-1)^p \eps$ and $s'=p'-s$. For $k<s \le \frac{p'}2$, the state $v^{\tinyx{\sigma',\eps'}}_{k,s'}$ involves a projector $P_{2s'}$ that is singular for $q \to \eE^{-\ir \pi p/p'}$. The following proposition states that the state $v^{\tinyx{\sigma',\eps'}}_{k,s'}$ is nonetheless non-singular in the limit. This is consistent with the examples in \eqref{eq:vks.examples}, which are free of singularities.

\begin{Proposition} \label{prop:v.is.regular}
Let $p$ and $p'$ be as in \eqref{eq:q.beta.standard}, $\sigma, \eps \in \{+1,-1\}$, and $k,s \in \frac12\mathbb Z_{\ge 0}$ be such that $s-k \in \mathbb Z$ and $k<s \le \frac{p'}2$. Let also $s'=p'-s$, $\sigma'=-\sigma$ and $\eps' = (-1)^p \eps$. Then the state $v^{\tinyx{\sigma',\eps'}}_{k,s'} \in \repW_{k,\eps q^{\sigma s}}(2s')$ is non-singular in the limit $q \to \eE^{-\ir \pi p/p'}$.
\end{Proposition}
\proof
From \cref{prop:Pnp'-1} and the recursive definition \eqref{eq:WJ.def} of the projectors, we know that the order of the singularity of $v^{\tinyx{\sigma,\eps}}_{k,s'}$ is at most one. We can therefore write the Laurent series
\be
\label{eq:v.Laurent}
v^{\tinyx{\sigma',\eps'}}_{k,s'} = \frac{y_{-1}}{q-\eE^{-\ir \pi p/p'}} + y_0 + \dots
\ee
for some states $y_{-1}$ and $y_0$ in $\repW_{k,\eps q^{\sigma s}}(2s')$. Let us make the assumption that $y_{-1} \neq 0$. From \eqref{eq:ejv=0}, we deduce that this state satisfies
\be
\label{eq:y-1.props}
\Omega \cdot y_{-1} = \eps' q^{\sigma' k} y_{-1}\,, 
\qquad 
e_j \cdot y_{-1} = 0\,, \quad j=1,2,\dots, 2s-1.
\ee

Let us denote by $\repW^\downarrow_{k,\eps q^{\sigma s}}(2s')$ the module $\repW_{k,\eps q^{\sigma s}}(2s')$ seen as a module over $\tl_{2s'}(\beta)$. This module has the filtration
\be
\repM^{\tinyx {k}}(2s') \subset \repM^{\tinyx {k\!+\!1}}(2s') \subset \cdots \subset \repM^{\tinyx {s'}}(2s')\,,
\ee
where $\repM^{\tinyx {m}}(2s')$ is the submodule spanned by link states with at most $m-k$ arches crossing the dashed segment. For all $q \in \mathbb C^\times$, the corresponding quotient modules satisfy
\be
\repM^{\tinyx {m}}(2s') / \repM^{\tinyx {m\!-\!1}}(2s') \simeq \repV_m(2s')\,.
\ee
For $q$ generic, we have
\be
\label{eq:Fvks'}
\wh{F} \cdot v^{\tinyx{\sigma',\eps'}}_{k,s'} = \wh{f}_{s'}\,v^{\tinyx{\sigma',\eps'}}_{k,s'}\,,
\ee
where $\wh{F} \in \tl_{2s'}(\beta)$ is the double-row braid transfer matrix, and $\wh{f}_{m}$ is its unique eigenvalue in $\repV_m$. We write their Taylor expansions as
\be
\wh{F} = \wh{F}_0 + (q-\eE^{-\ir \pi p/p'}) \wh{F}_1 + \dots\,, \qquad
\wh{f}_m = \wh{f}_{m,0} + (q-\eE^{-\ir \pi p/p'}) \wh{f}_{m,1} + \dots\,.
\ee
Because the identity term in $P_{2s'}$ has a unit prefactor, $y_{-1}$ has no component along the link state $(c_0^\dag)^{s'-k}\cdot u_k$. This implies that $y_{-1} \in \repM^{\tinyx m}(2s')$ for some $m<s'$. Let us then consider the module $\repW^\downarrow_{k,\eps q^{\sigma s}}(2s')/\repM^{\tinyx{m\!-\!1}}(2s')$. In this module, we have $\wh{F}_0 \cdot y_{-1} = \wh{f}_{m,0}\, y_{-1}$. By taking the leading order of~\eqref{eq:Fvks'}, we find that $\wh{F}_0 \cdot y_{-1} = \wh{f}_{s',0}\, y_{-1}$. We conclude that $f_{s',0} = f_{m,0}$, and hence $m=p'-1-s' = s-1$.\medskip

The state $y_{-1}$ therefore spans a one-dimensional submodule in $\repM^{\tinyx{s\!-\!1}}(2s') \simeq \repV_{s-1}(2s')$ isomorphic to $\repV_{s'}(2s')$. The Loewy diagram of $\repV_{s-1}(2s')$ is $[\repQ_{s-1}(2s') \to \repQ_{s'}(2s')]$, with $\repQ_{s'}(2s') \simeq \repV_{s'}(2s')$, so there is indeed a submodule isomorphic to $\repV_{s'}(2s')$ in $\repV_{s-1}(2s')$.\medskip

Using the same idea as in \eqref{eq:vk}, we construct a second state in $\repM^{\tinyx{s\!-\!1}}(2s')$ that spans a one-dimensional submodule isomorphic to $\repV_{s'}(2s')$, namely
\be
w = (P_{p'-1} \otimes \id_{s'-s+1}) \cdot \wh{u}_k\,, \qquad 
\wh{u}_k = c^\dag_{p'-1} c^\dag_{p'-2} \cdots c^\dag_{k+s} (c_0^\dag)^{s-1-k}\cdot u_k\,.
\ee
The state $w$ is nonzero and non-singular. It satisfies $e_j \cdot w = 0$ for $j=1,2,\dots, 2s'-1$, so it indeed spans the desired one-dimensional submodule in $\repM^{\tinyx{s\!-\!1}}(2s')$. We now prove that $w$ and $y_{-1}$ are linearly independent. To show this, we consider their components along the link state $\wh{u}_k$. For~$w$, it is straightforward to check that this component is equal to $1$. For $y_{-1}$, we first remark that the component of $v^{\tinyx{\sigma',\eps'}}_{k,s'}$ along $(c_0^\dag)^{s'-k} \cdot u_k$ is also equal to $1$. Because 
\be
\Omega^{s'-s+1}(c_0^\dag)^{s'-k} \cdot u_k = \wh{u}_k\,,
\ee
it follows from \eqref{eq:y-1.props} that the component of $v^{\tinyx{\sigma',\eps'}}_{k,s'}$ along $\wh{u}_k$ equals $(\eps' q^{\sigma' k})^{-s'+s-1}$. This component is non-singular for $q \to \eE^{-\ir \pi p/p'}$. As a result, the component of $y_{-1}$ along $\wh{u}_k$ is equal to zero, and $y_{-1}$ and $w$ are linearly independent.\medskip

We thus conclude that $y_{-1}$ and $w$ span two distinct one-dimensional submodules in $\repM^{\tinyx{s\!-\!1}}(2s')$ isomorphic to $\repV_{s-1}(2s')$. This is not possible, since $\repV_{s-1}(2s')$ has only one such submodule. By contradiction, we conclude that our assumption $y_{-1} \neq 0$ is false, and therefore that $v^{\tinyx{\sigma',\eps'}}_{k,s'}$ is not singular for $q\to \eE^{-\ir \pi p/p'}$.\eproof

We remark that Theorem 3.4 of \cite{GL98} constructs a family of homomorphisms from $\repW_{s,\eps q^{\sigma k}}$ to $\repW_{k,\eps q^{\sigma s}}$ whose coefficients in the basis of link states are polynomials in $q$. Moreover, Theorem 5.1 of \cite{GL98} states that this family of homomorphisms is unique, up to a multiplicative constant. Since the coefficient of $v^{\tinyx{\sigma,\eps}}_{k,s}$ along the state $(c_0^\dag)^{s-k}u_k$ is one, there exists a polynomial $p_{k,s}(q)$ such that the state $\wh{v}^{\tinyx{\sigma,\eps}}_{k,s}=p_{k,s}(q)v^{\tinyx{\sigma,\eps}}_{k,s}$ is nonzero, and its coefficients in the basis of link states are polynomials in $q$. This proves that $\wh{v}^{\tinyx{\sigma,\eps}}_{k,s}$ is non-singular for $q\to\eE^{-\ir \pi p/p'}$. The above proposition is a slightly stronger result, as it equivalently states that $p_{k,s}(q)$ is of degree zero.

\begin{Proposition}\label{prop:Q.properties}
Let $q=\eE^{-\ir\pi p/p'}$ with $(p,p')$ as in \eqref{eq:q.standard}, $k\in\frac12 \Zbb_{\ge 0}$ and $s-k\in \Zbb$ satisfying $0\leq k<s \le \frac{p'}2$, $\eps\in\{-1,+1\}$, and $N \in \mathbb Z_{\ge 0}$. 
\begin{enumerate}
\item[(i)] The dimension of $\repQ_{k,\eps q^s}(N)$ is given by
\be
\label{eq:Q.dim}
\dim \repQ_{k,\eps q^s}(N) = D_k(N) - D_{s}(N)- D_{p'-s}(N) + D_{p'-k}(N)\,,
\ee
where $D_k(N)$ is defined in \eqref{eq:DkN}.
\item[(ii)] Let $s'=p'-s$ and $\eps'=(-1)^p \eps$. The submodule $\repR_{k,\eps q^s}(N)$ is given by
\begin{equation} \label{eq:seed.R}
\repR_{k,\eps q^s}(N) = \cL(N,2s)\cdot v^{\tinyx{1,\eps}}_{k,s} + \cL(N,2p'-2s)\cdot v^{\tinyx{-1,\eps'}}_{k,s'} \,,
\end{equation}
and the quotient family $\repQ_{k,\eps q^s}$ is
\be
\repQ_{k,\eps q^s} = \repW_{k,\eps q^s}/\{v^{\tinyx{1,\eps}}_{k,s}\equiv 0\,,\,v^{\tinyx{-1,\eps'}}_{k,s'}\equiv 0\}\,.
\ee
\end{enumerate}
\end{Proposition}
\proof
Using \eqref{eq:W.Loewy}, we find for all $n \in \mathbb Z_{\ge 0}$
\begin{subequations}
\begin{alignat}{1}
& \dim\repR_{np'+k,\eps q^{np'+s}}(N) = d_{np'+s}(N)+d_{(n+1)p'-s}(N)-\dim\repR_{np'+s,\eps q^{np'+k}}(N) \,, \\[0.15cm]
& \dim\repR_{np'+s,\eps q^{np'+k}}(N) = d_{(n+1)p'+k}(N)+d_{(n+1)p'-k}(N)-\dim\repR_{(n+1)p'+k,\eps q^{(n+1)p'+s}}(N) \,,
\end{alignat}
\end{subequations}
and hence
\begin{equation}
\dim\repR_{np'+k,\eps q^{np'+s}}(N) = \sum_{n=0}^\infty\big(
d_{np'+s}(N)+d_{(n+1)p'-s}(N)-d_{(n+1)p'+k}-d_{(n+1)p'-k}(N)
\big) \,.
\end{equation}
Using $\dim \repQ_{np'+k,\eps q^{np'+s}} = \dim \repW_{np'+k,\eps q^{np'+s}} - \dim \repR_{np'+k,\eps q^{np'+s}}$, we readily obtain \eqref{eq:Q.dim}.\medskip

For the second property, using \eqref{eq:W.Loewy} we see that $\repW_{k,\eps q^s}$ admits exactly one subfamily isomorphic to $\repW_{s,\eps q^k}$, and exactly one subfamily isomorphic to $\repW_{p'-s,\eps q^{p'-k}}$, and note that $\repR_{k,\eps q^s}$ is the sum of these two families. This sum is however not a direct sum. Using \cref{prop:seed.W}, we conclude that the spaces of insertion states with parameters $(s,\eps q^k)$ and $(s',\eps' q^{-k})$ for $\repW_{k,\eps q^s}$ are each one-dimensional, and define families of homomorphisms, from $\repW_{s,\eps q^k} \to \repW_{k,\eps q^s}$ and $\repW_{s',\eps' q^{-k}} \to \repW_{k,\eps q^s}$, respectively. The corresponding maps are either injective or zero, depending on the size $N$.\medskip

Let us denote by $\phi^{\tinyx{1,\eps}}_{k,s}$ and $\phi^{\tinyx{-1,\eps'}}_{k,s'}$ the insertion maps, as in \eqref{eq:insertion.map.W}, associated to the states $v^{\tinyx{1,\eps}}_{k,s}$ and $v^{\tinyx{-1,\eps'}}_{k,s'}$, respectively. By construction, we have 
\be
\phi^{\tinyx{1,\eps}}_{k,s}\big(\repW_{s,\eps q^k}(N)\big)=\cL(N,2s)\cdot v^{\tinyx{1,\eps}}_{k,s}\,,
\qquad
\phi^{\tinyx{-1,\eps'}}_{k,s'}\big(\repW_{s',\eps q^{-k}}(N)\big)=\cL(N,2s')\cdot v^{\tinyx{-1,\eps'}}_{k,s'}\,. 
\ee
Since $\phi^{\tinyx{1,\eps}}_{k,s}$ is injective, we have $\phi_{k,s}(\repW_{s,\eps q^k})\simeq \repW_{s,\eps q^k}$. We thus conclude that the submodule of $\repW_{k,\eps q^s}(N)$ isomorphic to $\repW_{s,\eps q^k}(N)$ is $\cL(N,2s)\cdot v_{k,s}$. The same argument holds for the submodule isomorphic to $\repW_{s',\eps' q^{-k}}(N)$.
\eproof

%
\section{Scaling limit and characters}
\label{sec:scaling.limit}
%

In this section, we describe the scaling limit and conformal characters associated to the various modules over $\tl_N(\beta)$ and $\eptl_N(\beta)$ described in \cref{sec:EPTL}.

\subsection{Virasoro minimal models}

Let $p$ and $p'$ be coprime integers satisfying $2 \le p < p'$. The minimal model $\mathcal{M}(p,p')$ consists of the irreducible modules $\repK_{r,s}$ of heighest weight $h=h_{r,s}$, for the Virasoro algebra with central charge $c$, with
\begin{equation} \label{eq:chrs}
c = 1-\frac{6(p-p')^2}{pp'} \,,
\qquad h_{r,s}=\frac{(p'r-ps)^2-(p-p')^2}{4pp'} \,,
\end{equation}
where $r$ and $s$ are integers in the ranges $1\leq r \leq p-1$ and $1\leq s \leq p'-1$. The module $\repK_{r,s}$ is obtained as the quotient of the Verma module $\Verma_{r,s}$ by its maximal proper submodule. For $c$ and $h=h_{r,s}$ as in~\eqref{eq:chrs}, the module $\Verma_{r,s}$ has a double-ladder structure similar to \eqref{eq:W.Loewy}:
\be
\label{eq:V.Loewy}
\Verma_{r,s} \simeq
\left[\begin{pspicture}[shift=-1.4](-0.5,-1.5)(14.5,1.5)
\rput(0,0){$\repK_{r,s}$}
\rput(2,1){$\repK_{p+r,p'-s}$}
\rput(2,-1){$\repK_{r,2p'-s}$}
\rput(5,1){$\repK_{2p+r,s}$}
\rput(5,-1){$\repK_{r,2p'+s}$}
\rput(8,1){$\repK_{3p+r,p'-s}$}
\rput(8,-1){$\repK_{r,4p'-s}$}
\rput(11,1){$\repK_{4p+r,s}$}
\rput(11,-1){$\repK_{r,4p'+s}$}
\rput(14,1){$\dots$}
\rput(14,-1){$\dots$}
\psline{->}(0.4,0.2)(1.6,0.7)
\psline{->}(0.4,-0.2)(1.6,-0.7)
\psline{->}(2.9,1)(4.2,1)
\psline{->}(2.8,-1)(4.2,-1)
\psline{->}(5.8,1)(7.0,1)
\psline{->}(5.8,-1)(7.2,-1)
\psline{->}(9.0,1)(10.2,1)
\psline{->}(8.8,-1)(10.2,-1)
\psline{->}(11.8,1)(13.5,1)
\psline{->}(11.8,-1)(13.5,-1)
\psline{->}(2.2,-0.7)(4.8,0.7)
\psline{->}(2.2,0.7)(4.8,-0.7)
\psline{->}(5.2,-0.7)(7.8,0.7)
\psline{->}(5.2,0.7)(7.8,-0.7)
\psline{->}(8.2,-0.7)(10.8,0.7)
\psline{->}(8.2,0.7)(10.8,-0.7)
\psline{->}(11.2,-0.7)(13.8,0.7)
\psline{->}(11.2,0.7)(13.8,-0.7)
\end{pspicture}
\right]
\ee
The character of $\repK_{r,s}$, which we denote as $\chit_{r,s}$, is obtained as the alternating sum of Verma characters
\be
\label{eq:charVK}
\chit_{r,s} = \chit_{\Verma_{r,s}}
- \sum_{j=1}^\infty\Big(\chit_{\Verma_{(2j-1)p+r,p'-s}}+\chit_{\Verma_{r,2jp'-s}}\Big)
+ \sum_{j=1}^\infty\Big(\chit_{\Verma_{2jp+r,s}}+\chit_{\Verma_{r,2jp'+s}}\Big)
 \,,
\ee
with 
\be
\chit_{\Verma}(\qq) = \tr_{\Verma} \big(\qq^{L_0-c/24}\big)= \frac{\qq^{h-c/24}}{\prod_{n=1}^\infty (1-\qq^n)}\,,
\ee
where $\qq$ is the modular parameter.
The conformal dimensions satisfy the relations
\begin{equation} \label{eq:prop-h_{rs}}
h_{r+\alpha p,s}=h_{r,s-\alpha p'} \,,
\qquad h_{r,s}=h_{-r,-s}=h_{p-r,p'-s} \,,
\end{equation}
for $\alpha \in \mathbb R$. Two Verma modules with identical weights $h$ are isomorphic, and similarly for two irreducible modules. Using \eqref{eq:prop-h_{rs}}, we rearrange the sums in \eqref{eq:charVK} and find
\begin{equation} \label{eq:char.rs}
\chit_{r,s} = \sum_{j \in \Zbb}(\chit_{\Verma_{2jp+r,s}}-\chit_{\Verma_{2jp-r,s}}) \,.
\end{equation}

\subsection[Scaling limit of the families $\repV_{k}$ and $\repQ_{k}$]{Scaling limit of the families $\boldsymbol{\repV_{k}}$ and $\boldsymbol{\repQ_{k}}$}

\label{sec:scaling.TL}
Let $\repM$ be a family of $\tl_N(\beta)$-modules. We define the lattice character on $\repM$ as
\begin{equation}
\chit_{\repM}(M,N) = \tr_{\repM(N)}\Db^{M/2}\,,
\end{equation}
where $M \in 2\mathbb Z_{\ge 0}$ and $N \in \mathbb Z_{\ge 0}$, and $\Db$ is the transfer matrix defined in \eqref{eq:def.tm}.
The scaling limit $\chit_{\repM}(\qq)$ of this character, with $\qq = \eE^{-2 \pi \delta}$ and $\delta \in \mathbb R_{>0}$, is obtained by setting $M = \lfloor N \delta \rfloor$ and taking the limit $N \to \infty$
\be
\label{eq:chi(q)}
\chit_{\repM}(\qq) = \lim_{\substack{N \to \infty \\[0.05cm] M = \lfloor N \delta \rfloor}} 
\eE^{MN f_{\textrm{bulk}}+2Mf_{\textrm{bdy}}} 
\chit_{\repM}(M,N)\,,
\ee
where $f_{\textrm{bulk}}$ and $f_{\textrm{bdy}}$ are bulk and boundary free energies that depend on $\repM$ \cite{dFMS97}. For $\beta$ parameterised as in \eqref{eq:q.beta.standard} and $k\in \frac12 \mathbb Z_{\ge 0}$, the scaling limit of the character on $\repV_k$ is
\begin{equation} \label{eq:charVk}
\chit_{\repV_k}(\qq) = 
\chit_{\Verma_{1,2k+1}}(\qq) - \chit_{\Verma_{1,-2k-1}}(\qq) = \frac{\qq^{h_{1,2k+1}}(1-\qq^{2k+1})}{\prod_{n=1}^\infty (1-\qq^n)}\,.
\end{equation}
\medskip

Let $k\in\{0,\tfrac{1}{2},\dots, \tfrac{p'}{2}-1\}$. From the structure of $\repV_k$ described in \cref{sec:Vk}, we find that the lattice character on $\repQ_{k}$ is
\be
\chit_{\repQ_k}(M,N) = \sum_{j=0}^\infty \chit_{\repV_{k+p'j}}(M,N) - \sum_{j=1}^\infty \chit_{\repV_{-k-1+p'j}}(M,N) \,.
\ee
Using \eqref{eq:charVk}, \eqref{eq:prop-h_{rs}} and \eqref{eq:char.rs}, we obtain the scaling limit of this character
\begin{equation} \label{eq:charQk}
\chit_{\repQ_k}(\qq)
= \sum_{j \in \Zbb}\big(\chit_{\Verma_{1,2k+1+2p'j}}(\qq)-\chit_{\Verma_{-1,2k+1+2p'j}}(\qq)\big)
= \chit_{1,2k+1}(\qq)\,.
\end{equation}

\subsection[Scaling limit of the families $\repW_{k,x}$ and $\repQ_{k,\eps q^s}$]{Scaling limit of the families $\boldsymbol{\repW_{k,x}}$ and $\boldsymbol{\repQ_{k,\eps q^s}}$}

Let $\repM$ be a family of $\eptl_N(\beta)$-modules. We define the lattice character on $\repM$ as
\begin{equation}
\chit_{\repM}(M_1,M_2,N) = \tr_{\repM(N)}\big( \Omega^{M_1} \Tb^{M_2}\big)\,,
\end{equation}
whereas $M_1 \in \mathbb Z$ and $M_2,N \in \mathbb Z_{\ge 0}$, and $\Tb$ is the transfer matrix defined in \eqref{eq:def.tm}. The scaling limit $\chit_{\repM}(\qq)$ of this character, with $\qq = \eE^{2 \pi \ir \tau}$ and $\tau = \tau_r +\ir \tau_i$, is obtained by setting $M_1 = \lfloor N\tau_r \rfloor $ and $M_2 = \lfloor N \tau_i \rfloor$, and taking the limit $N \to \infty$
\be
\label{eq:chi(q)bulk}
\chit_{\repM}(\qq) = \lim_{\substack{N \to \infty \\[0.05cm] M_1 = \lfloor N \tau_r \rfloor \\[0.05cm] M_2 =\lfloor N \tau_i \rfloor}} \eE^{M_2N f_{\textrm{bulk}}} \chit_{\repM}(M_1,M_2,N)\,,
\ee
where $f_{\textrm{bulk}}$ is the same bulk free energy as in \eqref{eq:chi(q)}. The scaling limit of these characters was studied by Pasquier and Saleur \cite{PS90} in the context of the XXZ spin chain. For $\beta$ parameterised as in \eqref{eq:q.beta.standard}, $\mu \in \mathbb C$ and $k\in \frac12 \mathbb Z_{\ge 0}$, the scaling limit of the characters of the standard modules is
\begin{equation}
\label{eq:W.scaling}
\chit_{\repW_{k,\exp(\ir\pi\mu)}} = 
\sum_{m \in \Zbb} (-1)^{\eta m}\chit_{\Verma_{m+\mu,k}} \chib_{\Verma_{m+\mu,-k}} \,,
\end{equation}
where we drop the dependence on $\qq$ for convenience. Here, $\chib$ is obtained from $\chit$ by changing $\mathfrak q \mapsto \mathfrak q^{-1}$, and\footnote{We use the notation $a\equiv b$ mod $n$ when $a$ and $b$ have the same remainder $r$ when they are divided by $n$, with $0\le r < n$, and the notation $a= b$ mod $n$ when assigning to $a$ the value given by the remainder $r$ when $b$ is divided by $n$.} 
\be
\eta_1 = M_1 \mod 2\,,\qquad 
\eta_2 = M_2 \mod 2\,, \qquad
\eta= \eta_1 + \eta_2\,.
\ee
Below, we use the convention
\be
\label{eq:W.conv}
\repW_{-k,x} = \repW_{k,1/x}
\ee 
for standard modules with negative defect numbers. We can easily check that \eqref{eq:W.scaling} is consistent with this choice of convention. Using \eqref{eq:prop-h_{rs}}, we find
\be
\label{eq:charWV}
\chit_{\repW_{k+p'j,\eps q^s}} = 
\eps^{\eta} \sum_{m \in \Zbb} (-1)^{\eta m} \chit_{\Verma_{m-pj,k+s}} \, \chib_{\Verma_{m+pj,-k+s}} \,,
\ee
for $j,k \in \frac12\mathbb Z$, $s \in \mathbb R$ and $\eps \in \{+1,-1\}$.
\medskip

Let $k$, $s$ and $\eps$ be as in~\eqref{eq:xks}, namely $0\leq k<s<p'-k$ and $\eps\in\{-1,+1\}$.
From the Loewy diagrams \eqref{eq:W.Loewy} of $\repW_{k,\eps q^s}$, we find that the character on $\repQ_{k,\eps q^s}$ is
\be 
\chit_{\repQ_{k,\eps q^s}} = \sum_{j=0}^\infty \Big(
\chit_{\repW_{k+p'j,\eps q^{s + p'j}}}
- \chit_{\repW_{s+p'j,\eps q^{k+p'j}}}
- \chit_{\repW_{p'-s+p'j,\eps q^{-k+p'(j+1)}}}
+ \chit_{\repW_{p'-k+p'j,\eps q^{-s+p'(j+1)}}}
 \Big)\,.
\ee
Using \eqref{eq:W.conv} and \eqref{eq:charWV}, we simplify this expression to
\begin{alignat}{2}
\chit_{\repQ_{k,\eps q^{s}}} &= \sum_{j \in \mathbb Z} \Big(
\chit_{\repW_{k+p'j,\eps q^{s + p'j}}}
- \chit_{\repW_{s+p'j,\eps q^{k+p'j}}}
\Big) 
\nn\\&= \eps^\eta\sum_{j,m \in \Zbb}(-1)^{\eta m}
(\chit_{\Verma_{m,s+k}}\chib_{\Verma_{2pj+m,s-k}} - \chit_{\Verma_{m,s+k}}\chib_{\Verma_{2pj-m,s-k}})\,.
\end{alignat}
We now change the summation index $m$ to $m=2 t p+r$ with $r \in \{-p+1,\dots, p-1,p\}$ and $t \in \mathbb Z$. The terms $r=0$ and $r=p$ are easily found to vanish, so we instead set $m=2 t p + \sigma r$ with $t \in \mathbb Z$, $r \in \{1,2,\dots, p-1\}$ and $\sigma \in \{+1,-1\}$. Using \eqref{eq:charVK}, we simplify the result to
\be
\label{eq:charQK}
\chit_{\repQ_{k,\eps q^{s}}} = \eps^\eta\sum_{r=1}^{p-1}(-1)^{\eta r}
\chit_{r,s+k}\,\chib_{r,s-k}\,.
\ee
For $k=0$, we obtain a sum of {\it diagonal terms} $\chit_{r,s}\,\chib_{r,s}$,
whereas for $k>0$ all the contributions $\chit_{r,s+k}\,\chib_{r,s-k}$ are {\it non-diagonal}.

%
\section{ADE modules with boundaries}\label{sec:boundaryADE}
%

In this section, we describe the $\tl_N(\beta)$-modules $\repM_{\g,\mu,a,b}(N)$ associated to the ADE lattice models with fixed boundary conditions. We describe the action of the diagrams in $\cL(N,N')$ on these modules, decompose the space as direct sum of irreducible $\tl_N(\beta)$-modules, and recover the known conformal partition functions.

\subsection{State space and action of diagrams}

Let $N$ be a non-negative integer. The $\tl_N(\beta)$-module $\repM_{\g,\mu,a,b}(N)$ with two boundaries is fixed by a choice of a Lie algebra $\g$ associated to a Dynkin diagram $\mathcal G$, an index $\mu$ associated to an eigenvector with non-vanishing components $S_{a\mu}$, and two nodes $a,b \in \mathcal G$. The basis states are of the form $u_\ba = |a_0, a_1, a_2, \dots, a_N\rangle$, where (i) each $a_i$ is a height of $\mathcal G$, (ii) $a_i$ and $a_{i+1}$ are adjacent on $\mathcal G$ for $i = 0,1,\dots, N-1$, and (iii) the boundary sites are fixed to $a_0 = a$ and $a_N = b$. The dimension of $\repM_{\g,\mu,a,b}(N)$ is
\be
\dim \repM_{\g,\mu,a,b}(N) = \big(A^N\big)_{ab}\,.
\ee
The Dynkin diagrams of type $A,D,E$ are all bipartite, namely their nodes can be coloured in black and white, so that any edge connects two nodes of opposite colours. For $N$ even, $\repM_{\g,\mu,a,b}(N)$ is nonzero only if $a$ and $b$ are of the same color. For $N$ odd, it is nonzero only if $a$ and $b$ have different colors.\medskip

The vector space $\repM_{\g,\mu,a,b}(N)$ is endowed with an action of $\cL(N',N)$, with the loop weight
\begin{equation}
\beta = -q-q^{-1} \qquad \textrm{with} \qquad 
q=-\eE^{\ir\pi m_\mu/p'}
 = \eE^{-\ir\pi p/p'} \,, \qquad 
p=p'-m_\mu \,.
\end{equation}
The action of the generators of $\tl_N(\beta)$ on $\repM_{\g,\mu,a,b}(N)$ is given by 
\be
\label{eq:ejRSOS} 
e_j \cdot \ket{a_0,a_1, \dots, a_N}
= \sum_{a'_j=1}^n \delta_{a_{j-1}a_{j+1}} A_{a'_j,a_{j-1}} \frac{S_{a'_j\mu}}{S_{a_{j+1}\mu}} \, \ket{a_0,a_1, \dots, a_{j-1},a'_j,a_{j+1}, \dots, a_N}\,,
\ee
for $j=1,2,\dots, N$. As shown in \cite{IMD25}, this action can be factored as $e_j = c_j^\dag c_j$, where the operators $c_j: \repM_{\g,\mu,a,b}(N) \to \repM_{\g,\mu,a,b}(N-2)$ and $c^\dag_j: \repM_{\g,\mu,a,b}(N-2) \to \repM_{\g,\mu,a,b}(N)$ act as
\begin{subequations} \label{eq:cj.RSOS}
\begin{alignat}{2}
& c_j \cdot \ket{a_0, a_1, \dots, a_N}
= \, \frac{\delta_{a_{j-1},a_{j+1}}}{S_{a_{j+1}\mu}}\, \ket{a_0,a_1, \dots, a_{j-1},a_{j+2}, a_{j+3}, \dots, a_N} \,, 
\\
& c^\dag_j \cdot \ket{a_0, a_1, \dots, a_{N-2}} = \sum_{a'_j=1}^n A_{a_{j-1},a'_j}S_{a_{j'}\mu}
\, \ket{a_0, a_1, \dots, a_{j-1}, a'_j ,a_{j-1}, a_j, \dots, a_{N-2}} \,.
\end{alignat}
\end{subequations}
We also note that the transfer matrix $D_N$, defined in \eqref{eq:Dab}, is the matrix representative of $\Db$ in $\repM_{\g,\mu,a,b}(N)$.\medskip

Using the same arguments presented in \cite{IMD25} for the periodic case, we obtain the following proposition. 
\begin{Proposition}
The set 
\be 
\repM_{\g,\mu,a,b} = \{\repM_{\g,\mu,a,b}(N) \,|\, N_0, N_0+2, N_0+4, \dots \}\,,
\qquad N_0 = 
\left\{\begin{array}{ll}
0 & a \mathrm{\ and\ } b \mathrm{\ have\ the\ same\ color,}
\\[0.1cm]
\mathrm{1} & a \mathrm{\ and\ } b \mathrm{\ have\ different\ colors,}
\end{array}\right.
\ee
is a family of $\tl_N(\beta)$-modules.
\end{Proposition}

We introduce a symmetric sesquilinear form $\smallaver{\,,\,}: (\repM_{\g,\mu,a,b},\repM_{\g,\mu,a,b}) \to \mathbb C$, defined as
\begin{equation} \label{eq:def.aver}
\aver{u_{\boldsymbol a}, u_{\boldsymbol b}}
= \prod_{j=0}^N \frac{\delta_{a_j, b_j}}{S_{a_j\mu}}\,.
\end{equation}
Because this form is diagonal with nonzero elements, it is clear that its determinant is nonzero and therefore that it is a non-degenerate form. For all states $u_{\boldsymbol a} \in \repM_{\g,\mu,a,b}(N-2)$ and $u_{\boldsymbol b} \in \repM_{\g,\mu,a,b}(N)$, we have
\be
\aver{u_\ba, c_j \cdot u_{\boldsymbol b}}
= \aver{c_j^\dag \cdot u_\ba, u_{\boldsymbol b}} 
= \Bigg[\prod_{i=0}^{j-1}\frac{\delta_{a_i b_i}}{S_{a_i\mu}}\Bigg] 
\frac{\delta_{b_{j-1} b_{j+1}}}{S_{b_{j+1} \mu}} \Bigg[\prod_{i=j}^{N-2}\frac{\delta_{a_i b_{i+2}}}{S_{a_i\mu}}\Bigg]
\qquad j=1,2,\dots,N-1 \,.
\ee
Hence, the operators $c_j$ and $c_j^\dag$ are conjugate in $\repM_{\g,\mu,a,b}$, and 
\be
\aver{u_\ba,e_j \cdot u_{\boldsymbol b}}
= \aver{e_j \cdot u_\ba, u_{\boldsymbol b}} \,, 
\ee
We sometimes use the notation $\aver{u_{\boldsymbol a} |\lambda| u_{\boldsymbol b}}$, understanding that this is equal to both $\aver{u_{\boldsymbol a}, \lambda \cdot u_{\boldsymbol b}}$ and $\aver{\lambda^\dag \cdot u_{\boldsymbol a}, u_{\boldsymbol b}}$.\medskip

Let $K$ be an automorphism of $\mathcal G$, and $a$ and $b$ be nodes of $\mathcal G$ that are fixed points of $K$, namely $K(a) = a$ and $K(b) = b$. In these cases, we define the linear map $K_N:\repM_{\g,\mu,a,b}(N)\to\repM_{\g,\mu,a,b}(N)$ as
\be
\label{eq:KN.def}
K_N \cdot \ket{a_0, a_1, \dots, a_N} = \kappa_\mu^{(N+1)/2} 
\ket{K(a_0), K(a_1), \dots, K(a_N)} \,,
\ee 
where $\kappa_\mu$ is the eigenvalue of $K$ associated to $S_{a\mu}$, and we use the convention
\be
\label{eq:sqrt.kappa}
(\kappa_\mu)^{1/2} = 
\left\{\begin{array}{cl}
1 & \kappa_\mu = 1\,, \\[0.1cm]
\ir & \kappa_\mu = -1\,.
\end{array}\right.
\ee 
As a collection of maps on the family $\repM_{\g,\mu,a,b}$, the action of $K_N$ commutes with $c_j$ and $c_j^\dag$, namely
\be
c_j K_{N} = K_{N-2}\, c_j \,, \qquad
c^\dag_j K_{N-2} = K_{N}\, c^\dag_j\,.
\qquad 
\ee
Moreover, the operator obtained by replacing $K$ by $K^{-1}$ in \eqref{eq:KN.def} is well-defined on $\repM_{\g,\mu,a,b}(N)$. It is the inverse of $K_N$, and we denote it as $\bar K_N$. Hence, $K_N$ defines a family of invertible automorphisms on $\repM_{\g,\mu,a,b}$.
The operators $K_N$ and $\bar K_N$ satisfy
\be
\aver{u_{\boldsymbol{a}}, K_N \cdot u_{\boldsymbol{b}}}
= \aver{\bar K_N \cdot u_{\boldsymbol{a}}, u_{\boldsymbol{b}}} 
= \kappa_\mu^{(N+1)/2}\prod_{i=0}^{N} \frac{\delta_{a_i, K(b_i)}}{S_{a_i \mu}}\,.
\ee

\subsection[Symmetric gauge for $\mu = 1$]{Symmetric gauge for $\boldsymbol{\mu = 1}$}

For $\mu = 1$, we denote the eigenvector components as $S_a = S_{a1}$. They are all positive, and their square roots are real. Moreover, the automorphisms $K$ have the eigenvalue $\kappa_1 = 1$ in all cases.
Using the change of basis
\be
\label{eq:change.of.basis} 
v_{\boldsymbol{a}} = \kket{a_0,a_1,a_2,\dots,a_N} =\bigg[\prod_{j=0}^{N} \sqrt{S_{a_j}}\,\bigg] \ket{a_0,a_1,a_2,\dots,a_N}\,,
\ee
we find
\begingroup
\allowdisplaybreaks
\begin{subequations}
\begin{alignat}{2} 
\label{eq:ejRSOSv2}
& e_j \cdot \kket{a_0,a_1,\dots,a_N}
= \sum_{a'_j=1}^n \delta_{a_{j-1}a_{j+1}} A_{a'_j,a_{j+1}} \sqrt{\frac{S_{a_j}S_{a'_j}}{S_{a_{j-1}}S_{a_{j+1}}}} \, \kket{a_0,a_1,\dots, a_{j-1},a'_j,a_{j+1}, \dots, a_N} \,, 
\\[0.15cm]
&K_N \kket{a_0, a_1, \dots, a_N} = 
\kket{K(a_0), K(a_1), \dots, K(a_N)} \,,
\end{alignat}
\end{subequations}
and
\begin{subequations} \label{eq:cj.RSOS.v2}
\begin{alignat}{2}
& c_j \cdot \kket{a_0, a_1, \dots, a_N} 
= \delta_{a_{j-1},a_{j+1}} \, \sqrt{\frac{S_{a_j}}{S_{a_{j+1}}}}\, \kket{a_0, a_1, \dots, a_{j-1},a_{j+2}, a_{j+3}, \dots, a_N} \,, 
\\
& c^\dag_j \cdot \kket{a_0, a_1, \dots, a_{N-2}} 
= \sum_{a'_j=1}^n A_{a_{j-1},a'_j} \sqrt{\frac{S_{a'_j}}{S_{a_{j-1}}}}
\, \kket{a_0, a_1, \dots, a_{j-1}, a'_{j} ,a_{j-1}, a_j, \dots, a_{N-2}} \,.
\end{alignat}
\end{subequations}
\endgroup

In this gauge, we define the sesquilinear form
\begin{subequations}
\begin{equation}
\aaver{v_{\boldsymbol a},v_{\boldsymbol b}}
= \prod_{j=0}^N \delta_{a_j, b_j} \,.
\end{equation}
This is simply the canonical real scalar product. We then have
\begin{alignat}{3}
\aaver{v_{\boldsymbol a},c_j \cdot v_{\boldsymbol b}}
&= \aaver{c^\dag_j \cdot v_{\boldsymbol a},v_{\boldsymbol b}} \,, 
\qquad
&\aaver{v_{\boldsymbol a},c_j^\dag \cdot v_{\boldsymbol b}}
&= \aaver{c_j \cdot v_{\boldsymbol a}, v_{\boldsymbol b}}\,,
\\[0.15cm]
\aaver{v_{\boldsymbol a},e_j \cdot v_{\boldsymbol b}}
&= \aaver{e_j \cdot v_{\boldsymbol a},v_{\boldsymbol b}} \,, 
\qquad
&\aaver{v_{\boldsymbol a},K_N \cdot v_{\boldsymbol b}} 
&=\aaver{\bar K_N\cdot v_{\boldsymbol a}, v_{\boldsymbol b}}\,.
\end{alignat}
\end{subequations}
In this gauge specific to $\mu = 1$, the matrices representing the generators $e_j$ are real and symmetric, and the corresponding ADE models are unitary.

\subsection[Decomposition of $\repM_{\g,\mu,a,b}$]{Decomposition of $\boldsymbol{\repM_{\g,\mu,a,b}}$}
\label{sec:Mab.decomp}

In this section, we obtain the decomposition of $\repM_{\g,\mu,a,b}$ for each ADE model as a direct sum of irreducible families of modules over $\tl_N(\beta)$.
\medskip

Following \cite{BPZ98}, we define the fused adjacency matrices $J_s$ using the initial conditions and the recursion relation
\begin{equation} \label{eq:rec.J}
J_0= 0 \,, \qquad J_1= \id \,, \qquad J_s= J_{s-1}\,A - J_{s-2}\,.
\end{equation}
The matrices $J_s$ are square matrices of size $n$ with integer entries. They are defined for $s \in \mathbb Z$, and we have $J_{s+1} = U_s(\frac A2)$, where $U_s(x)$ is the $s$-th Chebyshev polynomial of the second kind. They satisfy the folding and periodicity relations
\be \label{eq:symm.J}
J_s = - J_{2p'-s} = J_{2p'+s}\,. 
\ee
Moreover, the matrix element $(J_s)_{ab}$ is nonzero only if (i) $a$ and $b$ have the same colour and $s$ is odd, or (ii) $a$ and $b$ have different colours and $s$ is even.
Since $J_s$ is a polynomial in $A$, each eigenvector of~$A$ with eigenvalue $\beta_\mu=2\cos(\tfrac{\pi m_\mu}{p'})$ is also an eigenvector of $J_s$, with eigenvalue $U_{s-1}\big(\cos(\tfrac{\pi m_\mu}{p'})\big)$. We then have
\begin{equation}
(J_s)_{ab} = \sum_{\mu=1}^n S_{a\mu}S^*_{b\mu}\,\frac{\sin\big(\frac{\pi s m_\mu}{p'}\big)}{\sin \big(\frac{\pi m_\mu}{p'}\big)}\,,
\end{equation}
which yields the identity
\begin{equation} \label{eq:diff.J}
(J_{2k+1})_{ab} - (J_{2k-1})_{ab} = 2\sum_{\mu=1}^n S_{a\mu}S^*_{b\mu} \,\cos\bigg(\frac{2\pi k m_\mu}{p'}\bigg) \,, \qquad k\in\tfrac12\Zbb \,.
\end{equation}
For $k\in\{0,\tfrac{1}{2},\dots,\tfrac{p'}{2}-1\}$, the matrix element $(J_{2k+1})_{ab}$ is a non-negative integer \cite{BPZ98}.

\begin{Proposition}
Let $N \in \mathbb Z_{\ge 0}$. The dimension of $\repM_{\g,\mu,a,b}(N)$ is
\begin{equation}
\label{eq:dim.Mab}
\dim \repM_{\g,\mu,a,b}(N)= \sum_{k=0,\tfrac{1}{2},\dots,\tfrac{p'}{2}-1} (J_{2k+1})_{ab} \, \dim\repQ_k(N) \,.
\end{equation}
\end{Proposition}
\proof
By writing $2\cos \theta = \eE^{\ir \theta}+\eE^{-\ir \theta}$ and using Newton's binomial formula, we obtain the identity
\begin{equation} \label{eq:binomial}
\left(2\cos \theta \right)^N = \sum_{j=-N/2}^{N/2} d_j(N)\, \cos(2j \theta)\,,
\end{equation}
where $d_j(N)$ is defined in \eqref{eq:dimV}.
Here, the sum on~$j$ runs over the integer values for $N$ even and half-integer values for $N$ odd. This yields
\begin{equation} \label{eq:cos^N}
\left( 2\cos \frac{\pi m}{p'} \right)^N = 
\left\{\begin{array}{cl}
\displaystyle\sum_{k=0}^{p'-1} \big(D_k(N)+D_{p'-k}(N)\big) \, \cos \left(\frac{2\pi k m}{p'}\right) & N \text{ even,} \\[0.45cm]
\displaystyle\sum_{k=1/2}^{p'-1/2} \big(D_k(N)+D_{p'-k}(N)\big) \, \cos \left(\frac{2\pi k m }{p'}\right) & N \text{ odd,}
\end{array}\right.
\end{equation}
for $m\in \Zbb$, where $D_k(N)$ is defined in \eqref{eq:DkN}. For $N$ even, we have
\begin{alignat}{1}
\dim \repM_{\g,\mu,a,b}(N)&= \left(A^N \right)_{ab}
= \sum_{\nu=1}^n S_{a\nu} S^*_{b\nu} \bigg( 2\cos \frac{\pi m_\nu}{p'} \bigg)^N \nn \\
&= \sum_{k=0}^{p'-1} \sum_{\nu=1}^n S_{a\nu} S^*_{b\nu} 
\big(D_k(N)+D_{p'-k}(N)\big)\cos\left(\frac{2k\pi m_\nu}{p'}\right)
\nn \\
&=\frac{1}{2} \sum_{k=0}^{p'-1} \big((J_{2k+1})_{ab}-(J_{2k-1})_{ab}\big)(D_k(N)+D_{p'-k}(N))
\nn \\
& = \frac{1}{2} \sum_{k=0}^{p'-1} (J_{2k+1})_{ab} \, \big(D_k(N) + D_{p'-k}(N) - D_{k+1}(N) - D_{p'-k-1}(N)\big) 
\nn \\
&= \sum_{k=0}^{\lfloor(p'-2)/2\rfloor } (J_{2k+1})_{ab} \, \dim \repQ_k(N)
= \sum_{s=1,3,\dots}^{2\lfloor{(p'-1)/2}\rfloor+1} (J_s)_{ab} \, \dim \repQ_{(s-1)/2}(N)
\nn \\
&= \sum_{s=1}^{p'-1} (J_s)_{ab} \, \dim \repQ_{(s-1)/2}(N) \,.
\end{alignat}
Here, we applied the identities \eqref{eq:QTL.dim} and \eqref{eq:diff.J}, and at the last step, we used the fact that $(J_s)_{ab}$ with $s$ even is zero for heights $a$ and $b$ of the same color.
For $N$ odd, the same calculation yields
\be
\dim \repM_{\g,\mu,a,b}(N)= \sum_{k=\frac12, \frac 32, \dots}^{\lfloor(p'-2)/2\rfloor} (J_{2k+1})_{ab} \, \dim \, \repQ_k(N) = \sum_{s=1}^{p'-1} (J_s)_{ab} \, \dim \repQ_{(s-1)/2}(N) \,,
\ee 
ending the proof.
\eproof

\begin{Proposition} \label{prop:dim.insertion}
Let $k\in\{0,\tfrac{1}{2},\dots, \tfrac{p'}{2}-1\}$. The dimension of the insertion space with $2k$ defects in $\repM_{\g,\mu,a,b}(2k)$ is equal to $(J_{2k+1})_{ab}$.
\end{Proposition}
\proof
The insertion space with $2k$ defects in $\repM_{\g,\mu,a,b}(2k)$ is defined as
\begin{equation}
I_k(a,b) = \big\{
u\in \repM_{\g,\mu,a,b}(2k) \, \big|\, c_1\cdot u=c_2\cdot u=\dots =c_{2k-1}\cdot u=0
\big\} \,.
\end{equation}
Let us denote its dimension by $d_k(a,b)$. For $k=0$, each basis state of $\repM_{\g,\mu,a,b}(0)$ is an insertion state with zero defects, so we have $d_0(a,b)=\delta_{ab}=(J_1)_{ab}$. Similarly, for $k=\frac12$, each basis state of $\repM_{\g,\mu,a,b}(1)$ is an insertion state with one defect, so $d_{1/2}(a,b)=A_{ab}=(J_2)_{ab}$. For the other cases with $1 \leq k \leq \frac{p'}2-1$, we introduce the vector space
\begin{equation}
H_k(a,b) = \big\{
u\in \repM_{\g,\mu,a,b}(2k) \, \big| \, c_1\cdot u=c_2\cdot u=\dots =c_{2k-2} \cdot u=0 
\big\} \,.
\end{equation}
Any basis state of $\repM_{\g,\mu,a,b}(2k)$ is of the form $\ket{a,a_1,\dots, a_{2k-2},b',b}$ for some height $b'$, where $\ket{a,a_1,\dots a_{2k-2},b'}$ is a basis state of $\repM_{\g,\mu,a,b'}(2k-1)$, and $b$ and $b'$ are adjacent in $\cal G$. This implies that $H_k(a,b)$ has dimension
\begin{equation}
\dim H_k(a,b) = \sum_{b'=1}^n d_{k-\frac12}(a,b') A_{b'b} \,.
\end{equation}
We now consider the linear map $f:H_k(a,b) \to I_{k-1}(a,b)$ defined as $f(u)=c_{2k-1}\cdot u$ for all $u\in H_k(a,b)$. By construction, the kernel of $f$ is $I_{k}(a,b)$. For each $v\in I_{k-1}(a,b)$, we define the state
\begin{equation}
u = -\frac{[2k-1]
}{[2k]}\, P_{2k-1} \, c^\dag_{2k-1}\cdot v \,.
\end{equation}
From \eqref{eq:WJ.relations}, we find that $f(u)=v$ as well as $c_1\cdot u=c_2\cdot u=\dots =c_{2k-2}\cdot u=0$, and thus $u\in H_k(a,b)$. Hence, $f$ is surjective, and therefore yields a bijection from $H_k(a,b)\big/ I_k(a,b)$ to $I_{k-1}(a,b)$.
As a consequence, the dimensions are related by
\begin{equation}
\dim I_k(a,b) = \dim H_k(a,b) - I_{k-1}(a,b) \,,
\end{equation}
which yields the recursion relation
\begin{equation}
d_k(a,b) = \bigg(\sum_{c=1}^n d_{k-\frac12}(a,c) A_{cb}\bigg) - d_{k-1}(a,b) \,,
\qquad 1\leq k\leq \tfrac{p'}2-1 \,.
\end{equation}
This recursion relation and the initial values $d_0(a,b),d_{1/2}(a,b)$ coincide with \eqref{eq:rec.J} for $s=2k+1$. We thus conclude that $d_k(a,b)=(J_{2k+1})_{ab}$ for each $k\in\{0,\tfrac{1}{2},\dots, \tfrac{p'}{2}-1\}$.
\eproof

We give the specific construction of the insertion states for each family of modules $\repM_{\g,\mu,a,b}$ in \cref{sec:insertionTL}. We now state the main result of this section.

\begin{Theorem} \label{thm:M.decomp.TL}
The family of modules $\repM_{\g,\mu,a,b}$ decomposes as
\begin{equation}
\label{eq:Mab.decomp}
\repM_{\g,\mu,a,b} \simeq \bigoplus_{k=0,\tfrac{1}{2},\dots,\tfrac{p'}{2}-1} (J_{2k+1})_{ab} \, \repQ_k \,.
\end{equation}
\end{Theorem}
\proof Let $k\in\{0,\tfrac{1}{2},\dots, \tfrac{p'}{2}-1\}$. Using \cref{prop:dim.insertion}, we know that there exist $(J_{2k+1})_{ab}$ linearly independent insertion states $w_{k,1},w_{k,2},\dots, w_{k,(J_{2k+1})_{ab}}$ with $2k$ defects in $\repM_{\g,\mu,a,b}(2k)$. The sesquilinear form defined in \eqref{eq:def.aver} is non-degenerate, so we can choose these states to be orthonormal, namely $\smallaver{w_{k,i},w_{k,j}}=\delta_{ij}$ for $i,j\in\{1,2,\dots,(J_{2k+1})_{ab}\}$. For each $i\in\{1,2,\dots,(J_{2k+1})_{ab}\}$ and for all admissible~$N$, we define the nonzero submodule of $\repM_{\g,\mu,a,b}(N)$
\begin{equation}
\repM_{k,i}(N) = \cL_0(N,2k) \cdot w_{k,i}
= \big\{ \lambda\cdot w_{k,i}\,|\, \lambda\in \cL_0(N,2k) \big\} \,,
\end{equation}
which is isomorphic to a quotient of $\repV_k(N)$. From the structure of $\repV_k$ described in \cref{sec:Vk}, we know that the modules $\repV_k(N)$ with distinct values of $k\in\{0,\tfrac{1}{2},\dots, \tfrac{p'}{2}-1\}$ have no common non-trivial submodules, and therefore two submodules $\repM_{k,i}(N)$ and $\repM_{k',j}(N)$ of $\repM_{\g,\mu,a,b}$ with $k\neq k'$ are necessarily direct summands in the decomposition of $\repM_{\g,\mu,a,b}(N)$.
\medskip

We now show that this is also true in the case $k=k'$. Let $k\in\{0,\tfrac{1}{2},\dots, \tfrac{p'}{2}-1\}$, and $i,j\in\{1,2,\dots,(J_{2k+1})_{ab}\}$ with $i\neq j$. For all diagrams $\lambda_1,\lambda_2\in \cL_0(N,2k)$, we have
\begin{equation}
\aver{\lambda_1\,w_{k,i}, \lambda_2\,w_{k,j}} = \aver{w_{k,i}, \lambda\,w_{k,j}}
= \aver{\lambda^\dag\,w_{k,i}, w_{k,j}}\,,
\end{equation}
where $\lambda=\lambda_1^\dag\lambda_2\in\tl_{2k}(\beta)$. If $\lambda$ is not proportional to $\id_{2k}$, then there exists $m\in\{1,2,\dots,2k-1\}$ and $\lambda'\in\cL_0(2k,2k-2)$ 
such that $\lambda=\lambda' c_m$, and thus $\lambda\,w_{k,j}=0$. If $\lambda$ is proportional to $\id_{2k}$, then $\smallaver{\lambda_1\,w_{k,i}, \lambda_2\,w_{k,j}} =0$ because $ \smallaver{w_{k,i}, w_{k,j}} = 0$ due to the orthogonality of the insertion states.
Hence, we have shown that $\smallaver{u,v}=0$ for all $u\in\repM_{k,i}(N)$ and $v\in\repM_{k,j}(N)$. As a consequence, two submodules $\repM_{k,i}(N)$ and $\repM_{k,j}(N)$ of $\repM_{\g,\mu,a,b}(N)$ with $i \neq j$ are direct summands in the decomposition of the module.\medskip

We therefore have constructed a large submodule of $\repM_{\g,\mu,a,b}(N)$:
\begin{equation} \label{eq:sum.Mik}
\bigoplus_{k=0,\tfrac{1}{2},\dots,\tfrac{p'}{2}-1} \bigoplus_{i=1}^{(J_{2k+1})_{ab}} \, \repM_{k,i}(N) \subseteq \repM_{\g,\mu,a,b}(N)\,.
\end{equation}
Each factor $\repM_{k,i}(N)$ is nonzero and isomorphic to quotient of $\repV_k(N)$. This implies that $\dim\repM_{k,i}(N)\geq \dim \repQ_k(N)$. With $\repM_{k,i}(N) \simeq \repQ_k(N)$, the large submodule in the left-hand side of \eqref{eq:sum.Mik} is the smallest possible, and it already exhausts the dimension \eqref{eq:dim.Mab}, which proves \eqref{eq:Mab.decomp}.
\eproof

This result is remarkably simple, and mimics the analogous decomposition found in \cite{BPZ98} in the scaling limit. Moreover, since the matrices $J_s$ are symmetric, we conclude that 
\be
\repM_{\g,\mu,a,b} \simeq \repM_{\g,\mu,b,a}\,, 
\ee
a fact which we have not been able to prove using simple symmetry arguments.
\medskip

For the $A_n$ and $D_n$ models, these module decompositions read
\begingroup
\allowdisplaybreaks
\begin{subequations}
\begin{alignat}{2}
\repM_{A_n,\mu,a,b} & \simeq \bigoplus_{k=|(a-b)/2|}^{\min[(a+b)/2-1,n-(a+b)/2]} \repQ_k \,,
\\[0.4cm]
\repM_{D_n,\mu,a,b} & \simeq 
\left\{\begin{array}{cl}
\displaystyle\bigoplus_{k=|(a-b)/2|}^{(a+b)/2-1} (\repQ_k\oplus \repQ_{n-k-2}) & a,b \le n-2\,,
\\[0.7cm] 
\displaystyle\bigoplus_{k=(n-a-1)/2}^{(n+a-3)/2} \repQ_k & a \le n-2, b \in\{n-1,n\},
\\[0.7cm] 
\displaystyle\bigoplus_{k=0}^{\lfloor(n-1)/2\rfloor} \repQ_{2k} & (a,b) = (n-1,n-1),(n,n),
\\[0.7cm] 
\displaystyle\bigoplus_{k=0}^{\lfloor(n-3)/2\rfloor} \repQ_{2k+1} & (a,b) = (n-1,n).
\end{array}\right.
\end{alignat}
\end{subequations}
\endgroup
The corresponding results for $E_6$, $E_7$ and $E_8$ give rise to lengthy expressions, which we collect in \cref{eq:MTL.E678}.

\subsection{Cylinder partition functions}

Let $\g$ be a Lie algebra, $\mu$ be an exponent satisfying $\textrm{gcd}(m_\mu,p')=1$, $a$ and $b$ be heights in $\mathcal G$, and $K$ be an automorphism of $\mathcal G$. We define the lattice partition function on the cylinder
\be
Z_{\g,\mu,a,b,K}(M,N) = \tr_{\repM_{\g,\mu,a,b}(N)}\big(K_N D^{M}\big)\,,
\ee
where $M,N \in \mathbb Z_{\ge 0}$ and $K_N$ is defined in \eqref{eq:KN.def}. We note that $Z_{\g,\mu,a,b,K}(M,N)$ is nonzero only if $a = K(a)$ and $b = K(b)$.
\medskip

For $K=\id$, from \eqref{eq:Mab.decomp}, we readily obtain
\begin{equation}
Z_{\g,\mu,a,b,\id}(M,N) =
\sum_{k=0,\tfrac{1}{2},\dots,\tfrac{p'}{2}-1} (J_{2k+1})_{ab} \, \chit_{\repQ_k}(M,N) \,.
\end{equation}
For a non-trivial automorphism $K$ of $\mathcal G$, we instead have 
\begin{equation}
Z_{\g,\mu,a,b,K}(M,N) = \sum_{k=0,\tfrac{1}{2},\dots,\tfrac{p'}{2}-1} n_k(a,b) \, \chit_{\repQ_k}(M,N) \,,
\end{equation}
where $n_k(a,b)$ is the trace of $K_{2k}$ on the insertion space with $2k$ defects in $\repM_{\g,\mu,a,b}$.
An argument similar to the proof of \cref{prop:dim.insertion} yields
\begin{equation}
n_k(a,b) = \kappa_\mu^{(2k+1)/2} \, (\wt{J}_{2k+1})_{ab},
\end{equation}
where $\wt{J}_{2k+1} = U_{2k}(\wt{A})$, and $\kappa_\mu$ is the eigenvalue of $K$ and associated to $S_{a\mu}$. By taking the scaling limit as in \cref{sec:scaling.TL}, and using \eqref{eq:charQk}, we obtain the following result.
\begin{Theorem} The conformal partition functions for ADE lattice models on the cylinder with $(a,b)$ boundary conditions are given by
\begin{equation} \label{eq:ZTL.untwisted}
Z_{\g,\mu,a,b,\id}(\qq) = \sum_{k=0,\tfrac{1}{2},\dots,\tfrac{p'}{2}-1} (J_{2k+1})_{ab} \, \chit_{1,2k+1}(\qq) \,,
\end{equation}
for $K = \id$, and by
\begin{equation}
Z_{\g,\mu,a,b,K}(\qq) = \sum_{k=0,\tfrac{1}{2},\dots,\tfrac{p'}{2}-1} n_k(a,b) \, \chit_{1,2k+1}(\qq) 
\end{equation}
for the non-trivial automorphisms.
\end{Theorem}
For $K=\id$, this coincides with the cylinder partition functions of the Virasoro ADE minimal models, found in \cite{BPZ98}. The explicit expressions for the $A_n$ and $D_n$ models with $K=\id$ read
\begin{subequations}
\begin{alignat}{2}
Z_{A_n,\mu,a,b,\id}(\qq) &= \sum_{k=|(a-b)/2|}^{\min[(a+b)/2-1,n-(a+b)/2]}\chit_{1,2k+1}(\qq)\,,
\\[0.1cm]
Z_{D_n,\mu,a,b,\id}(\qq) &= \left\{\begin{array}{cl}
\displaystyle\sum_{k=|(a-b)/2|}^{(a+b)/2-1}\big(\chit_{1,2k+1}(\qq) + \chit_{1,2n-2k-3}(\qq)\big) & a,b \le n-2 \,,
\\[0.6cm]
\displaystyle\sum_{k=(n-a-1)/2}^{(n+a-3)/2}\chit_{1,2k+1}(\qq) & a \le n-2 ,b \in \{n-1,n\},
\\[0.6cm]
\displaystyle\sum_{k=0}^{\lfloor(n-1)/2\rfloor}\chit_{1,4k+1}(\qq) & (a,b) = (n-1,n),
\\[0.6cm]
\displaystyle\sum_{k=0}^{\lfloor(n-1)/2\rfloor}\chit_{1,4k+3}(\qq) & (a,b) = (n-1,n-1),(n,n).
\end{array}\right.
\end{alignat}
\end{subequations}
The similar expressions for $E_6$, $E_7$ and $E_8$ are lengthy, and collected in \cref{eq:ZTL.E678}. Similarly, the twisted cylinder partition functions in the scaling limit read
\begin{subequations}
\begin{alignat}{2}
Z_{A_n,\mu,\frac{n+1}2,\frac{n+1}2,R}(\qq) &= \sum_{k=0}^{(n+1)/2}(-1)^k \,\chit_{1,2k+1}(\qq) \qquad n \textrm{ odd,}
\\[0.1cm]
Z_{D_n,\mu,a,b,P_{(n-1,n)}}(\qq) &= \sum_{k=|(a-b)/2|}^{(a+b)/2-1}\big(\chit_{1,2k+1}(\qq) - \chit_{1,2n-2k-3}(\qq) \big) \qquad a,b \le n-2\,,
\\[0.1cm]
Z_{D_4,\mu,2,2,P_{(134)}}(\qq) &= \chit_{1,1}(\qq) - \chit_{1,3}(\qq) + \chit_{1,5}(\qq)\,,
\\[0.1cm]
Z_{E_6,\mu,a,b,P_{(15,24)}}(\qq) & =
\left\{\begin{array}{cl}
\chit_{1,1}(\qq) - \chit_{1,5}(\qq) + \chit_{1,7}(\qq) - \chit_{1,11}(\qq)
& (a,b)=(3,3), (6,6),
\\[0.15cm] 
\chit_{1,2}(\qq) - \chit_{1,4}(\qq) + \chit_{1,8}(\qq) - \chit_{1,10}(\qq)
& (a,b)=(3,6).
\end{array}\right.
\end{alignat}
\end{subequations}

\subsection{Insertion states}
\label{sec:insertionTL}

In this section, we construct the insertion states $w_k$ for the families $\repM_{\g,\mu,a,b}$ using the ideas of {\it reduction} and {\it induction}. We note that this construction is not needed to prove \cref{thm:M.decomp.TL}. It is however used in \cref{sec:insertion.k>0} to construct the insertion states $w_{k,x}$ in the periodic case, and thus allows us to prove \cref{thm:M.decomp}.\medskip

Let us first suppose that we succeeded in constructing an insertion state $w_k(a,b)$ with $2k$ defects in $\repM_{\g,\mu,a,b}$ for some value $k$. We write it in component form as
\be
\label{eq:uNab}
w_k(a,b) = \sum_{\boldsymbol{a}} w_{a_1,a_2,\dots,a_{2k-1}}\ket{a,a_1,a_2, \dots, a_{2k-1},b}
\ee
where the sum over $\boldsymbol{a}$ runs over all the basis states of $\repM_{\g,\mu,a,b}(2k)$.
Reduction allows us to construct shorter insertion states from $w_k(a,b)$. For each neighbor $b'$ of $b$ in $\mathcal G$, we define the state
\be
w_{k-\frac12}(a,b') = \sum_{\boldsymbol{a'}} w_{a_1,a_2,\dots,a_{2k-2},b'}\ket{a,a_1,a_2, \dots, a_{2k-2},b'}\,,
\ee
where the sum over $\boldsymbol{a'}$ runs over all states in $\repM_{\g,\mu,a,b'}(2k-1)$, and the components are the same as in \eqref{eq:uNab}. Since $w_k(a,b)$ is an insertion state by assumption, we have $c_j \cdot w_{k-1/2}(a,b') = 0$ for $j=1,2, \dots, 2k-2$. It may happen that $w_{k-1/2}(a,b') = 0$ for some values of $b'$, but not for all neighbors of $b$, as otherwise $w_k(a,b)$ would also be zero. If $w_{k-1/2}(a,b') \neq 0$, then it is an insertion state with $(2k-1)$ defects for the family $\repM_{\g,\mu,a,b'}$. Reduction thus gives us a way to construct insertion states recursively on decreasing values of $k$.
\smallskip

Conversely, we use induction to construct insertion states on increasing values of $k$. Let us suppose that we managed to construct all the insertion states $w_{k-1/2}(a,b')$ with $(2k-1)$ defects in $\repM_{\g,\mu,a,b'}$ for some $k$, for all neighbors $b'$ of $b$. Then we can construct the states
\be
w_k(a,b) = \sum_{b'\sim b} y_{b'}\ket{w_{k-1/2}(a,b'),b}
\ee 
for some constants $y_{b'}$, where $\ket{w_{k-1/2}(a,b'),b}$ is obtained by extending each state in the linear combination $w_{k-1/2}(a,b')$ by $b$. One then searches for values of the constants $y_{b'}$ such that $c_{2k-1} \cdot u = 0$. A nonzero solution produces an insertion state with $2k$ defects in $\repM_{\g,\mu,a,b}$. Any insertion state can be constructed using this induction algorithm. Induction thus gives us a way to construct insertion states recursively on increasing values of $k$.
\medskip

Below, we apply these ideas and construct all the insertion states for the $A_n$ and $D_n$ models. For the $E_6$, $E_7$ and $E_8$ models, there are finitely many insertion states to construct. We checked using a computer software that the induction algorithm described above allows us to construct the corresponding insertion states, even in the cases where some factors $\repQ_k$ have non-trivial multiplicities.

\paragraph{$\boldsymbol{A_n}$ models.}

For the $A_n$ models, we construct an insertion state $w_k(a,b)$ for each $k$ satisfying 
\be
|a-b|\le 2k \le \min(a+b-2,2n-a-b)\,. 
\ee
Outside of this range, we set $w_k(a,b) = 0$. For $2k=|a-b|$, the module $\repM_{A_n,\mu,a,b}(2k)$ is one-dimensional, and we readily observe that its single state
\be
w_k(a,b) = \left\{\begin{array}{ll}
\ket{a,a+1,\dots,b} & 1\le a \le b \le n\,, \\[0.1cm]
\ket{a,a-1,\dots,b} & 1\le b \le a \le n\,,
\end{array}\right.
\qquad 2k = |a-b| \,,
\ee
is an insertion state. For larger values of $k$, the insertion state is defined from the recursive relation
\be
w_k(a,b) = \ket{w_{k-\frac12}(a,b-1),b}+(-1)^{(2k+b-a)/2}\ket{w_{k-\frac12}(a,b+1),b}\,.
\ee
From this definition, it readily follows that $c_j \cdot w_k(a,b) = 0$ for $j=1,2,\dots, 2k-2$. Moreover, we have
\begin{alignat}{2}
c_{2k-1} \cdot w_k(a,b) 
&= c_{2k-1} \cdot \ket{w_{k-1}(a,b-2),b-1,b}
+ c_{2k-1} \cdot \ket{w_{k-1}(a,b+2),b+1,b}
\nn\\
&+(-1)^{(2k+b-a)/2}c_{2k-1} \cdot \Big(\ket{w_{k-1}(a,b),b+1,b}
-\ket{w_{k-1}(a,b),b-1,b} \Big) = 0\,,
\label{eq:simple.proof}
\end{alignat}
thus proving that $w_k(a,b)$ is an insertion state. We also note that these insertion states have simple expressions for the components, namely
\begin{equation} \label{eq:u.An}
w_{a_1,a_2,\dots, a_{2k-1}} \propto \ir^{a_1+a_2+\dots+a_{2k-1}}\,,
\end{equation}
where $\propto$ denotes an equality up to a proportionality factor independent of $a_1,a_2, \dots, a_{2k-1}$. 

\paragraph{$\boldsymbol {D_n}$ models.} 
We construct insertion states for the $D_n$ models using similar ideas. We first note that the module $\repM_{D_n, \mu,a,b}(N)$ may have zero, one or two insertion states, depending on $N$. Moreover, for $a,b \le n-2$, the action of $K_N$ with $K = P_{(n-1,n)}$ commutes with the action of $\eptl_N(\beta)$. In these cases, we choose a basis of insertion states that are eigenvectors of $K_N$, with eigenvalues $+1$ and $-1$, and denote them by $w^+_k(a,b)$ and $w^-_k(a,b)$, respectively. In the cases where $a$ and/or $b$ are equal to $n-1$ or $n$, the boundary condition is not invariant under the action of $K_N$, and there is at most one insertion state, denoted by $w_k(a,b)$. We now construct the insertion states
\begin{alignat}{2}
w^+_k(a,b) \quad & \textrm{for } \ \ 1\le a,b \le n-2\,, \quad |a-b| \le 2k \le (a+b-2)\,,\nn\\[0.1cm]
w^-_k(a,b) \quad & \textrm{for } \ \ 1\le a,b \le n-2\,, \quad 2n-a-b-2 \le 2k \le 2n-|a-b|-4\,,\\[0.1cm]\nn
w_k(a,b) \quad & \textrm{for }
\left\{\begin{array}{ll}
1\le a \le n-2\,,\, b \in \{n-1,n\}\,, &n-a-1 \le 2k \le n+a-3\,,\\[0.1cm]
a \in \{n-1,n\}\,,\, 1\le b \le n-2\,, &n-b-1 \le 2k \le n+b-3\,,\\[0.1cm]
(a,b) \in\{(n-1,n-1),(n,n)\}\,, & 0 \le 2k \le 2n-4\,,\, 2k \in 4 \mathbb Z\,,\\[0.1cm]
(a,b) \in\{(n-1,n),(n,n-1)\}\,, &1 \le 2k \le 2n-5\,,\, 2k \in 4 \mathbb Z+2\,.
\end{array}\right.
\end{alignat}
Outside of these range, we set the corresponding states to zero. The initial condition for the inductive construction is
\be
\begin{array}{lll}
w^+_k(a,b) = \ket{a,a+1,\dots,b} & 1\le a \le b \le n-2\,, &2k=b-a\,,\\[0.15cm]
w^+_k(a,b) = \ket{a,a-1,\dots,b} & 1\le b \le a \le n-2\,, &2k=a-b\,,\\[0.15cm]
w_k(a,n-1) = \ket{a,a+1,\dots, n-2,n-1} &1\le a \le n-2\,, &2k=n-a-1\,,\\[0.15cm]
w_k(a,n) = \ket{a,a+1,\dots, n-2,n} &1\le a \le n-2\,, &2k=n-a-1\,,\\[0.15cm]
w_k(n-1,b) = \ket{n-1,n-2, \dots, b+1,b} &1\le b \le n-2\,, &2k=n-b-1\,,\\[0.15cm]
w_k(n,b) = \ket{n,n-2, \dots, b+1,b} &1\le b \le n-2\,, &2k=n-b-1\,,\\[0.15cm]
w_0(n-1,n-1) = \ket{n-1}\,,\\[0.15cm]
w_0(n,n) = \ket{n}\\[0.15cm]
w_2(n-1,n) = \ket{n-1,n-2,n}\,,\\[0.15cm]
w_2(n,n-1) = \ket{n,n-2,n-1}\,.\\[0.15cm]
\end{array}
\ee
In all these cases, $\repM_{D_n,\mu,a,b}(N)$ is one-dimensional and we easily check that these are indeed insertion states. The other insertion states are constructed recursively as
\begingroup
\allowdisplaybreaks
\begin{subequations}
\begin{alignat}{2}
&\begin{array}{ll}
w^+_k(a,b) &\hspace{-0.2cm}= \ket{w^+_{k-\frac12}(a,b-1),b} 
\\[0.2cm]&\hspace{-0.2cm}
+ (-1)^{(2k+b-a)/2} \ket{w^+_{k-\frac12}(a,b+1),b} 
\end{array}
&&\hspace{-1.7cm}\left\{\begin{array}{l}
1\le a \le n-2\,,\, 1\le b \le n-3\,, \\[0.15cm]
|a-b|<2k\le a+b-2\,,
\end{array}\right.
\\[0.25cm]
&\begin{array}{ll}
w^-_k(a,b) &\hspace{-0.2cm}= \ket{w^-_{k-\frac12}(a,b-1),b} 
\\[0.2cm]&\hspace{-0.2cm}
+ (-1)^{(2k+b-a)/2} \ket{w^-_{k-\frac12}(a,b+1),b} 
\end{array}
&&\hspace{-1.7cm}\left\{\begin{array}{l}
1\le a \le n-2\,,\, 1\le b \le n-3\,, \\[0.15cm]
2n-a-b-2\le 2k\le 2n-|a-b|-4\,,
\end{array}\right.
\\[0.15cm]
&\begin{array}{ll}
w^+_k(a,n-2) &\hspace{-0.2cm} = \ket{w^+_{k-\frac12}(a,n-3),n-2} 
\\[0.2cm]&\hspace{-0.2cm}
+ \tfrac12(-1)^{(2k+n-a-2)/2}\ket{w_{k-\frac12}(a,n-1),n-2} 
\\[0.2cm]&\hspace{-0.2cm}
+ \tfrac12(-1)^{(2k+n-a-2)/2}\ket{w_{k-\frac12}(a,n),n-2} 
\end{array}
&&\left\{\begin{array}{l}
1\le a \le n-2\,, \\[0.15cm]
n-a \le 2k\le n+a-4\,,
\end{array}\right.
\\[0.15cm]
&\begin{array}{ll}
w^-_k(a,n-2) &\hspace{-0.2cm}= \ket{w^-_{k-\frac12}(a,n-3),n-2} 
\\[0.2cm]&\hspace{-0.2cm}
+ \ket{w_{k-\frac12}(a,n-1),n-2} 
\\[0.2cm]&\hspace{-0.2cm}
-\ket{w_{k-\frac12}(a,n),n-2} 
\end{array}
&&\left\{\begin{array}{l}
1\le a \le n-2\,, \\[0.15cm]
n-a \le 2k\le n+a-2\,,
\end{array}\right.
\\[0.15cm]
&\begin{array}{ll}
w_k(a,n-1) &\hspace{-0.2cm}= \ket{w^+_{k-\frac12}(a,n-2),n-1} 
\\[0.2cm]&\hspace{-0.2cm}
+ \frac12(-1)^{(2k+n-a-1)/2}\ket{w^-_{k-\frac12}(a,n-2),n-1} 
\end{array}
&&\left\{\begin{array}{l}
1\le a \le n-2\,, \\[0.15cm]
n-a+1 \le 2k\le n+a-3\,,
\end{array}\right.
\\[0.15cm]
&\begin{array}{ll}
w_k(a,n) &\hspace{-0.2cm}= \ket{w^+_{k-1}(a,n-2),n} 
\\[0.2cm]&\hspace{-0.2cm}
+ \frac12 (-1)^{(2k+n-a+1)/2}\ket{w^-_{k-\frac12}(a,n-2),n} 
\end{array}
&&\left\{\begin{array}{l}
1\le a \le n-2\,, \\[0.15cm]
n-a+1 \le 2k\le n+a-3\,,
\end{array}\right.
\\[0.15cm]
&\begin{array}{ll}
w_k(a,b) &\hspace{-0.2cm}= \ket{w_{k-\frac12}(a,b-1),b} 
\\[0.2cm]&\hspace{-0.2cm}
+ (-1)^{(2k-n+b+1)/2}\ket{w_{k-\frac12}(a,b+1),b} 
\end{array}
&&\left\{\begin{array}{l}
a \in \{n-1,n\},\, 1\le b \le n-3\,, \\[0.15cm]
n-b+1 \le 2k\le n+b-3\,,
\end{array}\right.
\\[0.15cm]
&\begin{array}{ll}
w_k(a,n-2) &\hspace{-0.2cm}= \ket{w_{k-\frac12}(a,n-3),n-2} 
\\[0.2cm]&\hspace{-0.2cm}
+ (-1)^{(2k-1)/2}\ket{w_{k-\frac12}(a,n-1),n-2}
\\[0.2cm]&\hspace{-0.2cm}
+ (-1)^{(2k-1)/2}\ket{w_{k-\frac12}(a,n),n-2}
\end{array}
&&\left\{\begin{array}{l}
a \in \{n-1,n\}\,,\\[0.15cm]
1 \le 2k\le 2n-5\,,
\end{array}\right.
\\[0.15cm]
&\begin{array}{l}
w_k(n-1,n-1) = \ket{w_{k-\frac12}(n-1,n-2),n-1} \\[0.2cm]
w_k(n,n) = \ket{w_{k-\frac12}(n,n-2),n} 
\end{array}
&&\left\{\begin{array}{l}
0 \le 2k \le 2n-4\,,
\\[0.15cm]
2k\in 4 \mathbb Z\,,
\end{array}\right.
\\[0.15cm]
&\begin{array}{ll}
w_k(n-1,n) = \ket{w_{k-\frac12}(n-1,n-2),n} \\[0.15cm]
w_k(n,n-1) = \ket{w_{k-\frac12}(n,n-2),n-1} 
\end{array}
&&\left\{\begin{array}{l}
1 \le 2k \le 2n-5\,,
\\[0.15cm]
2k\in 4 \mathbb Z + 2\,.
\end{array}\right.
\end{alignat}
\end{subequations}
\endgroup
Using the same arguments as in \eqref{eq:simple.proof}, it is straightforward to check that these are indeed insertion states.

%
\section{ADE modules with periodic boundary conditions}
\label{sec:periodicADE}
%

\subsection{State space and action of diagrams}\label{sec:state.space.I}

Let $N$ be a non-negative even integer, $\g$ be a Lie algebra, and $\mu$ be an index of $A$ with $\textrm{gcd}(m_\mu,p')=1$. The module $\repM_{\g,\mu,\id}(N)$ with untwisted periodic boundary conditions is defined on the vector space spanned by configurations $u_{\boldsymbol{a}}=\ket{a_0,a_1,\dots,a_N}$ subject to the following conditions: (i) $a_i \in \mathcal G$ for $i = 0,1,\dots,N$, (ii) $a_i$ and $a_{i+1}$ are adjacent on $\mathcal G$ for $i = 0,1,\dots, N-1$, and (iii) $a_N=a_0$.
\medskip

Let $K$ be an automorphism of the Dynkin diagram of $\g$. The module $\repM_{\g,\mu,K}(N)$ with twisted periodic boundary conditions is constructed using a basis $\{S_{a\mu}\}$ of common eigenvectors of $A$ and $K$. The configurations spanning the vector space are like in the untwisted case, except for condition (iii) which is replaced by $a_N=K(a_0)$. The dimension of $\repM_{\g,\mu,K}(N)$ is
\begin{equation} \label{eq:M.dim}
\dim \, \repM_{\g,\mu,K}(N) = \tr(K A^N) = \sum_{\mu=1}^n \kappa_\mu (\beta_\mu)^N \,.
\end{equation}
We recall that the non-trivial automorphisms are listed in \cref{sec:Dynkin}, and that the Dynkin diagrams of type $A,D,E$ are all bipartite. If $K$ preserves the colour of the nodes of $\mathcal G$, the module $\repM_{\g,\mu,K}(N)$ is nonzero only for $N$ even. If $K$ exchanges the colour of the nodes of $\mathcal G$, the module $\repM_{\g,\mu,K}(N)$ is nonzero only for $N$ odd. For the ADE models, the second case occurs only for $\repM_{A_n,\mu,R}(N)$ with $K=R$ and $n$ even. We also note that this model with $\mu \in 2 \Zbb$ is the only one with $\kappa_\mu \neq 1$, with instead $\kappa_\mu = -1$.\medskip

The vector space $\repM_{\g,\mu,K}(N)$ is endowed with a representation of $\eptl_N(\beta)$. The action of the generators $e_j$ with $j=1,2,\dots,N-1$ is given by \eqref{eq:ejRSOS}, and is supplemented by
\begin{subequations}
\begin{alignat}{2}
\label{eq:e0RSOS}
e_0 \cdot \ket{a_0,a_1, \dots, a_N}
&= \sum_{a'_0=1}^n \delta_{a_{N-1},K(a_1)} A_{a'_0,a_1} \frac{S_{a'_0\mu}}{S_{a_1\mu}} \, \ket{a'_0,a_1, \dots, K(a'_0)}\,,
\\[0.15cm]
\label{eq:OmegaRSOS}
\Omega \cdot \ket{a_0,a_1, \dots, a_N} &= (\kappa_\mu)^{1/2}\,
\ket{a_1,a_2,\dots, a_N, K(a_1)} \,, \\[0.15cm]
\qquad 
\Omega^{-1} \cdot \ket{a_0,a_1, \dots, a_N}
&= (\kappa_\mu)^{-1/2} \ket{K^{-1}(a_{N-1}),a_0,a_1,\dots, a_{N-1}}\,,
\\[0.15cm]
\label{eq:f}
f \cdot \ket{a_0}&= \sum_{a'_0=1}^n \delta_{K(a'_0),a'_0} \, A_{a'_0,a_0} \, \ket{a'_0} \,,
\end{alignat}
\end{subequations}
where we use the convention \eqref{eq:sqrt.kappa}.\footnote{We note that by instead setting $(\kappa_\mu)^{1/2} = -1$ and $-\ir$ for $\kappa_\mu = 1$ and ${-1}$, respectively, one also obtains a module over $\eptl_N(\beta)$. By \cite[Eq.~(2.70)]{IMD25}, this module is isomorphic to $\repM_{\g,\mu,K^{-1}}(N)$.} The action of the generators $c_j:\repM_{\g,\mu,K}(N) \to \repM_{\g,\mu,K}(N-2)$ and $c^\dag_j:\repM_{\g,\mu,K}(N-2) \to \repM_{\g,\mu,K}(N)$ for $j =1,2,\dots,N-1$ is given in \eqref{eq:cj.RSOS}. For $j=0$, it reads
\begin{subequations} \label{eq:c0c0dag.RSOS}
\begin{alignat}{2}
\label{eq:c0.RSOS}
& c_0 \cdot \ket{a_0, a_1, \dots, a_N} 
= (\kappa_\mu)^{-1/2} \frac{\delta_{a_{N-1},K(a_1)}}{S_{a_1\mu}} 
\,\ket{a_1, a_2, \dots, a_{N-1}} \,, 
\\
& c_0^\dag \cdot \ket{a_0, a_1, \dots, a_{N-2}} 
= (\kappa_\mu)^{1/2}
\sum_{a'_0=1}^n A_{a_{N-2},K(a'_0)}
S_{a'_0\mu} \,\ket{a'_0,a_0, a_1, \dots, a_{N-2},K(a'_0)} \,.
\end{alignat}
\end{subequations}
These operators satisfy all the relations defining the diagram spaces $\cL(N,N')$ given in \cite[Eq.~(2.6)]{IMD25}. They also satisfy the relations
\be
e_j = c_j^\dag c_j\,, \qquad 
\Omega = c_1\, c_0^\dag\,, \qquad
\Omega^{-1} = c_0\, c_1^\dag\,.
\ee
We also note that the transfer matrix $T_N$, defined in \eqref{eq:Tab}, is the matrix representative of $\Tb$ in $\repM_{\g,\mu,K}(N)$.\medskip

\begin{Proposition}[\!\!\cite{IMD25}]
The set of modules 
\begin{equation}
\repM_{\g,\mu,K} = \{\repM_{\g,\mu,K}(N)\,|\,N=N_0,N_0+2,N_0+4, \dots \},
\qquad N_0 \in \{0,1\}\,,
\end{equation}
is a family of $\eptl_N(\beta)$-modules. 
\end{Proposition}
Here, we have $N_0 = 0$ in all cases except for $A_n$ with $K = R$ and $n$ even, in which case $N_0 = 1$.\medskip

We introduce a sesquilinear form $\smallaver{\,,\,}: (\repM_{\g,\mu,K^{-1}},\repM_{\g,\mu,K}) \to \mathbb C$, defined as
\begin{equation} \label{eq:bilinear.form}
\aver{u_{\boldsymbol a},u_{\boldsymbol b}}
= \prod_{j=1}^N \frac{\delta_{a_j, b_j}}{S_{a_j\mu}}
\end{equation}
for $N \ge 1$, and $\aver{a,b} = \delta_{a,b}$ for $N=0$.
For all states $u_{\boldsymbol{a}} \in \repM_{\g,\mu,K^{-1}}(N-2)$ and $u_{\boldsymbol{b}}\in \repM_{\g,\mu,K}(N)$, one can verify that
\be
 \aver{u_{\boldsymbol{a}}, c_j \cdot u_{\boldsymbol{b}}}
= \aver{c_j^\dag \cdot u_{\boldsymbol{a}}, u_{\boldsymbol{b}}}\,, 
\qquad j = 0,1, \dots, N-1\,.
\ee
Hence, the operators $c_j$ and $c_j^\dag$ are conjugate, and we have
\be
\aver{u_{\boldsymbol a},e_j \cdot u_{\boldsymbol b}}
= \aver{e_j \cdot u_{\boldsymbol a}, u_{\boldsymbol b}} \,, 
\qquad
\aver{u_{\boldsymbol a},\Omega^{\pm 1} \cdot u_{\boldsymbol b}}
= \aver{\Omega^{\mp 1} \cdot u_{\boldsymbol a}, u_{\boldsymbol b}} \,,
\qquad
\aver{u_{\boldsymbol a},f \cdot u_{\boldsymbol b}}
= \aver{f \cdot u_{\boldsymbol a}, u_{\boldsymbol b}} \,.
\ee
We sometimes use the notation $\aver{u_{\boldsymbol a}|\lambda|u_{\boldsymbol b}}$.\medskip

Let $L$ be an automorphism of $\mathcal G$ that may be different from $K$. For all indices $\mu$ satisfying $\textrm{gcd}(m_\mu,p')=1$, the eigenvalue $\beta_\mu$ of $A(\g)$ is non-degenerate. This implies that $S_{a\mu}$ is also an eigenvector of the matrix $L_{ab}=\delta_{L(a),b}$. We denote by $\gamma_\mu$ the corresponding eigenvalue, which can take the values in $\{+1,-1\}$ for the ADE series. We have
\begin{equation}
S_{L(a)\mu} = \sum_{b=1}^n \delta_{L(a),b} \, S_{b\mu}
= (LS)_{a\mu} = \gamma_\mu \, S_{a\mu} \,.
\end{equation}

We define the action of an invertible operator $L_N: \repM_{\g,\mu,K} \to \repM_{\g,\mu,LKL^{-1}}$ as
\be
\label{eq:LN.def}
L_N \cdot \ket{a_0, a_1, a_2, \dots, a_N} = 
(\gamma_\mu)^{N/2}\,\ket{L(a_0), L(a_1), L(a_2), \dots, L(a_N)} \,,
\ee 
with $(\gamma_\mu)^{1/2}$ fixed using the same convention as \eqref{eq:sqrt.kappa}. This operator satisfies
\be
\aver{u_{\boldsymbol{a}}, L_N \cdot u_{\boldsymbol{b}}}
= \aver{\bar L_N \cdot u_{\boldsymbol{a}}, u_{\boldsymbol{b}}} 
= (\gamma_\mu)^{N/2} \prod_{i=1}^N \frac{\delta_{a_i, L(b_i)}}{S_{a_i \mu}}\,,
\ee
where $\bar L_N = (L_N)^{-1}$

\begin{Proposition}
Let $K$ and $L$ be two graph automorphisms of the Dynkin diagram $\mathcal G$. Then the linear maps $L_N:\repM_{\g,\mu,K}(N)\to\repM_{\g,\mu,LKL^{-1}}(N)$ satisfy
\be
L_{N-2} \, c_j =c_j L_{N}\,,
\qquad
L_N \,c^\dag_j =c^\dag_j L_{N-2}\,,
\ee
for all $j$, and thus define a family of isomorphisms from $\repM_{\g,\mu,K}$ to $\repM_{\g,\mu,LKL^{-1}}$.
\end{Proposition}
\proof 
For $j=1,2,\dots,N-2$, we have
\begingroup
\allowdisplaybreaks
\begin{subequations}
\label{eq:cjL}
\begin{alignat}{2}
& c_j \, L_N \cdot \ket{a_0,a_1, \dots, a_N} \nn \\
&\qquad = (\gamma_\mu)^{N/2}\,\frac{\delta_{L(a_{j-1}),L(a_{j+1})}}{S_{L(a_{j+1}) \mu}} \, \ket{L(a_0),L(a_1),\dots, L(a_{j-1}),L(a_{j+2}),L(a_{j+3}), \dots, L(a_N)} 
\nn\\
&\qquad = (\gamma_\mu)^{(N-2)/2}\, \frac{\delta_{a_{j-1},a_{j+1}}}{S_{a_{j+1} \mu}} \, L_{N-2} \cdot \ket{a_0,a_1,\dots, a_{j-1},a_{j+2},a_{j+3}, \dots, a_N} 
\nn\\
&\qquad = L_{N-2} \, c_j \cdot \ket{a_0,a_1, \dots, a_N} \,,
\\[0.25cm]
& c^\dag_j \, L_{N-2} \cdot \ket{a_0,a_1, \dots, a_{N-2}} \nn \\
&\qquad =(\gamma_\mu)^{(N-2)/2} \sum_{b'_j=1}^n A_{L(a_{j-1}),b'_j} \, S_{b'_j\mu}
\, \ket{L(a_0), L(a_1), \dots, L(a_{j-1}), b'_j, L(a_{j-1}),\dots, L(a_{N-2})} 
\nn\\
&\qquad =(\gamma_\mu)^{N/2} \sum_{a'_j=1}^n A_{a_{j-1},a'_j} \, S_{a'_j\mu}
\, L_N \cdot \ket{a_0, a_1, \dots, a_{j-1}, a'_j, a_{j-1}, \dots, a_{N-2}} 
\nn\\
&\qquad = L_N \, c^\dag_j \cdot \ket{a_0,a_1, \dots, a_{N-2}}
 \,,
\end{alignat}
\end{subequations}
\endgroup
where we performed the change of variables $b'_j=L(a'_j)$. Similarly, we have
\begin{subequations}
\label{eq:c_0L}
\begin{alignat}{2}
c_0 \, L_N \, \ket{a_0,a_1, \dots, a_N} 
&= (\gamma_\mu)^{N/2} \frac{\delta_{L(a_{N-1}),KL(a_1)}}{S_{L(a_1)\mu}} \,\ket{L(a_1), L(a_2),\dots, L(a_{N-1})} 
\nn\\
&= (\gamma_\mu)^{(N-2)/2} \frac{\delta_{a_{N-1},K'(a_1)}}{S_{a_1\mu}} \,L_{N-2} \ket{a_1,a_2,\dots, a_{N-1}} 
\nn\\
&= L_{N-2}\, c_0 \,\ket{a_0, a_1, \dots, a_N}\,,
\\[0.25cm]
c_0^\dag \, L_{N-2}\, \ket{a_0, a_1, \dots, a_{N-2}} 
&= (\gamma_\mu)^{(N-2)/2}\sum_{b'_0=1}^n A_{L(a_0),b'_0}
S_{b'_0\mu} \,\ket{b'_0,L(a_0),\dots, L(a_{N-2}),KL(a'_0)} 
\nn\\
&= (\gamma_\mu)^{N/2}\sum_{a'_0=1}^n A_{a_0,a'_0}
S_{a'_0\mu} \,L_N\, \ket{a'_0,a_0,\dots, a_{N-2},K'(a'_0)} 
\nn\\
&=L_N\, c_0^\dag \, \ket{a_0, a_1, \dots, a_{N-2}}\,,
\end{alignat}
\end{subequations}
ending the proof.
\eproof

\noindent It thus suffices to study $\repM_{\g,\mu,K}$ for one automorphism $K$ in each conjugacy class of the symmetry group of $\mathcal G$. Setting $K=LKL^{-1}$, we obtain the following corollary.
\begin{Corollary}\label{prop:L.commutes}
Let $K$ and $L$ be two commuting graph automorphisms.
Then the action of $L_N$ defines a family of automorphisms of $\repM_{\g,\mu,K}$.
\end{Corollary}

We end this section by noting that, for the special case $L=K$, we have
\begin{equation}
\Omega^N\cdot u = K_N \cdot u \,, \qquad \textrm{for all } \quad u \in \repM_{\g,\mu,K}(N) \,.
\end{equation}
There always exists an integer $m$ such that $K^m=\id$, so that we have $\Omega^{mN}= K_N^m = (\kappa_\mu)^{mN/2}\,\id$ in $\repM_{\g,\mu,K}(N)$.

\subsection[Decomposition of $\repM_{\g,\mu,K}$]{Decomposition of $\boldsymbol{\repM_{\g,\mu,K}}$}

In this section, we obtain the decomposition of each family $\repM_{\g,\mu,K}$ as a direct sum of irreducible families of modules over $\eptl_N(\beta)$. We first discuss the insertion states with parameters $(0,x)$ in $\repM_{\g,\mu,K}$.
We use the short-hand notation
\begin{equation} \label{eq:def.x.nu}
x_\nu = \eE^{\ir \pi m_\nu/p'} \,, \qquad
\wt{x}_\nu = \eE^{\ir \pi \wt{m}_\nu/p'} \,,
\end{equation}
where $m_\nu$ and $\wt{m}_\nu$ respectively denote exponents of the adjacency matrices $A$ and $\wt A$.

\begin{Proposition} \label{prop:insertion.k=0}
Let $\nu\in\{1,2,\dots,n\}$. The insertion space with parameters $(0,x_\nu)$ in $\repM_{\g,\mu,\id}$ is spanned by the states of the form
\begin{equation} \label{eq:w.nu}
w_\lambda = \sum_{a=1}^n S_{a\lambda} \ket{a} \,, 
\qquad 
\lambda \in \{1,2,\dots, n\}\,
\quad \textrm{with} \quad
m_\lambda=m_\nu \,.
\end{equation}
\end{Proposition}
\proof
We first recall that the action of $f$ in $\repM_{\g,\mu,\id}(0)$ reads
\begin{equation}
f\cdot \ket{a} = \sum_{b=1}^n A_{ba}\, \ket{b} \,.
\end{equation}
From \eqref{eq:def.insertion.kx}, we see that a state $w \in \repM_{\g,\mu,\id}(0)$ is an insertion state with parameters $(0,x_\nu)$ if it satisfies $f\cdot w= 2\cos(\frac{\pi m_\nu}{p'})\, w$, namely if it is an eigenstate of the $A$ with eigenvalue $2\cos(\frac{\pi m_\nu}{p'})$.
\eproof

\noindent The dimension of the insertion space with parameters $(0,x_\nu)$ in $\repM_{\g,\mu,\id}$ is thus equal to the multiplicity of the exponent $m_\nu$ for the adjacency matrix of $\g$. Moreover, we refer to the state $w_\mu$ as the {\it vacuum state}.\medskip

With an argument similar to the above, we obtain the following result on the insertion space with no defects in $\repM_{\g,\mu,K}$.

\begin{Proposition} \label{prop:insertion.k=0.twist}
Let $K$ be an automorphism of $\cal G$, and $\nu\in\{1,2,\dots,\wt{n}\}$. The insertion space with parameters $(0,\wt{x}_\nu)$ in $\repM_{\g,\mu,K}$ is spanned by the states of the form
\begin{equation} \label{eq:w.nu.twist}
\wt{w}_\lambda = \sum_{a \in \wt{\cal G}} \wt{S}_{a\lambda}\, \ket{a} \,,
\qquad
\lambda \in \{1,2,\dots, \wt{n}\} \quad \textrm{with} \quad \wt{m}_\lambda=\wt{m}_\nu \,,
\end{equation}
where $\wt{S}_{a\nu}$ is the eigenvector of $\wt{A}$ associated to the exponent $\wt{m}_\nu$, given in \eqref{eq:St.a.mu}.
\end{Proposition}

For the ADE models with non-trivial automorphisms $K$, the subgraph $\wt{\cal G}$ of $\cal G$ is always a graph of type $A$, and in this case each exponent is non-degenerate, so the insertion space with parameters $(0,\wt{x}_\nu)$ of $\repM_{\g,\mu,K}$ is one-dimensional and spanned by $\wt{w}_\nu$.\medskip

We now describe the decomposition of $\repM_{\g,\mu,K}$. Before tackling the general case, we start by discussing the simplest example, namely the family $\repM_{A_n,1,\id}$ associated to the unitary $A$-series with periodic boundary conditions. For this model, we have $p'=n+1$ and $m_\nu = \nu$. The dimension of the module $\repM_{A_n,1,\id}(N)$ can be written as 
\be 
\dim\repM_{A_n,1,\id}(N) = \tr (A^N)
= \sum_{\nu=1}^n \left( 2 \cos \frac{\pi m_\nu}{p'} \right)^N 
= \sum_{s=0}^{p'-1} \alpha_s \, \big( D_s(N) + D_{p'-s}(N) \big) \,, \\
\ee 
where
\be
\alpha_s = \sum_{\nu=1}^n \cos \bigg(\frac{2\pi s m_\nu}{p'}\bigg)
= \left\{\begin{array}{cl}
p'-1 & s=0 \,, \\[0.15cm]
-1 & s\in\{1,2,\dots, p'-1\} \,,
\end{array}\right.
\ee
and we used the identity \eqref{eq:cos^N}. This allows us to express $\dim \repM_{A_n,1,\id}(N)$ as the sum of dimensions of the modules $\repQ_{0,(-q)^s}(N)$:
\begin{equation} \label{eq:dim.M.An}
\dim\repM_{A_n,1,\id}(N)
= \sum_{s=1}^{p'-1} \big( D_0(N) + D_{p'}(N) -D_s(N) - D_{p'-s}(N) \big)
= \sum_{s=1}^n \dim\repQ_{0,(-q)^s}(N) \,.
\end{equation}
For each $s\in\{1,2,\dots, n\}$, we have $x_s=(-q)^s$. From \cref{prop:insertion.k=0}, we deduce that the insertion space with parameters $(0,x_s)$ in $\repM_{A_n,1,\id}$ is one-dimensional and spanned by the state $w_s$ defined in \eqref{eq:w.nu}. Let us define the submodules $\repM_s(N)=\cL(N,0)\cdot w_s$. Each module $\repM_s(N)$ is nonzero and isomorphic to a quotient of $\repW_{0,x_s}(N)$, implying that $\dim\repM_s(N)\geq \dim\repQ_{0,x_s}(N)$. From the Loewy diagrams in~\eqref{eq:W.Loewy}, we conclude that two modules $\repW_{0,x_s}(N)$ associated to distinct values of $s\in\{1,2,\dots, n\}$ have no common submodules except for the zero module. This implies that the submodules $\repM_s(N)$ appear as direct summands in the decomposition of $\repM_{A_n,1,\id}(N)$. From \eqref{eq:dim.M.An}, we deduce that $\repM_s(N) \simeq \repQ_{0,(-q)^s}(N)$, as otherwise the dimension of $\bigoplus_{s=1}^n \repM_s(N)$ would exceed the dimension of $\repM_{\g,1,\id}(N)$. This proves the decomposition
\begin{equation}
\repM_{A_n,1,\id} \simeq \bigoplus_{s=1}^n \repQ_{0,(-q)^s} \,.
\end{equation}
We now turn to the main result of this section.

\begin{Theorem} \label{thm:M.decomp}
The families of modules $\repM_{\g,\mu,\id}$ with periodic boundary conditions decompose as
\begingroup
\allowdisplaybreaks
\begin{subequations} \label{eq:M.decomp}
\begin{alignat}{2}
\repM_{A_n,\mu,\id} & \simeq \bigoplus_{s=1}^n \repQ_{0,(-1)^{\mu s} q^s} \,, \label{eq:M.An.decomp} \\
\repM_{D_n,\mu,\id} & \simeq \bigg[\bigoplus_{s=1,3,\dots,2n-3,n-1} \repQ_{0,q^s}\bigg] \oplus \bigg[\bigoplus_{t=1}^{\lfloor(n-2)/2 \rfloor}\bigoplus_{\eps = \pm 1} \repQ_{2t, \eps q^{n-1}}\bigg] \,,
\\
\repM_{E_6,\mu,\id} & \simeq \bigg[\bigoplus_{s=1,4,5,7,8,11} \repQ_{0,q^s}\bigg] 
\oplus \bigg[\bigoplus_{\eps = \pm 1} \repQ_{2,\eps q^6}\bigg]
\oplus \bigg[\bigoplus_{\sigma, \eps = \pm 1} \repQ_{3,\eps q^{4\sigma}}\bigg] \,,
\\
\repM_{E_7,\mu,\id} & \simeq \bigg[\bigoplus_{s=1,5,7,9,11,13,17} \repQ_{0,q^s}\bigg] 
\oplus \bigg[\bigoplus_{k=2,4,8}\bigoplus_{\eps = \pm 1} \repQ_{k,\eps q^9}\bigg]
\oplus \bigg[\bigoplus_{\sigma, \eps = \pm 1} \repQ_{3,\eps q^{6\sigma}}\bigg]\,,
\\
\repM_{E_8,\mu,\id} & \simeq \bigg[\bigoplus_{s=1,7,11,13,17,19,23,29} \repQ_{0,q^s}\bigg] 
\oplus \bigg[\bigoplus_{k=2,4,8,14}\bigoplus_{\eps = \pm 1} \repQ_{k,\eps q^{15}}\bigg]
\oplus \bigg[\bigoplus_{k=3,9}\bigoplus_{\sigma, \eps = \pm 1} \repQ_{k,\eps q^{10\sigma}}\bigg] 
\nn\\&
\qquad\oplus \bigg[\bigoplus_{s=6,12}\bigoplus_{\sigma, \eps = \pm 1} \repQ_{5,\eps q^{\sigma s}}\bigg] \,.
\end{alignat}
\end{subequations}
\endgroup

Similarly, the families of modules $\repM_{\g,\mu,K}$ with twisted periodic boundary conditions decompose as
\begingroup
\allowdisplaybreaks
\begin{subequations} \label{eq:M.decomp.twist}
\begin{alignat}{2}
\repM_{A_n,\mu,R} & \simeq 
\left\{\begin{array}{ll}
\displaystyle\repQ_{0,q^{p'/2}}\oplus \bigg[\bigoplus_{k=1}^{(n-1)/2} \bigoplus_{\eps = \pm 1} \repQ_{k,\eps q^{p'/2}}\bigg] & n \textrm{ odd}\,,\\
\displaystyle\bigoplus_{k=1/2}^{(n-1)/2} \bigoplus_{\eps = \pm 1} \repQ_{k,\eps q^{p'/2}} & n \textrm{ even}\,,
\end{array}\right.
\\
\repM_{D_n,\mu,P_{(n-1,n)}} & \simeq\bigg[\bigoplus_{s=2,4,\dots,2n-4} \repQ_{0,q^s}\bigg] \oplus \bigg[\bigoplus_{t=1}^{\lfloor (n-1)/2 \rfloor}\bigoplus_{\eps = \pm 1} \repQ_{2t-1, \eps q^{n-1}}\bigg] \,,
\\
\repM_{D_4,\mu,P_{(134)}} & \simeq \repQ_{0,q^3} 
\oplus \bigg[\bigoplus_{\sigma, \eps = \pm 1} \repQ_{1,\eps q^{2\sigma}} \bigg]\,,
\\
\repM_{E_6,\mu,P_{(15),(24)}} & \simeq \bigg[\bigoplus_{s=4,8} \repQ_{0,q^s}\bigg] 
\oplus \bigg[\bigoplus_{k=1,2,5}\bigoplus_{\eps = \pm 1} \repQ_{k,\eps q^6}\bigg]
\oplus \bigg[\bigoplus_{\sigma, \eps = \pm 1} \repQ_{2,\eps q^{3\sigma}}\bigg]\,.
\end{alignat}
\end{subequations}
\endgroup
\end{Theorem}
\proof
We only discuss the proof of \eqref{eq:M.decomp}, as the arguments to prove \eqref{eq:M.decomp.twist} are similar. We divide the proof in three main steps. 
\begin{enumerate}
\item Insertion states. The factors of the form $\repQ_{0,\eps_s q^s}$ in \eqref{eq:M.decomp} correspond to the exponents of the adjacency matrix of $\g$, namely $s=m_\lambda$ with $\lambda\in \{1,2,\dots,n\}$. For each value of $s$, we define an associated insertion state $w_\nu$ using \cref{prop:insertion.k=0}. This is possible from \cref{prop:permut.nu} below, which allows us, for each $\nu$, to write $x_\nu+x^{-1}_\nu=\eps_s (q^{s}+q^{-s})$ for some integer $s$ and some sign $\eps_s$. For each factor of the form $\repQ_{k,x}$ with $k>0$ in \eqref{eq:M.decomp}, we construct an insertion states $w_{k,x}$, case by case, in \cref{sec:insertion.k>0}.
\item Dimension of the modules. For each non-negative even integer $N$, we use \eqref{eq:cos^N} and obtain
\begin{alignat}{1}
\dim \repM_{\g,\mu,\id}(N)
&= \sum_{m=1}^{p'-1} n_m\big(D_0(N) + D_{p'}(N) -D_m(N) - D_{p'-m}(N) \big) \nn\\
& + \sum_{m=1}^{p'-1} (n_m+\alpha_m)\big(D_m(N) + D_{p'-m}(N) \big) \,,
\end{alignat}
where $n_m$ is the multiplicity of the exponent $m$ for the adjacency matrix, and 
\begin{equation}
\alpha_s = \sum_{\nu=1}^n \cos\bigg(\frac{2\pi s m_\nu}{p'}\bigg) = \tr \,T_{2s}(\tfrac A2)\,.
\end{equation}
Here, $T_j(x)$ is the $j$-th Chebyshev polynomial of the first kind.
Comparing with \eqref{eq:Q.dim}, for any choice of signs $\eps_m\in\{-1,+1\}$ we find
\begin{equation} \label{eq:dim.M}
\dim \repM_{\g,\mu,\id}(N)
= \sum_{m=1}^{p'-1} n_m\dim \repQ_{0,\eps_m q^m}(N)
+ \sum_{m=1}^{p'-1} (n_m+\alpha_m) \, \big(D_m(N)+D_{p'-m}(N) \big) \,,
\end{equation}
where we used $n_{p'-m}=n_m$ and $\alpha_{p'-m}=\alpha_m$. The first sum in \eqref{eq:dim.M} corresponds to the factors $\repQ_{0,\eps_s q^s}(N)$ in \eqref{eq:M.decomp}. The second sum is a linear combination of the dimensions $D_m(N)$ with integer coefficients. In \cref{app:dims}, we show that this second sum equals the total dimension of the factors $\repQ_{k,x}$ with $k>0$ in \eqref{eq:M.decomp}, case by case. 

\item Structure of the module. For each term in \eqref{eq:M.decomp}, we denote by $w_{k,x}$ the corresponding insertion state. This state generates a nonzero submodule $\repM_{k,x}(N)=\cL(N,2k)\cdot w_{k,x}$, isomorphic to a quotient of $\repW_{k,x}(N)$. In all cases except for $D_n$ with $n$ even, the pairs $(k,x)$ of parameters are all pairwise distinct. Using the structure \eqref{eq:W.Loewy}, we conclude that the submodules $\repM_{k,x}(N)$ are direct summands in the module's decomposition. For $D_n$ with even $n$, the only degeneracy occurs for $(k,x)=(0,\pm \ir)$, and the two corresponding insertion states have distinct eigenvalues of $K_N$ for the automorphism $K=P_{(n,n-1)}$. The associated submodules are therefore also direct summands in the decomposition. Hence, in all cases, the direct sum of the modules $\repM_{k,x}(N)$ is a submodule of $\repM_{\g,\mu,\id}$ that exhausts the dimension of $\dim\repM_{\g,\mu,\id}(N)$ for $\repM_{k,x} = \repQ_{k,x}$, ending the proof.
\eproof
\end{enumerate}

We recall from \cref{sec:W.and.Q} that any family $\repQ_{k,\eps q^{\sigma s}}$ with $\sigma = -1$ arising in these decompositions can be equivalently rewritten as $\repQ_{k,\eps(-1)^p q^{p'-s}}$, which fits the criteria \eqref{eq:xks}. Moreover, in all cases, the factors that appear in the decompositions for $K=\id$ are all of the form $\repQ_{k,x}$ with $x$ satisfying $x^{2k} = 1$. This is because $\Omega^N = x^{2k} \id$ in $\repW_{k,x}(N)$ and $\repQ_{k,x}(N)$ whereas $\Omega^N = \id$ in $\repM_{\g,\mu,\id}$, so any factor $\repQ_{k,x}$ arising in the decomposition of $\repM_{\g,\mu,\id}$ must satisfy $x^{2k}=1$. Similarly, in the cases where $K^2 = \id$ or $K^3 = \id$, the only factors $\repQ_{k,x}$ that may arise are those where $x^{4k}=(\kappa_\mu)^N$ or $x^{6k}=1$, respectively.
\medskip

In the standard modules $\repW_{k,x}(N)$ and its quotient modules $\repQ_{k,x}(N)$, we have the relations 
\be
k=0: \ \ E\, \Omega\, E = \alpha E\,,
\qquad
k>0: \ \ E = 0\,,
\ee
where $\alpha = x+x^{-1}$ and $E$ is defined in \eqref{eq:def.E}. \cref{thm:M.decomp} implies the identities
\be
\repM_{\g,\mu,\id}(N): \quad E \prod_{\nu=1}^n (\Omega E - \beta_\nu \id) = 0\,,
\qquad
\repM_{\g,\mu,K}(N): \quad E \prod_{\nu=1}^{\wt n} (\Omega E - \wt \beta_\nu \id) = 0\,,
\ee
for $N$ even. These polynomials in $\Omega E$ are the minimal polynomials for the adjacency matrix of $\mathcal G$ and $\mathcal G'$, respectively. The relations specific to the different models are
\begingroup
\allowdisplaybreaks
\begin{subequations}
\label{eq:minimal.polys}
\begin{alignat}{2}
\repM_{A_n,\mu,\id}&: \qquad E\, U_n(\tfrac12 \Omega E) = 0\,, 
\\[0.15cm]
\repM_{A_n,\mu,R}&: \qquad 
E\,\Omega\, E=0 \qquad n \textrm{ odd},
\\[0.15cm]
\repM_{D_n,\mu,\id}&: \qquad 
\left\{\begin{array}{ll}
E\,\Omega \, E \, T_{n-1}(\tfrac12 \Omega E) = 0& n \textrm{ odd},\\[0.15cm]
E \, T_{n-1}(\tfrac12 \Omega E) = 0& n \textrm{ even},
\end{array}\right.
\\[0.15cm]
\repM_{D_n,\mu,P_{(n-1,n)}}&: \qquad E \, U_{n-2}(\tfrac12 \Omega E) = 0\,, 
\\[0.15cm]
\repM_{D_4,\mu,P_{(134)}}&: \qquad E\, \Omega E = 0\,,
\\[0.15cm]
\repM_{E_6,\mu,\id}&: \qquad E \big((\Omega E)^6 - 5 (\Omega E)^4 + 5 (\Omega E)^2-\id\big) = 0\,,
\\[0.15cm]
\repM_{E_6,\mu,P_{(15,24)}}&: \qquad E \big((\Omega E)^2 -\id\big) = 0\,,
\\[0.15cm]
\repM_{E_7,\mu,\id}&: \qquad E \big((\Omega E)^7 - 6 (\Omega E)^5 + 9 (\Omega E)^3-3\Omega E\big) = 0\,,
\\[0.15cm]
\repM_{E_8,\mu,\id}&: \qquad E \big((\Omega E)^8 - 7 (\Omega E)^6 + 14 (\Omega E)^4 -8 (\Omega E)^2+\id\big) = 0\,,
\end{alignat}
\end{subequations}
\endgroup
where $T_j(x)$ and $U_j(x)$ are Chebyshev polynomials of the first and second kind, respectively. It is thus apparent that $\repM_{\g,\mu,K}(N)$ is in general not a module over an uncoiled algebra.

\subsection{Permutations of the exponents}
\label{sec:permut}

For each exponent $m_\nu$ of $\g$, we constructed in \cref{prop:insertion.k=0} an insertion state $w_\nu$ with parameters $(0,x_\nu)$ in $\repM_{\g,\mu,\id}$. On the other hand, \cref{thm:M.decomp} states the decomposition of the families $\repM_{\g,\mu,K}$ in terms of the families $\repQ_{0,\eps q^s}$, where $s$ runs over sets of integers, and $\eps \in \{+1,-1\}$. The following results build the bridge between the two formulations. 

\begin{Lemma}
\label{lem:two.sets}
Let $q=\eE^{-\ir\pi p/p'}$ with $p,p'$ as in \eqref{eq:q.standard}. Then the two sets
\begin{subequations}
\begin{alignat}{2}
\mathcal S_1 &= \{(-1)^{s(p'-p)}(q^s+q^{-s})\,|\, s=1, 2,\dots, p'-1\}\,,
\\[0.15cm]
\mathcal S_2 &= \{\beta_\nu = 2 \cos (\tfrac{\pi \nu}{p'})\,|\, \nu=1, 2,\dots, p'-1\}\,,
\end{alignat}
\end{subequations}
are equal. Moreover, let $\mathcal S_{1,0}$ and $\mathcal S_{1,1}$ be the subsets of $\mathcal S_{1}$ with even and odd values of $s$, and similarly $\mathcal S_{2,0}$ and $\mathcal S_{2,1}$ be the subsets of $\mathcal S_{2}$ with even and odd values of $\nu$. Then we have $\mathcal S_{1,0} = \mathcal S_{2,0}$ and $\mathcal S_{1,1} = \mathcal S_{2,1}$. 
\end{Lemma}
\proof For $(p-p')$ odd, we have $\mathrm{gcd}(p-p',2p')=1$, and thus $-q=\eE^{2\ir\pi (p'-p)/(2p')}$ is a primitive root of unity of order $2p'$. Hence, there is a bijection between the sets $\{(-q)^s\,|\,s=0,1,\dots, 2p'-1\}$ and $\{\eE^{\ir\pi \nu/p'}\,|\, \nu=0, 1,\dots, 2p'-1\}$. Moreover, using $(-q)^0=1$, $(-q)^{p'}=-1$ and $(-q)^{2p'-s}=(-q)^{-s}$, we find that the set $\{(-q)^s\,|\,s=1,2,\dots, p'-1\}$ is in bijection with the set $\{\eE^{\ir\pi \sigma_\nu \nu/p'}\,|\, \nu=1,2,\dots, p'-1\}$, for some $\sigma_\nu \in \{-1,+1\}$ for each $\nu$. Since $\beta_\nu = \eE^{\ir \pi \nu/p'}+\eE^{-\ir \pi \nu/p'}$, this proves the bijection between $\mathcal S_1$ and $\mathcal S_2$. Moreover, if $s$ and $\nu$ are such that $\cos(\frac{\pi(p-p')s}{p'}) = \cos(\frac{\pi \nu}{p})$, then $\eE^{\ir \pi(p-p')s/p'} = \eE^{\ir \pi \sigma_\nu \nu/p'}$ for some sign $\sigma_\nu$, and then $(-1)^s = \eE^{\ir \pi(p-p')s} = \eE^{\ir \pi \sigma_\nu \nu} = (-1)^\nu$.\medskip

For $(p-p')$ even, both $p,p'$ are odd. We have $\mathrm{gcd}(p',2p)=1$, and conclude that $q=\eE^{-2\ir\pi p/(2p')}$ is a primitive root of unity of order $2p'$. Using the same argument as above, we find a bijection between the sets $\{q^s\,|\,s=0,1,\dots, 2p'-1\}$ and $\{\eE^{\ir\pi \nu/p'}\,|\, \nu=0, 1,\dots, 2p'-1\}$, and subsequently a bijection between the sets $\{q^s\,|\,s=1,2,\dots, p'-1\}$ and $\{\eE^{\ir\pi \sigma_\nu \nu/p'}\,|\, \nu=1,2,\dots, p'-1\}$ for some signs $\sigma_\nu$. This leads to a bijection between the sets $\mathcal S_1$ and $\mathcal S_2$. Moreover, if $s$ and $\nu$ are such that $\cos(\frac{\pi ps}{p'}) = \cos(\frac{\pi \nu}{p})$, then $\eE^{\ir \pi ps/p'} = \eE^{\ir \pi \sigma_\nu \nu/p'}$ for some sign $\sigma_\nu$, and then $(-1)^s = \eE^{\ir \pi p s} = \eE^{\ir \pi \sigma_\nu \nu} = (-1)^\nu$.
\eproof

\begin{Proposition} \label{prop:permut.nu}
For the ADE models, we have
\be
\{\beta_\nu = 2 \cos (\tfrac{\pi m_\nu}{p'})\,|\, \nu=1, 2,\dots, n\} =
\left\{\begin{array}{ll}
\{(-1)^{\mu s}(q^s+q^{-s})\,|\, s=1, 2,\dots, n\} & \g = A_n\,, \\[0.15cm]
\{q^s+q^{-s}\,|\, s=1,3,\dots,2n-3,n-1\}& \g = D_n\,, \\[0.15cm]
\{q^s+q^{-s}\,|\, s=1,4,5,7,8,11\} & \g = E_6\,, \\[0.15cm]
\{q^s+q^{-s}\,|\, s=1,5,7,9,11,13,17\} & \g = E_7\,, \\[0.15cm]
\{q^s+q^{-s}\,|\, s=1,7,11,13,17,19,23,29\} & \g = E_8\,. 
\end{array}\right.
\ee
\end{Proposition}
\proof
For the $A_n$ models, we have $p'=n+1$, $p'-p=\mu$ and $m_\nu=\nu$, and the result readily follows from \cref{lem:two.sets}. For the models $E_6$, $E_7$ and $E_8$, the sets are finite dimensional, and it is straightforward to check the equality. For the $D_n$ models, the Coxeter number is $p'=2n-2$, and $p$ is odd. We first note that, for $\nu = n$, we have $\beta_n = 0=q^{n-1}+q^{-(n-1)}$. Moreover, we have
\begin{alignat}{2}
\{2 \cos (\tfrac{\pi m_\nu}{p'})\,|\, \nu=1, 2,\dots, n-1\} &= \{2 \cos (\tfrac{\pi m}{p'})\,|\, m=1, 3,\dots, 2n-3\}
\nn\\[0.15cm]
&=\{(-1)^{(p'-p)}(q^s+q^{-s})\,|\, s=1, 3,\dots, p'-1\},
\end{alignat} 
where we used the equality of the odd subsets given in \cref{lem:two.sets}. Because $q^{p'-s} = (-1)^p q^{-s} = - q^{-s}$, we conclude that if the set $\{(-1)^{(p'-p)}(q^s+q^{-s})\,|\, s=1, 3,\dots, p'-1\}$ contains an element $\beta$, it also contains $-\beta$, allowing us to write this set simply as $\{(q^s+q^{-s})\,|\, s=1, 3,\dots, p'-1\}$.
\eproof

\begin{Proposition} \label{prop:permut.nu.twist}
For the ADE models with non-trivial automorphisms, we have
\be
\{\wt{\beta}_\nu = 2 \cos (\tfrac{\pi \wt{m}_\nu}{p'})\,|\, \nu=1, 2,\dots, \wt{n}\} =
\left\{\begin{array}{ll}
\{q^{p'/2}+q^{-p'/2}\} & (\g,K) = (A_n,R)\,, \\[0.15cm]
\{q^s+q^{-s}\,|\, s=2,4,\dots,2n-4\}& (\g,K) = (D_n,P_{(n-1,n)})\,, \\[0.15cm]
\{q^3+q^{-3}\} & (\g,K) = (D_4,P_{(134)})\,, \\[0.15cm]
\{q^4+q^{-4}, q^8+q^{-8}\} & (\g,K) = (E_6,P_{(15)(24)})\,.
\end{array}\right.
\ee
\end{Proposition}
\proof For $(\g,K) = (D_n,P_{(n-1,n)})$, the Coxter number is $p'=2n-2$, and we have
\begin{alignat}{2}
\{\wt{\beta}_\nu = 2 \cos (\tfrac{\pi \wt{m}_\nu}{p'})\,|\, \nu=1, 2,\dots, \wt{n}\} 
&= \{2 \cos (\tfrac{\pi m}{p'})\,|\, m=2, 4,\dots, 2n-4\}
\nn\\[0.15cm]
&=\{(q^s+q^{-s})\,|\, s=2, 4,\dots, 2n-4\},
\end{alignat} 
where we used the equality of the even subsets given in \cref{lem:two.sets}. For the other cases, the sets contain either one or two elements and verifying the equality is straightforward.
\eproof

\subsection[Insertion states for $k>0$]{Insertion states for $\boldsymbol{k>0}$}
\label{sec:insertion.k>0}

In this section, we give a general construction of the insertion states for $k>0$ for all the models. We however start by giving some examples for the $A_n$ models with $n$ odd.

\paragraph{Examples of insertion states for $\boldsymbol{\repM_{A_n,\mu,R}}$ with $\boldsymbol{n}$ odd.} 

In this case, $N$ and $p'=n+1$ are even, and $\mathrm{gcd}(p',\mu)=1$ imposes that $\mu$ is odd, with $\kappa_\mu=1$. In $\repM_{A_n,\mu,R}(N)$, we have the identities $\Omega^N=R$, and $\Omega^{2N}=\id$. We introduce the projectors
\begin{equation}
\Pi_{\pm}(N) = \frac{1}{2N}\sum_{j=0}^{2N-1} (\mp \ir \, \Omega)^j \,,
\end{equation}
and use them to construct insertion states $w_{1,\pm\ir}$ associated to the submodule $\repW_{k,\pm \ir}(2k)$ of $\repM_{A_n,\mu,R}(2k)$ for $k=1,2,3$:
\begin{subequations}
\begin{alignat}{2}
&k=1: \quad 
w_{1,\pm \ir} &&= \Pi_{\pm}(2) \cdot \ket{\tfrac{n+1}2,\tfrac{n+3}2,\tfrac{n+1}2}
\\[0.1cm]\nn
& &&= \tfrac 14 \big(\,
\ket{\tfrac{n+1}2,\tfrac{n+3}2,\tfrac{n+1}2} 
\mp \ir \ket{\tfrac{n+3}2,\tfrac{n+1}2,\tfrac{n-1}2}
- \ket{\tfrac{n+1}2,\tfrac{n-1}2,\tfrac{n+1}2} 
\pm \ir \ket{\tfrac{n-1}2,\tfrac{n+1}2,\tfrac{n+3}2} \,\big) \,,
\\[0.25cm]
&k=2: \quad 
w_{2,\pm \ir} &&= \Pi_{\pm}(4) \cdot \big(\,\ket{\tfrac{n+1}2,\tfrac{n+3}2,\tfrac{n+1}2,\tfrac{n+3}2,\tfrac{n+1}2} - \ket{\tfrac{n+1}2,\tfrac{n+3}2,\tfrac{n+5}2,\tfrac{n+3}2,\tfrac{n+1}2}\,\big) \,,
\\[0.25cm]
&k=3: \quad 
w_{3,\pm \ir} &&= \Pi_{\pm}(6) \cdot \big(\,
\ket{\tfrac{n+1}2,\tfrac{n+3}2,\tfrac{n+1}2,\tfrac{n+3}2,\tfrac{n+1}2,\tfrac{n+3}2,\tfrac{n+1}2} 
-\ket{\tfrac{n+1}2,\tfrac{n+3}2,\tfrac{n+1}2,\tfrac{n+3}2,\tfrac{n+5}2,\tfrac{n+3}2,\tfrac{n+1}2} 
\nn\\&&&\hspace{1.8cm}
-\ket{\tfrac{n+1}2,\tfrac{n+3}2,\tfrac{n+5}2,\tfrac{n+3}2,\tfrac{n+1}2,\tfrac{n+3}2,\tfrac{n+1}2} +
\ket{\tfrac{n+1}2,\tfrac{n+3}2,\tfrac{n+5}2,\tfrac{n+3}2,\tfrac{n+5}2,\tfrac{n+3}2,\tfrac{n+1}2} 
\nn\\&&&\hspace{1.8cm}-
\ket{\tfrac{n+1}2,\tfrac{n+3}2,\tfrac{n+5}2,\tfrac{n+7}2,\tfrac{n+5}2,\tfrac{n+3}2,\tfrac{n+1}2}
-\tfrac13
\ket{\tfrac{n+1}2,\tfrac{n+3}2,\tfrac{n+1}2,\tfrac{n-1}2,\tfrac{n+1}2,\tfrac{n+3}2,\tfrac{n+1}2}
\,\big)
\,.
\end{alignat}
\end{subequations}
In all three cases, one can check that these states are nonzero and that they satisfy
\begin{equation}
c_0 \cdot w_{k,\pm \ir} =0 \,,
\qquad
\Omega \cdot w_{k,\pm \ir} = \pm \ir \, w_{k,\pm \ir} \,.
\end{equation}
The factor of $\frac13$ for the last state for $k=3$ is perhaps unexpected, and is not a typo.

\paragraph{General construction.} 

We now give a general construction of the insertion states $w_{k,x}$ for $k>0$, as 
\be
\label{u.def}
w_{k,x} = \wh P_{2k,x} \cdot w_k\,,
\ee
where $\wh P_{N,x}$ is a Jones--Wenzl projectors for the uncoiled algebras defined in \eqref{eq:hatP}, and $w_k$ is an insertion states for the boundary, discussed in \cref{sec:insertionTL}. The fact that we can use these projectors is non-trivial, since as observed previously the modules $\repM_{\g,\mu,K}(N)$ are not modules over the uncoiled algebras.\medskip

We first discuss the case $K=\id$, for which $N=2k$ is even and $\Omega^{2k} = \id$. As is clear from \eqref{eq:minimal.polys}, the identities $E=0$, and $E\,\Omega\, E = \alpha E$ for some $\alpha \in \mathbb C$, are in general not satisfied in $\repM_{\g,\mu,\id}(2k)$, implying that it is neither a module over $\eptl^{\tinyx 1}_{2k}(\beta,\alpha)$ nor $\eptl^{\tinyx 2}_{2k}(\beta,1)$. We note however that $\repM_{\g,\mu,\id}(2k)$ is a module over $\sum_{\nu=1}^n \eptl^{\tinyx 1}_{2k}(\beta,\beta_\nu)$. This sum of algebras is not a direct sum, as the summands differ only in the behaviour of the diagrams with zero bridges. Importantly, we see from \eqref{eq:M.decomp} that any factor $\repQ_{k,x}$ for which we wish to construct an insertion state has $x \notin\{+1,-1\}$. We also see from the coefficients $\Gamma_{s,\ell}$, given in \eqref{eq:Gammakl} and \eqref{eq:Gamma.n/2}, that the projector $\wh P_{2k,x}$ defined in \eqref{eq:hatP1} with $x^{2k}=1$ but $x \notin\{+1,-1\}$ is independent of $\alpha$. As a result, $\wh P_{2k,x}$ is a Jones--Wenzl projector over $\eptl^{\tinyx 1}_{2k}(\beta,\beta_\nu)$ for each index $\nu$, it is thus also a projector in $\sum_{\nu=1}^n \eptl^{\tinyx 1}_{2k}(\beta,\beta_\nu)$, and the state $w_{k,x} = \wh P_{2k,x} \cdot w_k$ satisfies the relations \eqref{eq:def.insertion.kx}. To show that it is an insertion state, one must prove that $w_{k,x} \neq 0$, which we discuss below and in \cref{app:insertion}.\medskip

Second, we discuss the case $K^2 = \id$. Because $K_{2k}^2 = \Omega^{4k} =(\kappa_\mu)^{2k} \id$, the eigenvalues of $K_{2k}$ are $\pm (\kappa_\mu)^{k}$, with the convention \eqref{eq:sqrt.kappa}. The projectors on the corresponding eigenspaces are 
\be
\Xi_0 = \tfrac12(\id + (\kappa_\mu)^{k}\,K_{2k}), \qquad 
\Xi_1 = \tfrac12(\id - (\kappa_\mu)^{k}\,K_{2k})\,. 
\ee
Because $K_{2k}$ commutes with the action of $\eptl_{2k}(\beta)$, we can use these projectors to write
\be
\repM_{\g,\mu,K}(2k) = \Xi_0\cdot\repM_{\g,\mu,K}(2k) \oplus \Xi_1\cdot\repM_{\g,\mu,K}(2k)\,.
\ee
For $2k$ odd, $\Xi_0\cdot\repM_{\g,\mu,K}(2k)$ and $\Xi_1\cdot\repM_{\g,\mu,K}(2k)$ are modules over the uncoiled algebras $\eptl_{2k}\big(\beta,(\kappa_\mu)^{1/2}\big)$ and $\eptl_{2k}\big(\beta,-(\kappa_\mu)^{1/2}\big)$, respectively. We construct the following states in these two submodules as 
\be
w^{\tinyx 0}_{k,x} = \wh P_{2k,x}\, \Xi_0 \cdot w_k\,, \qquad 
w^{\tinyx 1}_{k,x} = \wh P_{2k,x}\, \Xi_1 \cdot w_k\,,
\ee
respectively. For $2k$ even, $\Pi_0 \cdot \repM_{\g,\mu,K}(2k)$ is a module over $\sum_{\nu=1}^{\wt n}\eptl^{\tinyx 1}_{2k}\big(\beta,\wt\beta_{\nu}\big)$, and $\Pi_1 \cdot \repM_{\g,\mu,K}(2k)$ is a module over $\eptl_{2k}\big(\beta,-(\kappa_\mu)^{1/2})$. In the former case, all the factors $\repQ_{k,x}$ for which we wish to construct an insertion state have $x \notin\{+1,-1\}$, and using the same argument as above, we construct the states $w^{\tinyx 0}_{k,x} = \wh P_{2k,x}\, \Xi_0 \cdot w_k$ and $w^{\tinyx 1}_{k,x} = \wh P_{2k,x}\, \Xi_1 \cdot w_k$, respectively. To show that these are insertion states, we must show that they are nonzero.\medskip

In the special case $D_4$ with $K^3 = \id$, one can use the same arguments, with the projectors 
\be
\Pi_n = \tfrac13(\id + \omega^n\,K_{2k} + \omega^{2n}\, (K_{2k})^{2})\,, \qquad
\omega = \eE^{2 \pi \ir /3}\,, \qquad n = 0,1,2\,.
\ee

Thus in all cases, we construct insertion states $w_{k,x} \in \repM_{\g,\mu,K}(N)$ as in \eqref{u.def}, for some states $w_k$ specific to each model. We detail these choices for the different ADE models in \cref{app:insertAn,app:insertDn,app:insertE678}. In each case, $w_k$ is an insertion state for $\tl_N(\beta)$, described in \cref{sec:Mab.decomp}, so it satisfies $e_j \cdot w_k$ for $j=1,2,\dots, 2k-1$ and thus $P_N \cdot w_k = w_k$. Moreover, we choose $w_k$ so that it is also an eigenstate of $K_{2k}$, and thus also of the projectors $\Xi_n$.\medskip

The only remaining difficulty is to show that the states $w_{k,x}$ are nonzero. In fact, such states can be defined for all values of $x$ of the form \eqref{eq:x.values}, but only a select few of them are nonzero. These precisely correspond to the cases where $\repM_{\g,\mu,K}(2k)$ has a one-dimensional module isomorphic to $\repW_{N/2,x}(2k) \simeq \repQ_{N/2,x}(2k)$, in a way that is consistent with the decompositions \eqref{eq:M.decomp}. To show that $w_{k,x} \neq 0$, we use the bilinear form \eqref{eq:bilinear.form}. For $\kappa_\mu = 1$, we compute $\aver{w_k,w_{k,x}}$ and show that it is non zero. For the special case $\kappa_\mu = -1$, there are two states $w_k$ and $\bar w_k$ satisfying $K_{2k} \cdot w_k = \ir w_k$ and $K_{2k} \cdot \bar w_k = -\ir \bar w_k$. To show that $w_{k,x}$ is nonzero in this case, we instead compute $\aver{\bar w_k,w_{k,x}}$ and show that it is nonzero. Using the invariance property of the bilinear form, we obtain
\begin{alignat}{2}
\label{eq:(w,u)}
\big\langle w, Z_{s,\ell} \cdot w_{k}\big\rangle
&=\big\langle w, P_{2k}\,(c_0^\dag)^s\, \Omega^\ell\, (c_0)^s\,P_{2k} \cdot w_{k}\big\rangle
\nn\\[0.1cm]&
= \big\langle(c_0)^s P_{2k} \cdot w, \Omega^\ell\, (c_0)^s\,P_{2k} \cdot w_{k}\big\rangle
\nn\\[0.1cm]&
= \big\langle(c_0)^s \cdot w, \Omega^\ell\, (c_0)^s \cdot w_{k}\big\rangle\,,
\end{alignat}
with $w = w_k$ or $w = \bar w_k$. The right side is then easily computed. This therefore reduces the evaluation of $\aver{w_k,w_{k,x}}$ or $\aver{\bar w_k,w_{k,x}}$ to a weighted sum of constants $\Gamma_{s,\ell}$, which can then be evaluated explicitly and shown to be nonzero. This is discussed in further detail in \cref{app:insertion}.

\subsection{Torus partition functions}

Let $\g$ be a Lie algebra, $\mu$ be an index satisfying $\textrm{gcd}(m_\mu,p')=1$, and $K$ and $K'$ be two automorphisms of the graph $\mathcal G$, belonging to some finite group $G$. We define the lattice partitions function with doubly-twisted boundary conditions
\be
Z_{\g,\mu,K,K'}(M_1,M_2,N) = \tr_{\repM_{\g,\mu,K}(N)}\big(K'_N\Omega^{M_1} T^{M_2}\big)\,,
\ee
where $M_1 \in \mathbb Z$ and $M_2,N \in \mathbb Z_{\ge 0}$. These partition functions satisfy
\be
Z_{\g,LKL^{-1},LK'L^{-1}}(M_1,M_2,N) = Z_{\g,K,K'}(M_1,M_2,N)\,,
\ee
for all $K,K',L\in G$. Moveover, the identity $\Omega^N = K_N$ in $\repM_{\g,\mu,K}(N)$, implies that 
\be
\label{eq:T.lattice}
Z_{\g,\mu,K,K'}(M_1+N,M_2,N) = Z_{\g,\mu,K,K'K}(M_1,M_2,N)\,.
\ee

We now focus on the case where $K$ and $K'$ commute: $[K,K']=0$. \cref{prop:L.commutes} states that the operator $K'_N$ yields a family automorphism of $\repM_{\g,\mu,K}$.
Since $\repM_{\g,\mu,K}$ is a direct sum of irreducible subfamilies, $K'_N$ acts as a multiple of the identity in each summand $\repQ_{k,x}(N)$ of \eqref{eq:M.decomp} or \eqref{eq:M.decomp.twist}. Moreover, each of these summands is of the form $\cL(N,2k)\cdot w_{k,x}$ where $w_{k,x}$ is an insertion state. This state is an eigenstate of $K'_{2k}$. We denote its eigenvalue by $\kappa'_{k,x}$. For the insertion states with $k=0$, we have
\be
\kappa'_{0,x_\nu} = \kappa'_\nu \quad \text{in } \repM_{\g,\mu,\id}\,,
\qquad 
\kappa'_{0,\wt{x}_\nu} = (-1)^{\nu+1} \quad \text{in } \repM_{\g,\mu,K} \,, \quad \textrm{for} \quad K \neq \id \,,
\ee
where $\kappa'_\nu$ is the eigenvalue of $K'$ associated to the eigenvector of $A$ with components $S_{a\nu}$. For the insertion states with $k>0$, we compute the eigenvalues $\kappa'_{k,x}$ case by case using the explicit forms of the insertion states. In the special case where $K'=K^n$ for some $n \in \mathbb Z_{\ge 0}$, we have $K'_{2k} = \Omega^{2nk}$ and thus $\kappa'_{k,x} = x^{2nk}$.
\medskip

Using the decompositions \eqref{eq:M.decomp} and \eqref{eq:M.decomp.twist}, and the action of $K'$ on $\repM_{\g,\mu,K}$, we obtain the lattice partition function $Z_{\g,\mu,K,K'}(M_1,M_2,N)$ as a linear combination of the lattice characters $\chit_{\repQ_{k,x}}$:
\be
Z_{\g,\mu,K,K'}(M_1,M_2,N) = \sum_{k,x} \kappa'_{k,x}\, \chit_{\repQ_{k,x}}(M_1,M_2,N)\,.
\ee

The scaling limit of the partition function is defined as in \eqref{eq:chi(q)bulk}, namely
\be
Z_{\g,\mu,K,K'}(\qq) = \lim_{\substack{N \to \infty \\[0.05cm] M_1 = \lfloor N \tau_r \rfloor \\[0.05cm] M_2 =\lfloor N \tau_i \rfloor}} \eE^{M_2N f_{\textrm{bulk}}} Z_{\g,\mu,K,K'}(M_1,M_2,N)\,.
\ee
We then have
\be
\label{eq:Z.decomp.general}
Z_{\g,\mu,K,K'}(\qq) = \sum_{k,x} \kappa'_{k,x}\, \chit_{\repQ_{k,x}}(\qq)\,.
\ee
This can be written as a sesquilinear form in the characters $\chit_{r,s}(\qq)$, using \eqref{eq:charQK}. The details are given in \cref{app:Z}, and lead to the following theorems. We state separately the results for the model $D_4$ and all the other models, starting with the latter. 

\begin{Theorem}
The conformal partition functions for ADE lattice models on the torus with periodic boundary conditions are given by
\begingroup
\allowdisplaybreaks
\begin{subequations}
\begin{alignat}{2}
Z_{A_n,\mu,\id,\id}(\qq) &= \sum_{r=1}^{p-1}\sum_{s=1}^{p'-1}|\chiK{r,s}|^2\,, \\[0.1cm]
Z_{D_n,\mu,\id,\id}(\qq) &= \sum_{r=1}^{p-1} \sum_{s=1,3,\dots,p'-1} |\chiK{r,s}|^2
+ \sum_{r=1}^{p-1} \sum_{\substack{s=1\\[0.05cm]s\,\equiv\, n+1 \mmod 2}}^{p'-1} \chiK{r,s}\chiKb{r,p'-s}\,, \\[0.1cm]
Z_{E_6,\mu,\id,\id}(\qq) &= \sum_{r=1}^{p-1}\Big(|\chiK{r,1}+\chiK{r,7}|^2
+|\chiK{r,4}+\chiK{r,8}|^2+|\chiK{r,5}+\chiK{r,11}|^2\Big)\,, \\[0.1cm]
Z_{E_7,\mu,\id,\id}(\qq) &= \sum_{r=1}^{p-1}\Big(|\chiK{r,1}+\chiK{r,17}|^2
-|\chiK{r,3}+\chiK{r,15}|^2+|\chiK{r,5}+\chiK{r,13}|^2+|\chiK{r,7}+\chiK{r,11}|^2 \nn\\[0.1cm]
&\hspace{1cm} +|\chiK{r,3}+\chiK{r,9}+\chiK{r,15}|^2\Big)\,,\\[0.1cm]
Z_{E_8,\mu,\id,\id}(\qq) &= \sum_{r=1}^{p-1}\Big(|\chiK{r,1}+\chiK{r,11}+\chiK{r,19}+\chiK{r,29}|^2
+|\chiK{r,7}+\chiK{r,13}+\chiK{r,17}+\chiK{r,23}|^2\Big)\,.
\end{alignat}
\end{subequations}
The conformal partition functions for ADE lattice models on the torus with twisted periodic boundary conditions are given by
\begin{subequations}
\begin{alignat}{2}
Z_{A_n,\mu,\id,R}(\qq) &= (-1)^{pp'}\sum_{r=1}^{p-1}\sum_{s=1}^{p'-1}(-1)^{\lambda_{r,s}}|\chit_{r,s}|^2\,,\\[0.1cm]
Z_{A_n,\mu,R,\id}(\qq) &= \sum_{r=1}^{p-1}\sum_{s=1}^{p'-1}\chit_{r,s}\chib_{r,p'-s}\,,\\[0.1cm]
Z_{A_n,\mu,R,R}(\qq) &= \ir^{pp'}\sum_{r=1}^{p-1}\sum_{s=1}^{p'-1}(-1)^{\lambda_{r,s}+1}\chit_{r,s}\chib_{r,p'-s}\,,\\[0.1cm]
Z_{D_n,\mu,\id,P_{(n-1,n)}}(\qq) &= \sum_{r=1}^{p-1} \sum_{s=1,3,\dots,p'-1} |\chit_{r,s}|^2 - \sum_{r=1}^{p-1} \sum_{\substack{s=1\\[0.05cm]s\,\equiv\, n+1 \mmod 2}}^{p'-1} \chit_{r,s}\chib_{r,p'-s}\,,\\[0.1cm]
Z_{D_n,\mu,P_{(n-1,n),\id}}(\qq) &= \sum_{r=1}^{p-1} \sum_{s=2,4,\dots,p'-2} |\chit_{r,s}|^2 + \sum_{r=1}^{p-1} \sum_{\substack{s=1\\[0.05cm]s\,\equiv\, n \mmod 2}}^{p'-1} \chit_{r,s}\chib_{r,p'-s}\,,\\[0.1cm]
Z_{D_n,\mu,P_{(n-1,n)},P_{(n-1,n)}}(\qq) &= \sum_{r=1}^{p-1} \sum_{s=2,4,\dots,p'-2} |\chit_{r,s}|^2 - \sum_{r=1}^{p-1} \sum_{\substack{s=1\\[0.05cm]s\,\equiv\, n \mmod 2}}^{p'-1} \chit_{r,s}\chib_{r,p'-s}\,,\\[0.1cm]
Z_{E_6,\mu,\id,P_{(15),(24)}}(\qq) &= \sum_{r=1}^{p-1}\Big(|\chit_{r,1}+\chit_{r,7}|^2-|\chit_{r,4}+\chit_{r,8}|^2+|\chit_{r,5}+\chit_{r,11}|^2\Big)\,,\\[0.1cm]
Z_{E_6,\mu,P_{(15),(24)},\id}(\qq) &= \sum_{r=1}^{p-1}\Big(
-|\chit_{r,1}+\chit_{r,7}|^2
+|\chit_{r,4}+\chit_{r,8}|^2
-|\chit_{r,5}+\chit_{r,11}|^2
\nn\\[0.1cm]
&\hspace{1.3cm}
+ |\chit_{r,1}+\chit_{r,5}+\chit_{r,7}+\chit_{r,11}|^2
\Big)\,,\\[0.1cm]
Z_{E_6,\mu,P_{(15),(24)},P_{(15),(24)}}(\qq) &=\sum_{r=1}^{p-1} \Big(
|\chit_{r,1}+\chit_{r,7}|^2
+|\chit_{r,4}+\chit_{r,8}|^2
+|\chit_{r,5}+\chit_{r,11}|^2
\nn\\[0.1cm]
&\hspace{1.3cm}
- |\chit_{r,1}+\chit_{r,5}+\chit_{r,7}+\chit_{r,11}|^2
\Big)\,,
\end{alignat}
\end{subequations}
where 
\be
\label{eq:lambda.rs}
\lambda_{r,s} = p'(r+1)+p(s+1)\,.
\ee
\endgroup
\end{Theorem}
In all cases, we have 
\be
Z_{\g,\mu,\id,\id}(\eE^{2 \pi \ir (\tau+1)}) = Z_{\g,\mu,\id,\id}(\eE^{2 \pi \ir \tau})\,, \qquad
Z_{\g,\mu,\id,\id}(\eE^{-2 \pi \ir /\tau}) = Z_{\g,\mu,\id,\id}(\eE^{2 \pi \ir \tau})\,,
\ee
so that $Z_{\g,\mu,\id,\id}(\eE^{2 \pi \ir \tau})$ is modular invariant. In the cases where there is a non-trivial automorphism~$K$, we have
\begin{subequations}
\begin{alignat}{3}
Z_{\g,\mu,\id,K}(\eE^{2 \pi \ir (\tau+1)}) &= Z_{\g,\mu,\id,K}(\eE^{2 \pi \ir \tau})\,,\qquad
&&Z_{\g,\mu,\id,K}(\eE^{-2 \pi \ir /\tau}) &&= Z_{\g,\mu,K,\id}(\eE^{2 \pi \ir \tau})\,,
\\[0.1cm]
Z_{\g,\mu,K,\id}(\eE^{2 \pi \ir (\tau+1)}) &= Z_{\g,\mu,K,K}(\eE^{2 \pi \ir \tau})\,,\qquad
&&Z_{\g,\mu,K,\id}(\eE^{-2 \pi \ir /\tau}) &&= Z_{\g,\mu,\id,K}(\eE^{2 \pi \ir \tau})\,,
\\[0.1cm]
Z_{\g,\mu,K,K}(\eE^{2 \pi \ir (\tau+1)}) &= \kappa_\mu\, Z_{\g,\mu,K,\id}(\eE^{2 \pi \ir \tau})\,,\qquad
&&Z_{\g,\mu,K,K}(\eE^{-2 \pi \ir /\tau}) &&= \kappa_\mu\, Z_{\g,\mu,K,K}(\eE^{2 \pi \ir \tau})\,.
\end{alignat}
\end{subequations}
These identities are easily verified using 
\be
\chit_{r,s}(\eE^{2 \pi \ir (\tau+1)}) = \eE^{2 \pi \ir h_{r,s}}\chit_{r,s}(\eE^{2 \pi \ir \tau})\,,
\qquad
\chit_{r,s}(\eE^{-2 \pi \ir/\tau}) = \frac12\sum_{r'=1}^{p-1}\sum_{s'=1}^{p'-1} S_{(r,s),(r',s')} \chit_{r',s'}(\eE^{2 \pi \ir\tau})\,,
\ee
where the modular $S$-matrix elements are \cite{dFMS97}
\begin{equation}
S_{(r,s),(r',s')} = \sqrt{\frac{8}{pp'}} (-1)^{rs+r's'+1}
\sin\bigg( \frac{\pi p' r r'}{p} \bigg)\sin\bigg( \frac{\pi p s s'}{p'} \bigg) \,,
\end{equation}
and we have the identities
\be
S_{(r,s),(r',p'-s')} = (-1)^{p'r'+ps'+1} \, S_{(r,s),(r',s')}\,, 
\qquad 
\eE^{2\ir\pi(h_{r,s}-h_{r,p'-s})}=\ir^{pp'} (-1)^{\lambda_{r,s}+1}\,.
\ee
The factor of $\frac12$ in the formula for $\chit_{r,s}(\eE^{-2 \pi \ir/\tau})$ accounts for the fact that each character appears twice in the double sum, because $h_{r,s} = h_{p-r,p'-s}$. The three partition functions $Z_{\g,\mu,\id,K}(\qq)$, $Z_{\g,\mu,K,\id}(\qq)$ and $Z_{\g,\mu,K,K}(\qq)$ yield a three-dimensional representation of the modular group, with
\be
\mathsf{S}=
\begin{pmatrix} 
0&1&0\\
1&0&0\\
0&0&\kappa_\mu
\end{pmatrix},
\qquad 
\mathsf{T}=\begin{pmatrix}
1&0&0\\
0&0&\kappa_\mu\\
0&1&0
\end{pmatrix},
\ee
satisfying
\be
\mathsf{S}^2=(\mathsf{S}\mathsf{T})^3=\id
\qquad
\textrm{and}
\quad
\left\{\begin{array}{ll}
\mathsf{T}^2 = \id & \kappa_\mu = 1\,, \\[0.1cm]
\mathsf{T}^4 = \id & \kappa_\mu = -1\,.
\end{array}\right.
\ee

Let us justify the presence of the signs $\kappa_\mu$ arising in these modular transformations. First, for the $T$-transform, this sign is explained by recalling that the $T$-transform on finite lattices is given in~\eqref{eq:T.lattice}, and that applying it twice yields a sign $\kappa_\mu$ because of the relation $(K_N)^{2N} = (\kappa_\mu)^N \id$. For the $S$-transform, we note that for arbitrary automorphisms $K$ and $K'$, it maps
\be
Z_{A_n,\mu,K,K'}(\qq)
\mapsto
Z_{A_n,\mu,K',K^{-1}}(\qq)\,.
\ee 
For $K = K' = R$, we remark that, even though $R$ is its own inverse, the choice $(\kappa_\mu)^{1/2} = \ir$ for $\kappa_\mu = -1$ is not invariant under this transformation, so in changing $R\to R^{-1}$, one must also replace $\ir$ by $-\ir$, which in the end amounts to multiplying the partition function by $\kappa_\mu = -1$.\medskip

We now discuss the case $\g = D_4$, for which the only relevant index is $\mu = 1$ corresponding to $p = 5$, and the automorphisms are in $G = S_3$. A full set of partition functions $Z_{D_4,1,K,K'}$ is obtained by considering the commuting pairs $(K,K')$ in $S_3 \otimes S_3$ modulo the equivalence classes under the equivalence relation $(K,K') \equiv (LKL^{-1},LK'L^{-1})$. This yields a set of eight inequivalent partition functions, with
\begin{alignat}{2}
(K,K') \in \big\{&
(\id,\id),(\id,P_{(13)}),(\id,P_{(134)}),(P_{(13)},\id),(P_{(13)},P_{(13)}),
\nn\\[0.1cm]
&(P_{(134)},\id),(P_{(134)},P_{(134)}),(P_{(134)},P_{(143)})\big\}\,.
\end{alignat}
This construction is the well-known modular group action on the space of commuting pairs $(g,h)$ in $S_3 \otimes S_3$ modulo simultaneous conjugation, which is associated to Drinfeld's quantum double \cite{D88} for the group $S_3$.

\begin{Theorem}
We have
\begingroup
\allowdisplaybreaks
\begin{subequations}
\begin{alignat}{2}
Z_{D_4,1,\id,\id}(\qq) &= 
\sum_{r=1}^{4} \big(|\chit_{r,1}+\chit_{r,5}|^2+2|\chit_{r,3}|^2\big)\,, \\[0.1cm]
Z_{D_4,1,\id,P_{(13)}}(\qq) &= \sum_{r=1}^{4} |\chit_{r,1}-\chit_{r,5}|^2\,, \\[0.1cm]
Z_{D_4,1,\id,P_{(134)}}(\qq) &= \sum_{r=1}^{4} \big(|\chit_{r,1}+\chit_{r,5}|^2-|\chit_{r,3}|^2\big)\\[0.1cm]
Z_{D_4,1,P_{(13)},\id}(\qq) &= \sum_{r=1}^{4} |\chit_{r,2}+\chit_{r,4}|^2\,, \\[0.1cm]
Z_{D_4,1,P_{(13)},P_{(13)}}(\qq) &= \sum_{r=1}^{4} |\chit_{r,2}-\chit_{r,4}|^2\,,\\[0.1cm]
Z_{D_4,1,P_{(134)},\id}(\qq) &= \sum_{r=1}^{4} \big(|\chit_{r,1}+\chit_{r,3}|^2+|\chit_{r,3}+\chit_{r,5}|^2-|\chit_{r,1}|^2-|\chit_{r,3}|^2-|\chit_{r,5}|^2\big)\,,\\[0.1cm]
Z_{D_4,1,P_{(134)},P_{(134)}}(\qq) &= \sum_{r=1}^{4} \big(|\chit_{r,1}+\omega^{-1}\chit_{r,3}|^2+|\chit_{r,3}+\omega\,\chit_{r,5}|^2-|\chit_{r,1}|^2-|\chit_{r,3}|^2-|\chit_{r,5}|^2\big)\,, \\[0.1cm]
Z_{D_4,1,P_{(134)},P_{(143)}}(\qq) &= \sum_{r=1}^{4} \big(|\chit_{r,1}+\omega\,\chit_{r,3}|^2+|\chit_{r,3}+\omega^{-1}\chit_{r,5}|^2-|\chit_{r,1}|^2-|\chit_{r,3}|^2-|\chit_{r,5}|^2\big)\,, 
\end{alignat}
\end{subequations}
\endgroup
where $\omega = \eE^{2 \pi \ir/3}$.
\end{Theorem}
The details are given in \cref{app:Z.D4}. These partition functions satisfy
\be
Z_{\g,1,K,K'}(\eE^{2 \pi \ir (\tau+1)}) = Z_{\g,1,K,K'K}(\eE^{2 \pi \ir \tau})\,, \qquad
Z_{\g,1,K,K'}(\eE^{-2 \pi \ir/\tau}) = Z_{\g,1,K',K^{-1}}(\eE^{2 \pi \ir \tau})\,,
\ee
and realise a representation of the modular group of dimension $8$, with
\be
\mathsf{S}=
\left(\begin{array}{ccccccccccc} 
1 & 0 & 0 & 0 & 0 & 0 & 0 & 0\\
0 & 0 & 0 & 1 & 0 & 0 & 0 & 0 \\
0 & 0 & 0 & 0 & 0 & 1 & 0 & 0 \\
0 & 1 & 0 & 0 & 0 & 0 & 0 & 0 \\
0 & 0 & 0 & 0 & 1 & 0 & 0 & 0 \\
0 & 0 & 1 & 0 & 0 & 0 & 0 & 0 \\
0 & 0 & 0 & 0 & 0 & 0 & 0 & 1 \\
0 & 0 & 0 & 0 & 0 & 0 & 1 & 0 
\end{array}\right),
\quad 
\mathsf{T}=
\left(\begin{array}{ccccccccccc} 
 1 & 0 & 0 & 0 & 0 & 0 & 0 & 0 \\
 0 & 1 & 0 & 0 & 0 & 0 & 0 & 0 \\
 0 & 0 & 1 & 0 & 0 & 0 & 0 & 0 \\
 0 & 0 & 0 & 0 & 1 & 0 & 0 & 0 \\
 0 & 0 & 0 & 1 & 0 & 0 & 0 & 0 \\
 0 & 0 & 0 & 0 & 0 & 0 & 1 & 0 \\
 0 & 0 & 0 & 0 & 0 & 0 & 0 & 1 \\
 0 & 0 & 0 & 0 & 0 & 1 & 0 & 0
\end{array}\right),
\ee
satisfying $\mathsf{S}^2=(\mathsf{S}\mathsf{T})^3=\id$ as well as $\mathsf{T}^6 = \id$.\medskip

Lastly, we note that it is possible to define partition functions $Z_{\g,\mu,K,K'}$ for non-commuting automorphisms $K$ and $K'$. For the model $D_4$, this allows for three more conjugacy classes of partition functions, represented by $(K,K') \in \{(P_{(13)},P_{(14)}), (P_{(13)},P_{(134)})(P_{(134)},P_{(13)})\}$. Because $K'_N$ does not commute with the action of $\eptl_N(\beta)$ in $\repM_{\g,\mu,K}(N)$, the decomposition of this module as a direct sum of irreducible $\eptl_N(\beta)$-modules does not allow us to compute these partition functions from the characters of the standard modules. Nonetheless, these partition functions are expected to realise a three-dimensional representation of the modular group, with
\be
\mathsf{S}=
\begin{pmatrix} 
1&0&0\\
0&0&1\\
0&1&0
\end{pmatrix},
\qquad 
\mathsf{T}=\begin{pmatrix}
0&1&0\\
1&0&0\\
0&0&1
\end{pmatrix}.
\ee

%
\section{Local operators and linear difference equations}
\label{sec:correlators}
%

In this section, we define connectivity operators and local operators, and obtain linear difference equations satisfied by these local operators.

\subsection{Bulk connectivity operators}
\label{sec:connectivity.ops}

Let $k\in\frac12\Zbb_{\geq 0}$, $x\in\Cbb^\times$, and $\repM$ and $\repM'$ be two families of modules over $\eptl_N(\beta)$. We define a \emph{connectivity operator $\Op$ of type $(k,x)$ from $\repM$ to $\repM'$} to be a collection $\Op_{N}$ of linear maps from $\repM(N)$ to $\repM'(N+2k)$, for each admissible integer $N$ in $\repM$, satisfying the relations
\begin{subequations} \label{eq:def.Oj}
\begin{alignat}{2}
c_{N+2k,j}\,\Op_{N} &= \left\{\begin{array}{cl}
0 & j=1,2,\dots,2k-1\,, \\[0.1cm]
\Op_{N-2} \,c_{N,j-2k} & j = 2k+1,k+2,\dots,N+2k-1\,, \quad N \ge 2,
\end{array}\right.\label{eq:def.Oj.a}
\\[0.2cm] 
c^\dag_{N+2k+2,j}\,\Op_N &= \left\{\begin{array}{cl}
\Op_{N+2,1} \,c^\dag_{N+2,0} & j=0\,, \\[0.1cm]
\Op_{N+2,2} \,c^\dag_{N+2,1} & j=1\,, \\[0.1cm]
\Op_{N+2} \,c^\dag_{N+2,j-2k} & j=2k+1,2k+2, \dots, N+2k+1\,,
\end{array}\right.\label{eq:def.Oj.b} 
\\[0.2cm]
c_{N+2k,0}\, \Op_{N}\, c^\dag_{N,0} &= \left\{\begin{array}{cl}
(x+x^{-1})\, \Op_{N-2} & k=0\,, \\[0.1cm]
x\, \Op_{N-2} & k>0 \,,
\end{array}\right. \quad N \ge 2,\label{eq:def.Oj.c}
\end{alignat}
\end{subequations}
where
\be
\Op_{N,j} = \Omega^{-j} \,\Op_N\, \Omega^j\,.
\ee
The operator $\Op_N$ is represented by the diagram 
\be
k=0: \qquad
\Op_N = \
\begin{pspicture}[shift=-1.3](-1.2,-1.4)(1.4,1.4)
\psarc[linecolor=black,linewidth=0.5pt,fillstyle=solid,fillcolor=lightlightblue]{-}(0,0){1.2}{0}{360}
\psarc[linecolor=black,linewidth=0.5pt,fillstyle=solid,fillcolor=white]{-}(0,0){0.45}{0}{360}
\psline[linecolor=blue,linewidth=1.5pt]{-}(-0.116469,0.434667)(-0.310583,1.15911)
\psline[linecolor=blue,linewidth=1.5pt]{-}(0.116469,0.434667)(0.310583,1.15911)
\psline[linecolor=blue,linewidth=1.5pt]{-}(0.318198, 0.318198)(0.848528, 0.848528)
\psline[linecolor=blue,linewidth=1.5pt]{-}(0.434667, 0.116469)(1.15911, 0.310583)
\psline[linecolor=blue,linewidth=1.5pt]{-}(-0.318198, 0.318198)(-0.848528, 0.848528)
\psline[linecolor=blue,linewidth=1.5pt]{-}(-0.434667, 0.116469)(-1.15911, 0.310583)
\psline[linecolor=blue,linewidth=1.5pt]{-}(-0.434667, -0.116469)(-1.15911, -0.310583)
\psline[linecolor=blue,linewidth=1.5pt]{-}(-0.318198, -0.318198)(-0.848528, -0.848528)
\psline[linecolor=blue,linewidth=1.5pt]{-}(-0.116469, -0.434667)(-0.310583, -1.15911)
\psline[linecolor=blue,linewidth=1.5pt]{-}(0.116469, -0.434667)(0.310583, -1.15911)
\psline[linecolor=blue,linewidth=1.5pt]{-}(0.318198, -0.318198)(0.848528, -0.848528)
\psline[linecolor=blue,linewidth=1.5pt]{-}(0.434667, -0.116469)(1.15911, -0.310583)
\psline[linestyle=dashed, dash= 1.5pt 1.5pt,linewidth=0.5pt]{-}(0,-1.2)(0.05,-0.85)
\psbezier[linestyle=dashed, dash= 1.5pt 1.5pt,linewidth=0.5pt]{-}(0,-1.2)(-0.1,-1.2)(-0.1,-0.45)(0,-0.45)
\psarc[linecolor=black,linewidth=0.5pt,fillstyle=solid,fillcolor=darkgreen]{-}(0.05,-0.85){0.09}{0}{360}
\rput(0.362347, -1.3523){$_1$}
\rput(0.989949, -0.989949){$_2$}
\rput(1.3523, -0.362347){$_{...}$}
\rput(-0.362347, -1.3523){$_N$}
\psarc[linecolor=black,linewidth=0.5pt]{-}(0,0){1.2}{0}{360}
\end{pspicture} 
\ ,
\qquad
k>0: \qquad
\Op_N = \ 
\begin{pspicture}[shift=-1.3](-1.2,-1.4)(1.4,1.4)
\psarc[linecolor=black,linewidth=0.5pt,fillstyle=solid,fillcolor=lightlightblue]{-}(0,0){1.2}{0}{360}
\psarc[linecolor=black,linewidth=0.5pt,fillstyle=solid,fillcolor=white]{-}(0,0){0.45}{0}{360}
\rput{-90}(0,0){
\psline[linecolor=blue,linewidth=1.5pt]{-}(0.848528, 0.848528)(0.4125, 0.714471)
\psline[linecolor=blue,linewidth=1.5pt]{-}(0.310583, 1.15911)(0.4125, 0.714471)
\psline[linecolor=blue,linewidth=1.5pt]{-}(-0.310583, 1.15911)(0.4125, 0.714471)
\psline[linecolor=blue,linewidth=1.5pt]{-}(1.15911, 0.310582)(0.4125, 0.714471)
\psarc[linecolor=black,linewidth=0.5pt,fillstyle=solid,fillcolor=darkgreen]{-}(0.4125, 0.714471){0.09}{0}{360}
\psline[linecolor=blue,linewidth=1.5pt]{-}(-0.318198, 0.318198)(-0.848528, 0.848528)
\psline[linecolor=blue,linewidth=1.5pt]{-}(-0.434667, 0.116469)(-1.15911, 0.310583)
\psline[linecolor=blue,linewidth=1.5pt]{-}(-0.434667, -0.116469)(-1.15911, -0.310583)
\psline[linecolor=blue,linewidth=1.5pt]{-}(-0.318198, -0.318198)(-0.848528, -0.848528)
\psline[linecolor=blue,linewidth=1.5pt]{-}(-0.116469, -0.434667)(-0.310583, -1.15911)
\psline[linecolor=blue,linewidth=1.5pt]{-}(0.116469, -0.434667)(0.310583, -1.15911)
\psline[linecolor=blue,linewidth=1.5pt]{-}(0.318198, -0.318198)(0.848528, -0.848528)
\psline[linecolor=blue,linewidth=1.5pt]{-}(0.434667, -0.116469)(1.15911, -0.310583)
}
\psline[linestyle=dashed, dash= 1.5pt 1.5pt,linewidth=0.5pt]{-}(0,-1.2)(0.714471, -0.4125)
\psline[linestyle=dashed, dash= 1.5pt 1.5pt,linewidth=0.5pt]{-}(0,-1.2)(0,-0.45)
\psarc[linecolor=black,linewidth=0.5pt,fillstyle=solid,fillcolor=darkgreen]{-}(0.714471, -0.4125){0.09}{0}{360}
\rput(0.362347, -1.3523){$_1$}
\rput(0.989949, -0.989949){$_2$}
\rput(1.3523, -0.362347){$_{...}$}
\rput(1.4423, 0.362347){$_{2k}$}
\rput(0.989949, 0.989949){$_{...}$}
\rput(-0.362347, -1.3523){$_{N+2k}$}
\psarc[linecolor=black,linewidth=0.5pt]{-}(0,0){1.2}{0}{360}
\end{pspicture} \ \ .
\ee
Similarly, the diagrams for the operators $\Op_{N,j}$ with $0 \le j\le N$ are
\be
k=0: \quad
\Op_{N,j} = \
\begin{pspicture}[shift=-1.3](-1.2,-1.4)(1.4,1.4)
\psarc[linecolor=black,linewidth=0.5pt,fillstyle=solid,fillcolor=lightlightblue]{-}(0,0){1.2}{0}{360}
\psarc[linecolor=black,linewidth=0.5pt,fillstyle=solid,fillcolor=white]{-}(0,0){0.45}{0}{360}
\psline[linecolor=blue,linewidth=1.5pt]{-}(-0.116469,0.434667)(-0.310583,1.15911)
\psline[linecolor=blue,linewidth=1.5pt]{-}(0.116469,0.434667)(0.310583,1.15911)
\psline[linecolor=blue,linewidth=1.5pt]{-}(0.318198, 0.318198)(0.848528, 0.848528)
\psline[linecolor=blue,linewidth=1.5pt]{-}(0.434667, 0.116469)(1.15911, 0.310583)
\psline[linecolor=blue,linewidth=1.5pt]{-}(-0.318198, 0.318198)(-0.848528, 0.848528)
\psline[linecolor=blue,linewidth=1.5pt]{-}(-0.434667, 0.116469)(-1.15911, 0.310583)
\psline[linecolor=blue,linewidth=1.5pt]{-}(-0.434667, -0.116469)(-1.15911, -0.310583)
\psline[linecolor=blue,linewidth=1.5pt]{-}(-0.318198, -0.318198)(-0.848528, -0.848528)
\psline[linecolor=blue,linewidth=1.5pt]{-}(-0.116469, -0.434667)(-0.310583, -1.15911)
\psline[linecolor=blue,linewidth=1.5pt]{-}(0.116469, -0.434667)(0.310583, -1.15911)
\psline[linecolor=blue,linewidth=1.5pt]{-}(0.318198, -0.318198)(0.848528, -0.848528)
\psline[linecolor=blue,linewidth=1.5pt]{-}(0.434667, -0.116469)(1.15911, -0.310583)
\psline[linestyle=dashed, dash= 1.5pt 1.5pt,linewidth=0.5pt]{-}(0,-1.2)(0,-0.45)
\psline[linestyle=dashed, dash= 1.5pt 1.5pt,linewidth=0.5pt]{-}(0,-1.2)(0.825,0)
\psarc[linecolor=black,linewidth=0.5pt,fillstyle=solid,fillcolor=darkgreen]{-}(0.825,0){0.09}{0}{360}
\rput(0.362347, -1.3523){$_1$}
\rput(0.989949, -0.989949){$_{...}$}
\rput(1.48, -0.362347){$_{j-1}$}
\rput(1.37, 0.362347){$_{j}$}
\rput(0.989949, 0.989949){$_{...}$}
\rput(-0.362347, -1.3523){$_N$}
\psarc[linecolor=black,linewidth=0.5pt]{-}(0,0){1.2}{0}{360}
\end{pspicture} 
\ \ \ ,
\qquad
k>0: \quad
\Op_{N,j} = \ 
\begin{pspicture}[shift=-1.3](-1.2,-1.4)(1.4,1.4)
\psarc[linecolor=black,linewidth=0.5pt,fillstyle=solid,fillcolor=lightlightblue]{-}(0,0){1.2}{0}{360}
\psarc[linecolor=black,linewidth=0.5pt,fillstyle=solid,fillcolor=white]{-}(0,0){0.45}{0}{360}
\rput{0}(0,0){
\psline[linecolor=blue,linewidth=1.5pt]{-}(0.848528, 0.848528)(0.4125, 0.714471)
\psline[linecolor=blue,linewidth=1.5pt]{-}(0.310583, 1.15911)(0.4125, 0.714471)
\psline[linecolor=blue,linewidth=1.5pt]{-}(-0.310583, 1.15911)(0.4125, 0.714471)
\psline[linecolor=blue,linewidth=1.5pt]{-}(1.15911, 0.310582)(0.4125, 0.714471)
\psline[linecolor=blue,linewidth=1.5pt]{-}(-0.318198, 0.318198)(-0.848528, 0.848528)
\psline[linecolor=blue,linewidth=1.5pt]{-}(-0.434667, 0.116469)(-1.15911, 0.310583)
\psline[linecolor=blue,linewidth=1.5pt]{-}(-0.434667, -0.116469)(-1.15911, -0.310583)
\psline[linecolor=blue,linewidth=1.5pt]{-}(-0.318198, -0.318198)(-0.848528, -0.848528)
\psline[linecolor=blue,linewidth=1.5pt]{-}(-0.116469, -0.434667)(-0.310583, -1.15911)
\psline[linecolor=blue,linewidth=1.5pt]{-}(0.116469, -0.434667)(0.310583, -1.15911)
\psline[linecolor=blue,linewidth=1.5pt]{-}(0.318198, -0.318198)(0.848528, -0.848528)
\psline[linecolor=blue,linewidth=1.5pt]{-}(0.434667, -0.116469)(1.15911, -0.310583)
}
\psbezier[linestyle=dashed, dash= 1.5pt 1.5pt,linewidth=0.5pt]{-}(0,-1.2)(0.2,-1)(1,0)(0.4125, 0.714471)
\psline[linestyle=dashed, dash= 1.5pt 1.5pt,linewidth=0.5pt]{-}(0,-1.2)(0,-0.45)
\psarc[linecolor=black,linewidth=0.5pt,fillstyle=solid,fillcolor=darkgreen]{-}(0.4125, 0.714471){0.09}{0}{360}
\rput(0.362347, -1.3523){$_1$}
\rput(0.989949, -0.989949){$_{...}$}
\rput(1.35, -0.362347){$_{j}$}
\rput(1.52, 0.362347){$_{j+1}$}
\rput(0.989949, 0.989949){$_{...}$}
\rput(-0.375288, 1.40059){$_{j+2k}$}
\rput(0.989949, 0.989949){$_{...}$}
\rput(-0.362347, -1.3523){$_{N+2k}$}
\psarc[linecolor=black,linewidth=0.5pt]{-}(0,0){1.2}{0}{360}
\end{pspicture} \ \ \ .
\ee
We note that the relations \eqref{eq:def.Oj} for $k>0$ imply that
\be
c_{N+2k,2k}\,\Op_{N} \, c^\dag_{N,0} = x^{-1}\, \Omega\, \Op_{N-2}\,.
\ee
Indeed, changing $N\to N+2$ in \eqref{eq:def.Oj.c} and acting with $x^{-1}c_{2k}$ from the left and $c_0^\dag$ from the right, we obtain
\begin{alignat}{1}
c_{2k}\, \Op_N\, c_0^\dag &= x^{-1}\, c_{2k}\,c_0\, \Op_{N+2}\, (c_0^\dag)^2
= x^{-1}\, c_0\, c_{2k+1}\, \Op_{N+2}\, (c_0^\dag)^2
= x^{-1}c_0\, \Op_N\, c_1\, (c_0^\dag)^2 
\nn \\&
= x^{-1}c_0\, \Op_N\, c_{N-3}^\dag 
= x^{-1}c_0\, c_{N-3+2k}^\dag\, \Op_{N-2} = x^{-1}\Omega\, \Op_{N-2} \,,
\end{alignat}
where we used the relations \eqref{eq:def.Oj.a} and \eqref{eq:def.Oj.b}.\medskip

Similarly, let $k \in \mathbb Z_{\ge 0}$, $x\in\Cbb^\times$, and $\repM$ and $\repM'$ be families of modules over $\tl_N(\beta)$. In this case, a {\it connectivity operator of type $(k,x)$ from $\repM$ to $\repM'$} is a collection of linear maps $\Op_N: \repM(N) \to \repM(N+2k)$, satisfying
\begin{subequations}
\begin{alignat}{2}
c_{N+2k,j}\,\Op_{N} &= \left\{\begin{array}{cl}
0 & j=1,2,\dots,2k-1\,, \\[0.1cm]
\Op_{N-2} \,c_{N,j-2k} & j = 2k+1,k+2,\dots,N+2k-1\,,
\quad N \ge 2
\end{array}\right.
\\ 
c^\dag_{N+2k+2,j}\,\Op_N &= \left\{\begin{array}{cl}
(\id_2 \otimes \Op_{N}) \,c^\dag_{N+2,1} & j=1\,, \\[0.1cm]
\Op_{N+2} \,c^\dag_{N+2,j-2k} & j=2k+1,2k+2, \dots, N+2k+1\,,
\end{array}\right.
\\[0.2cm]
c_{N+2k,1}\, (\id_1 \otimes \Op_{N-2}\otimes \id_1)\, c^\dag_{N,1} &= \left\{\begin{array}{cl}
(x+x^{-1})\, \Op_{N-2} & k=0\,, \\[0.1cm]
x\, \Op_{N-2} & k>0 \,,
\end{array}\right.
\quad N \ge 2\,.
\end{alignat}
\end{subequations}

We now discuss how connectivity operators are realised in the critical ADE models. More precisely, for each irreducible factor $\repQ_{k,x}$ in the decompositions \eqref{eq:M.decomp} and \eqref{eq:M.decomp.twist}, we construct a connectivity operator of type $(k,x)$ using the corresponding insertion state $w_{k,x}$. The construction applies to both families of modules associated to periodic boundary conditions and to fixed boundary conditions. Starting with the periodic boundary conditions, we let $K$ and $L$ be two commuting automorphisms of the graph~$\cal G$ of $\g$. Each state $w\in\repM_{\g,\mu,K}(2k)$ can be written in component form as
\begin{equation}
w = \sum_{\boldsymbol {a}} w_{a_0,a_1,\dots,a_{2k}} \ket{a_0,a_1,\dots,a_{2k}} \,,
\end{equation}
where the sum over $\boldsymbol {a}$ runs on all states $|a_0,a_1,\dots,a_{2k}\rangle$ in $\repM_{\g,\mu,K}(2k)$. For each such state $w$ and for each integer $N$ admissible for $\repM_{\g,\mu,L}$, we define $\phi(w)$ to be the collection of linear maps
$\phi_N(w):\repM_{\g,\mu,L}(N)\to\repM_{\g,\mu,KL}(N+2k)$ given by
\begin{equation} \label{eq:def.phi.j(w)}
\phi_N(w) \cdot \ket{b_0,b_1,\dots,b_N}
= \kappa_\mu^{N/2} \, 
\sum_{\boldsymbol a} \delta_{a_0,b_0}
 \frac{w_{a_0,a_1,\dots,a_{2k}}}{S_{a_0\mu}}
\ket{a_0,a_1,\dots,a_{2k},K(b_1),K(b_{2}),\dots,K(b_N)} \,.
\end{equation}
The action of the connectivity operators on the vacuum state $w_\mu$ defined in \eqref{eq:w.nu} yields
\begin{equation}
\phi_0(w) \cdot w_\mu = w \quad \textrm{for }L=\id\,,
\qquad \textrm{and}\qquad
\phi_N(w_\mu) \cdot w = w\quad \textrm{for }K=\id\,.
\end{equation}
In general, for all commuting automorphisms $K$ and $L$, all states $w\in\repM_{\g,\mu,K}(2k)$ and all admissible integers $N$, we have the relations
\begin{subequations} \label{eq:O(c.w)}
\begin{alignat}{2}
& \phi_N(c_{j}\cdot w) = c_{j}\, \phi_N(w) \,,
\qquad &&j=1,2,\dots, 2k-1 \,, \qquad 
\label{eq:O(c.w)a} \\
& \phi_N(c^\dag_{j}\cdot w) = c^\dag_{j}\,\phi_N(w) \,,
\qquad &&j=1,2, \dots, 2k+1 \,, \label{eq:O(c.w)b} \\
& \phi_N(c_{0}\cdot w) = c_{0}\,c_{2k}\,\phi_{N+2}(w) \, c_{0}^\dag \,,
\qquad && k\geq 1 \,,\label{eq:O(c.w)c} \\
& \phi_N(c^\dag_{0}\cdot w) = \Omega^{-1}\phi_{N+2}(w)\,c_{0}^\dag \,.
\label{eq:O(c.w)d}
\end{alignat}
\end{subequations}
We prove these relations as follows.
For \eqref{eq:O(c.w)a}, we write
\begin{alignat}{1}
&\phi_N(c_j\cdot\ket{a_0,a_1,\dots,a_{2k}}) \cdot \ket{b_0,b_1,\dots,b_N} \nn \\
&\qquad= \frac{\delta_{a_{j-1},a_{j+1}}}{S_{a_{j+1}\mu}}
\phi_N(\ket{a_0,a_1,\dots,a_{j-1},a_{j+2},a_{j+3},\dots,a_{2k}})
\cdot\ket{b_0,b_1,\dots,b_N}\nn \\
&\qquad= \kappa_\mu^{N/2} \frac{\delta_{a_0,b_0}}{S_{a_0\mu}}
\frac{\delta_{a_{j-1},a_{j+1}}}{S_{a_{j+1}\mu}}
\ket{a_0,a_1,\dots,a_{j-1},a_{j+2},a_{j+3},\dots,a_{2k},K(b_1),K(b_2),\dots,K(b_N)}\nn \\
&\qquad= c_j\,\phi_N(\ket{a_0,a_1,\dots,a_{2k}}) \cdot \ket{b_0,b_1,\dots,b_N} \,.
\end{alignat}
A similar argument holds for \eqref{eq:O(c.w)b}.
For \eqref{eq:O(c.w)c}, we write
\begin{alignat}{1}
&\phi_N(c_0\cdot\ket{a_0,a_1,\dots,a_{2k}}) \cdot \ket{b_0,b_1,\dots,b_N} \nn \\
&\qquad= \kappa_\mu^{-1/2} \frac{\delta_{a_{2k-1},K(a_1)}}{S_{a_1\mu}}
\phi_N(\ket{a_1,a_2,\dots,a_{2k-1}})\cdot \ket{b_0,b_1,\dots,b_N} \nn \\
&\qquad= \kappa_\mu^{(N-1)/2} \frac{\delta_{a_1,b_0}\delta_{a_{2k-1},K(a_1)}}{S_{a_1\mu}^2} \ket{a_1,a_2,\dots,a_{2k-1},K(b_1),K(b_2),\dots,K(b_N)} \,,
\end{alignat}
whereas
\begin{alignat}{1}
&c_0\,c_{2k}\,\phi_{N+2}(\ket{a_0,a_1,\dots,a_{2k}})\, c_0^\dag\cdot \ket{b_0,b_1,\dots,b_N} \nn \\
&\qquad= \gamma_\mu^{1/2} \sum_{b'_0=1}^n A_{b_N,L(b'_0)} S_{b'_0\mu} c_0\,c_{2k}\,\phi_{N+2}(\ket{a_0,a_1,\dots,a_{2k}}) \cdot\ket{b'_0,b_0,b_1,\dots,b_N,L(b'_0)} \nn \\
&\qquad= \gamma_\mu^{1/2} \kappa_\mu^{(N+2)/2}
A_{b_N,L(a_0)} c_0\,c_{2k} \cdot\ket{a_0,a_1,\dots,a_{2k},K(b_0),K(b_1),\dots,K(b_N),KL(a_0)} \nn \\
&\qquad= \gamma_\mu^{1/2} \kappa_\mu^{(N+2)/2}
A_{b_N,L(a_0)} \frac{\delta_{a_{2k-1},K(b_0)}}{S_{K(b_0)\mu}} \, c_0\cdot\ket{a_0,a_1,\dots,a_{2k},K(b_0),K(b_1),\dots,K(b_N),KL(a_0)} \nn \\
&\qquad= \kappa_\mu^{(N+1)/2}
A_{b_N,L(a_0)} \frac{\delta_{a_{2k-1},K(b_0)}}{S_{K(b_0)\mu}}
\, \frac{\delta_{K(b_N),KL(a_1)}}{S_{a_1\mu}}
\cdot\ket{a_1,a_2,\dots,a_{2k},K(b_0),K(b_1),\dots,K(b_N)} \,.
\end{alignat}
We then use the fact that $b_N=L(b_0)$, which yields $\delta_{K(b_N),KL(a_1)}=\delta_{b_0,a_1}$, and $A_{b_N,L(a_0)}=A_{b_0,a_0}=1$ for $b_0=a_1$. As a result, we obtain
\begin{alignat}{1}
&c_0\,c_{2k}\,\phi_{N+2}(\ket{a_0,a_1,\dots,a_{2k}})\, c_0^\dag \cdot\ket{b_0,b_1,\dots,b_N} \nn \\
&\qquad = \kappa_\mu^{(N-1)/2} \frac{\delta_{a_1,b_0}\delta_{a_{2k-1},K(a_1)}}{S_{a_1\mu}^2} \,\ket{a_1,a_2,\dots,a_{2k-1},K(b_1),K(b_2),\dots,K(b_N)} \,,
\end{alignat}
which proves \eqref{eq:O(c.w)c}.
For \eqref{eq:O(c.w)d}, a similar calculation yields
\begin{alignat}{1}
&\phi_N(c_0^\dag\cdot\ket{a_0,a_1,\dots,a_{2k}})\cdot\ket{b_0,b_1,\dots,b_N}
= \Omega^{-1}\, \phi_{N+2}(\ket{a_0,a_1,\dots,a_{2k}})\,c_0^\dag\cdot\ket{b_0,b_1,\dots,b_N} \nn \\
&\qquad= \kappa_\mu^{(N+1)/2} \, A_{a_{2k},K(b_0)} \, \ket{b_0,a_0,a_1,\dots,a_{2k},K(b_0),K(b_1),\dots,K(b_N)} \,.
\end{alignat}
These relations allow us to prove the following result.
\begin{Proposition}
Let $K$ and $L$ be commuting automorphisms of $\cal G$, and $w$ be an insertion state with parameters $(k,x)$ in $\repM_{\g,\mu,K}$. Then $\phi(w)$ is a connectivity operator of type $(k,x)$ from $\repM_{\g,\mu,L}$ to $\repM_{\g,\mu,KL}$.
\end{Proposition}
\proof
By construction, the state $w$ satisfies \eqref{eq:def.insertion.kx}.
For \eqref{eq:def.Oj.a}, we have for $j=1,2,\dots,2k-1$
\begin{equation}
c_j\, \phi_N(w) = \phi_N(c_j\,w) = 0 \,.
\end{equation}
Similarly, for $j=1,2,\dots,N-1$, we have
\begin{alignat}{1}
&c_{2k+j}\, \phi_N(w) \cdot \ket{b_0,b_1,\dots,b_N} \nn \\
&\qquad= \kappa_\mu^{N/2}\,\frac{\delta_{a_0,b_0}}{S_{a_0\mu}} \,\frac{\delta_{b_{j-1},b_{j+1}}}{S_{b_{j+1}\mu}}
\, \ket{a_0,a_1,\dots,a_{2k},K(b_1),K(b_2),\dots,K(b_{j-1}),K(b_{j+2}),K(b_{j+3}),\dots,K(b_N)} \nn \\
&\qquad= \phi_{N-2}(w) \,c_j \cdot \ket{b_0,b_1,\dots,b_N} \,.
\end{alignat}
The proof for \eqref{eq:def.Oj.b} is similar.
For \eqref{eq:def.Oj.c}, we write
\be
c_0\, \phi_N(w) \, c^\dag_{0} = 
c_1\, \Omega^{-1} \phi_N(w)\, c_{0}^\dag = c_1\, \phi_{N-2}(c_0^\dag w) = \phi_{N-2}(c_1 c_0^\dag w)
= \left\{\begin{array}{cl}
(x+x^{-1})\, \phi_{N-2}(w)& k=0 \,,\\[0.1cm] 
x\, \phi_{N-2}(w) & k>0\,,
\end{array}\right.
\ee
where we used the identites $c_1c_0^\dag=f$ for $k=0$ and $c_1c_0^\dag=\Omega$ for $k>0$.
\eproof 

For the ADE modules with fixed boundary conditions, we use the same ideas and define $\phi(w)$ as the collection of linear maps $\phi_N(w): \repM_{\g,\mu,a,b}(N) \to \repM_{\g,\mu,a,K(b)}(N+2k)$ given by \eqref{eq:def.phi.j(w)}. With the same arguments as above, we obtain the following proposition.
\begin{Proposition}
Let $K$ be an automorphism of $\mathcal G$, $a$ and $b$ be heights of $\cal G$, and $w$ be an insertion state with parameters $(k,x)$ in $\repM_{\g,\mu,K}$. Then $\phi(w)$ is a connectivity operator of type $(k,x)$ from $\repM_{\g,\mu,a,b}$ to $\repM_{\g,\mu,a,K(b)}$.
\end{Proposition}

For $K=\id$, using the insertion states $w_\nu$ with $(k,x)=(0,x_\nu)$ defined in \eqref{eq:w.nu}, we recover the local operators introduced in \cite{PasquierOpContent87}, namely
\begin{equation} \label{eq:O.nu}
\Op^{\tinyx \nu}_{N,j}= \Omega^{-j}\,\phi_N(w_\nu)\,\Omega^j\,,
\qquad 
\Op^{\tinyx \nu}_{N,j}\cdot \ket{a_0, a_1,\dots, a_N}
= \frac{S_{a_j,\nu}}{S_{a_j,\mu}} \, \ket{a_0, a_1, \dots, a_N} \,,
\end{equation}
where $\nu \in \{1,2,\dots,n\}$ is an index of the adjacency matrix, and $j \in \{0,1,\dots, N\}$. Thus $\Op^{\tinyx \nu}$ is a connectivity operator of type $(0,x_\nu)$. As discussed in \cite{PasquierOpContent87}, it follows from the fact that the eigenvectors form an orthonormal basis that
\begin{equation}
\frac{S_{a\nu}S_{a\nu'}}{S_{a\mu}} = 
\sum_{\nu''=1}^n C_{\nu,\nu'}^{\nu''} \, S_{a\nu''} \,,
\qquad \textrm{where}\qquad 
C_{\nu,\nu'}^{\nu''} = \sum_{a=1}^n \frac{S_{a\nu}S_{a\nu'}S_{a\nu''}^*
}{S_{a\mu}} \,,
\end{equation}
for all $a,\nu,\nu'\in\{1,2,\dots,n\}$. This implies the decomposition of the operator product
\begin{equation}
\Op^{\tinyx \nu}_{N,j} \, \Op^{\tinyx {\nu'}}_{N,j}
= \sum_{\nu''=1}^n C_{\nu,\nu'}^{\nu''} \, \Op^{\tinyx {\nu''}}_{N,j} \,.
\end{equation}

In the case $K\neq \id$, using the twisted insertion states $\wt{w}_\nu$ with $(k,x)=(0,\wt{x}_\nu)$ defined in \eqref{eq:w.nu.twist}, we construct the connectivity operators $\wt{\Op}^{\tinyx \nu}_{N,j}=\Omega^{-j}\,\phi_N(\wt{w}_\nu)\,\Omega^j$, whose action is
\begin{equation} \label{eq:Ot.nu}
\wt{\Op}^{\tinyx \nu}_{N,j}\cdot \ket{a_0, a_1,\dots, a_N}
= \kappa_\mu^{(N-j)/2}\, \delta_{K(a_j),a_j}\,
\frac{\wt{S}_{a_j,\nu}}{S_{a_j,\mu}} \, \ket{a_0, a_1, \dots, a_j,K(a_{j+1}),K(a_{j+2}),\dots,K(a_N)} \,.
\end{equation}
Using a similar argument to the case $K=\id$, we obtain the operator product
\begin{equation}
\wt\Op^{\tinyx \nu}_{N,j} \, \wt\Op^{\tinyx {\nu'}}_{N,j}
= \sum_{\nu''=1}^{\wt n} \wt{C}_{\nu,\nu'}^{\nu''} \, \wt\Op^{\tinyx {\nu''}}_{N,j} \,,
\qquad \wt{C}_{\nu,\nu'}^{\nu''} = \sum_{a \in \wt{\cal G}}
\frac{\wt{S}_{a\nu}\wt{S}_{a\nu'}\wt{S}_{a\nu''}^*}{S_{a\mu}} \,.
\end{equation}

\subsection{Singular-vector relations and difference equations}

In this section, we set $K = \id$ and use the decomposition of the family $\repM_{\g,\mu,\id}$ given in \cref{thm:M.decomp} to show that the local operators $\mathcal O^{\tinyx \nu}_{N,j}$ satisfy extra quotient relations in addition to the relations \eqref{eq:def.Oj}, which in turn lead to linear relations for their correlation functions. Let us recall from \cref{prop:Q.properties} that, for $s\in \mathbb Z_{\ge 0}$ satisfying $s\le \frac{p'}2$, we have
\begin{equation}
\repQ_{0,\eps q^s} = \repW_{0,\eps q^s}
\big/ \{v^{\tinyx{1,\eps}}_{0,s} \equiv 0 \,,\, v^{\tinyx{-1,\eps'}}_{0,s'} \equiv 0\} \,,
\end{equation}
where $s'=p'-s$ and $\eps'=(-1)^p \eps$. In analogy with the Verma modules of the Virasoro algebra, we say that $v^{\tinyx{1,\eps}}_{0,s}$ and $v^{\tinyx{-1,\eps'}}_{0,s'}$ are \emph{singular vectors} in the family $\repW_{0,\eps q^s}$. Moreover, we can write
\begin{equation}
v^{\tinyx{\sigma,\eps}}_{0,s} = \mu^{\tinyx{\sigma,\eps}}_{0,s}\cdot u_0\,,\end{equation}
for some unique $\mu^{\tinyx{\sigma,\eps}}_{0,s}\in\cLk{0,0}(2s,0)$, where $u_0$ is the unique empty link state of $\repW_{0,\eps q^s}(0)$. For example, we have
\begin{subequations} \label{eq:mu.0s}
\begin{alignat}{2}
\mu^{\tinyx{\sigma,\eps}}_{0,1} &= c_0^\dag + \eps c_1^\dag \,,
\\[0.1cm]
\mu^{\tinyx{\sigma,\eps}}_{0,2} &= 
(c_0^\dag)^2 + \eps\,[2]\, c_2^\dag c_0^\dag + c_2^\dag c_1^\dag
+ \eps\, c_1^\dag c_0^\dag + [2]\, c_3^\dag c_1^\dag + \eps\, c_3^\dag c_0^\dag \,,
\\[0.1cm]
\mu^{\tinyx{\sigma,\eps}}_{0,3} &= (c_0^\dag)^3 + \eps\, [3]\, c_3^\dag(c_0^\dag)^2 + [3]\, c_3^\dag c_2^\dag c_0^\dag + \eps\, c_3^\dag c_2^\dag c_1^\dag 
\nn\\& 
+ \eps\,[2]\, c_2^\dag(c_0^\dag)^2 + \eps\,[2]\, c_4^\dag(c_0^\dag)^2 + \eps\, [2]\, c_4^\dag c_2^\dag c_1^\dag + [2][3]\, c_4^\dag c_2^\dag c_0^\dag
\nn\\& 
+ \eps\, c_5^\dag(c_0^\dag)^2 + [3]\, c_5^\dag c_2^\dag c_0^\dag + \eps\, [3]\, c_5^\dag c_2^\dag c_1^\dag + c_2^\dag c_1^\dag c_0^\dag
\nn\\& 
+ \eps\, c_1^\dag(c_0^\dag)^2 + [3]\, c_4^\dag c_1^\dag c_0^\dag + \eps\, [3]\, c_4^\dag c_3^\dag c_1^\dag + c_4^\dag c_3^\dag c_0^\dag
\nn\\& 
+ [2]\, c_5^\dag c_3^\dag c_0^\dag + [2]\, c_5^\dag c_1^\dag c_0^\dag + [2]\, c_3^\dag c_1^\dag c_0^\dag + \eps\, [2][3]\, c_5^\dag c_3^\dag c_1^\dag\,.
\end{alignat}
\end{subequations}
Moreover, we note that $\mu^{\tinyx{\sigma, \eps}}_{0,s}$ does not depend on $\sigma$, because $\repW_{0,\eps q^s}$ is invariant under $q\to q^{-1}$.\medskip

In the ADE family $\repM_{\g,\mu,\id}$, let $w_\nu$ be an insertion state with parameters $(0,x_\nu)$, as defined in~\eqref{eq:w.nu}. The insertion homomorphism $\varphi_\nu:\repW_{0,x_\nu}\to \repM_{\g,\mu,\id}$ vanishes on $\repR_{0,x_\nu}$, as otherwise $\repR_{0,x_\nu}$ would appear in the decomposition of $\repM_{\g,\mu,\id}$. From \cref{prop:permut.nu}, we know that there is a pair $(s,\eps)$ with $s \in \{1,2,\dots,\frac{p'}2\}$ and $\eps\in\{-1,+1\}$ satisfying $\eps (q^s+q^{-s}) = x_\nu+x_\nu^{-1}$. Hence, we have
\begin{equation}
\label{eq:sv.relation.derivation}
\mu^{\tinyx{1,\eps}}_{0,s}\cdot w_\nu = \mu^{\tinyx{1,\eps}}_{0,s}\cdot\varphi_\nu(u_0) = \varphi_\nu\big(\mu^{\tinyx{1,\eps}}_{0,s}\cdot u_0\big)
= \varphi_\nu\big(v^{\tinyx{1,\eps}}_{0,s}\big)=0 \,,
\end{equation}
and similarly for $\mu^{\tinyx{-1,\eps'}}_{0,s'}\cdot w_\nu$.
Thus, each insertion state $w_\nu$ defined in \eqref{eq:w.nu} satisfies the \emph{singular-vector relations}
\begin{equation}
\mu^{\tinyx{1,\eps}}_{0,s}\cdot w_\nu = 0 \,, \qquad \mu^{\tinyx{-1,\eps'}}_{0,s'}\cdot w_\nu = 0 \,,
\end{equation}
with $x_\nu+x^{-1}_\nu = \eps (q^s+q^{-s}) = \eps' (q^{s'}+q^{-s'})$.
Writing
\be
c_{s+1}c_{s+2}\cdots c_{2s} \, \phi_{N} (\mu^{\tinyx{1,\eps}}_{0,s}\cdot w_\nu)=0\,,\end{equation}
we use \eqref{eq:O(c.w)} and obtain a local difference equation for the operators $\Op_{N,j}^{\tinyx \nu}$ associated to the pair $(s,\eps)$. For example, we have
\begin{subequations}
\begin{alignat}{2}
s=1: \qquad &\Op_{N,1}^{\tinyx \nu} + \eps\, \Op_{N,0}^{\tinyx \nu} = 0 \,, 
\\[0.1cm]
s=2: \qquad &
\Op_{N,2}^{\tinyx \nu} + \eps\, [2]\,\Op_{N,1}^{\tinyx \nu} + \Op_{N,0}^{\tinyx \nu} +\eps\, e_1\,\Op_{N,1}^{\tinyx \nu} + [2]\, e_1\, \Op_{N,0}^{\tinyx \nu} + \Op_{N,1}^{\tinyx \nu}\,e_1 =0 \,,
\\[0.1cm]\nn
s=3: \qquad &
\mathcal O^{\tinyx \nu}_{N,3}
+ \eps \, [3]\, \mathcal O^{\tinyx \nu}_{N,2}
+ [3]\, \mathcal O^{\tinyx \nu}_{N,1}
+ \eps\, \mathcal O^{\tinyx \nu}_{N,0}
\\[0.1cm]&\nn
+ [2]\, e_{j+1}\,\mathcal O(j+3)
+ \eps\,[2]\,[3]\, e_{j+1}\, \mathcal O^{\tinyx \nu}_{N,2} 
+ [2]\, e_{j+1} \,\mathcal O^{\tinyx \nu}_{N,1}
+ [2]\, \mathcal O^{\tinyx \nu}_{N,1}\, e_{j+1}
\\[0.1cm]&
+ \eps\,[2]\, e_{j+2} \, \mathcal O^{\tinyx \nu}_{N,2} 
+ \eps\,[2]\, \mathcal O^{\tinyx \nu}_{N,2} \, e_{j+2}
+ [2]\,[3]\, e_{j+2} \, \mathcal O^{\tinyx \nu}_{N,1}
+ \eps \,[2]\, e_{j+2} \, \mathcal O^{\tinyx \nu}_{N,0}
\\[0.1cm]&\nn
+ \eps\, e_{j+1}\, e_{j+2}\, \mathcal O^{\tinyx \nu}_{N,2}
+ [3]\, e_{j+1}\, e_{j+2}\, \mathcal O^{\tinyx \nu}_{N,1}
+ \mathcal O^{\tinyx \nu}_{N,1}\, e_{j+1}\, e_{j+2}
+ \eps \, [3]\, e_{j+1}\, e_{j+2}\, \mathcal O^{\tinyx \nu}_{N,0}
\\[0.1cm]&\nn
+ [3]\, e_{j+2}\, e_{j+1} \, \mathcal O(j+3)
+ \eps \, [3]\, e_{j+2}\, e_{j+1} \, \mathcal O^{\tinyx \nu}_{N,2}
+ \eps\, \mathcal O^{\tinyx \nu}_{N,2}\, e_{j+2}\, e_{j+1}
+ e_{j+2}\, e_{j+1} \, \mathcal O^{\tinyx \nu}_{N,1}\,.
\end{alignat}
\end{subequations}
The singular-vector relation $\mu^{\tinyx{-1,\eps'}}_{0,s'}\cdot w_\nu=0$ results in a second difference equation for $\Op_{N,j}^{\tinyx \nu}$, obtained from the above construction simply by changing $s\to s'$ and $\eps \to \eps'$.\medskip 

We have thus obtained two singular-vector relations satisfied by the connectivity operators of type $(0,x_\nu)$ in the ADE lattice models. We note that the same ideas allow us to derive singular-vector relations for the connectivity operators of type $(k,x)$ with $k>0$ constructed in \cref{sec:connectivity.ops}.

\subsection{Boundary operators}

Let $k\in\frac12\Zbb_{\geq 0}$, and $\repM$ and $\repM'$ be two families of modules over $\tl_N(\beta)$. We define a \emph{boundary operator~$\Op$ of type $k$ from $\repM$ to $\repM'$} to be a collection $\Op_{N}$ of linear maps from $\repM(N)$ to $\repM'(N+2k)$, for each admissible integer $N$ in $\repM$, satisfying the relations
\begin{subequations}
\begin{alignat}{2}
c_{N+2k,j}\,\Op_{N} &= \left\{\begin{array}{cl}
0 & j=1,2,\dots,2k-1\,, \\[0.1cm]
\Op_{N-2} \,c_{N,j-2k} & j = 2k+1,k+2,\dots,N+2k-1\,,\quad N \ge 2\,,
\end{array}\right. 
\\ 
c^\dag_{N+2k+2,j}\,\Op_N &= \left\{\begin{array}{cl}
(\id_2 \otimes \Op_{N}) \,c^\dag_{N+2,1} & j=1\,, \\[0.1cm]
\Op_{N+2} \,c^\dag_{N+2,j-2k} & j=2k+1,2k+2, \dots, N+2k+1\,.
\end{array}\right.
\end{alignat}
\end{subequations}
The operator $\Op_N$ is represented by the diagram 
\be
\Op_N = \
\begin{pspicture}[shift=-0.6](0,-0.7)(3.6,0.7)
\pspolygon[fillstyle=solid,fillcolor=lightlightblue,linewidth=0pt](0,-0.5)(3.6,-0.5)(3.6,0.5)(0,0.5)
\psbezier[linecolor=blue,linewidth=1.5pt]{-}(0.2,-0.5)(0.2,-0.2)(0.2,-0.1)(0,0)
\psbezier[linecolor=blue,linewidth=1.5pt]{-}(0.6,-0.5)(0.6,-0.2)(0.6,-0.1)(0,0)
\psbezier[linecolor=blue,linewidth=1.5pt]{-}(1.0,-0.5)(1.0,-0.2)(1.0,-0.1)(0,0)
\psbezier[linecolor=blue,linewidth=1.5pt]{-}(1.4,-0.5)(1.4,-0.2)(1.4,0)(0,0)
\psarc[linecolor=black,linewidth=0.5pt,fillstyle=solid,fillcolor=darkgreen]{-}(0,0){0.09}{0}{360}
\multiput(0,0)(0.4,0){5}{\psline[linecolor=blue,linewidth=1.5pt]{-}(1.8,-0.5)(1.8,0.5)}
\rput(0.2,-0.75){$_1$}
\rput(0.6,-0.75){$_2$}
\rput(1.0,-0.75){$_{...}$}
\rput(1.4,-0.75){$_{2k}$}
\rput(3.4,-0.75){$_{N+2k}$}
\end{pspicture}\ \ .
\ee
For $k=0$, the operator $\Op_{N}$ acts as the identity.
\medskip

Let $a$, $b$ and $c$ be three heights of $\cal G$. Each state $w\in\repM_{\g,\mu,a,b}(2k)$ can be written in component form as
\begin{equation}
w = \sum_{\boldsymbol a} w_{a_0,a_1,\dots,a_{2k}} \ket{a_0,a_1,\dots,a_{2k}} \,,
\end{equation}
where the sum runs on all states in $\repM_{\g,\mu,a,b}(2k)$. We define the boundary operator $\psi(w)$ as the collection of linear maps $\psi_N(w):\repM_{\g,\mu,b,c}(N) \to \repM_{\g,\mu,a,c}(N+2k)$ given by
\begin{equation}
\psi_N(w)\cdot \ket{b_0,b_1,\dots,b_N}
= \sum_{\boldsymbol a} w_{a_0,a_1,\dots,a_{2k}} \, \ket{a_0,a_1,\dots, a_{2k},b_1,b_2,\dots,b_N} \,.
\end{equation}
This operator satisfies
\begin{subequations}
\begin{alignat}{2}
\psi_N(c_j\cdot w) &= c_j \,\psi_N(w)\qquad &&j=1,2,\dots, 2k-1 \,,
\\[0.1cm] 
\psi_N(c_j^\dag\cdot w) &= c_j^\dag \,\psi_N(w)\qquad &&j=1,2,\dots, 2k+1 \,.
\end{alignat}
\end{subequations}
Moreover, for all states $w\in\repM_{\g,\mu,a,b}(2k)$, we have
\begin{equation}
\psi_N(w) \cdot \ket{b} = w \,.
\end{equation}
Using similar arguments to the case of bulk operators, we obtain the following result. 
\begin{Proposition}
Let $k \in \frac12 \mathbb Z_{\ge 0}$ with $k \le \tfrac{p'}{2}-1$, and $w$ be an insertion state with $2k$ defects in $\repM_{\g,\mu,a,b}$. Then $\psi(w)$ is a boundary operator of type $k$ from $\repM_{\g,\mu,b,c}$ to $\repM_{\g,\mu,a,c}$.
\end{Proposition}

We recall from \cref{prop:Qk.Rk} that the family $\repV_{k}$ with $q = \eE^{-\ir \pi p/p'}$ has the singular vector $v_k$ defined in \eqref{eq:vk}. For $k \le \tfrac{p'}{2}-1$, we have $r_k=1$ and thus $k'=p'-1-k$. The state $v_k$ admits a unique expression of the form $v_k=\lambda_k\cdot u_k$ for some unique $\lambda_k\in\cL_0^{\tinyx k}(2k',2k)$. For example, for $p'=2,3,4$, we have
\begin{subequations}
\begin{alignat}{2}
&p'=2: \quad 
&&\lambda_0 = c_1^\dag\,, 
\qquad \qquad
p'=3: \quad 
\left\{\begin{array}{l}
\lambda_0 = c_2^\dag c_1^\dag + (-1)^{p+1} c_3^\dag c_1^\dag\,,\\[0.1cm]
\lambda_{\frac12} = c_2^\dag + (-1)^{p+1} c_1^\dag\,,
\end{array}\right.
\\[0.1cm]
&p'=4: \quad &&
\left\{\begin{array}{l}
\lambda_0 = c_1^\dag c_2^\dag c_3^\dag + c_1^\dag c_2^\dag c_5^\dag + c_1^\dag c_3^\dag c_4^\dag + [2]\, c_1^\dag c_3^\dag c_5^\dag + [2]\, c_1^\dag c_2^\dag c_4^\dag\,,\\[0.1cm]
\lambda_{\frac12} = c_1^\dag c_2^\dag + c_1^\dag c_4^\dag + c_2^\dag c_3^\dag + [2]\, c_1^\dag c_3^\dag + [2]\, c_2^\dag c_4^\dag\,,\\[0.1cm]
\lambda_{1} = c_1^\dag + [2]\, c_2^\dag + c_3^\dag\,.
\end{array}\right.
\end{alignat}
\end{subequations}

Repeating the argument used in \eqref{eq:sv.relation.derivation}, we obtain in the ADE models the singular-vector relation $\lambda_k\cdot w=0$. We then define
\begin{equation}
\wh\lambda_k=c_{k+k'+1}c_{k+k'+2}\dots c_{2k'}\, (\lambda_k \otimes \id_{k'-k}) \in \tl_{p'-1}(\beta)
\end{equation}
and obtain the local difference equation for $\psi_N$ on $\repM_{\g,\mu,b,c}(N)$
\begin{equation}
(\wh\lambda_k\otimes \id_{N-k'+k}) \, \psi_N = 0 \,, \qquad N \ge k'-k\,.
\end{equation}
For example, for $p'=3,4$, we have
\begin{subequations}
\begin{alignat}{1}
&p'=3: \quad (\id + (-1)^{p+1} e_1)\, \psi_N = 0 \qquad k=0,\tfrac12,
\\[0.16cm]
&p'=4: \quad 
\left\{\begin{array}{ll}
(\id + [2] e_1 + [2]e_2 + e_1 e_2 + e_2 e_1)\, \psi_N = 0 \quad & k=0,\frac12\,, \\[0.15cm]
(\id + [2] e_2 + e_1 e_2)\, \psi_N = 0 \quad & k=1\,, 
\end{array}\right.
\end{alignat}
\end{subequations}
where we recall that the ADE models are only defined for $p' \ge 3$.

%
\section{Conclusion}\label{sec:conclusion}
%

In this work, we obtained the decomposition of the space of states of the critical ADE lattice models in terms of the irreducible modules over the ordinary or periodic Temperley--Lieb algebra. From this decomposition, we recovered the known expressions for the cylinder and torus partition functions in the scaling limit. To each irreducible module in the decomposition, we constructed an associated lattice operator, and derived a linear difference equation for this operator, which is the lattice analog of the null-state condition for degenerate operators in CFT.
\medskip

The linear difference equations that we derived pertain to contours surrounding the insertion points that are of minimal length. For instance, for the periodic case, we constructed an operator $\mu^{\tinyx{\sigma, \eps}}_{0,s} \in \cL(2s,0)$ that annihilates the insertion state $w_\nu$, so in this case the contour has length $2s$. In general, one can construct an operator ${\cal U} \in \tl_N(\beta)$ that annihilates the states $\lambda \cdot w_\nu$, for all $\lambda \in \cL(N,0)$. Let $\wh P_{N,s,x}$ be the projector on the module $\repW_{s,x}(N)$, which in general is not one-dimensional. This operator $\cal U$ is obtained as the lowest order coefficient of $\wh P_{N,s,x}$ in its Laurent series at $q = \eE^{-\ir \pi p/p'}$, and thus evaluates to zero in $\repQ_{k,x}(N)$. The cases considered in this paper correspond to $N=2s$ in which case $\repW_{s,x}(N)$ is one-dimensional and $\wh P_{2s,s,x}=\wh P_{2s,x}$ is the Jones--Wenzl projector on this module. The cases with $N>2s$ are interesting, since one can expect in the scaling limit that $\cal U$ will scale to the combination of Virasoro modes whose action on the highest weight state produces a singular vector, that is set to zero in the irreducible modules $\repK_{r,s}$. Clearly, more work is needed to make this observation more concrete.\medskip

We moreover believe that our results can serve as a starting point for more thorough studies of the correlation functions in ADE lattice models. A lattice analog of the operator product expansion may be obtained by computing the fusion of the families $\repQ_{k,x}$ of irreducibles modules over $\eptl_N(\beta)$ at roots of unity. This is likely to be a difficult task. Furthermore, in their present form the difference equations for the lattice operators involve the action of Temperley--Lieb generators, and for this reason, it is not easy to analyse their solutions. It is however tempting to speculate that, for well-chosen correlation functions, these difference equations can be expressed in terms of lattice derivatives, like for correlation functions of primary operators in CFT \cite{BPZ84}.
\medskip

Another interesting question would be to understand how our approach applies to other lattice models with Temperley--Lieb symmetry, such as the anyonic chains considered in \cite{Blakeney25}.

\subsection*{Acknowledgments}

The authors thank Alexis Langlois-R\'emilllard for going over the manuscript, as well as Jean-Bernard Zuber for useful discussions.

\bigskip

\appendix

%
\section{Formulas for the eigenvector components $\boldsymbol{S_{a\mu}}$}
\label{app:Samu}
%

In this appendix, we give the formulas for the eigenvector components $S_{a\mu}$ for each model.

\paragraph{$\boldsymbol{A_n}$ models.}

The general formula for the components of the orthonormal eigenvectors of the adjacency matrix for the $A_n$ models reads
\begin{equation}
S_{a\mu} = \sqrt{\frac2{n+1}}\sin \left(\frac{\pi a \mu}{n+1}\right) \,, \qquad a,\mu \in \{1,2, \dots, n\}\,,
\end{equation}
and we recall that $p' = n+1$.

\paragraph{$\boldsymbol{D_n}$ models.}

For the $D_n$ models, the general formula for the components of the orthogonal eigenvectors in the basis that also diagonalises $P_{(n-1,n)}$ is
\be
S_{a\mu} = \sqrt{\frac 2{n-1}}\times
\left\{\begin{array}{cl}
\displaystyle\sin \left(\frac{\pi a m_\mu}{2n-2}\right) & a \in \{1,2, \dots, n-2\}\,, \\[0.4cm]
\displaystyle\frac{(-1)^\mu}2 & a \in \{n-1,n\}\,,\\
\end{array}\right.\qquad \,
\ee
for $\mu \in \{1,2, \dots, n-1\}$, and
\be
S_{an} = \frac 1{\sqrt 2} \times
\left\{\begin{array}{cl}
0 & a \in \{1,2, \dots, n-2\}\,,\\[0.1cm]
-1 & a = n-1\,,\\[0.1cm]
1 & a = n\,,
\end{array}\right.
\ee
where we recall that $p' = 2n-2$. For the special case of $D_4$ with $K=P_{(134)}$, the components $S_{a\mu}$ of the orthonormal eigenvectors of $P_{(134)}$ and of the adjacency matrix are given by the matrix
\be
\mathbf{S} = \begin{pmatrix}
\frac{1}{\sqrt{6}} & \frac{1}{\sqrt{3}} & \frac{1}{\sqrt{6}} & \frac{1}{\sqrt{3}} \\
\frac{1}{\sqrt{2}} & 0 & -\frac{1}{\sqrt{2}} & 0 \\
\frac{1}{\sqrt{6}} & \frac{\omega^{-1}}{\sqrt{3}} & \frac{1}{\sqrt{6}} & \frac{\omega}{\sqrt{3}} \\
\frac{1}{\sqrt{6}} & \frac{\omega }{\sqrt{3}} & \frac{1}{\sqrt{6}} & \frac{\omega^{-1}}{\sqrt{3}} 
\end{pmatrix}, \qquad \omega = \eE^{2 \pi \ir /3}\,.
\ee

\paragraph{$\boldsymbol{E_6, E_7, E_8}$ models.}

The formulas for the components of the orthonormal eigenvectors of the adjacency matrix for $E_6$, $E_7$ and $E_8$ are
\begin{subequations}
\begin{alignat}{2}
E_6&: \quad\mathbf{S}_{\mu} = \frac1{\eta_\mu}\Big(
s(5),
s(2)s(5),
s(2)^2 s(3), 
s(2)^3,
s(2)^2,
s(2)s(3)\Big)\,,
\\[0.2cm]
E_7&: \quad\mathbf{S}_{\mu} = 
\left\{\begin{array}{ll}
\displaystyle
\frac1{\eta_\mu}\Big(
s(6),
s(2)s(6),
s(2)^2 s(4), 
s(2)^2 s(3),
s(2)^3,
s(2)^2,
s(2)s(4)
\Big) & \mu \neq 4\,,
\\[0.4cm]
\displaystyle
\frac{1}{\sqrt 3}(0,0,0,1,0,-1,-1) &
\mu = 4\,,
\end{array}\right.
\\[0.2cm]
E_8&: \quad \mathbf{S}_\mu = \frac1{\eta_\mu}\Big(
s(7),
s(2)s(7),
s(2)^2 s(5), 
s(2)^2 s(4), 
s(2)^2 s(3),
s(2)^3,
s(2)^2,
s(2) s(5)
\Big)\,,
\end{alignat}
\end{subequations}
where 
\be
s(k) = \frac{\sin (\frac{\pi k m_\mu}{p'})}{\sin (\frac{\pi m_\mu}{p'})}\,,
\ee
and each $\eta_\mu$ with $\mu = 1,2, \dots, n$ is a normalising constant. We recall that $p' = 12,18,30$ for $n=6,7,8$.

%
\section{Dimension counting for $\boldsymbol{\repM_{\g,\mu,K}}$}
\label{app:dims}
%

In this section, we show that the dimensions on the left and right sides of \eqref{eq:M.decomp} coincide, for all families $\repM_{\g,\mu,K}$. To show this, we start from \eqref{eq:M.dim}, use \eqref{eq:binomial} and find after simplification
\begin{equation}
\label{eq:M.with.Dj}
\dim \, \repM_{\g,\mu,K}(N) = 
\left\{\begin{array}{cl}
\displaystyle d_0(N)\, \tr\, K + 2 \sum_{j=1}^{p'} D_j(N) \sum_{\nu=1}^n\kappa_{\nu} \cos\left(\frac{2\pi j m_\nu}{p'}\right)
& N \textrm{ even,}
\\
\displaystyle 2\sum_{j=1/2}^{p'-1/2} D_j(N) \sum_{\nu =1}^n\kappa_{\nu} \cos\left(\frac{2\pi j m_\nu}{p'}\right)
& N \textrm{ odd,}
\end{array}\right.
\end{equation}
with $D_j(N)$ defined in \eqref{eq:DkN}. The rest of the proof is then done separately for the different models.

\subsection[$A_n$ models]{$\boldsymbol{A_n}$ models}

\paragraph{1) $\boldsymbol{\repM_{A_n,\mu,\id}}$.}

In this case, $N$ is even and $\tr\, K = n$. Using the identity
\begin{equation}
\sum_{\nu=1}^n \cos\left(\frac{2\pi m_\nu j}{p'}\right) = p'\delta_{j,p'}-1\,, 
\qquad j \in \{1,2,\dots,p'\} \,,
\end{equation}
we find
\begin{alignat}{2}
\dim \, \repM_{A_n,\mu,\id}(N) 
&= (p'-1) \big(d_0(N)+2 D_{p'}(N)\big) - 2\sum_{j=1}^{p'-1} D_j(N)
\nn\\&= \sum_{j=1}^{p'-1} \big(D_0(N) - D_j(N) - D_{p'-j}(N) + D_{p'}\big)
\nn\\&
= \sum_{s=1}^{p'-1} \dim \repQ_{0,q^s}(N)\,,
\end{alignat}
where we used the identity
\be
D_0(N) = d_0(N) + D_{p'}(N)\,.
\ee

\paragraph{2) $\boldsymbol{\repM_{A_n,\mu,R}}$ with $\boldsymbol n$ odd.}

In this case, $N$ is even, $\tr\, K = 1$ and $\kappa_\nu = (-1)^{\nu+1}$. Using the identity
\begin{equation}
\sum_{\nu=1}^n (-1)^{\nu}\cos\left(\frac{2\pi m_\nu j}{p'}\right) = p'\delta_{j,p'/2}-1\,, 
\qquad j \in \{1,2,\dots,p'\} \,,
\end{equation}
we find
\begin{alignat}{2}
\dim \, \repM_{A_n,\mu,R}(N) 
&= d_0(N) -2 (p'-1) D_{p'/2}(N) + 2 D_{p'}(N) + \sum_{j=1}^{(p'-2)/2}2\big(D_j(N)+D_{p'-j}(N)\big)
\nn\\&= \big(D_0(N) -2 D_{p'/2}+D_{p'}(N)\big) + \sum_{j=1}^{(p'-2)/2} 2\big(D_j(N) - 2D_{p'/2}(N) + D_{p'-j}(N)\big)
\nn\\&
= \dim \repQ_{0,p'/2}(N) + \sum_{k=1}^{(p'-2)/2} 2\dim \repQ_{k,q^{p'/2}}(N)\,.
\end{alignat}

\paragraph{3) $\boldsymbol{\repM_{A_n,\mu,R}}$ with $\boldsymbol n$ even.}

In this case, $N$ is odd and $\kappa_\nu = (-1)^{\mu+1}$. Using the identity
\begin{equation}
\sum_{\nu=1}^n (-1)^{\nu}\cos\left(\frac{2\pi m_\nu j}{p'}\right) = p'\delta_{j,p'/2}-1\,, 
\qquad j \in \{\tfrac12,\tfrac32,\dots,p'-\tfrac12\} \,,
\end{equation}
we find
\begin{alignat}{2}
\dim \, \repM_{A_n,\mu,R}(N) 
&= -2 (p'-1) D_{p'/2}(N) + \sum_{j=1/2}^{(p'-1)/2}2\big(D_j(N)+D_{p'-j}(N)\big)
\nn\\&= \sum_{j=1/2}^{(p'-1)/2} 2\big(D_j(N) - 2D_{p'/2}(N) + D_{p'-j}(N)\big)
\nn\\&
= \sum_{k=1/2}^{(p'-1)/2} 2\dim \repQ_{k,q^{p'/2}}(N)\,,
\end{alignat}
where the sums run over half-integers.

\subsection[$D_n$ models]{$\boldsymbol{D_n}$ models}

\paragraph{1) $\boldsymbol{\repM_{D_n,\mu,\id}}$.}

In this case, we have $\tr\, K = n$. Using the identity
\begin{equation}
\sum_{\nu=1}^n \cos\left(\frac{2\pi m_\nu j}{p'}\right) = \frac{p'}2\delta_{j,p'}-\frac{p'}2\delta_{j,p'/2}+(-1)^j\,, 
\qquad j \in \{1,2,\dots,p'\} \,,
\end{equation}
we find
\be
\dim \, \repM_{D_n,\mu,\id}(N) 
= \frac{p'+2}2 \big(D_0(N)+D_{p'}(N)\big) - p' D_{p'/2}(N) +\sum_{j=1}^{p'-1}
(-1)^j\big(D_j(N)+D_{p'-j}(N)\big)\,.
\ee
We now split the discussion between $n$ odd and even, corresponding to $p' \equiv 0 \mod 4$ and $p' \equiv 2 \mod 4$ respectively. For $n$ odd, we have
\begin{alignat}{2}
\dim \, \repM_{D_n,\mu,\id}(N) 
&= \big(D_0(N)-2D_{p'/2}+D_{p'}(N)\big)+\sum_{j=1,3,\dots, p'-1} \big(D_0(N)-D_j(N)-D_{p'-j}(N)+D_{p'}(N)\big) 
\nn\\&
+ \sum_{j=2,4,\dots, p'/2-2}2\big(D_j(N)-2D_{p'/2}(N)+D_{p'-j}(N)\big)
\nn\\& 
= \dim \repQ_{0,q^{p'/2}} +\sum_{j=1,3,\dots, 2n-3} \dim \repQ_{0,q^s} + \sum_{s=1}^{(n-3)/2} 2 \dim \repQ_{2s,q^{p'/2}} \,.
\end{alignat}
For $n$ even, we instead have
\begin{alignat}{2}
\dim \, \repM_{D_n,\mu,\id}(N) 
&= \big(D_0(N)-2D_{p'/2}+D_{p'}(N)\big)+\sum_{j=1,3,\dots, p'-1} \big(D_0(N)-D_j(N)-D_{p'-j}(N)+D_{p'}(N)\big) 
\nn\\&
+ \sum_{j=2,4,\dots, p'/2-1}2\big(D_j(N)-2D_{p'/2}(N)+D_{p'-j}(N)\big)
\nn\\& 
= \dim \repQ_{0,q^{p'/2}} +\sum_{j=1,3,\dots, 2n-3} \dim \repQ_{0,q^s} + \sum_{s=1}^{(n-2)/2} 2 \dim \repQ_{2s,q^{p'/2}} \,.
\end{alignat}

\paragraph{2) $\boldsymbol{\repM_{D_n,\mu,P_{(n-1,n)}}}$.}

In this case, $N$ is even and we have $p'=2n-2$ and $\tr\, K = n-2$. The eigenvalues of $K$ are $\kappa_{\mu}=1$ for $\mu=1,2,\dots,n-1$ and $\kappa_n=-1$. Using the identity
\begin{equation}
\sum_{\nu=1}^n \kappa_\nu \cos\left(\frac{2\pi m_\nu j}{p'}\right) = \frac{p'}2\delta_{j,p'}-\frac{p'}2\delta_{j,p'/2}+(-1)^{j+1}\,, 
\qquad j \in \{1,2,\dots,p'\} \,,
\end{equation}
we find
\be
\dim \, \repM_{D_n,\mu,P_{(n-1,n)}}(N) 
= \frac{p'-2}2 \big(D_0(N)+D_{p'}(N)\big) - p' D_{p'/2}(N) - \sum_{j=1}^{p'-1}
(-1)^j\big(D_j(N)+D_{p'-j}(N)\big)\,.
\ee
The discussion again splits between the two parities of $n$. For $n$ odd, we have
\begin{alignat}{2}
\dim \, \repM_{D_n,\mu,P_{(n-1,n)}}(N) 
&= \sum_{j=2,4,\dots, p'-2} \big(D_0(N)-D_j(N)-D_{p'-j}(N)+D_{p'}(N)\big) 
\nn\\&
+ \sum_{j=1,3,\dots, (p'-2)/2}2\big(D_j(N)-2D_{p'/2}(N)+D_{p'-j}(N)\big)
\nn\\& 
= \sum_{j=2,4,\dots, 2n-4} \dim \repQ_{0,q^s} + \sum_{s=1}^{(n-1)/2} 2\dim \repQ_{2s-1,q^{p'/2}} \,.
\end{alignat}
For $n$ even, we instead have
\begin{alignat}{2}
\dim \, \repM_{D_n,\mu,P_{(n-1,n)}}(N) 
&= \sum_{j=2,4,\dots, p'-2} \big(D_0(N)-D_j(N)-D_{p'-j}(N)+D_{p'}(N)\big) 
\nn\\&
+ \sum_{j=1,3,\dots, (p'-4)/2} 2 \big(D_j(N)-2D_{p'/2}(N)+D_{p'-j}(N)\big)
\nn\\& 
= \sum_{j=2,4,\dots, 2n-4} \dim \repQ_{0,q^s} + \sum_{s=1}^{(n-2)/2} 2\dim \repQ_{2s-1,q^{p'/2}} \,.
\end{alignat}

\paragraph{3) $\boldsymbol{\repM_{D_n,\mu,P_{(134)}}}$.}

In this case, we have $\tr\, K = 1$. The exponents are $\{1,3,5,3\}$, and the eigenvectors of $D_4$ are chosen such that they are also eigenstates of $K$ with the respective eigenvalues $\kappa_\mu=\{1,\eE^{2 \pi \ir /3},1,\eE^{-2 \pi \ir /3}\}$. The sums over $\nu$ in \eqref{eq:M.with.Dj} contain only four terms and can be computed explicitly for $j=1,2, \dots, 6$. This yields
\begin{alignat}{2}
\dim \, \repM_{D_4,\mu,P_{(134)}}(N) 
&= D_0(N) + 4 D_1(N) - 4D_2(N) - 2 D_3(N) - 4 D_4(N) + 4 D_5(N) + D_6(N) 
\nn\\[0.1cm]&= \dim \repQ_{0,q^3} + 4 \dim \repQ_{1,q^2}\,.
\end{alignat}

\subsection[$E_6$, $E_7$ and $E_8$ models]{$\boldsymbol{E_6}$, $\boldsymbol{E_7}$ and $\boldsymbol{E_8}$ models}

For the $E_6$, $E_7$ and $E_8$ models, we have $p'=12,18,30$ respectively, and $N$ is even in all cases.

\paragraph{1) $\boldsymbol{\repM_{E_n,\mu,\id}}$ with $\boldsymbol{n=6,7,8}$.}

In this case, we gave $\tr\, K = n$. The sums over~$m$ in \eqref{eq:M.with.Dj} contain $n$~terms and can be evaluated explicitly. This yields
\begingroup
\allowdisplaybreaks
\begin{subequations}
\begin{alignat}{2}
\dim \, \repM_{D_6,\mu,\id}(N) 
&= 6D_0(N) -2 D_1(N) +2 D_2(N) + 4 D_3(N) - 6 D_4(N) - 2 D_5(N) - 4 D_6(N) \nn\\
&-2 D_7(N) - 6 D_8(N) + 4 D_9(N) + 2 D_{10}(N) - 2 D_{11}(N) + 6 D_{12}(N)
\nn\\[0.1cm]&= \sum_{s=1,4,5,7,8,11}\dim \repQ_{0,q^s}(N) + 2 \dim \repQ_{2,q^6}(N) + 4 \dim \repQ_{2,q^4}(N)\,,\\[0.1cm]
\dim \, \repM_{D_7,\mu,\id}(N) 
&= 7D_0(N) -2 D_1(N) +2 D_2(N) + 4 D_3(N) + 2 D_4(N) - 2 D_5(N) - 4 D_6(N) \nn\\
&-2 D_7(N) +2 D_8(N) - 14 D_9(N) + 2 D_{10}(N) - 2 D_{11}(N) -4 D_{12}(N) \nn\\
&-2 D_{13}(N) +2 D_{14}(N) + 4 D_{15}(N) + 2 D_{16}(N) - 2 D_{17}(N) +7 D_{18}(N)
\nn\\[0.1cm]&= \sum_{s=1,5,7,9,11,13,17}\dim \repQ_{0,q^s}(N) + \sum_{k=2,4,8}2 \dim \repQ_{k,q^9}(N) + 4 \dim \repQ_{3,q^6}(N)\,,\\[0.1cm]
\dim \, \repM_{D_8,\mu,\id}(N) 
&= 8D_0(N) -2 D_1(N) +2 D_2(N) + 4 D_3(N) + 2 D_4(N) +8 D_5(N) - 4 D_6(N) \nn\\
&-2 D_7(N) +2 D_8(N) +4 D_9(N) -8 D_{10}(N) - 2 D_{11}(N) -4 D_{12}(N) \nn\\
&-2 D_{13}(N) +2 D_{14}(N) -16 D_{15}(N) + 2 D_{16}(N) - 2 D_{17}(N) -4 D_{18}(N)\nn\\
& -2D_{19}(N) -8 D_{20}(N) +4 D_{21}(N) +2 D_{22}(N) -2 D_{23}(N) -4 D_{24}(N)\nn\\
& +8D_{25}(N) +2 D_{26}(N) +4 D_{27}(N) + 2 D_{28}(N) -2 D_{29}(N) + 8D_{30}(N)\nn\\
\nn\\[0.1cm]&
= \sum_{s=1,7,11,13,17,19,23,29}\dim \repQ_{0,q^s}(N) + \sum_{k=2,4,8,14}2 \dim \repQ_{k,q^{15}}(N) 
\nn\\&
+ 4 \sum_{k=3,9} \dim \repQ_{k,q^{10}}(N) + 4 \sum_{s=6,12} \dim \repQ_{5,q^{s}}(N)\,.
\end{alignat}
\end{subequations}
\endgroup

\paragraph{2) $\boldsymbol{\repM_{E_6,\mu,P_{(15),(24)}}}$.}

In this case, we have $\tr\, K = 2$. The sums over~$m$ in \eqref{eq:M.with.Dj} are evaluated explicitly and yield
\begin{alignat}{2}
\dim \, \repM_{D_6,\mu,P_{(15)(24)}}(N) 
&= 2D_0(N) +2 D_1(N) +6 D_2(N) - 4 D_3(N) - 2 D_4(N) + 2 D_5(N) - 12 D_6(N)\nn\\&
+2 D_7(N) - 2 D_8(N) - 4 D_9(N) + 6 D_{10}(N) + 2 D_{11}(N) + 2 D_{12}(N)
\nn\\[0.1cm]&= \sum_{s=4,8}\dim \repQ_{0,q^s}(N) + \sum_{k=1,2,5} 2 \dim \repQ_{k,q^6}(N) + 4 \dim \repQ_{2,q^3}(N)\,.
\end{alignat}

%
\section{Insertion states for $\boldsymbol{\repM_{\g,\mu,K}}$ with $\boldsymbol{k>0}$}\label{app:insertion}
%

\subsection[The constants $\Gamma_{s,\ell}$]{The constants $\boldsymbol{\Gamma_{s,\ell}}$}
\label{app:GammasL}

In \cite{LRMD22}, the constant $\Gamma_{s,\ell}$ were found to be given by $\Gamma_{0,\ell} = \frac {x^{-\ell}}N$ for $s=0$, and for $1 \le s \le \lfloor\frac{N-1}2\rfloor$ by
\begin{alignat}{2}
\label{eq:Gammakl}
\Gamma_{s,\ell} &= 
\frac{x^{-\ell}}{N(q-q^{-1})^{2s-1}[s]![s-1]!} 
\sum_{\sigma = \pm 1}\sum_{\kappa = 1}^s\sum_{\tau = 0}^{s-\kappa} \frac{(-1)^{s+\kappa} \sigma q^{\sigma (\ell \kappa+N \tau)}}{x^2 q^{\sigma (N-2\kappa)}-1} 
\frac{\left[\begin{matrix} s-1\\ \kappa-1\end{matrix}\right]}
{\left[\begin{matrix} N-\kappa-1\\ N-s-1\end{matrix}\right]}
\nonumber\\[0.15cm]&\hspace{7cm}\times
\left[\begin{matrix} N - s -\ell - \kappa - \tau - 1\\ s-\kappa-\tau\end{matrix}\right]
\left[\begin{matrix} \ell + \tau - 1\\ \tau\end{matrix}\right]\,,
\end{alignat}
where the $q$-factorial and the $q$-binomial are
\be
[n]! = \prod_{j=1}^n [j]\,, 
\qquad
\left[\begin{matrix}
n \\ m
\end{matrix}\right] = \frac{[n]!}{[n-m]![m]!}\,,
\ee
and we use the convention for $\ell = 0$
\be
\left[\begin{matrix} \ell + \tau - 1\\ \tau\end{matrix}\right] \, \xrightarrow{\ell = 0} \, \delta_{\tau,0}\,.
\ee
These formulas are valid for $\eptl_N(\beta,\gamma)$, $\eptl^{\tinyx 1}_N(\beta,\gamma)$, and $\eptl^{\tinyx 2}_N(\beta,\gamma)$. For $\eptl^{\tinyx 1}_N(\beta,\gamma)$, there is the extra constant
\be
\label{eq:Gamma.n/2}
\Gamma_{N/2,0} = \left\{\begin{array}{cl}
\displaystyle-\frac12 \frac1{(q-q^{-1})^{N-2} [\frac{N-2}2]!^2}\frac{1}{\alpha-(q^{N/2}+q^{-N/2})}
& x = 1,
\\[0.5cm]
\displaystyle\frac12 \frac1{(q-q^{-1})^{n-2} [\frac{N-2}2]!^2}\frac{1}{\alpha+(q^{N/2}+q^{-N/2})}
& x=-1,
\\[0.5cm]
0 & \textrm{otherwise.}
\end{array}
\right.
\ee

For the derivations in \cref{app:insertAn,app:insertDn,app:insertE678}, we are particulary interested in the constants $\Gamma_{s,0}$ for $0 \le s \le \lfloor\frac{N-1}2\rfloor$. By expanding the $q$-binomials in \eqref{eq:Gammakl}, we rewrite $\Gamma_{s,0}$ as
\be
\label{eq:Gammas0}
\Gamma_{s,0} 
= \frac1{N (q-q^{-1})^{2s-1}}
\left[\begin{matrix}
N-s-1 \\ s
\end{matrix}\right]
\sum_{\sigma = \pm 1} \sum_{\kappa = 1}^s \frac{(-1)^{s+\kappa}}{x^2 q^{\sigma(N-2\kappa)}-1}
\frac{[N-s-\kappa-1]!}{[s-\kappa]![\kappa-1]![N-\kappa-1]!}\,.
\ee
We now prove two results about these constants.
\begin{Proposition}
The constants $\Gamma_{s,0}$ are equivalently given by
\be
\label{eq:Gammas0simple}
\Gamma_{s,0} = (-1)^s \frac{x^{2s}}N
\left[\begin{matrix}
N-s-1 \\ s
\end{matrix}\right] 
\prod_{\sigma = \pm 1}\prod_{\kappa=1}^s\frac1{x^2 q^{\sigma (N-2 \kappa)} - 1}\,,
\qquad
0 \le s \le \Big\lfloor\frac{N-1}2\Big\rfloor\,.
\ee
\end{Proposition}
\proof
We first remark that the right sides of \eqref{eq:Gammas0} and \eqref{eq:Gammas0simple} are both rational functions in $x^2$ that vanish in the limit $x\to\infty$. The first expression is in fact a partial fraction decomposition of the second. To prove this, it suffices to show that they have the same residues at $x^2 = q^{-\sigma(N-2\kappa)}$ for $\sigma \in \{-1,+1\}$ and $\kappa \in \{1,2,\dots, s\}$, which is a straightforward verification.
\eproof

\begin{Proposition}
Let $\gamma \in \{+1,-1\}$. The constants $\Gamma_{s,0}$ satisfy the identities
\begin{subequations}
\label{eq:sum.Gammas0}
\begin{alignat}{2}
N \textrm{ even}&:\ \sum_{s=0}^{\frac{N-2}2} \frac{(-1)^s \Gamma_{s,0} }{q^{\frac N2-s}+\gamma\, q^{s-\frac N2}}= \frac{1}{N (q^{\frac N2}+\gamma\, q^{-\frac N 2})} \frac{(1-x^N)(1+\gamma\, x^N)}{(1-x^2)(1+\gamma\, x^2)} \prod_{\sigma = \pm 1} \prod_{\kappa=1}^{\frac{N-2}2} \frac1{x^2 q^{\sigma (N-2 \kappa)} - 1}\,,
\label{eq:sum.Gammas0.E}\\\label{eq:sum.Gammas0.O}
N \textrm{ odd}&:\ \sum_{s=0}^{\frac{N-1}2} \frac{(-1)^s \Gamma_{s,0}}{q^{\frac N2-s}+\gamma\, q^{s-\frac N2}} = \frac{1}{N (q^{\frac N2}+ \gamma\, q^{-\frac N2})} \frac{1+\gamma \,x^{2N}}{1+\gamma\, x^2} \prod_{\sigma = \pm 1} \prod_{\kappa=1}^{\frac {N-1}2} \frac1{x^2 q^{\sigma (N-2 \kappa)} - 1}\,.
\end{alignat}
\end{subequations}
\end{Proposition}
\proof 
Starting from the expression \eqref{eq:Gammas0simple} for $\Gamma_{s,0}$, we first write
\begin{alignat}{2}
\prod_{\sigma = \pm 1} \prod_{\kappa = 1}^s \frac1{x^2 q^{\sigma(N-2\kappa)}-1} &= 
\mu \prod_{\sigma = \pm 1} \prod_{\kappa = 1}^{\lfloor\frac{N-1}2\rfloor} \frac1{x^2 q^{\sigma(N-2\kappa)}-1} \prod_{\kappa = -(\frac N2-s-1)}^{\frac N2-s-1} (1-x^2 q^{2 \kappa})
\nn\\
& = \mu \prod_{\sigma = \pm 1} \prod_{\kappa = 1}^{\lfloor\frac{N-1}2\rfloor} \frac1{x^2 q^{\sigma(N-2\kappa)}-1} \sum_{j=0}^{\frac N2-s-1}(-1)^j x^{2j} 
\left[\begin{matrix}N-2s-1\\j\end{matrix}\right]\,,
\end{alignat}
where 
\be
\mu = \left\{\begin{array}{cl}
\displaystyle\frac1{1-x^2} & N \textrm{ even,}\\[0.3cm]
\displaystyle 1 & N \textrm{ odd,}
\end{array}\right.
\ee
and we used the identity
\be
\sum_{j=0}^n z^j \left[\begin{matrix}n\\j\end{matrix}\right] = \prod_{j=-\frac{n-1}2}^{\frac{n-1}2}(1+z q^2j)\,.
\ee
This yields
\begin{alignat}{2}
\sum_{s=0}^{\lfloor\frac{N-1}2\rfloor} &\frac{(-1)^s \Gamma_{s,0} }{q^{\frac N2-s}+\gamma \, q^{s-\frac N2}}=
\frac \mu N 
\prod_{\sigma = \pm 1} \prod_{\kappa = 1}^{\lfloor\frac{N-1}2\rfloor} \frac1{x^2 q^{\sigma(N-2\kappa)}-1} \sum_{s=0}^{\lfloor\frac{N-1}2\rfloor} \sum_{j=0}^{\frac N2-s-1} \frac{(-1)^j x^{2(s+j)}}{q^{\frac N2-s}+\gamma\, q^{s-\frac N2}} \left[\begin{matrix}N-s-1\\s\end{matrix}\right] \left[\begin{matrix}N-2s-1\\j\end{matrix}\right]
\nn\\[0.15cm]
& \hspace{-0.5cm} = \frac \mu N 
\prod_{\sigma = \pm 1} \prod_{\kappa = 1}^{\lfloor\frac{N-1}2\rfloor} \frac1{x^2 q^{\sigma(N-2\kappa)}-1} \sum_{t=0}^{N-1} \sum_{s=0}^{\min(t,N-t-1)} \frac{(-1)^{s+t}x^{2t}}{q^{\frac N2-s}+\gamma\,q^{s-\frac N2}} \left[\begin{matrix}N-s-1\\s\end{matrix}\right] \left[\begin{matrix}N-2s-1\\t-s\end{matrix}\right],
\label{eq:sum.of.Gammas}
\end{alignat}
where we changed the summation index $j$ to $t = j+s$ and exchanged the order of the sums.\medskip

Next, we introduce the function 
\be
g(z) = \frac1{z+\gamma} \prod_{\kappa = 0}^{\min(t,N-t-1)} \frac 1{z\, q^{-\frac N2+\kappa}-q^{\frac N2-\kappa}} \prod_{\kappa = 1}^{\min(t,N-t-1)} (z\, q^{\frac N2-\kappa}-q^{-\frac N2+\kappa})\,. 
\ee
It has the residues
\be
\textrm{Res}\big(g(z),z \to q^{N-2s}\big) = \frac{(-1)^s}{q^{\frac N2-s}+\gamma\, q^{s-\frac N2}}
\left[\begin{matrix}N-s-1\\s\end{matrix}\right] \left[\begin{matrix}N-2s-1\\t-s\end{matrix}\right]
\ee
for $s = 0,1, \dots, \min(t,N-t-1)$. We can therefore write 
\be
\sum_{s=0}^{\min(t,N-t-1)} \frac{(-1)^{s}}{q^{\frac N2-s}+\gamma\,q^{s-\frac N2}} 
\left[\begin{matrix}N-s-1\\s\end{matrix}\right] \left[\begin{matrix}N-2s-1\\t-s\end{matrix}\right]
= \oint_{\mathcal C} \frac{\dd z}{2 \pi \ir}\, g(z)\,,
\ee
where $\mathcal C$ is a contour encircling the poles at $z=q^{N-2s}$ with $s = 0,1, \dots, \min(t,N-t-1)$ in the counter-clockwise direction, but not the pole at $z=-\gamma$. One can check that the function $\frac1{y^{2}} g(\frac 1y)$ is regular at $y=0$, implying that $g(z)$ has no pole at infinity. We can therefore deform the contour $\mathcal C$ so that it encircles only the pole at $z=-\gamma$ in the clockwise direction. This yields
\be
\sum_{s=0}^{\min(t,N-t-1)} \frac{(-1)^{s}}{q^{\frac N2-s}+q^{s-\frac N2}} 
\left[\begin{matrix}N-s-1\\s\end{matrix}\right] \left[\begin{matrix}N-2s-1\\t-s\end{matrix}\right]
= -\textrm{Res}\big(g(z), z \to -\gamma\big) = \frac {\gamma^{\min(t,N-t-1)}}{q^{\frac N2}+\gamma \, q^{-\frac N2}}\,.
\ee
Inserting this result in \eqref{eq:sum.of.Gammas}, we use 
\be
\sum_{t=0}^{N-1} \gamma^{\min(t,N-t-1)}(-x^2)^t = 
\left\{\begin{array}{cl}
\displaystyle\frac{(1- x^N)(1+\gamma\, x^N)}{1+\gamma\, x^2} & N \textrm{ even,} \\[0.35cm]
\displaystyle\frac{1 + \gamma\, x^{2N}}{1+ \gamma\, x^2} & N \textrm{ odd,} 
\end{array}\right.
\ee 
and obtain the desired expressions \eqref{eq:sum.Gammas0}, ending the proof.
\eproof

In \cref{app:insertAn,app:insertDn,app:insertE678}, we use these results to show that certain overlaps involving the insertion states $w_{k,x}$ are non-zero, for the models $A_n$, $D_n$ and $E_n$, respectively, thus showing that these insertion states are also non-zero.

\subsection[$A_n$ models]{$\boldsymbol{A_n}$ models}\label{app:insertAn}

The decomposition of $\repM_{A_n,\mu,\id}$ only involves factors $\repQ_{0,q^s}$, so there are no insertion state with $k>0$ to construct in this case. For $\repM_{A_n,\mu,R}$, we construct the insertion state $w_{k,x}=\wh P_{2k,x} \cdot w_k$ associated to $\repQ_{k,x}(2k)$, with $x = \eps q^{p'/2}$, and
\be
w_k = \frac1{\sqrt 2}\big(u_1 + \rho\, u_2)\,,
\quad
\left\{\begin{array}{ll}
u_1 = | \frac{n-2k+1}2, \frac{n-2k+3}2, \dots, \frac{n+2k+1}2 \rangle\,,\\[0.15cm]
u_2 = | \frac{n+2k+1}2,\frac{n+2k-1}2, \dots, \frac{n-2k+1}2\rangle\,,
\end{array}\right.
\ee
with
\be
\rho = 
\left\{\begin{array}{ll}
(-1)^{k}& n \textrm{ odd},\, \mu \textrm{ odd},\\[0.1cm]
\eps (-1)^{p/2}& n \textrm{ even},\, \mu \textrm{ odd},\\[0.1cm]
\eps (-1)^{(p+1)k}& n \textrm{ even},\, \mu \textrm{ even},\\[0.1cm]
\end{array}\right.
\ee
and with $k=1,2,\dots, \frac{n-1}2$ for $n$ odd and $k=\frac12, \frac32, \dots, \frac{n-1}2$ for $n$ even. This choice of $\rho$ ensures that $w_k$ is an eigenvector of $K_{2k}$ with eigenvalue $x^{2k}$. We note moreover that $x \in \{+1,-1\}$ for $n$ even and $\mu$~odd, whereas $x \in \{+\ir,-\ir\}$ in the two other cases. Using \eqref{eq:c0.RSOS} and \eqref{eq:bilinear.form}, we find 
\begin{subequations}
\begin{alignat}{2}
\aver{c_0^s \cdot u_1, \Omega^\ell\, c_0^s \cdot u_1} &= \delta_{\ell,0}\,(\kappa_\mu)^s
\prod_{a = (n-2k+3)/2}^{(n-2k+2s+1)/2} \frac1{S_{a\mu}} 
\prod_{a = (n-2k+3)/2}^{(n+2k-2s+1)/2} \frac1{S_{a\mu}}\,,
\\[0.1cm]
\aver{c_0^s \cdot u_2, \Omega^\ell\, c_0^s \cdot u_2} &= \delta_{\ell,0}\,(\kappa_\mu)^s 
\prod_{a = (n-2k+2s+1)/2}^{(n+2k-1)/2} \frac1{S_{a\mu}}
\prod_{a = (n+2k-2s+1)/2}^{(n+2k-1)/2} \frac1{S_{a\mu}}\,,
\\[0.1cm]
\aver{c_0^s \cdot u_1, \Omega^\ell\, c_0^s \cdot u_2} &= \aver{c_0^s \cdot u_2, \Omega^\ell\, c_0^s \cdot u_1} = 0\,.
\end{alignat}
\end{subequations}
Because $\Omega^{2k-2s} c_0^s \cdot u_1 = \kappa_\mu^{k-s} c_0^s \cdot u_2$, we have
\be
\aver{c_0^s u_2, c_0^s u_2} = \kappa_\mu \aver{\Omega^{2k-2s} c_0^s u_1,\Omega^{2k-2s} c_0^s u_1} = \kappa_\mu \aver{c_0^s u_1, c_0^s u_1}\,,
\ee
so we need only compute $\aver{c_0^s \cdot u_1, c_0^s \cdot u_1}$. Using the identities
\begin{subequations}
\begin{alignat}{2}
\label{eq:Samu}
S_{a\mu} &= (\tfrac2{p'})^{1/2} \sin(\tfrac {\pi p}{p'}) (-1)^{a+1} [a] \,, 
\\[0.15cm]
[\tfrac{p'}2+ a] &=
\left\{\begin{array}{cl}
\displaystyle-\ir^{p}\, \frac{q^a + q^{-a}}{q-q^{-1}} & n \textrm{ odd with } \mu \textrm{ odd}, \textrm{ or } n \textrm{ even with } \mu \textrm{ even},
\\[0.35cm]
\displaystyle (-1)^{p/2}\, \frac{q^a - q^{-a}}{q-q^{-1}} & n \textrm{ even with } \mu \textrm{ odd},
\end{array}\right.
\\[0.15cm]
[\tfrac{p'}2+a]![\tfrac{p'}2-a]!&=
\left\{\begin{array}{cl}
[\tfrac{p'}2]![\tfrac{p'}2-1]![\tfrac{p'}2+a] & n \textrm{ odd},\, \mu \textrm{ odd},\, a \in \mathbb Z\,,\\[0.15cm]
(-1)^{a-1/2}[\tfrac{p'-1}2]!^2[\tfrac{p'}2+a] & n \textrm{ even},\, \mu \textrm{ odd},\,a \in \mathbb Z+\frac12\,,
\\[0.15cm]
[\tfrac{p'-1}2]!^2[\tfrac{p'}2+a] & n \textrm{ even},\, \mu \textrm{ even},\,a \in \mathbb Z+\frac12\,,
\end{array}\right.
\end{alignat}
\end{subequations}
we obtain after simplications
\be
\big\langle c_0^s \cdot u_1, c_0^s \cdot u_1 \big\rangle = \frac{2\, (\tfrac{p'}2)^{k}[\frac{p'}2-k]!^2}{\sin(\frac{\pi p}{p'})^{2k-1}}
\left\{\begin{array}{cl}
\displaystyle \frac{(-1)^{(2k+p-1)/2}}{[\frac{p'}2]![\frac{p'}2-1]!}\frac{(-1)^s}{q^{k-s}+q^{s-k}}
& n \textrm{ odd},\, \mu \textrm{ odd}\,, 
\\[0.55cm]
\displaystyle \frac{\ir \, (-1)^{(2k+p+n+1)/2}}{[\frac{p'-1}2]!^2}\frac{(-1)^s}{q^{k-s}-q^{s-k}}
& n \textrm{ even},\, \mu \textrm{ odd}\,, 
\\[0.55cm]
\displaystyle \frac{(-1)^{(p+n-1)/2}}{[\frac{p'-1}2]!^2}\frac{(-1)^s}{q^{k-s}+q^{s-k}}
& n \textrm{ even},\, \mu \textrm{ even}\,. 
\end{array}
\right.
\ee
For $n$ odd and $\mu$ odd, we use \eqref{eq:sum.Gammas0.E} with $\gamma = 1$ and find
\begin{alignat}{2}
\aver{w_k,w_{k,x}} &= \frac{2\, (-1)^{(2k+p-1)/2} (\tfrac{p'}2)^{k}}{\sin(\frac{\pi p}{p'})^{2k-1}} \frac{[\frac{p'}2-k]!^2}{[\frac{p'}2]![\frac{p'}2-1]!} \sum_{s=0}^{\lfloor (2k-1)/2\rfloor} \frac{(-1)^s}{q^{k-s}+q^{s-k}}\,\Gamma_{s,0} \\[0.15cm]\nn
& = \frac{2\, (-1)^{(2k+p-1)/2} (\tfrac{p'}2)^{k}}{\sin(\frac{\pi p}{p'})^{2k-1}} \frac{[\frac{p'}2-k]!^2}{[\frac{p'}2]![\frac{p'}2-1]!} \frac 1{2k} \frac 1{q^{k}+q^{-k}} \frac{1-x^{4k}}{1-x^4} \prod_{\sigma = \pm 1} \prod_{\kappa = 1}^{(2k-2)/2} \frac1{x^2 q^{2\sigma(k-\kappa)}-1}\,,
\end{alignat}
which is nonzero for $x \to \pm \ir$. Similarly, for $n$ even and $\mu$ odd, we use \eqref{eq:sum.Gammas0.O} with $\gamma = -1$ and find
\begin{alignat}{2}
\aver{w_k,w_{k,x}} &= \frac{2 \ir \, (-1)^{(2k+p+n+1)/2} (\tfrac{p'}2)^{k}}{\sin(\frac{\pi p}{p'})^{2k-1}} \frac{[\frac{p'}2-k]!^2}{[\frac{p'-1}2]!^2} \sum_{s=0}^{\lfloor (2k-1)/2\rfloor} \frac{(-1)^s}{q^{k-s}-q^{s-k}}\,\Gamma_{s,0} 
\\[0.15cm]\nn
& = \frac{2 \ir \, (-1)^{(2k+p+n+1)/2} (\tfrac{p'}2)^{k}}{\sin(\frac{\pi p}{p'})^{2k-1}} \frac{[\frac{p'}2-k]!^2}{[\frac{p'-1}2]!^2} \frac 1{2k} \frac 1{q^{k}-q^{-k}} \frac{1-x^{4k}}{1-x^2} \prod_{\sigma = \pm 1} \prod_{\kappa = 1}^{(2k-1)/2} \frac1{x^2 q^{2\sigma(k-\kappa)}-1}\,,
\end{alignat}
which is nonzero for $x \to \pm 1$. Lastly, for $n$ even and $\mu$ even, we define $\bar w_k =\frac 1{\sqrt 2}(u_1 - \rho\, u_2)$, use \eqref{eq:sum.Gammas0.O} with $\gamma = 1$ and find
\begin{alignat}{2}
\aver{\bar w_k,w_{k,x}} &= \frac{2 \, (-1)^{(p+n-1)/2} (\tfrac{p'}2)^{k}}{\sin(\frac{\pi p}{p'})^{2k-1}} \frac{[\frac{p'}2-k]!^2}{[\frac{p'-1}2]!^2} \sum_{s=0}^{\lfloor (2k-1)/2\rfloor} \frac{(-1)^s}{q^{k-s}-q^{s-k}}\,\Gamma_{s,0} 
\\[0.15cm]\nn
& = \frac{2 \, (-1)^{(p+n-1)/2} (\tfrac{p'}2)^{k}}{\sin(\frac{\pi p}{p'})^{2k-1}} \frac{[\frac{p'}2-k]!^2}{[\frac{p'-1}2]!^2} \frac 1{2k} \frac 1{q^{k}+q^{-k}} \frac{1+x^{4k}}{1+x^2} \prod_{\sigma = \pm 1} \prod_{\kappa = 1}^{(2k-1)/2} \frac1{x^2 q^{2\sigma(k-\kappa)}-1}\,,
\end{alignat}
which is nonzero for $x \to \pm \ir$. 

\subsection[$D_n$ models]{$\boldsymbol{D_n}$ models}\label{app:insertDn}

For the $D_n$ models, we consider simultaneously the cases $K=\id$, for which we set $N = 2k= 4t$ with $t=1, 2, \dots, \lfloor\frac{n-2}2\rfloor$, and $K=P_{(n-1,n)}$ with $N = 2k = 4t-2$ with $t=1, 2, \dots, \lfloor\frac{n-1}2\rfloor$. We construct the corresponding insertion states~$w_{k,x}$ for $\repQ_{k,\eps q^{p'/2}}(2k)$ as in \eqref{u.def}, with
\be
w_k = \frac1{\sqrt 2}(u_1 - u_2)\,,
\qquad
\left\{\begin{array}{ll}
u_1 = | n-k-1,n-k, \dots, n-2, n-1, n-2, \dots, n-k-1 \rangle\,,\\[0.15cm]
u_2 = | n-k-1,n-k, \dots, n-2, n, n-2, \dots, n-k-1 \rangle\,.
\end{array}\right.
\ee
It is easy to see that $e_j \cdot u = 0$ for $j = 1, 2, \dots, 2k-1$. With a similar calculation as for the $A_n$ models, we find 
\begin{subequations}
\begin{alignat}{2}
\aver{c_0^s \cdot u_1, \Omega^\ell\, c_0^s \cdot u_1} &= \aver{c_0^s \cdot u_1, \Omega^\ell\, c_0^s \cdot u_1}
= \delta_{\ell,0}\, \frac{(p')^{1/2}}{S_{n-1,\mu}}\frac{(-1)^{k+\mu+s}}{q^{k-s}+q^{s-k}}
 \prod_{a=n-k}^{n-2}\frac1{S^2_{a\mu}}\,, \\[0.1cm]
\aver{c_0^s \cdot u_1, \Omega^\ell\, c_0^s \cdot u_2} &= \aver{c_0^s \cdot u_2, \Omega^\ell\, c_0^s\cdot u_1} = 0\,.
\end{alignat}
\end{subequations}
and therefore
\begin{alignat}{2}
\aver{w_k,w_{k,x}} &= (-1)^{k+\mu}\frac{(p')^{1/2}}{S_{n-1,\mu}}
 \prod_{a=n-k}^{n-2}\frac1{S^2_{a\mu}}\sum_{s=0}^{k-1} \frac{(-1)^s}{q^{k-s}+q^{s-k}}\Gamma_{s,0}
\\\nn&=
(-1)^{k+\mu}\frac{(p')^{1/2}}{S_{n-1,\mu}}
\frac1{2k} \frac1{q^{k}+q^{-k}}
\frac{1-x^{4k}}{1-x^4}
\prod_{a=n-k}^{n-2}\frac1{S^2_{a\mu}} 
 \prod_{\sigma = \pm 1} \prod_{\kappa=1}^{k-1} \frac1{x^2 q^{2\sigma (k-\kappa)} - 1}\,,
\end{alignat}
which is nonzero for $x \to \pm\ir$.\medskip

Lastly, for $\repM_{D_4,\mu,P_{(134)}}$, we construct the states $w_1$ and $w_{1,\eps q^{2 \sigma}}$ of $\repQ_{1,\eps q^{2 \sigma}}(2)$ as
\begin{subequations}
\begin{alignat}{2}
w_1 &= |2,1,2 \rangle + \omega^{-\sigma}|2,3,2 \rangle + \omega^{\sigma}|2,4,2 \rangle
\\[0.15cm]
w_{1,\eps q^{2 \sigma}} &= \wh{P}_{2,\eps q^{2\sigma}} \cdot w_1=\tfrac12 \big(|2,1,2 \rangle + \omega^{-\sigma} |2,3,2 \rangle + \omega^{\sigma} |2,4,2 \rangle - \eps\, \omega^{\sigma} |1,2,3 \rangle - \eps |3,2,4 \rangle - \eps\, \omega^{-\sigma} |4,2,1 \rangle\big)\,.
\end{alignat}
\end{subequations}
The state $w_{1,\eps q^{2 \sigma}}$ is nonzero, and it is easy to check that $w_{1,\eps q^{2 \sigma}}$ satisfies \eqref{eq:def.insertion.kx,k>0} with $x = \eps q^{2 \sigma}$.

\subsection[$E_6$, $E_7$ and $E_8$ models]{$\boldsymbol{E_6}$, $\boldsymbol{E_7}$ and $\boldsymbol{E_8}$ models}\label{app:insertE678}

Below, we construct the states $w_k$ and $w_{k,x}$ for the models $\repM_{E_6,\mu,\id}$, $\repM_{E_6,\mu,P_{(15),(24)}}$, $\repM_{E_7,\mu,\id}$, $\repM_{E_8,\mu,\id}$. The resulting overlaps $\aver{w_k,w_{k,x}}$ are sums of finitely many constants $\Gamma_{s,\ell}$, which can be computed to show that $\aver{w_k,w_{k,x}} \neq 0$ in each case.

\paragraph{$\boldsymbol{\repM_{E_6,\mu,\id}}$.}

We construct the insertion states for $\repQ_{2,\eps q^{6}}(4)$ and $\repQ_{3,\eps q^{4 \sigma}}(6)$ as in \eqref{u.def}, with
\begin{subequations}
\begin{alignat}{2}
w_{2} &= \frac1{\sqrt 2}(u_1 - u_2)\,,
\qquad
&&\left\{\begin{array}{ll}
u_1 = | 6,3,2,3,6 \rangle\,,\\[0.15cm]
u_2 = | 6,3,4,3,6 \rangle\,,
\end{array}\right.
\\[0.2cm]
w_{3} &= \frac12(u_1 - u_2 + u_3 - u_4)\,,
\qquad
&&\left\{\begin{array}{ll}
u_1 = | 1,2,3,4,3,2,1 \rangle\,,\\[0.15cm]
u_2 = | 1,2,3,6,3,2,1 \rangle\,,\\[0.15cm]
u_3 = | 5,4,3,2,3,4,5 \rangle\,,\\[0.15cm]
u_4 = | 5,4,3,6,3,4,5 \rangle\,,
\end{array}\right.
\end{alignat}
\end{subequations}
respectively.

\paragraph{$\boldsymbol{\repM_{E_6,\mu,P_{(15),(24)}}}$.}

We construct the insertion states for $\repQ_{1,\eps q^{6}}(2)$, $\repQ_{2,\eps q^{6}}(4)$, $\repQ_{2,\eps q^{3 \sigma}}(4)$ and $\repQ_{5,\eps q^{4 \sigma}}(10)$ as in \eqref{u.def}, with
\begingroup
\allowdisplaybreaks
\begin{subequations}
\begin{alignat}{2}
w_1 &= \frac1{\sqrt 2}(u_1 - u_2)\,,
\qquad
&&\left\{\begin{array}{ll}
u_1 = | 3,2,3 \rangle\,,\\[0.15cm]
u_2 = | 3,4,3 \rangle\,,
\end{array}\right.
\\[0.2cm]
w_2 &= \frac1{\sqrt 2}(u_1 + u_2)\,,
\qquad
&&\left\{\begin{array}{ll}
u_1 = | 1,2,3,4,5 \rangle\,,\\[0.15cm]
u_2 = | 5,4,3,2,1 \rangle\,,
\end{array}\right.
\\[0.2cm]
w'_2 &= \frac1{\sqrt 2}(u_1 - u_2)\,,
\qquad
&&\left\{\begin{array}{ll}
u_1 = | 6,3,2,3,6 \rangle\,,\\[0.15cm]
u_2 = | 6,3,4,3,6 \rangle\,,
\end{array}\right.
\\[0.2cm]
w_5 &= \Xi_1 \cdot \frac1{\sqrt{12}}\sum_{i=1}^{12} \alpha_i u_i\,,
\qquad
&&\left\{\begin{array}{ll}
u_1 = | 1, 2, 3, 4, 3, 2, 1, 2, 3, 4, 5 \rangle & \alpha_1 = 1\,,\\[0.15cm]
u_2 = | 1, 2, 3, 4, 3, 2, 3, 2, 3, 4, 5 \rangle & \alpha_2 = -1\,,\\[0.15cm]
u_3 = | 1, 2, 3, 4, 3, 2, 3, 6, 3, 4, 5 \rangle & \alpha_3 = 1\,,\\[0.15cm]
u_4 = | 1, 2, 3, 4, 3, 4, 3, 2, 3, 4, 5 \rangle & \alpha_4 = 1\,,\\[0.15cm]
u_5 = | 1, 2, 3, 4, 3, 4, 3, 6, 3, 4, 5 \rangle & \alpha_5 = -1\,,\\[0.15cm]
u_6 = | 1, 2, 3, 4, 5, 4, 3, 2, 3, 4, 5 \rangle & \alpha_6 = -1\,,\\[0.15cm]
u_7 = | 1, 2, 3, 4, 5, 4, 3, 6, 3, 4, 5 \rangle & \alpha_7 = 1\,,\\[0.15cm]
u_8 = | 1, 2, 3, 6, 3, 2, 1, 2, 3, 4, 5 \rangle & \alpha_8 = -1\,,\\[0.15cm]
u_9 = | 1, 2, 3, 6, 3, 2, 3, 2, 3, 4, 5 \rangle & \alpha_9 = 1\,,\\[0.15cm]
u_{10} = | 1, 2, 3, 6, 3, 2, 3, 6, 3, 4, 5 \rangle & \alpha_{10} = -1\,,\\[0.15cm]
u_{11} = | 1, 2, 3, 6, 3, 4, 3, 2, 3, 4, 5 \rangle & \alpha_{11} = -1\,,\\[0.15cm]
u_{12} = | 1, 2, 3, 6, 3, 4, 3, 6, 3, 4, 5 \rangle & \alpha_{12} = 1\,,\\[0.15cm]
\end{array}\right.
\end{alignat}
\end{subequations}
\endgroup
respectively.

\paragraph{$\boldsymbol{\repM_{E_7,\mu,\id}}$.}

We construct the insertion states for $\repQ_{2,\eps q^{9}}(4)$, $\repQ_{3,\eps q^{6 \sigma}}(6)$, $\repQ_{4,\eps q^{9}}(8)$ and $\repQ_{8,\eps q^{9}}(16)$ as in \eqref{u.def}, with
\begingroup
\allowdisplaybreaks
\begin{subequations}
\begin{alignat}{2}
w_2 &= \frac1{\sqrt 2}(u_1 - u_2)\,,
\qquad
&&\left\{\begin{array}{ll}
u_1 = | 4, 3, 2, 3, 4 \rangle\,,\\[0.15cm]
u_2 = | 4, 3, 7, 3, 4 \rangle\,,
\end{array}\right.
\\[0.2cm]
w_3 &= \frac1{\sqrt 2}(u_1 - u_2)\,,
\qquad
&&\left\{\begin{array}{ll}
u_1 = | 5, 4, 3, 2, 3, 4, 5 \rangle\,,\\[0.15cm]
u_2 = | 5, 4, 3, 7, 3, 4, 5 \rangle\,,
\end{array}\right.
\\[0.2cm]
w_4 &= \frac1{\sqrt 2}(u_1 - u_2)\,,
\qquad
&&\left\{\begin{array}{ll}
u_1 = | 6, 5, 4, 3, 2, 3, 4, 5, 6 \rangle\,,\\[0.15cm]
u_2 = | 6, 5, 4, 3, 7, 3, 4, 5, 6 \rangle\,,
\end{array}\right.
\\[0.2cm]
w_8 &= \frac1{\sqrt {78}}\sum_{i=1}^{78} \alpha_i u_i\,,
\qquad
&& 
\left\{
\begin{array}{ll}
u_{1} = |6,5,4,3,2,1,2,3,4,3,2,1,2,3,4,5,6\rangle & \alpha_{1} = 1, \\[0.15cm]
u_{2} = |6,5,4,3,2,1,2,3,4,3,2,3,2,3,4,5,6\rangle & \alpha_{2} = -1, \\[0.15cm]
u_{3} = |6,5,4,3,2,1,2,3,4,3,2,3,7,3,4,5,6\rangle & \alpha_{3} = 1, \\[0.15cm]
u_{4} = |6,5,4,3,2,1,2,3,4,3,4,3,2,3,4,5,6\rangle & \alpha_{4} = 1, \\[0.15cm]
u_{5} = |6,5,4,3,2,1,2,3,4,3,4,3,7,3,4,5,6\rangle & \alpha_{5} = -1, \\[0.15cm]
u_{6} = |6,5,4,3,2,1,2,3,4,5,4,3,2,3,4,5,6\rangle & \alpha_{6} = -1, \\[0.15cm]
u_{7} = |6,5,4,3,2,1,2,3,4,5,4,3,7,3,4,5,6\rangle & \alpha_{7} = 1, \\[0.15cm]
u_{8} = |6,5,4,3,2,1,2,3,7,3,2,1,2,3,4,5,6\rangle & \alpha_{8} = -1, \\[0.15cm]
u_{9} = |6,5,4,3,2,1,2,3,7,3,2,3,2,3,4,5,6\rangle & \alpha_{9} = 1, \\[0.15cm]
u_{10} = |6,5,4,3,2,1,2,3,7,3,2,3,7,3,4,5,6\rangle & \alpha_{10} = -1, \\[0.15cm]
\textrm{\hspace{1.7cm} \dots and $68$ more states}
\end{array}
\right.
\end{alignat}
\end{subequations}
\endgroup
respectively.

\paragraph{$\boldsymbol{\repM_{E_8,\mu,\id}}$.}

We construct the insertion states for $\repQ_{2,\eps q^{15}}(4)$, $\repQ_{3,\eps q^{10 \sigma}}(6)$, $\repQ_{4,\eps q^{15}}(8)$, $\repQ_{5,\eps q^{\sigma s}}(10)$ with $s \in \{6,12\}$, $\repQ_{8,\eps q^{15}}(16)$, $\repQ_{9,\eps q^{10 \sigma}}(18)$ and $\repQ_{14,\eps q^{15}}(28)$ as in \eqref{u.def}, with
\begingroup
\allowdisplaybreaks
\begin{subequations}
\begin{alignat}{2}
w_2 &= \frac1{\sqrt 2}(u_1 - u_2)\,,
\ \
&&\left\{\begin{array}{ll}
u_1 = | 4,3,2,3,4 \rangle\,,\\[0.15cm]
u_2 = | 4,3,8,3,4 \rangle\,,
\end{array}\right.
\\[0.2cm]
w_3 &= \frac1{\sqrt 2}(u_1 - u_2)\,,
\ \
&&\left\{\begin{array}{ll}
u_1 = | 5, 4, 3, 2, 3, 4, 5 \rangle\,,\\[0.15cm]
u_2 = | 5, 4, 3, 8, 3, 4, 5 \rangle\,,
\end{array}\right.
\\[0.2cm]
w_4 &= \frac1{\sqrt 2}(u_1 - u_2)\,,
\ \
&&\left\{\begin{array}{ll}
u_1 = | 6, 5, 4, 3, 2, 3, 4, 5, 6 \rangle\,,\\[0.15cm]
u_2 = | 6, 5, 4, 3, 8, 3, 4, 5, 6 \rangle\,,
\end{array}\right.
\\[0.2cm]
w_5 &= \frac1{\sqrt 2}(u_1 - u_2)\,,
\qquad
&&\left\{\begin{array}{ll}
u_1 = | 7, 6, 5, 4, 3, 2, 3, 4, 5, 6, 7 \rangle\,,\\[0.15cm]
u_2 = | 7, 6, 5, 4, 3, 8, 3, 4, 5, 6, 7 \rangle\,,
\end{array}\right.
\\[0.2cm]
w_8 &= \frac1{\sqrt {78}}\sum_{i=1}^{78} \alpha_i u_i\,,
\ \
&&
\left\{
\begin{array}{ll}
u_{1} = |6,5,4,3,2,1,2,3,4,3,2,1,2,3,4,5,6\rangle & \alpha_{1} = 1, \\[0.15cm]
u_{2} = |6,5,4,3,2,1,2,3,4,3,2,3,2,3,4,5,6\rangle & \alpha_{2} = -1, \\[0.15cm]
u_{3} = |6,5,4,3,2,1,2,3,4,3,2,3,8,3,4,5,6\rangle & \alpha_{3} = 1, \\[0.15cm]
u_{4} = |6,5,4,3,2,1,2,3,4,3,4,3,2,3,4,5,6\rangle & \alpha_{4} = 1, \\[0.15cm]
u_{5} = |6,5,4,3,2,1,2,3,4,3,4,3,8,3,4,5,6\rangle & \alpha_{5} = -1, \\[0.15cm]
u_{6} = |6,5,4,3,2,1,2,3,4,5,4,3,2,3,4,5,6\rangle & \alpha_{6} = -1, \\[0.15cm]
u_{7} = |6,5,4,3,2,1,2,3,4,5,4,3,8,3,4,5,6\rangle & \alpha_{7} = 1, \\[0.15cm]
u_{8} = |6,5,4,3,2,1,2,3,8,3,2,1,2,3,4,5,6\rangle & \alpha_{8} = -1, \\[0.15cm]
u_{9} = |6,5,4,3,2,1,2,3,8,3,2,3,2,3,4,5,6\rangle & \alpha_{9} = 1, \\[0.15cm]
u_{10} = |6,5,4,3,2,1,2,3,8,3,2,3,8,3,4,5,6\rangle & \alpha_{10} = -1, \\[0.15cm]
\textrm{\hspace{1.7cm} \dots and $68$ more states}
\end{array}
\right.\\[0.2cm]
w_9 &= \frac1{\sqrt {78}}\sum_{i=1}^{78} \alpha_i u_i\,,
\ \
&&
\left\{
\begin{array}{ll}
u_{1} = |7,6,5,4,3,2,1,2,3,4,3,2,1,2,3,4,5,6,7\rangle & \alpha_{1} = 1, \\[0.15cm]
u_{2} = |7,6,5,4,3,2,1,2,3,4,3,2,3,2,3,4,5,6,7\rangle & \alpha_{2} = -1, \\[0.15cm]
u_{3} = |7,6,5,4,3,2,1,2,3,4,3,2,3,8,3,4,5,6,7\rangle & \alpha_{3} = 1, \\[0.15cm]
u_{4} = |7,6,5,4,3,2,1,2,3,4,3,4,3,2,3,4,5,6,7\rangle & \alpha_{4} = 1, \\[0.15cm]
u_{5} = |7,6,5,4,3,2,1,2,3,4,3,4,3,8,3,4,5,6,7\rangle & \alpha_{5} = -1, \\[0.15cm]
u_{6} = |7,6,5,4,3,2,1,2,3,4,5,4,3,2,3,4,5,6,7\rangle & \alpha_{6} = -1, \\[0.15cm]
u_{7} = |7,6,5,4,3,2,1,2,3,4,5,4,3,8,3,4,5,6,7\rangle & \alpha_{7} = 1, \\[0.15cm]
u_{8} = |7,6,5,4,3,2,1,2,3,8,3,2,1,2,3,4,5,6,7\rangle & \alpha_{8} = -1, \\[0.15cm]
u_{9} = |7,6,5,4,3,2,1,2,3,8,3,2,3,2,3,4,5,6,7\rangle & \alpha_{9} = 1, \\[0.15cm]
u_{10} = |7,6,5,4,3,2,1,2,3,8,3,2,3,8,3,4,5,6,7\rangle & \alpha_{10} = -1, \\[0.15cm]
\textrm{\hspace{1.7cm} \dots and $68$ more states}
\end{array}\right.\\[0.2cm]
w_{14} &= \frac1{\sqrt {13110}}\sum_{i=1}^{13110} \alpha_i u_i\,,
\ \
&&
\left\{\begin{array}{ll}
u_{1} = |7,6,5,4,3,2,1,2,3,4,3,2,1,2,3,4,3,2,3,4,3,2,1,2,3,4,5,6,7\rangle & \alpha_{1} = 1, \\[0.15cm]
u_{2} = |7,6,5,4,3,2,1,2,3,4,3,2,1,2,3,4,3,2,3,4,3,2,3,2,3,4,5,6,7\rangle & \alpha_{2} = -1, \\[0.15cm]
u_{3} = |7,6,5,4,3,2,1,2,3,4,3,2,1,2,3,4,3,2,3,4,3,2,3,8,3,4,5,6,7\rangle & \alpha_{3} = 1, \\[0.15cm]
u_{4} = |7,6,5,4,3,2,1,2,3,4,3,2,1,2,3,4,3,2,3,4,3,4,3,2,3,4,5,6,7\rangle & \alpha_{4} = 1, \\[0.15cm]
u_{5} = |7,6,5,4,3,2,1,2,3,4,3,2,1,2,3,4,3,2,3,4,3,4,3,8,3,4,5,6,7\rangle & \alpha_{5} = -1, \\[0.15cm]
u_{6} = |7,6,5,4,3,2,1,2,3,4,3,2,1,2,3,4,3,2,3,4,5,4,3,2,3,4,5,6,7\rangle & \alpha_{6} = -1, \\[0.15cm]
u_{7} = |7,6,5,4,3,2,1,2,3,4,3,2,1,2,3,4,3,2,3,4,5,4,3,8,3,4,5,6,7\rangle & \alpha_{7} = 1, \\[0.15cm]
u_{8} = |7,6,5,4,3,2,1,2,3,4,3,2,1,2,3,4,3,2,3,8,3,2,1,2,3,4,5,6,7\rangle & \alpha_{8} = -1, \\[0.15cm]
u_{9} = |7,6,5,4,3,2,1,2,3,4,3,2,1,2,3,4,3,2,3,8,3,2,3,2,3,4,5,6,7\rangle & \alpha_{9} = 1, \\[0.15cm]
u_{10} = |7,6,5,4,3,2,1,2,3,4,3,2,1,2,3,4,3,2,3,8,3,2,3,8,3,4,5,6,7\rangle & \alpha_{10} = -1, \\[0.15cm]
\textrm{\hspace{1.7cm} \dots and $13100$ more states}
\end{array}\right.
\end{alignat}
\end{subequations}
\endgroup
respectively.

%
\section{Conformal partition functions for $\boldsymbol{\repM_{\g,\mu,K}}$}\label{app:Z}
%

In this section, we use \eqref{eq:Z.decomp.general}, \eqref{eq:charQK} and the decompositions \eqref{eq:M.decomp} to compute the conformal partition functions for each model. We repeatedly use the identity
\be
\sum_{r=1}^{p-1} \chit_{r,s} \chib_{r,s'} = 
\sum_{r=1}^{p-1} \chit_{r,p'-s} \chib_{r,p'-s'}\,,
\ee
obtained by using $h_{r,s} = h_{p-r,p'-s}$ and changing the summation variable $r \mapsto p-r$.

\subsection[$A_n$ models]{$\boldsymbol{A_n}$ models}

\paragraph{1) $\boldsymbol{Z_{A_n,\mu,\id,\id}}$\,.} 

In this case, we have $\eta = 0$, and
\be
Z_{A_n,\mu,\id,\id} = \sum_{s=1}^{p'-1} \chit_{\repQ_{0,(-1)^{s \mu} q^s}}
= \sum_{r=1}^{p-1} \sum_{s=1}^{p'-1} \chit_{r,s}\chib_{r,s}\,.
\ee

\paragraph{2) $\boldsymbol{Z_{A_n,\mu,\id,R}}$\,.} 

In this case, we have
\be
Z_{A_n,\mu,\id,R} = \sum_{s=1}^{p'-1} \kappa_{0,(-1)^{s \mu}q^s} \chit_{\repQ_{0,(-1)^{s \mu}q^s}}
\ee
with $\kappa_{0,(-1)^{s \mu}q^s} = (-1)^{s+1}$. For $p'$ even, we have $\eta = 0$ and 
\be
Z_{A_n,\mu,\id,R} = \sum_{r=1}^{p-1} \sum_{s=1}^{p'-1} (-1)^{s+1} \chit_{r,s}\chib_{r,s}\,.
\ee
For $p'$ odd, we instead have $\eta = 1$, and 
\be
Z_{A_n,\mu,\id,R} = \sum_{r=1}^{p-1} \sum_{s=1}^{p'-1} (-1)^{r+p s+1} \chit_{r,s}\chib_{r,s}\,.
\ee
We express these two results together as
\be
Z_{A_n,\mu,\id,R} = (-1)^{pp'}\sum_{r=1}^{p-1} \sum_{s=1}^{p'-1} (-1)^{\lambda_{r,s}} \chit_{r,s}\chib_{r,s}\,,
\ee
with $\lambda_{r,s}$ defined in \eqref{eq:lambda.rs} and satisfying
\be
(-1)^{\lambda_{r,s}+pp'} = 
\left\{\begin{array}{ll}
(-1)^{s+1} & p' \textrm{ even,} \\[0.1cm]
(-1)^{r+ps+1} & p' \textrm{ odd.} 
\end{array}\right.
\ee

\paragraph{3) $\boldsymbol{Z_{A_n,\mu,R,\id}}$\,.} 

In this case, we have $\eta = 0$. For $p'$ even, we have
\begin{alignat}{2}
Z_{A_n,\mu,R,\id} &=\chit_{\repQ_{0,q^{p'/2}}} + \sum_{k=1}^{(p'-2)/2} \sum_{\eps = \pm 1} \chit_{\repQ_{k,\eps q^{p'/2}}}
= \sum_{r=1}^{p-1} \chit_{r,p'/2} \chib_{r,p'/2}
+2\sum_{r=1}^{p-1} \sum_{k=1}^{(p'-2)/2} \chit_{r,p'/2+k}\chib_{r,p'/2-k}
\nn\\
&= \sum_{r=1}^{p-1} \sum_{s=1}^{p'-1} \chit_{r,s}\chib_{r,p'-s}\,.
\end{alignat}
For $p'$ odd, we have
\begin{alignat}{2}
Z_{A_n,\mu,R,\id} &=\sum_{k=1/2}^{(p'-2)/2} \sum_{\eps = \pm 1}\chit_{\repQ_{k,\eps q^{p'/2}}}
= 2\sum_{r=1}^{p-1} \sum_{k=1/2}^{(p'-2)/2} \chit_{r,p'/2+k}\chib_{r,p'/2-k}
= \sum_{r=1}^{p-1} \sum_{s=1}^{p'-1} \chit_{r,s}\chib_{r,p'-s}\,.
\end{alignat}

\paragraph{4) $\boldsymbol{Z_{A_n,\mu,R,R}}$\,.} 

For $p'$ even, we have $\eta = 0$ and
\begin{alignat}{2}
Z_{A_n,\mu,R,R} &=\kappa_{0,q^{p'/2}}\chit_{\repQ_{0,q^{p'/2}}} 
+\sum_{k=1}^{(p'-2)/2} \sum_{\eps = \pm 1} \kappa_{k,\eps q^{p'/2}}\chit_{\repQ_{k,\eps q^{p'/2}}}
\nn\\
&= \sum_{r=1}^{p-1} \chit_{r,p'/2} \chib_{r,p'/2}
+2\sum_{r=1}^{p-1} \sum_{k=1}^{(p'-2)/2} (-1)^k\chit_{r,p'/2+k}\chib_{r,p'/2-k}
\nn\\
&= \sum_{r=1}^{p-1} \sum_{s=1}^{p'-1} (-1)^{s+p'/2}\chit_{r,s}\chib_{r,p'-s}\,,
\end{alignat}
where we used $\kappa_{k,\eps q^{p'/2}} = (-1)^k$.
For $p'$ odd, we have $\eta = 1$ and
\begin{alignat}{2}
Z_{A_n,\mu,R,R} &=\sum_{k=1/2}^{(p'-2)/2} \sum_{\eps = \pm 1}\kappa_{k,\eps q^{p'/2}} \chit_{\repQ_{k,\eps q^{p'/2}}}
= 2\sum_{r=1}^{p-1} \sum_{k=1/2}^{(p'-2)/2} \ir^{-2 p k} (-1)^r \chit_{r,p'/2+k}\chib_{r,p'/2-k}
\nn\\
&= \ir^{-p p'} \sum_{r=1}^{p-1} \sum_{s=1}^{p'-1} (-1)^{r+p s} \chit_{r,s}\chib_{r,p'-s}\,,
\end{alignat} 
where we used 
$\kappa_{k,\eps q^{p'/2}} = \eps\, \ir^{-2 p k}$.
We express these two results together as
\be
Z_{A_n,\mu,R,R} = \ir^{pp'}\sum_{r=1}^{p-1} \sum_{s=1}^{p'-1} (-1)^{\lambda_{r,s}+1} \chit_{r,s}\chib_{r,p'-s}\,,
\ee
with
\be
\ir^{pp'}(-1)^{\lambda_{r,s}+1} = 
\left\{\begin{array}{ll}
(-1)^{s+p'/2} & p' \textrm{ even,} \\[0.1cm]
\ir^{-pp'}(-1)^{r+ps} & p' \textrm{ odd.} 
\end{array}\right.
\ee

\subsection[$D_n$ models]{$\boldsymbol{D_n}$ models}

Here, we have $\eta = 0$ in all cases.

\paragraph{1) $\boldsymbol{Z_{D_n,\mu,\id,\id}}$ and $\boldsymbol{Z_{D_n,\mu,\id,P_{(n-1,n)}}}$\,.} 

For $v \in \{0,1\}$, we have
\begin{alignat}{2}
&Z_{D_n,\mu,\id,(P_{(n-1,n)})^v} =(\wh\kappa_{0,q^{p'/2}})^v\chit_{\repQ_{0,q^{p'/2}}} 
+ \sum_{s = 1,3,\dots, p'-1} (\kappa_{0,q^s})^v\chit_{\repQ_{0,q^s}}
+ \sum_{\ell=1}^{\lfloor (p'-2)/4 \rfloor} \sum_{\eps = \pm 1} (\kappa_{2\ell,\eps q^{p'/2}})^v\chit_{\repQ_{2\ell,\eps q^{p'/2}}}
\nn\\
&= (-1)^v\sum_{r=1}^{p-1} \chit_{r,p'/2} \chib_{r,p'/2}
+ \sum_{r=1}^{p-1} \sum_{s = 1,3,\dots, p'-1} \chit_{r,s}\chib_{r,s}
+2(-1)^v\sum_{r=1}^{p-1} \sum_{\ell=1}^{\lfloor (p'-2)/4 \rfloor} \chit_{r,p'/2+2 \ell}\chib_{r,p'/2-2\ell}
\nn\\
&= \sum_{r=1}^{p-1} \sum_{s=1,3,...,p'-1} \chit_{r,s}\chib_{r,s} + (-1)^v \sum_{r=1}^{p-1}\sum_{\substack{s=1\\[0.05cm]s\,\equiv\, n+1 \mmod 2}}^{p'-1} \chit_{r,s}\chib_{r,p'-s}
\end{alignat}
where we used $\wh\kappa_{0,q^{p'/2}} = -1$, $\kappa_{0,q^s} = 1$ for all $s$ and $\kappa_{2\ell,\eps q^{p'/2}} = -1$ for all $\ell$.

\paragraph{2) $\boldsymbol{Z_{D_n,\mu,P_{(n-1,n)},\id}}$ and $\boldsymbol{Z_{D_n,\mu,P_{(n-1,n)},P_{(n-1,n)}}}$\,.} 

For $v \in \{0,1\}$, we have
\begin{alignat}{2}
Z_{D_n,\mu,P_{(n-1,n)},(P_{(n-1,n)})^v} 
&= \sum_{s = 2,4,\dots, p'-2} (\kappa_{0,q^s})^v\chit_{\repQ_{0,q^s}}
+ \sum_{\eps = \pm 1}\sum_{\ell=1}^{\lfloor p'/4 \rfloor} (\kappa_{2\ell-1,\eps q^{p'/2}})^v\chit_{\repQ_{2\ell-1,\eps q^{p'/2}}}
\nn\\
&= \sum_{r=1}^{p-1} \sum_{s = 2,4,\dots, p'-2} \chit_{r,s}\chib_{r,s}
+2(-1)^v\sum_{r=1}^{p-1} \sum_{\ell=1}^{\lfloor p'/4 \rfloor} \chit_{r,p'/2+2 \ell-1}\chib_{r,p'/2-2\ell+1}
\nn\\
&= \sum_{r=1}^{p-1} \sum_{s=2,4,...,p'-2} \chit_{r,s}\chib_{r,s} + (-1)^v \sum_{r=1}^{p-1}\sum_{\substack{s=1\\[0.05cm]s\,\equiv\, n \mmod 2}}^{p'-1} \chit_{r,s}\chib_{r,p'-s}
\end{alignat}
where we used $\kappa_{0,q^s} = 1$ for all $s$ and $\kappa_{2\ell-1,\eps q^{p'/2}} = -1$ for all $\ell$.

\subsection[$E_6$, $E_7$ and $E_8$ models]{$\boldsymbol{E_6}$, $\boldsymbol{E_7}$ and $\boldsymbol{E_8}$ models}

For these partition functions, we have $\eta = 0$ in all cases.

\paragraph{1) $\boldsymbol{Z_{E_6,\mu,\id,\id}}$ and $\boldsymbol{Z_{E_6,\mu,\id,P_{(15),(24)}}}$\,.} 

For $v \in \{0,1\}$, we have
\begin{alignat}{2}
Z_{E_6,\mu,P_{(15),(24)},(P_{(15),(24)})^v} &= \sum_{s=1,4,5,7,8,11} (\kappa_{0,q^s})^v\chit_{\repQ_{0,q^s}} + \sum_{\eps = \pm 1} (\kappa_{2,\eps q^6})^v\chit_{\repQ_{2,\eps q^6}} + \sum_{\sigma,\eps = \pm 1} (\kappa_{3,\eps q^{4 \sigma}})^v\chit_{\repQ_{3,\eps q^{4 \sigma}}}
\nn\\&= \sum_{r=1}^{p-1}\sum_{s=1,5,7,11} \chit_{r,s}\chib_{r,s} 
+ (-1)^v \sum_{r=1}^{p-1}\sum_{s=4,8} \chit_{r,s}\chib_{r,s} 
\nn\\&+ 2 \sum_{r=1}^{p-1} \Big( 
(-1)^v\chit_{r,8}\chib_{r,4} 
+ \chit_{r,7}\chib_{r,1} 
+ \chit_{r,11}\chib_{r,5}
\Big)
\nn\\& 
= \sum_{r=1}^{p-1}\Big(|\chit_{r,1}+\chit_{r,7}|^2+(-1)^v |\chit_{r,4}+\chit_{r,8}|^2+|\chit_{r,5}+\chit_{r,11}|^2\Big)\,,
\end{alignat}
where we used 
\be
\kappa_{0,q^s} = \left\{\begin{array}{cl}
1 & s \in \{1,5,7,11\}\,, \\[0.1cm]
-1 & s \in \{4,8\}\,,
\end{array}\right.
\qquad 
\kappa_{2, \eps q^6} = -1\,,
\qquad
\kappa_{3, \eps q^{4 \sigma}} = 1\,.
\ee

\paragraph{2) $\boldsymbol{Z_{E_6,\mu,P_{(15),(24)},\id}}$ and $\boldsymbol{Z_{E_6,\mu,P_{(15),(24)},P_{(15),(24)}}}$\,.} 

For $v \in \{0,1\}$, we have
\begin{alignat}{2}
Z_{E_6,\mu,P_{(15),(24)},(P_{(15),(24)})^v} &= \sum_{s=4,8} (\kappa_{0,q^s})^v\chit_{\repQ_{0,q^s}} + \sum_{k=1,2,5}\sum_{\eps = \pm 1} (\kappa_{k,\eps q^6})^v\chit_{\repQ_{k,\eps q^6}} + \sum_{\sigma,\eps = \pm 1} (\kappa_{2,\eps q^{3 \sigma}})^v\chit_{\repQ_{2,\eps q^{3 \sigma}}}
\nn\\&= \sum_{r=1}^{p-1}\sum_{s=4,8} \chit_{r,s}\chib_{r,s} 
+ 2 \sum_{r=1}^{p-1} \Big( 
(-1)^v\chit_{r,7}\chib_{r,5} 
+ \chit_{r,8}\chib_{r,4} 
+ (-1)^v \chit_{r,11}\chib_{r,1}
\nn\\&\hspace{5cm}+ (-1)^v \chit_{r,5}\chib_{r,1}
+ (-1)^v \chit_{r,11}\chib_{r,7}
\Big)
\nn\\& 
= \sum_{r=1}^{p-1} \Big(
(-1)^{v+1}|\chit_{r,1}+\chit_{r,7}|^2
+|\chit_{r,4}+\chit_{r,8}|^2
+(-1)^{v+1}|\chit_{r,5}+\chit_{r,11}|^2
\nn\\[0.1cm]
&\hspace{1.3cm}
+(-1)^v |\chit_{r,1}+\chit_{r,5}+\chit_{r,7}+\chit_{r,11}|^2
\Big)\,,
\end{alignat}
where we used 
\be
\kappa_{0,q^4} = \kappa_{0,q^8} = 1
\qquad 
\kappa_{k,\eps q^6} = \left\{\begin{array}{cl}
-1 & k \in \{1,5\}\,,\\[0.1cm]
1 & k =2\,, 
\end{array}\right.
\qquad
\kappa_{2, \eps q^{3 \sigma}} = -1\,.
\ee

\paragraph{3) $\boldsymbol{Z_{E_7,\mu,\id,\id}}$\,.} 

We have
\begin{alignat}{2}
Z_{E_7,\mu,\id,\id} &= \sum_{s=1,5,7,9,11,13,17} \chit_{\repQ_{0,q^s}} + \sum_{k=2,4,8}\sum_{\eps = \pm 1} \chit_{\repQ_{k,\eps q^9}} + \sum_{\sigma,\eps = \pm 1} \chit_{\repQ_{3,\eps q^{6 \sigma}}}
\nn\\&= \sum_{r=1}^{p-1}\sum_{s=1,5,7,9,11,13,17} \chit_{r,s}\chib_{r,s} 
\nn\\&+ 2 \sum_{r=1}^{p-1} \Big( \chit_{r,11}\chib_{r,7} 
+ \chit_{r,13}\chib_{r,5} 
+ \chit_{r,17}\chib_{r,1}
+ \chit_{r,9}\chib_{r,3}
+ \chit_{r,15}\chib_{r,9}
\Big)
\nn\\& 
= \sum_{r=1}^{p-1}\Big(|\chit_{r,1}+\chit_{r,17}|^2-|\chit_{r,3}+\chit_{r,15}|^2+|\chit_{r,5}+\chit_{r,13}|^2+|\chit_{r,7}+\chit_{r,11}|^2
\\& \hspace{1cm}+ |\chit_{r,3}+\chit_{r,9}+\chit_{r,15}|^2\Big)\,.
\end{alignat}

\paragraph{4) $\boldsymbol{Z_{E_8,\mu,\id,\id}}$\,.} 

We have
\begin{alignat}{2}
Z_{E_8,\mu,\id,\id} &= \sum_{s=1,7,11,13,17,19,23,29} \chit_{\repQ_{0,q^s}} + \sum_{k=2,4,8,14}\sum_{\eps = \pm 1} \chit_{\repQ_{k,\eps q^{15}}} + \sum_{k=3,9}\sum_{\sigma,\eps = \pm 1} \chit_{\repQ_{k,\eps q^{10 \sigma}}}
+ \sum_{s=6,12}\sum_{\sigma,\eps = \pm 1} \chit_{\repQ_{5,\eps q^{\sigma s}}}
\nn\\&= \sum_{r=1}^{p-1}\sum_{s=1,7,11,13,17,19,23,29} \chit_{r,s}\chib_{r,s} 
\nn\\&+ 2 \sum_{r=1}^{p-1} \Big( 
\chit_{r,17}\chib_{r,13} 
+ \chit_{r,19}\chib_{r,11} 
+ \chit_{r,23}\chib_{r,7}
+ \chit_{r,29}\chib_{r,1}
+ \chit_{r,13}\chib_{r,7}
+ \chit_{r,23}\chib_{r,17}
\nn\\&\hspace{1.4cm}
+ \chit_{r,19}\chib_{r,1}
+ \chit_{r,29}\chib_{r,11}
+ \chit_{r,11}\chib_{r,1}
+ \chit_{r,29}\chib_{r,19}
+ \chit_{r,17}\chib_{r,7}
+ \chit_{r,23}\chib_{r,13}
\Big)
\nn\\& 
= \sum_{r=1}^{p-1}\Big(|\chit_{r,1}+\chit_{r,11}+\chit_{r,19}+\chit_{r,29}|^2+|\chit_{r,7}+\chit_{r,13}+\chit_{r,17}+\chit_{r,23}|^2\Big)\,.
\end{alignat}

\subsection[$D_4$ model]{$\boldsymbol{D_4}$ models}\label{app:Z.D4}

For these partition functions, we have $\eta = 0$ in all cases.

\paragraph{1) $\boldsymbol{Z_{D_4,1,\id,\id}}$, $\boldsymbol{Z_{D_4,1,\id,P_{(13)}}}$ and $\boldsymbol{Z_{D_4,\mu,\id,P_{(134)}}}$\,.} 

For $K \in \{\id,P_{(13)},P_{(134)}\}$, we have
\begin{alignat}{2}
Z_{D_4,1,\id,K} &= \wh\kappa_{0,q^3}\chit_{\repQ_{0,q^3}} + \sum_{s=1,3,5} \kappa_{0, q^s}\chit_{\repQ_{0,q^s}} + \sum_{\eps = \pm 1}\kappa_{2,\eps q^3}\chit_{\repQ_{2,\eps q^3}}
\nn\\&= 
\sum_{r=1}^4 \Big( 
\kappa_{0,q}\chit_{r,1}\chib_{r,1} 
+ (\kappa_{0,q^3} + \wh\kappa_{0,q^3})\chit_{r,3}\chib_{r,3}
+ \kappa_{0,q^5}\chit_{r,5}\chib_{r,5}
+ (\kappa_{2,q^3}+\kappa_{2,-q^3}) \chit_{r,1}\chib_{r,5}
\Big)
\nn\\&= 
\left\{\begin{array}{cl}
\displaystyle\sum_{r=1}^{p-1} \big(|\chit_{r,1}+\chit_{r,5}|^2+2|\chit_{r,3}|^2\big) & K = \id\,, \\
\displaystyle\sum_{r=1}^{p-1} |\chit_{r,1}-\chit_{r,5}|^2 & K = P_{(13)}\,, \\
\displaystyle\sum_{r=1}^{p-1} \big(|\chit_{r,1}+\chit_{r,5}|^2-|\chit_{r,3}|^2\big)
& K = P_{(134)}\,, \\
\end{array}\right.
\end{alignat}
where we have $\kappa_{0,q} = \kappa_{0,q^5} = 1$ in all cases, and
\be
\kappa_{0,q^3} = 
\left\{\begin{array}{cl}
1 & K = \id\,, \\[0.1cm]
1 & K = P_{(13)}\,, \\[0.1cm]
\omega & K = P_{(134)}\,, \\[0.1cm]
\end{array}\right.
\quad \ 
\wh\kappa_{0,q^3} = 
\left\{\begin{array}{cl}
1 & K = \id\,, \\[0.1cm]
-1 & K = P_{(13)}\,, \\[0.1cm]
\omega^{-1} & K = P_{(134)}\,, \\[0.1cm]
\end{array}\right.
\quad \ 
\kappa_{2, \eps q^3} = 
\left\{\begin{array}{cl}
1 & K = \id\,, \\[0.1cm]
-1 & K = P_{(13)}\,, \\[0.1cm]
1 & K = P_{(134)}\,. \\[0.1cm]
\end{array}\right.
\ee
with $\omega = \eE^{2 \pi \ir/3}$.

\paragraph{2) $\boldsymbol{Z_{D_4,1,P_{(13)},\id}}$ and $\boldsymbol{Z_{D_4,1,P_{(13)},P_{(13)}}}$\,.} 

For $K \in \{\id,P_{(13)}\}$, we have
\begin{alignat}{2}
Z_{D_4,1,P_{(13)},K} &= \sum_{s=2,4} \kappa_{0, q^s}\chit_{\repQ_{0,q^s}} + \sum_{\eps = \pm 1}\kappa_{1,\eps q^3}\chit_{\repQ_{1,\eps q^3}}
\nn\\&= 
\sum_{r=1}^4 \Big( 
\kappa_{0,q^2}\chit_{r,2}\chib_{r,2}
+ \kappa_{0,q^4}\chit_{r,4}\chib_{r,4} 
+ (\kappa_{1,q^3} + \kappa_{1,-q^3})\chit_{r,4}\chib_{r,2}
\Big)
\nn\\&= 
\left\{\begin{array}{cl}
\displaystyle\sum_{r=1}^{p-1} |\chit_{r,2}+\chit_{r,4}|^2 & K = \id\,, \\
\displaystyle\sum_{r=1}^{p-1} |\chit_{r,2}-\chit_{r,4}|^2 & K = P_{(13)}\,, 
\end{array}\right.
\end{alignat}
where we used
\be
\kappa_{0,q^2} = \kappa_{0,q^4} = 1\,,
\qquad
\kappa_{1,\eps q^3} = 
\left\{\begin{array}{cl}
1 & K = \id\,, \\[0.1cm]
-1 & K = P_{(13)}\,. 
\end{array}\right.
\ee

\paragraph{3) $\boldsymbol{Z_{D_4,1,P_{(134)},\id}}$, $\boldsymbol{Z_{D_4,1,P_{(134)},P_{(134)}}}$ and $\boldsymbol{Z_{D_4,1,P_{(134)},P_{(143)}}}$\,.} 

For $K \in \{\id,P_{(134)},P_{(143)}\}$, we have
\begin{alignat}{2}
&Z_{D_4,1,P_{(134)},K} = \kappa_{0,q^3}\chit_{\repQ_{0,q^3}} + \sum_{\sigma,\eps = \pm 1}\kappa_{1,\eps q^{2\sigma}}\chit_{\repQ_{1,\eps q^{2\sigma}}}
\nn\\&= 
\sum_{r=1}^4 \Big( 
\kappa_{0,q^3}\chit_{r,3}\chib_{r,3} 
+ (\kappa_{1,q^2} + \kappa_{1,-q^2})\chit_{r,3}\chib_{r,1}
+ (\kappa_{1,q^{-2}} + \kappa_{1,-q^{-2}})\chit_{r,5}\chib_{r,3}\Big)
\nn\\&= 
\left\{\begin{array}{cl}
\displaystyle\sum_{r=1}^{4} \big(|\chit_{r,1}+\chit_{r,3}|^2+|\chit_{r,3}+\chit_{r,5}|^2-|\chit_{r,1}|^2-|\chit_{r,3}|^2-|\chit_{r,5}|^2\big) 
& K = \id\,, \\[0.3cm]
\displaystyle\sum_{r=1}^{4} \big(|\chit_{r,1}+\omega^{-1}\chit_{r,3}|^2+|\chit_{r,3}+\omega\,\chit_{r,5}|^2-|\chit_{r,1}|^2-|\chit_{r,3}|^2-|\chit_{r,5}|^2\big)
& K = P_{(134)}\,, \\[0.3cm]
\displaystyle\sum_{r=1}^{4} \big(|\chit_{r,1}+\omega\,\chit_{r,3}|^2+|\chit_{r,3}+\omega^{-1}\chit_{r,5}|^2-|\chit_{r,1}|^2-|\chit_{r,3}|^2-|\chit_{r,5}|^2\big)
& K = P_{(143)}\,, \\
\end{array}\right.
\end{alignat}
where $\omega = \eE^{2 \pi \ir/3}$ and we used
\be
\kappa_{0,q^3} = 1\,
\qquad
\kappa_{1,\eps q^2} = 
\left\{\begin{array}{cl}
1 & K = \id\,, \\[0.1cm]
\omega^{-1} & K = P_{(134)}\,, \\[0.1cm]
\omega & K = P_{(143)}\,, \\[0.1cm]
\end{array}\right.
\qquad
\kappa_{1,\eps q^{-2}} = 
\left\{\begin{array}{cl}
1 & K = \id\,, \\[0.1cm]
\omega & K = P_{(134)}\,, \\[0.1cm]
\omega^{-1} & K = P_{(143)}\,. 
\end{array}\right.
\ee

%
\section{Decompositions and partition functions for $\boldsymbol{E_6}$, $\boldsymbol{E_7}$ and $\boldsymbol{E_8}$}\label{app:E678}
%

\subsection{Decompositions}\label{eq:MTL.E678}

The decomposition of $\repM_{E_6,\mu,a,b}$, $\repM_{E_7,\mu,a,b}$ and $\repM_{E_8,\mu,a,b}$ given by \cref{thm:M.decomp.TL} is
\begingroup
\allowdisplaybreaks
\begin{subequations}
\begin{alignat}{2}
\label{eq:Mab.decomp.E6}
\repM_{E_6,\mu,a,b} & \simeq 
\left\{\begin{array}{cl}
\repQ_0 \oplus \repQ_3
& (a,b)=(1,1), (5,5)
\\[0.15cm] 
\repQ_0 \oplus \repQ_1 \oplus \repQ_2 \oplus 2\repQ_3 \oplus \repQ_4
& (a,b)=(2,2), (4,4)
\\[0.15cm] 
\repQ_0 \oplus 2\repQ_1 \oplus 3\repQ_2 \oplus 3\repQ_3 \oplus 2\repQ_4 \oplus \repQ_5
& (a,b)=(3,3)
\\[0.15cm] 
\repQ_0 \oplus \repQ_2 \oplus \repQ_3 \oplus \repQ_5
& (a,b)=(6,6)
\\[0.15cm] 
\repQ_1 \oplus \repQ_2 \oplus \repQ_3 \oplus \repQ_4
& (a,b)=(1,3),(3,5)
\\[0.15cm] 
\repQ_2 \oplus \repQ_5
& (a,b)=(1,5)
\\[0.15cm] 
\repQ_1 \oplus 2\repQ_2 \oplus \repQ_3 \oplus \repQ_4 \oplus \repQ_5
& (a,b)=(2,4)
\\[0.15cm] 
\repQ_1 \oplus \repQ_2 \oplus \repQ_3 \oplus \repQ_4
& (a,b)=(2,6)
\\[0.15cm] 
\repQ_{\frac12} \oplus \repQ_{\frac52} \oplus \repQ_{\frac72}
& (a,b)=(1,2),(4,5)
\\[0.15cm] 
\repQ_{\frac32} \oplus \repQ_{\frac52} \oplus \repQ_{\frac92}
& (a,b)=(1,4),(2,5)
\\[0.15cm] 
\repQ_{\frac32} \oplus \repQ_{\frac72}
& (a,b)=(1,6),(5,6)
\\[0.15cm] 
\repQ_{\frac12} \oplus 2\repQ_{\frac32} \oplus 2\repQ_{\frac52} \oplus 2\repQ_{\frac72} \oplus \repQ_{\frac92}
& (a,b)=(2,3),(3,4)
\\[0.15cm] 
\repQ_{\frac12} \oplus \repQ_{\frac32} \oplus 2\repQ_{\frac52} \oplus \repQ_{\frac72} \oplus \repQ_{\frac92}
& (a,b)=(3,6)
\end{array}\right.
\\[0.4cm]\label{eq:Mab.decomp.E7}
\repM_{E_7,\mu,a,b} & \simeq 
\left\{\begin{array}{cl}
\repQ_0 \oplus \repQ_3 \oplus \repQ_5 \oplus \repQ_8
& (a,b)=(1,1)
\\[0.15cm] 
\repQ_0 \oplus \repQ_1 \oplus \repQ_2 \oplus 2\repQ_3 \oplus 2\repQ_4 \oplus 2\repQ_5 \oplus \repQ_6 \oplus \repQ_7 \oplus \repQ_8 
& (a,b)=(2,2)
\\[0.15cm] 
\repQ_0 \oplus 2\repQ_1 \oplus 3\repQ_2 \oplus 4\repQ_3 \oplus 4\repQ_4 \oplus 4\repQ_5 \oplus 3\repQ_6 \oplus 2\repQ_7 \oplus \repQ_8
& (a,b)=(3,3)
\\[0.15cm] 
\repQ_0 \oplus \repQ_1 \oplus 2\repQ_2 \oplus 2\repQ_3 \oplus 3\repQ_4 \oplus 2\repQ_5 \oplus 2\repQ_6 \oplus \repQ_7 \oplus \repQ_8 
& (a,b)=(4,4)
\\[0.15cm] 
\repQ_0 \oplus \repQ_1 \oplus \repQ_3 \oplus 2\repQ_4 \oplus \repQ_5 \oplus \repQ_7 \oplus \repQ_8 
& (a,b)=(5,5)
\\[0.15cm] 
\repQ_0 \oplus \repQ_4 \oplus \repQ_8 
& (a,b)=(6,6)
\\[0.15cm] 
\repQ_0 \oplus \repQ_2 \oplus \repQ_3 \oplus \repQ_4 \oplus \repQ_5 \oplus \repQ_6 \oplus \repQ_8 
& (a,b)=(7,7)
\\[0.15cm] 
\repQ_1 \oplus \repQ_2 \oplus \repQ_3 \oplus 2\repQ_4 \oplus \repQ_5 \oplus \repQ_6 \oplus \repQ_7
& (a,b)=(1,3)
\\[0.15cm] 
\repQ_2 \oplus \repQ_3 \oplus \repQ_5 \oplus \repQ_6
& (a,b)=(1,5)
\\[0.15cm] 
\repQ_1 \oplus 2\repQ_2 \oplus 2\repQ_3 \oplus 2\repQ_4 \oplus 2\repQ_5 \oplus 2\repQ_6 \oplus \repQ_7
& (a,b)=(2,4)
\\[0.15cm] 
\repQ_2 \oplus \repQ_3 \oplus \repQ_5 \oplus \repQ_6
& (a,b)=(2,6)
\\[0.15cm] 
\repQ_1 \oplus \repQ_2 \oplus \repQ_3 \oplus 2\repQ_4 \oplus \repQ_5 \oplus \repQ_6 \oplus \repQ_7
& (a,b)=(2,7)
\\[0.15cm] 
\repQ_1 \oplus 2\repQ_2 \oplus 2\repQ_3 \oplus 2\repQ_4 \oplus 2\repQ_5 \oplus 2\repQ_6 \oplus 2\repQ_7
& (a,b)=(3,5)
\\[0.15cm] 
\repQ_1 \oplus \repQ_3 \oplus \repQ_4 \oplus \repQ_5 \oplus \repQ_7 
& (a,b)=(4,6)
\\[0.15cm] 
\repQ_1 \oplus \repQ_2 \oplus 2\repQ_3 \oplus \repQ_4 \oplus 2\repQ_5 \oplus \repQ_6 \oplus \repQ_7
& (a,b)=(4,7)
\\[0.15cm] 
\repQ_2 \oplus \repQ_4 \oplus \repQ_6
& (a,b)=(6,7)
\\[0.15cm] 
\repQ_{\frac12} \oplus \repQ_{\frac52} \oplus \repQ_{\frac72} \oplus \repQ_{\frac92} \oplus \repQ_{\frac{11}2} \oplus \repQ_{\frac{15}2}
& (a,b)=(1,2)
\\[0.15cm] 
\repQ_{\frac32} \oplus \repQ_{\frac52} \oplus \repQ_{\frac72} \oplus \repQ_{\frac92} \oplus \repQ_{\frac{11}2} \oplus \repQ_{\frac{13}2}
& (a,b)=(1,4)
\\[0.15cm] 
\repQ_{\frac52} \oplus \repQ_{\frac{11}2}
& (a,b)=(1,6)
\\[0.15cm] 
\repQ_{\frac32} \oplus \repQ_{\frac72} \oplus \repQ_{\frac92} \oplus \repQ_{\frac{13}2}
& (a,b)=(1,7)
\\[0.15cm] 
\repQ_{\frac12} \oplus 2\repQ_{\frac32} \oplus 2\repQ_{\frac52} \oplus 3\repQ_{\frac72} \oplus 3\repQ_{\frac92} \oplus 2\repQ_{\frac{11}2} \oplus 2\repQ_{\frac{13}2} \oplus \repQ_{\frac{15}2} 
& (a,b)=(2,3)
\\[0.15cm] 
\repQ_{\frac32} \oplus 2\repQ_{\frac52} \oplus \repQ_{\frac72} \oplus \repQ_{\frac92} \oplus 2\repQ_{\frac{11}2} \oplus \repQ_{\frac{13}2}
& (a,b)=(2,5)
\\[0.15cm] 
\repQ_{\frac12} \oplus 2\repQ_{\frac32} \oplus 3\repQ_{\frac52} \oplus 3\repQ_{\frac72} \oplus 3\repQ_{\frac92} \oplus 3\repQ_{\frac{11}2} \oplus 2\repQ_{\frac{13}2} \oplus \repQ_{\frac{15}2} 
& (a,b)=(3,4)
\\[0.15cm] 
\repQ_{\frac32} \oplus \repQ_{\frac52} \oplus \repQ_{\frac72} \oplus \repQ_{\frac92} \oplus \repQ_{\frac{11}2} \oplus \repQ_{\frac{13}2} 
& (a,b)=(3,6)
\\[0.15cm] 
\repQ_{\frac12} \oplus \repQ_{\frac32} \oplus 2\repQ_{\frac52} \oplus 2\repQ_{\frac72} \oplus 2\repQ_{\frac92} \oplus 2\repQ_{\frac{11}2} \oplus \repQ_{\frac{13}2} \oplus \repQ_{\frac{15}2} 
& (a,b)=(3,7)
\\[0.15cm] 
\repQ_{\frac12} \oplus \repQ_{\frac32} \oplus \repQ_{\frac52} \oplus 2\repQ_{\frac72} \oplus 2\repQ_{\frac92} \oplus \repQ_{\frac{11}2} \oplus \repQ_{\frac{13}2} \oplus \repQ_{\frac{15}2} 
& (a,b)=(4,5)
\\[0.15cm] 
\repQ_{\frac12} \oplus \repQ_{\frac72} \oplus \repQ_{\frac92} \oplus \repQ_{\frac{15}2} 
& (a,b)=(5,6)
\\[0.15cm] 
\repQ_{\frac32} \oplus \repQ_{\frac52} \oplus \repQ_{\frac72} \oplus \repQ_{\frac92} \oplus \repQ_{\frac{11}2} \oplus \repQ_{\frac{13}2} 
& (a,b)=(5,7)
\end{array}\right.
\\[0.4cm]\label{eq:Mab.decomp.E8}
\repM_{E_8,\mu,a,b} & \simeq 
\left\{\begin{smallmatrix}
\repQ_0 \oplus \repQ_3 \oplus \repQ_5 \oplus \repQ_6 \oplus \repQ_8 \oplus \repQ_9 \oplus \repQ_{11} \oplus \repQ_{14}
& (a,b)=(1,1)
\\[0.15cm] 
\repQ_0 \oplus \repQ_1 \oplus \repQ_2 \oplus 2\repQ_3 \oplus 2\repQ_4 \oplus 3\repQ_5 \oplus 3\repQ_6 \oplus 2\repQ_7 \oplus 3\repQ_8 \oplus 3\repQ_9 \oplus 2\repQ_{10} \oplus 2\repQ_{11} \oplus \repQ_{12} \oplus \repQ_{13} \oplus \repQ_{14}
& (a,b)=(2,2)
\\[0.15cm] 
\repQ_0 \oplus 2\repQ_1 \oplus 3\repQ_2 \oplus 4\repQ_3 \oplus 5\repQ_4 \oplus 6\repQ_5 \oplus 6\repQ_6 \oplus 6\repQ_7 \oplus 6\repQ_8 \oplus 6\repQ_9 \oplus 5\repQ_{10} \oplus 4\repQ_{11} \oplus 3\repQ_{12} \oplus 2\repQ_{13} \oplus \repQ_{14}
& (a,b)=(3,3)
\\[0.15cm] 
\repQ_0 \oplus \repQ_1 \oplus 2\repQ_2 \oplus 3\repQ_3 \oplus 3\repQ_4 \oplus 4\repQ_5 \oplus 4\repQ_6 \oplus 4\repQ_7 \oplus 4\repQ_8 \oplus 4\repQ_9 \oplus 3\repQ_{10} \oplus 3\repQ_{11} \oplus 2\repQ_{12} \oplus \repQ_{13} \oplus \repQ_{14}
& (a,b)=(4,4)
\\[0.15cm] 
\repQ_0 \oplus \repQ_1 \oplus \repQ_2 \oplus \repQ_3 \oplus 2\repQ_4 \oplus 3\repQ_5 \oplus 2\repQ_6 \oplus 2\repQ_7 \oplus 2\repQ_8 \oplus 3\repQ_9 \oplus 2\repQ_{10} \oplus \repQ_{11} \oplus \repQ_{12} \oplus \repQ_{13} \oplus \repQ_{14}
& (a,b)=(5,5)
\\[0.15cm] 
\repQ_0 \oplus \repQ_1 \oplus \repQ_4 \oplus 2\repQ_5 \oplus \repQ_6 \oplus \repQ_8 \oplus 2\repQ_9 \oplus \repQ_{10} \oplus \repQ_{13} \oplus \repQ_{14}
& (a,b)=(6,6)
\\[0.15cm] 
\repQ_0 \oplus \repQ_5 \oplus \repQ_9 \oplus \repQ_{14}
& (a,b)=(7,7)
\\[0.15cm] 
\repQ_0 \oplus \repQ_2 \oplus \repQ_3 \oplus \repQ_4 \oplus 2\repQ_5 \oplus \repQ_6 \oplus 2\repQ_7 \oplus \repQ_8 \oplus 2\repQ_9 \oplus \repQ_{10} \oplus \repQ_{11} \oplus \repQ_{12} \oplus \repQ_{14}
& (a,b)=(8,8)
\\[0.15cm] 
\repQ_1 \oplus \repQ_2 \oplus \repQ_3 \oplus 2\repQ_4 \oplus 2\repQ_5 \oplus 2\repQ_6 \oplus 2\repQ_7 \oplus 2\repQ_8 \oplus 2\repQ_9 \oplus 2\repQ_{10} \oplus \repQ_{11} \oplus \repQ_{12} \oplus \repQ_{13}
& (a,b)=(1,3)
\\[0.15cm] 
\repQ_2 \oplus \repQ_3 \oplus \repQ_4 \oplus \repQ_5 \oplus \repQ_6 \oplus 2\repQ_7 \oplus \repQ_8 \oplus \repQ_9 \oplus \repQ_{10} \oplus \repQ_{11} \oplus \repQ_{12} 
& (a,b)=(1,5)
\\[0.15cm] 
\repQ_3 \oplus \repQ_6 \oplus \repQ_8 \oplus \repQ_{11}
& (a,b)=(1,7)
\\[0.15cm] 
\repQ_1 \oplus 2\repQ_2 \oplus 2\repQ_3 \oplus 3\repQ_4 \oplus 3\repQ_5 \oplus 3\repQ_6 \oplus 4\repQ_7 \oplus 3\repQ_8 \oplus 3\repQ_9 \oplus 3\repQ_{10} \oplus 2\repQ_{11} \oplus 2\repQ_{12} \oplus \repQ_{13}
& (a,b)=(2,4)
\\[0.15cm] 
\repQ_2 \oplus 2\repQ_3 \oplus \repQ_4 \oplus \repQ_5 \oplus 2\repQ_6 \oplus 2\repQ_7 \oplus 2\repQ_8 \oplus \repQ_9 \oplus \repQ_{10} 2\oplus \repQ_{11} \oplus \repQ_{12}
& (a,b)=(2,6)
\\[0.15cm] 
\repQ_1 \oplus \repQ_2 \oplus \repQ_3 \oplus 2\repQ_4 \oplus 2\repQ_5 \oplus 2\repQ_6 \oplus 2\repQ_7 \oplus 2\repQ_8 \oplus 2\repQ_9 \oplus 2\repQ_{10} \oplus \repQ_{11} \oplus \repQ_{12} \oplus \repQ_{13}
& (a,b)=(2,8)
\\[0.15cm] 
\repQ_1 \oplus 2\repQ_2 \oplus 3\repQ_3 \oplus 3\repQ_4 \oplus 3\repQ_5 \oplus 4\repQ_6 \oplus 4\repQ_7 \oplus 4\repQ_8 \oplus 3\repQ_9 \oplus 3\repQ_{10} \oplus 3\repQ_{11} \oplus 2\repQ_{12} \oplus \repQ_{13}
& (a,b)=(3,5)
\\[0.15cm] 
\repQ_2 \oplus \repQ_3 \oplus \repQ_4 \oplus \repQ_5 \oplus \repQ_6 \oplus 2\repQ_7 \oplus \repQ_8 \oplus \repQ_9 \oplus \repQ_{10} \oplus \repQ_{11} \oplus \repQ_{12} 
& (a,b)=(3,7)
\\[0.15cm] 
\repQ_1 \oplus \repQ_2 \oplus \repQ_3 \oplus 2\repQ_4 \oplus 2\repQ_5 \oplus 2\repQ_6 \oplus 2\repQ_7 \oplus 2\repQ_8 \oplus 2\repQ_9 \oplus 2\repQ_{10} \oplus \repQ_{11} \oplus \repQ_{12} \oplus \repQ_{13}
& (a,b)=(4,6)
\\[0.15cm] 
\repQ_1 \oplus \repQ_2 \oplus 2\repQ_3 \oplus 2\repQ_4 \oplus 2\repQ_5 \oplus 3\repQ_6 \oplus 2\repQ_7 \oplus 3\repQ_8 \oplus 2\repQ_9 \oplus 2\repQ_{10} \oplus 2\repQ_{11} \oplus \repQ_{12} \oplus \repQ_{13}
& (a,b)=(4,8)
\\[0.15cm] 
\repQ_1 \oplus \repQ_4 \oplus \repQ_5 \oplus \repQ_6 \oplus \repQ_8 \oplus \repQ_9 \oplus \repQ_{10} \oplus \repQ_{13}
& (a,b)=(5,7)
\\[0.15cm] 
\repQ_2 \oplus \repQ_3 \oplus \repQ_4 \oplus \repQ_5 \oplus \repQ_6 \oplus 2\repQ_7 \oplus \repQ_8 \oplus \repQ_9 \oplus \repQ_{10} \oplus \repQ_{11} \oplus \repQ_{12} 
& (a,b)=(6,8)
\\[0.15cm] 
\repQ_{\frac12} \oplus \repQ_{\frac52} \oplus \repQ_{\frac72} \oplus \repQ_{\frac92} \oplus 2\repQ_{\frac{11}2} \oplus \repQ_{\frac{13}2} \oplus \repQ_{\frac{15}2} \oplus 2\repQ_{\frac{17}2} \oplus \repQ_{\frac{19}2} \oplus \repQ_{\frac{21}2} \oplus \repQ_{\frac{23}2} \oplus \repQ_{\frac{27}2} 
& (a,b)=(1,2)
\\[0.15cm] 
\repQ_{\frac32} \oplus \repQ_{\frac52} \oplus \repQ_{\frac72} \oplus 2\repQ_{\frac92} \oplus \repQ_{\frac{11}2} \oplus 2\repQ_{\frac{13}2} \oplus 2\repQ_{\frac{15}2} \oplus \repQ_{\frac{17}2} \oplus 2\repQ_{\frac{19}2} \oplus \repQ_{\frac{21}2} \oplus \repQ_{\frac{23}2} \oplus \repQ_{\frac{25}2}
& (a,b)=(1,4)
\\[0.15cm] 
\repQ_{\frac52} \oplus \repQ_{\frac72} \oplus \repQ_{\frac{11}2} \oplus \repQ_{\frac{13}2} \oplus \repQ_{\frac{15}2} \oplus \repQ_{\frac{17}2} \oplus \repQ_{\frac{21}2} \oplus \repQ_{\frac{23}2}
& (a,b)=(1,6)
\\[0.15cm] 
\repQ_{\frac32} \oplus \repQ_{\frac72} \oplus \repQ_{\frac92} \oplus \repQ_{\frac{11}2} \oplus \repQ_{\frac{13}2} \oplus \repQ_{\frac{15}2} \oplus \repQ_{\frac{17}2} \oplus \repQ_{\frac{19}2} \oplus \repQ_{\frac{21}2} \oplus \repQ_{\frac{25}2}
& (a,b)=(1,8)
\\[0.15cm] 
\repQ_{\frac12} \oplus 2\repQ_{\frac32} \oplus 2\repQ_{\frac52} \oplus 3\repQ_{\frac72} \oplus 4\repQ_{\frac92} \oplus 4\repQ_{\frac{11}2} \oplus 4\repQ_{\frac{13}2} \oplus 4\repQ_{\frac{15}2} \oplus 4\repQ_{\frac{17}2} \oplus 4\repQ_{\frac{19}2} \oplus 3\repQ_{\frac{21}2} \oplus 2\repQ_{\frac{23}2} \oplus 2\repQ_{\frac{25}2} \oplus \repQ_{\frac{27}2} 
& (a,b)=(2,3)
\\[0.15cm] 
\repQ_{\frac32} \oplus 2\repQ_{\frac52} \oplus 2\repQ_{\frac72} \oplus 2\repQ_{\frac92} \oplus 2\repQ_{\frac{11}2} \oplus 3\repQ_{\frac{13}2} \oplus 3\repQ_{\frac{15}2} \oplus 2\repQ_{\frac{17}2} \oplus 2\repQ_{\frac{19}2} \oplus 2\repQ_{\frac{21}2} \oplus 2\repQ_{\frac{23}2} \oplus \repQ_{\frac{25}2}
& (a,b)=(2,5)
\\[0.15cm] 
\repQ_{\frac52} \oplus \repQ_{\frac72} \oplus \repQ_{\frac{11}2} \oplus \repQ_{\frac{13}2} \oplus \repQ_{\frac{15}2} \oplus \repQ_{\frac{17}2} \oplus \repQ_{\frac{21}2} \oplus \repQ_{\frac{23}2}
& (a,b)=(2,7)
\\[0.15cm] 
\repQ_{\frac12} \oplus 2\repQ_{\frac32} \oplus 3\repQ_{\frac52} \oplus 4\repQ_{\frac72} \oplus 4\repQ_{\frac92} \oplus 5\repQ_{\frac{11}2} \oplus 5\repQ_{\frac{13}2} \oplus 5\repQ_{\frac{15}2} \oplus 5\repQ_{\frac{17}2} \oplus 4\repQ_{\frac{19}2} \oplus 4\repQ_{\frac{21}2} \oplus 3\repQ_{\frac{23}2} \oplus 2\repQ_{\frac{25}2} \oplus \repQ_{\frac{27}2} 
& (a,b)=(3,4)
\\[0.15cm] 
\repQ_{\frac32} \oplus 2\repQ_{\frac52} \oplus 2\repQ_{\frac72} \oplus 2\repQ_{\frac92} \oplus 2\repQ_{\frac{11}2} \oplus 3\repQ_{\frac{13}2} \oplus 3\repQ_{\frac{15}2} \oplus 2\repQ_{\frac{17}2} \oplus 2\repQ_{\frac{19}2} \oplus 2\repQ_{\frac{21}2} \oplus 2\repQ_{\frac{23}2} \oplus \repQ_{\frac{25}2}
& (a,b)=(3,6)
\\[0.15cm] 
\repQ_{\frac12} \oplus \repQ_{\frac32} \oplus 2\repQ_{\frac52} \oplus 2\repQ_{\frac72} \oplus 3\repQ_{\frac92} \oplus 3\repQ_{\frac{11}2} \oplus 3\repQ_{\frac{13}2} \oplus 3\repQ_{\frac{15}2} \oplus 3\repQ_{\frac{17}2} \oplus 3\repQ_{\frac{19}2} \oplus 2\repQ_{\frac{21}2} \oplus 2\repQ_{\frac{23}2} \oplus \repQ_{\frac{25}2} \oplus \repQ_{\frac{27}2} 
& (a,b)=(3,8),(4,5)
\\[0.15cm] 
\repQ_{\frac32} \oplus \repQ_{\frac72} \oplus \repQ_{\frac92} \oplus \repQ_{\frac{11}2} \oplus \repQ_{\frac{13}2} \oplus \repQ_{\frac{15}2} \oplus \repQ_{\frac{17}2} \oplus \repQ_{\frac{19}2} \oplus \repQ_{\frac{21}2} \oplus \repQ_{\frac{25}2}
& (a,b)=(4,7)
\\[0.15cm] 
\repQ_{\frac12} \oplus \repQ_{\frac32}\oplus \repQ_{\frac72} \oplus 2\repQ_{\frac92} \oplus 2\repQ_{\frac{11}2} \oplus \repQ_{\frac{13}2} \oplus \repQ_{\frac{15}2} \oplus 2\repQ_{\frac{17}2} \oplus 2\repQ_{\frac{19}2} \oplus \repQ_{\frac{21}2} \oplus \repQ_{\frac{25}2} \oplus \repQ_{\frac{27}2} 
& (a,b)=(5,6)
\\[0.15cm] 
\repQ_{\frac32} \oplus \repQ_{\frac52} \oplus 2\repQ_{\frac72} \oplus \repQ_{\frac92} \oplus 2\repQ_{\frac{11}2} \oplus 2\repQ_{\frac{13}2} \oplus 2\repQ_{\frac{15}2} \oplus 2\repQ_{\frac{17}2} \oplus \repQ_{\frac{19}2} \oplus 2\repQ_{\frac{21}2} \oplus \repQ_{\frac{23}2} \oplus \repQ_{\frac{25}2}
& (a,b)=(5,8)
\\[0.15cm] 
\repQ_{\frac12} \oplus \repQ_{\frac92} \oplus \repQ_{\frac{11}2} \oplus \repQ_{\frac{17}2} \oplus \repQ_{\frac{19}2} \oplus \repQ_{\frac{27}2} 
& (a,b)=(6,7)
\\[0.15cm] 
\repQ_{\frac52} \oplus \repQ_{\frac92} \oplus \repQ_{\frac{13}2} \oplus \repQ_{\frac{15}2} \oplus \repQ_{\frac{19}2} \oplus \repQ_{\frac{23}2}
& (a,b)=(7,8)
\end{smallmatrix}\right.
\end{alignat}
\end{subequations}
\endgroup

\subsection{Untwisted conformal partition functions}\label{eq:ZTL.E678}

The conformal partition functions for $E_6$, $E_7$ and $E_8$ with $K=\id$ given in \eqref{eq:ZTL.untwisted} read
\begingroup
\allowdisplaybreaks
\begin{subequations}
\begin{alignat}{2}
Z_{E_6,\mu,a,b,\id}(\qq) & = \left\{\begin{array}{cl}
\chit_{1,1} + \chit_{1,7}
& (a,b)=(1,1), (5,5)
\\[0.15cm] 
\chit_{1,1} + \chit_{1,3} + \chit_{1,5} + 2\chit_{1,7} + \chit_{1,9}
& (a,b)=(2,2), (4,4)
\\[0.15cm] 
\chit_{1,1} + 2\chit_{1,3} + 3\chit_{1,5} + 3\chit_{1,7} + 2\chit_{1,9} + \chit_{1,11}
& (a,b)=(3,3)
\\[0.15cm] 
\chit_{1,1} + \chit_{1,5} + \chit_{1,7} + \chit_{1,11}
& (a,b)=(6,6)
\\[0.15cm] 
\chit_{1,1} + \chit_{1,5} + \chit_{1,7} + \chit_{1,9}
& (a,b)=(1,3),(3,5)
\\[0.15cm] 
\chit_{1,5} + \chit_{1,11}
& (a,b)=(1,5)
\\[0.15cm] 
\chit_{1,3} + 2\chit_{1,5} + \chit_{1,7} + \chit_{1,9} + \chit_{1,11}
& (a,b)=(2,4)
\\[0.15cm] 
\chit_{1,3} + \chit_{1,5} + \chit_{1,7} + \chit_{1,9}
& (a,b)=(2,6)
\\[0.15cm] 
\chit_{1,2} + \chit_{1,6} + \chit_{1,8}
& (a,b)=(1,2),(4,5)
\\[0.15cm] 
\chit_{1,4} + \chit_{1,6} + \chit_{1,10}
& (a,b)=(1,4),(2,5)
\\[0.15cm] 
\chit_{1,4} + \chit_{1,8}
& (a,b)=(1,6),(5,6)
\\[0.15cm] 
\chit_{1,2} + 2\chit_{1,4} + 2\chit_{1,6} + 2\chit_{1,8} + \chit_{1,10}
& (a,b)=(2,3),(3,4)
\\[0.15cm] 
\chit_{1,2} + \chit_{1,4} + 2\chit_{1,\repK_{6}} + \chit_{1,8} + \chit_{1,10}
& (a,b)=(3,6)
\end{array}\right.
\end{alignat}
\end{subequations}

\be 
Z_{E_7,\mu,a,b,\id} =
\left\{\begin{array}{cl}
\chit_{1,1} + \chit_{1,7} + \chit_{1,11} + \chit_{1,17}
& (a,b)=(1,1)
\\[0.15cm] 
\chit_{1,1} + \chit_{1,3} + \chit_{1,5} + 2\chit_{1,7} + 2\chit_{1,9} + 2\chit_{1,11} + \chit_{1,13} + \chit_{1,15} + \chit_{1,17} 
& (a,b)=(2,2)
\\[0.15cm] 
\chit_{1,1} + 2\chit_{1,3} + 3\chit_{1,5} + 4\chit_{1,7} + 4\chit_{1,9} + 4\chit_{1,11} + 3\chit_{1,13} + 2\chit_{1,15} + \chit_{1,17}
& (a,b)=(3,3)
\\[0.15cm] 
\chit_{1,1} + \chit_{1,3} + 2\chit_{1,5} + 2\chit_{1,7} + 3\chit_{1,9} + 2\chit_{1,11} + 2\chit_{1,13} + \chit_{1,15} + \chit_{1,17} 
& (a,b)=(4,4)
\\[0.15cm] 
\chit_{1,1} + \chit_{1,3} + \chit_{1,7} + 2\chit_{1,9} + \chit_{1,11} + \chit_{1,15} + \chit_{1,17} 
& (a,b)=(5,5)
\\[0.15cm] 
\chit_{1,1} + \chit_{1,9} + \chit_{1,17}
& (a,b)=(6,6)
\\[0.15cm] 
\chit_{1,1} + \chit_{1,5} + \chit_{1,7} + \chit_{1,9} + \chit_{1,11} + \chit_{1,13} + \chit_{1,17} 
& (a,b)=(7,7)
\\[0.15cm] 
\chit_{1,3} + \chit_{1,5} + \chit_{1,7} + 2\chit_{1,9} + \chit_{1,11} + \chit_{1,13} + \chit_{1,15}
& (a,b)=(1,3)
\\[0.15cm] 
\chit_{1,5} + \chit_{1,7} + \chit_{1,11} + \chit_{1,13}
& (a,b)=(1,5)
\\[0.15cm] 
\chit_{1,3} + 2\chit_{1,5} + 2\chit_{1,7} + 2\chit_{1,9} + 2\chit_{1,11} + 2\chit_{1,13} + \chit_{1,15}
& (a,b)=(2,4)
\\[0.15cm] 
\chit_{1,5} + \chit_{1,7} + \chit_{1,11} + \chit_{1,13}
& (a,b)=(2,6)
\\[0.15cm] 
\chit_{1,3} + \chit_{1,5} + \chit_{1,7} + 2\chit_{1,9} + \chit_{1,11} + \chit_{1,13} + \chit_{1,15}
& (a,b)=(2,7)
\\[0.15cm] 
\chit_{1,3} + 2\chit_{1,5} + 2\chit_{1,7} + 2\chit_{1,9} + 2\chit_{1,11} + 2\chit_{1,13} + 2\chit_{1,15}
& (a,b)=(3,5)
\\[0.15cm] 
\chit_{1,3} + \chit_{1,7} + \chit_{1,9} + \chit_{1,11} + \chit_{1,15}
& (a,b)=(4,6)
\\[0.15cm] 
\chit_{1,3} + \chit_{1,5} + 2\chit_{1,7} + \chit_{1,9} + 2\chit_{1,11} + \chit_{1,13} + \chit_{1,15}
& (a,b)=(4,7)
\\[0.15cm] 
\chit_{1,5} + \chit_{1,9} + \chit_{1,13}
& (a,b)=(6,7)
\\[0.15cm] 
\chit_{1,2} + \chit_{1,6} + \chit_{1,8} + \chit_{1,10} + \chit_{1,12} + \chit_{1,16}
& (a,b)=(1,2)
\\[0.15cm] 
\chit_{1,4} + \chit_{1,6} + \chit_{1,8} + \chit_{1,10} + \chit_{1,12} + \chit_{1,14}
& (a,b)=(1,4)
\\[0.15cm] 
\chit_{1,6} + \chit_{1,12}
& (a,b)=(1,6)
\\[0.15cm] 
\chit_{1,4} + \chit_{1,8} + \chit_{1,10} + \chit_{1,14}
& (a,b)=(1,7)
\\[0.15cm] 
\chit_{1,2} + 2\chit_{1,4} + 2\chit_{1,6} + 3\chit_{1,8} + 3\chit_{1,10} + 2\chit_{1,12} + 2\chit_{1,14} + \chit_{1,16} 
& (a,b)=(2,3)
\\[0.15cm] 
\chit_{1,4} + 2\chit_{1,6} + \chit_{1,8} + \chit_{1,10} + 2\chit_{1,12} + \chit_{1,14}
& (a,b)=(2,5)
\\[0.15cm] 
\chit_{1,2} + 2\chit_{1,4} + 3\chit_{1,6} + 3\chit_{1,8} + 3\chit_{1,10} + 3\chit_{1,12} + 2\chit_{1,14} + \chit_{1,16} 
& (a,b)=(3,4)
\\[0.15cm] 
\chit_{1,4} + \chit_{1,6} + \chit_{1,8} + \chit_{1,10} + \chit_{1,12} + \chit_{1,14} 
& (a,b)=(3,6)
\\[0.15cm] 
\chit_{1,2} + \chit_{1,4} + 2\chit_{1,6} + 2\chit_{1,8} + 2\chit_{1,10} + 2\chit_{1,12} + \chit_{1,14} + \chit_{1,16} 
& (a,b)=(3,7)
\\[0.15cm] 
\chit_{1,2} + \chit_{1,4} + \chit_{1,6} + 2\chit_{1,8} + 2\chit_{1,10} + \chit_{1,12} + \chit_{1,14} + \chit_{1,16}
& (a,b)=(4,5)
\\[0.15cm] 
\chit_{1,2} + \chit_{1,8} + \chit_{1,10} + \chit_{1,16}
& (a,b)=(5,6)
\\[0.15cm] 
\chit_{1,4} + \chit_{1,6} + \chit_{1,8} + \chit_{1,10} + \chit_{1,12} + \chit_{1,14} 
& (a,b)=(5,7)
\end{array}\right.
\ee
\be
Z_{E_8,\mu,a,b,\id} =
\left\{\begin{smallmatrix}
\chit_{1,1} + \chit_{1,7} + \chit_{1,11} + \chit_{1,13} + \chit_{1,17} + \chit_{1,19} + \chit_{1,23} + \chit_{1,29}
& (a,b)=(1,1)
\\[0.15cm] 
\chit_{1,1} + \chit_{1,3} + \chit_{1,5} + 2\chit_{1,7} + 2\chit_{1,9} + 3\chit_{1,11} + 3\chit_{1,13} + 2\chit_{1,15} 
\\[0.05cm]+ 3\chit_{1,17}+ 3\chit_{1,19} + 2\chit_{1,21} + 2\chit_{1,23} + \chit_{1,25} + \chit_{1,27} + \chit_{1,29}
& (a,b)=(2,2)
\\[0.15cm] 
\chit_{1,1} + 2\chit_{1,3} + 3\chit_{1,5} + 4\chit_{1,7} + 5\chit_{1,9} + 6\chit_{1,11} + 6\chit_{1,13} 
\\[0.05cm] + 6\chit_{1,15} + 6\chit_{1,17} + 6\chit_{1,19} + 5\chit_{1,21} + 4\chit_{1,23} + 3\chit_{1,25} + 2\chit_{1,27} + \chit_{1,29}
& (a,b)=(3,3)
\\[0.15cm] 
\chit_{1,1} + \chit_{1,3} + 2\chit_{1,5} + 3\chit_{1,7} + 3\chit_{1,9} + 4\chit_{1,11} + 4\chit_{1,13} + 4\chit_{1,15} +
\\[0.05cm] 4\chit_{1,17} + 4\chit_{1,19} + 3\chit_{1,21} + 3\chit_{1,23} + 2\chit_{1,25} + \chit_{1,27} + \chit_{1,29}
& (a,b)=(4,4)
\\[0.15cm] 
\chit_{1,1} + \chit_{1,3} + \chit_{1,5} + \chit_{1,7} + 2\chit_{1,9} + 3\chit_{1,11} + 2\chit_{1,13} 
\\[0.05cm]
+ 2\chit_{1,15} + 2\chit_{1,17} + 3\chit_{1,19} + 2\chit_{1,21} + \chit_{1,23} + \chit_{1,25} + \chit_{1,27} + \chit_{1,29}
& (a,b)=(5,5)
\\[0.15cm] 
\chit_{1,1} + \chit_{1,3} + \chit_{1,9} + 2\chit_{1,11} + \chit_{1,13} + \chit_{1,17} + 2\chit_{1,19} + \chit_{1,21} + \chit_{1,27} + \chit_{1,29}
& (a,b)=(6,6)
\\[0.15cm] 
\chit_{1,1} + \chit_{1,11} + \chit_{1,19} + \chit_{1,29}
& (a,b)=(7,7)
\\[0.15cm] 
\chit_{1,1} + \chit_{1,5} + \chit_{1,7} + \chit_{1,9} + 2\chit_{1,11} + \chit_{1,13} + 2\chit_{1,15} 
\\[0.15cm]
+ \chit_{1,17} + 2\chit_{1,19} + \chit_{1,21} + \chit_{1,23} + \chit_{1,25} + Q_{14}
& (a,b)=(8,8)
\\[0.15cm] 
\chit_{1,3} + \chit_{1,5} + \chit_{1,7} + 2\chit_{1,9} + 2\chit_{1,11} + 2\chit_{1,13} + 2\chit_{1,15} 
\\[0.15cm]
+ 2\chit_{1,17} + 2\chit_{1,19} + 2\chit_{1,21} + \chit_{1,23} + \chit_{1,25} + \chit_{1,27}
& (a,b)=(1,3)
\\[0.15cm] 
\chit_{1,5} + \chit_{1,7} + \chit_{1,9} + \chit_{1,11} + \chit_{1,13} + 2\chit_{1,15} + \chit_{1,17} + \chit_{1,19} + \chit_{1,21} + \chit_{1,23} + \chit_{1,25} 
& (a,b)=(1,5)
\\[0.15cm] 
\chit_{1,7} + \chit_{1,13} + \chit_{1,17} + \chit_{1,23}
& (a,b)=(1,7)
\\[0.15cm] 
\chit_{1,3} + 2\chit_{1,5} + 2\chit_{1,7} + 3\chit_{1,9} + 3\chit_{1,11} + 3\chit_{1,13} + 4\chit_{1,15} 
\\[0.15cm]
+ 3\chit_{1,17} + 3\chit_{1,19} + 3\chit_{1,21} + 2\chit_{1,23} + 2\chit_{1,25} + \chit_{1,27}
& (a,b)=(2,4)
\\[0.15cm] 
\chit_{1,5} + 2\chit_{1,7} + \chit_{1,9} + \chit_{1,11} + 2\chit_{1,13} + 2\chit_{1,15} + 2\chit_{1,17} + \chit_{1,19} + \chit_{1,21} 2+ \chit_{1,23} + \chit_{1,25}
& (a,b)=(2,6)
\\[0.15cm] 
\chit_{1,3} + \chit_{1,5} + \chit_{1,7} + 2\chit_{1,9} + 2\chit_{1,11} + 2\chit_{1,13} + 2\chit_{1,15} 
\\[0.05cm]
+ 2\chit_{1,17} + 2\chit_{1,19} + 2\chit_{1,21} + \chit_{1,23} + \chit_{1,25} + \chit_{1,27}
& (a,b)=(2,8)
\\[0.15cm] 
\chit_{1,3} + 2\chit_{1,5} + 3\chit_{1,7} + 3\chit_{1,9} + 3\chit_{1,11} + 4\chit_{1,13} + 4\chit_{1,15} 
\\[0.05cm]+ 4\chit_{1,17} + 3\chit_{1,19} + 3\chit_{1,21} + 3\chit_{1,23} + 2\chit_{1,25} + \chit_{1,27}
& (a,b)=(3,5)
\\[0.15cm] 
\chit_{1,5} + \chit_{1,7} + \chit_{1,9} + \chit_{1,11} + \chit_{1,13} + 2\chit_{1,15} + \chit_{1,17} + \chit_{1,19} + \chit_{1,21} + \chit_{1,23} + \chit_{1,25} 
& (a,b)=(3,7)
\\[0.15cm] 
\chit_{1,3} + \chit_{1,5} + \chit_{1,7} + 2\chit_{1,9} + 2\chit_{1,11} + 2\chit_{1,13} + 2\chit_{1,15} 
\\[0.05cm]+ 2\chit_{1,17} + 2\chit_{1,19} + 2\chit_{1,21} + \chit_{1,23} + \chit_{1,25} + \chit_{1,27}
& (a,b)=(4,6)
\\[0.15cm] 
\chit_{1,3} + \chit_{1,5} + 2\chit_{1,7} + 2\chit_{1,9} + 2\chit_{1,11} + 3\chit_{1,13} + 2\chit_{1,15} 
\\[0.05cm]+ 3\chit_{1,17} + 2\chit_{1,19} + 2\chit_{1,21} + 2\chit_{1,23} + \chit_{1,25} + \chit_{1,27}
& (a,b)=(4,8)
\\[0.15cm] 
\chit_{1,3} + \chit_{1,9} + \chit_{1,11} + \chit_{1,13} + \chit_{1,17} + \chit_{1,19} + \chit_{1,21} + \chit_{1,27}
& (a,b)=(5,7)
\\[0.15cm] 
\chit_{1,5} + \chit_{1,7} + \chit_{1,9} + \chit_{1,11} + \chit_{1,13} + 2\chit_{1,15} + \chit_{1,17} + \chit_{1,19} + \chit_{1,21} + \chit_{1,23} + \chit_{1,25} 
& (a,b)=(6,8)
\end{smallmatrix}\right.
\ee
for $N$ even, and
\be
Z_{E_8,\mu,a,b,\id} = 
\left\{\begin{smallmatrix}
\chit_{1,2} + \chit_{1,6} + \chit_{1,8} + \chit_{1,10} + 2\chit_{1,12} + \chit_{1,14} 
\\[0.05cm]
+ \chit_{1,16} + 2\chit_{1,18} + \chit_{1,20} + \chit_{1,22} + \chit_{1,24} + \chit_{1,28} 
& (a,b)=(1,2)
\\[0.15cm] 
\chit_{1,4} + \chit_{1,6} + \chit_{1,8} + 2\chit_{1,10} + \chit_{1,12} + 2\chit_{1,14} 
\\[0.05cm]
+ 2\chit_{1,16} + \chit_{1,18} + 2\chit_{1,20} + \chit_{1,22} + \chit_{1,24} + \chit_{1,26}
& (a,b)=(1,4)
\\[0.15cm] 
\chit_{1,6} + \chit_{1,8} + \chit_{1,12} + \chit_{1,14} + \chit_{1,16} + \chit_{1,18} + \chit_{1,22} + \chit_{1,24}
& (a,b)=(1,6)
\\[0.15cm] 
\chit_{1,4} + \chit_{1,8} + \chit_{1,10} + \chit_{1,12} + \chit_{1,14} + \chit_{1,16} + \chit_{1,18} + \chit_{1,20} + \chit_{1,22} + \chit_{1,26}
& (a,b)=(1,8)
\\[0.15cm] 
\chit_{1,2} + 2\chit_{1,4} + 2\chit_{1,6} + 3\chit_{1,8} + 4\chit_{1,10} + 4\chit_{1,12} + 4\chit_{1,14} 
\\[0.05cm]
+ 4\chit_{1,16} + 4\chit_{1,18} + 4\chit_{1,20} + 3\chit_{1,22} + 2\chit_{1,24} + 2\chit_{1,26} + \chit_{1,28} 
& (a,b)=(2,3)
\\[0.15cm] 
\chit_{1,4} + 2\chit_{1,6} + 2\chit_{1,8} + 2\chit_{1,10} + 2\chit_{1,12} + 3\chit_{1,14} 
\\[0.05cm]
+ 3\chit_{1,16} + 2\chit_{1,18} + 2\chit_{1,20} + 2\chit_{1,22} + 2\chit_{1,24} + \chit_{1,26}
& (a,b)=(2,5)
\\[0.15cm] 
\chit_{1,6} + \chit_{1,8} + \chit_{1,12} + \chit_{1,14} + \chit_{1,16} + \chit_{1,18} + \chit_{1,22} + \chit_{1,24}
& (a,b)=(2,7)
\\[0.15cm] 
\chit_{1,2} + 2\chit_{1,4} + 3\chit_{1,6} + 4\chit_{1,8} + 4\chit_{1,10} + 5\chit_{1,12} 
\\[0.05cm]
+ 5\chit_{1,14} + 5\chit_{1,16} + 5\chit_{1,18} + 4\chit_{1,20} + 4\chit_{1,22} + 3\chit_{1,24} + 2\chit_{1,26} + \chit_{1,28} 
& (a,b)=(3,4)
\\[0.15cm] 
\chit_{1,4} + 2\chit_{1,6} + 2\chit_{1,8} + 2\chit_{1,10} + 2\chit_{1,12} + 3\chit_{1,14} 
\\[0.05cm]
+ 3\chit_{1,16} + 2\chit_{1,18} + 2\chit_{1,20} + 2\chit_{1,22} + 2\chit_{1,24} + \chit_{1,26}
& (a,b)=(3,6)
\\[0.15cm] 
\chit_{1,2} + \chit_{1,4} + 2\chit_{1,6} + 2\chit_{1,8} + 3\chit_{1,10} + 3\chit_{1,12} + 3\chit_{1,14} 
\\[0.05cm]
+ 3\chit_{1,16} + 3\chit_{1,18} + 3\chit_{1,20} + 2\chit_{1,22} + 2\chit_{1,24} + \chit_{1,26} + \chit_{1,28} 
& (a,b)=(3,8),(4,5)
\\[0.15cm] 
\chit_{1,4} + \chit_{1,8} + \chit_{1,10} + \chit_{1,12} + \chit_{1,14} + \chit_{1,16} + \chit_{1,18} + \chit_{1,20} + \chit_{1,22} + \chit_{1,26}
& (a,b)=(4,7)
\\[0.15cm] 
\chit_{1,2} + \chit_{1,4}+ \chit_{1,8} + 2\chit_{1,10} + 2\chit_{1,12} + \chit_{1,14} 
\\[0.05cm]
+ \chit_{1,16} + 2\chit_{1,18} + 2\chit_{1,20} + \chit_{1,22} + \chit_{1,26} + \chit_{1,28} 
& (a,b)=(5,6)
\\[0.15cm] 
\chit_{1,4} + \chit_{1,6} + 2\chit_{1,8} + \chit_{1,10} + 2\chit_{1,12} + 2\chit_{1,14} 
\\[0.05cm]
+ 2\chit_{1,16} + 2\chit_{1,18} + \chit_{1,20} + 2\chit_{1,22} + \chit_{1,24} + \chit_{1,26}
& (a,b)=(5,8)
\\[0.15cm] 
\chit_{1,2} + \chit_{1,10} + \chit_{1,12} + \chit_{1,18} + \chit_{1,20} + \chit_{1,28} 
& (a,b)=(6,7)
\\[0.15cm] 
\chit_{1,6} + \chit_{1,10} + \chit_{1,14} + \chit_{1,16} + \chit_{1,20}+ \chit_{1,24}
& (a,b)=(7,8)
\end{smallmatrix}\right.
\ee
for $N$ odd.
\endgroup


\end{document}